\pdfoutput=1

\documentclass[a4paper,11pt,twoside]{book}

\usepackage[a4paper,top=3cm,bottom=3cm,left=2.67cm,right=2cm]{geometry}

\usepackage[nottoc]{tocbibind}

\usepackage[T2A]{fontenc} 
\usepackage[cp1251]{inputenc} 
\usepackage[english]{babel} 

\usepackage{amssymb}
\usepackage{amsfonts}
\usepackage{amsmath}
\usepackage{mathtext}
\usepackage{amsthm}
\usepackage{bbm}
\usepackage{mathtools}
\usepackage{mathrsfs}
\usepackage{empheq}  

\usepackage{booktabs} 
\usepackage{epigraph}

\usepackage{fancyhdr}
\newenvironment{abstract} 
{\clearpage\null \vfill \begin{center}%
\bfseries \abstractname \end{center}}%
{\vfill \null}

\usepackage[tocindentauto,tocgraduated]{tocstyle}  
\usetocstyle{standard}  

\usepackage{stmaryrd} 

\usepackage[protrusion=true,expansion=true]{microtype} 

\usepackage{soul} 
\usepackage[onehalfspacing]{setspace} 

\usepackage{cite}  
\usepackage{enumerate}  
\usepackage{enumitem} 
\usepackage{tabularx}  
\usepackage{cases} 
\usepackage{esvect} 
\usepackage{tikz-cd} 
\usepackage{float} 
\usepackage[autostyle]{csquotes}

\usepackage{graphicx, color, calc} 
\usepackage[list=true, font=large, labelfont=bf, labelformat=brace, position=top]{subcaption}
\usepackage{dsfont}
\usepackage{wrapfig}
\usepackage[hyphens]{url}
\usepackage[linktocpage=true]{hyperref}
\hypersetup{breaklinks=true}
\DeclareFontFamily{OT1}{pzc}{}
\DeclareFontShape{OT1}{pzc}{m}{it}{<-> s * [1.10] pzcmi7t}{}
\DeclareMathAlphabet{\mathpzc}{OT1}{pzc}{m}{it}

\newcommand*{\mb}{\mathbf} 
\newcommand*{\mr}{\mathrm} 
\newcommand*{\mc}{\mathcal} 
\newcommand*{\ms}{\mathscr} 
\newcommand*{\mf}{\mathfrak} 
\newcommand*{\ph}{\phantom} 

\newcommand*{\al}{\ensuremath{\alpha}}
\newcommand*{\be}{\ensuremath{\beta}}
\newcommand*{\ga}{\ensuremath{\gamma}}
\newcommand*{\Ga}{\ensuremath{\Gamma}}
\newcommand*{\de}{\ensuremath{\delta}}
\newcommand*{\De}{\ensuremath{\Delta}}
\newcommand*{\ka}{\ensuremath{\kappa}}
\newcommand*{\eps}{\ensuremath{\epsilon}}

\newcommand*{\la}{\ensuremath{\lambda}}
\newcommand*{\La}{\ensuremath{\Lambda}}
\newcommand*{\si}{\ensuremath{\sigma}}
\newcommand*{\Si}{\ensuremath{\Sigma}}

\newcommand*{\om}{\ensuremath{\omega}}
\newcommand*{\Om}{\ensuremath{\Omega}}

\newcommand{\cg}{\textnormal{\textsl{g}}} 

\newcommand*{\pa}{\partial}
\newcommand*{\bra}{\langle}
\newcommand*{\ket}{\rangle}
\newcommand*{\ra}{\rightarrow}
\newcommand*{\Ra}{\Rightarrow}
\newcommand*{\rt}{\triangleright}
\newcommand*{\tr}{\mathrm{tr}}
\newcommand*{\Tr}{\mathrm{Tr}}

\newtheorem{definition}{Definition}[chapter]
\newtheorem{theorem}{Theorem}[chapter]
\newtheorem*{lemma}{Lemma}

\newtheorem{proposition}{Proposition}[chapter]

\begin{document}


\frontmatter

		
\begin{titlepage}
	\addtolength{\hoffset}{0.5\evensidemargin-0.5\oddsidemargin} 
	\noindent%
	\begin{center}
	\end{center}
	\noindent%
	\begin{tabular}{@{}p{\textwidth}@{}}
		\toprule[2pt]
		\midrule
		\vspace{0.2cm}
		\begin{center}
			\Large{\textbf{On Geometry and Symmetries in Classical and Quantum Theories of Gauge Gravity}}
		\end{center}
		\vspace{0.2cm}\\
		\midrule
		\toprule[2pt]
	\end{tabular}
	\vspace{1.8 cm}
  {\small
  	\begin{center}
  		{\Large \textbf{Dissertation}}
  		
  		\vspace{1.2cm}
  		\large
  		
  		zur Erlangung des Doktorgrades
  		
  		
  		\vspace{0.2cm}
  		an der Fakult{\"a}t f{\"u}r Mathematik, 

        \vspace{0.2cm}
        Informatik und Naturwissenschaften
  		
  		\vspace{0.2cm}
  		Fachbereich Physik 
  		
  		
  		\vspace{0.2cm}
  		der Universit{\"a}t Hamburg

  		\vspace{1.6cm}

  		vorgelegt von
  		
  		\vspace{0.5cm}

  		{\Large Vadim Belov}
  		
  		\vspace{0.5cm}
  		   
  		{\large
  		      aus Sankt Petersburg
  		    }
  	\end{center}
  }
  \vspace{1.5cm}
  {\large \begin{center}Hamburg\\2019\end{center} }
\end{titlepage}
\clearpage


\newpage
\thispagestyle{empty}

\vspace*{\fill}


\noindent 
\begin{tabular}{@{}ll}
	Gutachter der Dissertation: 					& Dr.~Benjamin~Bahr \\
	\hspace{10cm}									& Prof.~Dr.~Gleb~Arutyunov \\
	\\
	Zusammensetzung der Pr\"ufungskommission: 		
													& Dr.~Benjamin~Bahr \\
													& Prof.~Dr.~Gleb~Arutyunov \\
													& Prof.~Dr.~Sven-Olaf~Moch \\
													& Prof.~Dr.~J\"{o}rg Teschner \\
													& Prof.~Dr.~Dieter~Horns \\
	\\
	Vorsitzender der Pr\"ufungskommission:			& Prof.~Dr.~Dieter~Horns \\
	\\
	Datum der Disputation:							& 10.04.2019 \\
	\\
	Vorsitzender Fach-Promotionsausschusses PHYSIK: & Prof.~Dr.~Michael~Potthoff \\
	\\
	Leiter des Fachbereichs PHYSIK:					& Prof.~Dr.~Wolfgang~Hansen \\
	\\
	Dekan der Fakult\"at MIN:							& Prof.~Dr.~Heinrich~Graener
\end{tabular}


\large
\newpage
\thispagestyle{empty}

\newpage
\begin{abstract}
Spin Foam and Loop approaches to Quantum Gravity reformulate Einstein's theory of relativity in terms of connection variables. The metric properties are encoded in face bivectors/conjugate fluxes that are required to satisfy certain conditions, in order to allow for their geometric interpretation. We show that the (sub-)set of the so-called `volume simplicity constraints' is not implemented properly in the current EPRL-FK spinfoam vertex amplitude, if extended to arbitrary polyhedra. 

We then propose that a certain knot-invariant of the bivector geometry, induced on the boundary graph, encodes the missing conditions, allowing for reconstruction of a polytope from its two-dimensional faces. Implemented in the quantum amplitude, this leads to corrected semi-classical asymptotics for a hypercuboid, and is conjectured to be non-trivial in more general situations.

The analysis of linear version of `volume simplicity' suggests to switch from hypersurface normals to edge lengths, that is from 3-forms directly to tetrads -- in the extended configuration space of the Plebanski constrained formulation. We then give the corresponding dual version of linear simplicity constraints, which prescribe 3d volume for the polyhedral faces in the boundary of a 4d polytope.

We also analyse the status of metric/vielbein degrees of freedom and the role of local translations in the classical Einstein-Cartan gravity, viewed as a Poincare gauge theory. The relation with the diffeomorphism symmetry is established through the key concept of development, which generalizes parallel transport of vectors in the geometric theory of Cartan connections. We advocate the latter to be the natural gauge-theoretic framework for the theory of relativity.\end{abstract}

\renewcommand{\abstractname}{Zusammenfassung}

\begin{abstract}
Spinschaum-Modelle und Schleifenquantengravitation reformulieren Einsteins Allgemeine Relativit\"{a}tstheorie mit Hilfe von Zusammenhangsvariablen. Die metrischen Eigenschaften sind in den konjugierten Fl\"{u}ssen kodiert. Diese m\"{u}ssen gewisse Bedingungen erf\"{u}llen, um eine geometrische Interpretation der Variablen zuzulassen. Wir zeigen, dass eine dieser Bedingungen, die ``Volumenzwangsbedingung'', nicht richtig im EPRL-FK Modell implementiert ist, sobald sie auf generelle Polyeder verallgemeinert wird. 

Wir schlagen dann eine gewisse Knoten-Invariante der Bivektorgeometrie im Randgraphen vor, die die fehlende Bedingung enthatlen k\"{o}nnte. Diese erlaubt eine Rekonstruktion des Polyeders aus dessen zweidimensionalen Fl\"{a}chen. Implementiert in der Quantenamplitude, hat diese dann die richtige semiklassische Asymptotik f\"{u}r den Hyperquader. Es wird vermutet, dass sie f\"{u}r allgemeinere Polyeder ebenfalls anwendbar ist. 

Die Analyse der linearen Volumenzwangsbedingung legt nahe, von den Hyperfl\"{a}chenormalen zu den Kantenl\"{a}ngen \"{u}berzugehen, also von 3-Formen direkt zu den Dreibeinen im erweiterten Konfigurationsraum des Plebanskiformalismus. Wir entwickeln die duale Version der linearen Volumenzwangsbedingung, welche 3d Volumina f\"{u}r die polyhedralen Fl\"{a}chen im Rand der 4d Polytope vorgibt.

Ebenfalls analysieren wir den Status der metrischen/Vielbein-Freiheitsgrade, sowie die Rolle der lokalen Translationen in klassischer Einstein-Cartan-Gravitation, formuliert als Poincare-Eichtheorie. Wir stellen diese als den nat\"{u}rlichen eichtheoretischen Rahmen f\"{u}r Relativit\"{a}tstheorie dar.
\end{abstract}

\thispagestyle{empty}
\selectlanguage{english}

\newpage
\thispagestyle{empty}

\begin{flushright}
	\vspace*{3cm}
	\emph{The Book of Nature is written in the language of mathematics, and the symbols are triangles, circles and other geometrical figures, without whose help it is impossible to comprehend a single word of it; without which one wanders in vain through a dark labyrinth. [G. Galilei]}
\end{flushright}

\begin{flushright}
	\vspace*{3cm}
	\emph{I attach special importance to the view of geometry which I have just set forth, because without it I should have been unable to formulate the theory of relativity. [A. Einstein]}
\end{flushright}

\begin{flushright}
	\vspace*{3cm}
	\emph{I think I can safely say that nobody understands quantum mechanics. [R. Feynman]}
\end{flushright}


\newpage
\thispagestyle{empty}
\mbox{}

\thispagestyle{empty}
\clearpage
\thispagestyle{empty}

\cleardoublepage
\pagenumbering{gobble}
\setcounter{tocdepth}{2}
\tableofcontents
\cleardoublepage

\pagenumbering{arabic}

\mainmatter


\chapter{Introduction}


The two cornerstones of fundamental physics are the Quantum theory and General Relativity, they provide some of the deepest insights into the nature of observable universe. The first one is primarily concerned with the elementary constituents of matter and their behaviour at the smallest scales, resulting in the Standard Model of particle physics and the description of three fundamental interactions (electroweak and strong nuclear forces) in terms of Quantum Field Theory (QFT). Whereas the second one presents the classical theory of gravitation, prescribing the dynamics of the large collection of masses and extending its effects across the whole universe, studied in cosmology. It endows the space-time itself with dynamical properties, usually formulated in the language of (pseudo-)Riemannian geometry.


The two disciplines thus rarely overlap in their usual domains of applicability, where the effects of one theory are negligibly small as compared to the other. It is only in the exotic physical conditions of high energy densities of matter concentrated in the small region -- such as in the center of black holes, or at the big bang inception of expanding universe -- where both theories should work together, describing the quantum properties of a strongly curved space-time geometry. Up until now, consistent theory of Quantum Gravity -- that would combine the essential features of both types of description of dynamics -- has not been fully developed. A number of promising candidates have been put forward, however, each emphasizing various aspects of the problem~\footnote{Such as string theory, loop and foam approaches, dynamical triangulations and quantum Regge calculus, causal sets, asymptotic safety hypothesis, enthropic and modified gravity -- just to name a few.}. It seems clear that the usual methods of the perturbative QFT do not work properly, leading to intractable singularities, and one must explore other directions.

In some of the modern non-perturbative approaches to the quantization of gravity, the role of Riemannian metrics is very much reduced, in favour of mathematically more tractable connection variables. For instance, in the canonical Loop Quantum Gravity (LQG)~\cite{Ashtekar1991Non-pertubativeLQG,Rovelli2004QG,Thiemann2007ModCanQuantGR}, the densitiezed triad of the hypersurface plays the part of momenta variables, conjugate to the configuration of the certain Ashtekar-Barbero connection. The emerging picture of `quantum space' may be read in terms of a certain `twisted geometries'~\cite{FreidelSpeziale2010Twisted-geometries,FreidelSpeziale2010Twistors-to-twisted-geometries,RovelliSpeziale2010LQG-on-a-graph,BianchiDonaSpeziale2011Polyhedra}, consisting of the discrete polyhedra of `fuzzy' shapes, glued non-trivially alond their two-dimensional faces. In the closely related Spin Foam (SF) state-sum/path-integral quantization~\cite{Baez2000SF-BF,Perez2013SF-review,RovelliVidotto2014CovariantLQG}, the analogous general bivectors of the 2-dimensional surfaces are first integrated out completely, in order to obtain the vacuum state of the so-called BF Topological Quantum Field Theory (TQFT), which can be rigorously quantized. (Part of) the metric information is then restored via restriction of summation in the partition function to the states that could attain the meaning of some `discrete metric geometry'. Thereby one also says that the respective quantum topological BF model is being `reduced' to that of General Relativity (GR).

Both these strategies then struggle to recover some more familiar spacetime geometry from their formalisms, for instance, within some sort of semi-classical approximation (`infra-red' limit). Their reliance on the discrete structures for the regularization also leads to the continuum limit problems (`ultra-violet' completion), that is consistency of the results under arbitrary refinements. The inability to identify what constitutes a \emph{`geometry'} in these models impedes the sensible interpretation of their outcomes. In result, one cannot achieve an unequivocal conclusion, whether the gravitational field has been satisfactorily quantized or not. In this thesis we address the nature of the degrees of freedom that are usually associated with the (pseudo-)Riemannian metric field $\cg$, both from the discrete viewpoint of the above Spin Foam quantum amplitudes, as well as in the classical continuum field-theoretic framework.

Our \textbf{first contribution} consists in realization that one of the most reliable and studied SF proposals so far -- the Engle-Pereira-Rovelli-Livine-Freidel-Krasnov (EPRL/FK) amplitude~\cite{EPRL-FK2008flipped2,EPRL-FK2008finiteImmirzi,EPRL-FK2008FK,EPRL-FK2008LS,ConradyHnybida2010timelike1} -- originally constructed and working well-enough for simplicial discretizations, \ul{does not extend trivially to the arbitrary polyhedral cells}. Such an extension is required to match the general-valence graph kinematical states of the boundary configurations (as LQG suggests). This was indeed put forward by Kami\'{n}ski-Kisielowski-Lewandowski (KKL)~\cite{KKL2010AllLQG,KKL2010correctedEPRL,BHKKL2011OperatorSF}, following the original EPRL-recipe. By scrutinizing the instructive case of the hypercuboid in Ch.~\ref{ch:problem}, \ul{we explicitly demonstrate that the part of the so-called `simplicity constraints', reducing BF to GR, is not implemented properly in the model}~\cite{Belov2018Poincare-Plebanski}. In result, the vertex amplitude cannot be assigned the meaning of describing contribution from some (unique) flat polyhedron, like it was for the rigid 4-simplex. This is evidenced by the appearance in the large-$j$ asymptotic limit of the same type face-mismatch (and torsion) as in canonical LQG. This constitutes a major problem for interpretation of such states in terms of some (semi-)classical 4d geometry~\footnote{We are basing our analysis on the numerical studies first performed in~\cite{BahrSteinhaus2016Cuboidal-EPRL} and subsequently generalized in~\cite{Dona-etal2017KKL-asympt}. This part of our work provides a possible interpretation for their results.}.

Our \textbf{second contribution} -- after identification of the problem -- was to propose some tentative solution. This is done in Ch.~\ref{ch:prop-1}, more or less straightforwardly. The purely  mathematical offspring of the present analysis was the formula, expressing the volume of 4-dim polyhedron in terms of bivectors of its 2-dim faces, proved in~\cite{Bahr2018polytope-volume,Bahr2018non-convex-polytope}. Based on these developments, \ul{we link the missing set of constraints with certain type of knot invariants, related to embedding of the boundary graph in 4d}~\cite{BahrBelov2017VolumeSimplicity}. We show how the condition of independence on the chosen particular knot, imposed on the amplitude (weakly), is capable of resolving the above-indicated issue, in the case of hypercuboid.

As has been noted, Spin Foams usually avoid operating full-fledged metric/vielbein variables, preferring instead to constrain the bivectors accordingly. The follow up of our analysis in Ch.~\ref{ch:problem} was another another realization that \ul{the vielbein/co-frame field is inherently present in the formalism: its degrees of freedom are encoded in the hypersurface normals} (relatives of lapse-shift variables of canonical theory) that are exploited for constraints imposition. Arguably, they are as independent as the connection variables (cf.~\cite{AlexandrovRoche2011CovariantLQG-critique}). This led us to \ul{suggest as the starting point not the Lorentz-BF, but the Poincar\'{e}-BF}~\cite{Belov2018Poincare-Plebanski}, whose Klein's homogeneous space suits the idea of the flat Minkowski vacuum more naturally. The corresponding classical theory is developed in Ch.~\ref{ch:prop-2}, as well as \ul{the novel version of simplicity constraints is proposed} (dual to the currently used). This is our \textbf{third contribution} to the analysis of quantum SF models.

The above findings can be concisely put in the form of a question: \emph{``What are the independent degrees of freedom of gravitational field that are actually being quantized?''} This naturally led us to address the status of the metric/vielbein $\cg\sim\theta$ most directly in the classical theory of Einstein-Cartan gravity that underlies the quantum framework. In particular: since the quantization of connections is fairly well understood, it is all the more natural to phrase GR in these terms. Indeed, this was the driving force of Loop Quantum Gravity since Ashtekar uncovered his variables -- in the canonical setting -- related to local Lorentz symmetry (at the point, or `internal'). It is desirable thus to relate metric/vielbein variables as well as the associated symmetry w.r.t. diffeomorphism transformations (`external') -- which are distinguishing features of gravity -- to the local translations of the Poincar\'{e} gauge theory~\footnote{In particular, given their controversial status in the discrete framework, cf.~\cite{Dittrich2008QG-diffeos,BahrDittrich2009broken-gauge-sym}.}. This constitutes the first half of the thesis, where the comprehensive study of classical Cartan connections is performed. (Also, in order to set up the stage for the rest of the work.)



The idea is certainly not new and (re-)appeared in many forms~\cite{BlagojevicHehl2013,Utiyama1956gauge-trick,Kibble1961gauge-trick-ECSK,GotzesHirshfeld1990MacDowellMansouri-geometry,StelleWest1980broken-deSitter-holonomy,Wise2010Cartan-geometry}. The most noticeable proponents of Cartan connections are R. W. Sharpe~\cite{Sharpe1997Diff-Geometry-Cartan} in mathematics of differential geometry, and D. K. Wise~\cite{Wise2007thesis,Wise2010Cartan-geometry} in the physics of gravity, respectively. In~particular, regarding the gauge theoretic description of gravity and its symmetries, those are frequently viewed as arising from the sort of `symmetry breaking' of some related topological theory (conceptually similar to Spin Foams).

We choose more conservative standpoint, actually dating back to \'{E}. Cartan himself. The mathematical theory of connections is laid out in Ch.~\ref{ch:background} in accordance with his geometric intuition, following very closely to the presentation in~\cite{Sharpe1997Diff-Geometry-Cartan} (to which we refer for the most of proofs). Nevertheless, we make more extensive use of the theory of geometric $G$-structures, adapted to the bundle framework of modern gauge theories~\cite{Bleecker1981gauge-variational-principles}, in what we called a `(locally) Klein bundle'. It enjoys the full action of the principal Poincar\'{e} group of geometry, while the role of a base manifold is diminished. The notion of Cartan connection then embodies the geometric picture of ``rolling'' the affine space on the ``lumpy'' surface of a curved manifold.

The main \textbf{outcome} of our analysis is the \ul{implementation of diffeomorphism transformations as the gauge group of translations, which is `not broken' but an exact local symmetry}. This requires a refined notion of tensors in affine space, which can be either `free' or `bound to a point'. One then shows how the \ul{covariant derivative could be seen as a Lie dragging} that acts both on vectors and points. In Ch.~\ref{ch:Einstein-Cartan}, the physical theory of relativity is given in terms of Cartan geometry, following~\cite{Cartan1986Affine-connections}. In particular, the geometric discretization of forms by means of integration is naturally seen in terms of vector summation. The series of results is obtained on torsion in relation to (non-) closure of surfaces with `defects'. We then discuss tangentially the possible implications for quantization, when the approaches of LQG and Spin Foams are briefly reviewed in Ch.~\ref{ch:quantization}, bridging the gap between `classical' and `quantum' parts of the thesis.

\chapter{Background of `background-independence'}
\label{ch:background}

The ultimate goal of the search for the theory of Quantum Gravity (QG), broadly defined, is to develop a conceptual and mathematical scheme which combines the relativistic field-theoretical description of gravitation with the main principles of quantum physics within a coherent framework. The standard methods of the perturbative QFT lead to intractable divergences when applied to the gravitational field and space-time itself, prompting to revisit the role of background structure(s). Let us retrace consecutively the basic assumptions and fundamental principles to be thoroughly implemented in the modern non-perturbative approaches to the quantization of gravity. The selection of material is quite subjective and may differ somewhat from the traditional presentation. Whereas the language may appear less formal, for the sake of accessibility, we supplied all the references to the rigorous proofs and statements wherever possible.


There are two pieces of the puzzle -- 1) the gravity and 2) the quantum, respectively, that one has to specify first in order to make the whole subject tangible. The best description of the gravitational phenomena up to date is provided by the Einstein's classical theory of General Relativity (GR), so that its basic assumptions are to be taken at face value and (at least some of them) realized quantum mechanically. This emphasis on GR postulates is characteristic to such conservative~\footnote{'Conservative' in the sense of not altering the theory's foundations without compelling evidences. This is what distinguishes it from other approaches where general relativistic description is viewed as effective (or emergent), such as string theory.} approaches as Loop Quantum Gravity (LQG), Spin Foams (SF) and related programs. The specialty of gravity is that its field, in comparison with other known interactions of the Standard Model, directly relates to the mutual displacement of objects, prescribing the change of their \emph{relative} situation in the manifold $\mc M$ of conceivable `points-events'. According to Einstein, this type of gravitational pull is realized through the properties of the spacetime itself, that is endowing the abstractly defined (topological) $\mc M$ with~\emph{geometry}. 

The notion of the manifold is the standard one and intuitively clear, though we remind it promptly. Roughly speaking, it serves the purpose of introducing the way to label points and differentiate functions in a consistent manner.~\footnote{According to~\cite[p.2]{Sharpe1997Diff-Geometry-Cartan}:
\blockquote{In the hierarchy of geometry (whose ``spine'' rises from homotopy theory 
through cell complexes, through topological and smooth manifolds to analytic varieties), the category of smooth manifolds and maps lies ``halfway'' between the global rigidity of the analytic category and the almost total flabbiness of the topological category. We might say that a smooth manifold possesses \emph{full infinitesimal rigidity} governed by Taylor's theorem while at the same time having absolutely no rigidity relating points that are not 
``infinitesimally near'' each other, as is seen by the existence of partitions of unity . . . Smooth manifolds are sufficiently rigid to act as a support for the structures of differential geometry while at the same time being sufficiently flexible to act as \ul{a model for many physical and mathematical circumstances that allow independent local perturbations}.} 
The latter is precisely our main focus in the context of modern QG approaches, where this relation between global-and-local, discrete-and-continuous is often debated.} [Since we are interested in the local field excitations, for the mathematical details on the notions from topology, concerning global aspects, we refer to the textbooks, e.g.~\cite{Hatcher2002alg-topology,Munkres1984alg-topology,Bourbaki1966topology}] The crucial part is that its points -- and hence the coordinate labels that can be chosen arbitrarily -- have no internal significance per se. They attain characterization only through the dynamical elements of the theory themselves, such as particle collisions (an example of Einstein's `spacetime coincidences', cf.~\cite[sec.2.2.5]{Rovelli2004QG}) or certain field value measurements. In other words, whereas the fields of the Standard Model are localized w.r.t. inert/non-dynamical Minkowski (or de Sitter) space, in GR they are localized w.r.t. one another, what embodies the relational view of the world. The contiguity relation becomes dynamical entity in its own right that we know under the name of gravitational field. Quoting from~\cite[p.10]{Rovelli2004QG}:
\blockquote{This absence of the familiar spacetime ``stage'' is called the \emph{background independence} of the classical theory. Technically, it is realized by the gauge invariance of the action under (active) diffeomorphisms. A diffeomorphism is a transformation that smoothly drags all dynamical fields and particles from one region of the four-dimensional manifold to another. In turn, gauge invariance under diffeomorphism (or \emph{diffeomorphism invariance}) is the consequence of the combination of two properties of the action: its invariance under arbitrary changes of coordinates and the fact that there is no nondynamical ``background'' field.}

In respect, all the (non-gravitational) fields inhabiting the arena $\mc M$ react backwards on the geometry, universally coupling it to energy-momentum -- the quantity of motion contained in the `medium'. All this makes Einstein's relativity stand out as the general theory of motion, with the `shape' of spacetime actively responding to the configuration of the matter. The challenge of Quantum Gravity is precisely to fully incorporate this drastic novelty and to understand what is a general-relativistic, or background-independent, Quantum Field Theory (QFT). 


Note that we kept discussion as general as possible so far and did not specify yet what we understand by `geometry' and how exactly it relates to the displacement. This is to stress the universal significance of the above fundamental concepts, having roots in the very nature of how one experiences reality and describes observations. In fact, the physical content of GR can be expressed in multiple possible ways~\footnote{This is not to mention almost the infinitude of conceivable corrections and modifications of the Einstein's GR, compatible with the empirical data.}, putting various its aspects in the forefront, depending on what is considered essential [e.g., metric or connection?]. It is already at this early stage, the preference given to one formalism over the other may play a crucial role, and the choice of the appropriate mathematical language is the first step in the quantization process. 

\blockquote{When a long computation gives a short answer, then one looks for a better method}, -- these words preceded in~\cite[p.344]{MTU1973Gravitation} the introduction of an effective computational device that, however, demands a heavier investment in the calculus of differential forms~\blockquote{than anyone would normally find needful for any introductory survey of relativity}. In contrast, the present chapter provides an account of the mathematical theory of connections from a geometric standpoint -- avoiding as much as possible what Cartan~\footnote{According to~\cite[p.261]{Willmore1959Diff-Geometry}:~\blockquote{a powerful mathematician who possessed a remarkable geometrical intuition which enabled him to see the geometrical content of very complicated calculations. In fact Cartan often used geometrical arguments to replace some of the calculations, and a reader who does not possess this remarkable gift is often baffled by his arguments. It is hoped that as a result of reading this brief introduction the reader will be encouraged to make a serious study of Cartan's book}~\cite{Cartan1983Riemannian-geometry}. We wholeheartedly subscribe under the above words, and may only recommend adding~\cite{Cartan2001Orthogonal-frame,Cartan1986Affine-connections,Sharpe1997Diff-Geometry-Cartan} to the list.} called ``les d\'{e}bauches d'indices'' -- in an effort to stress the naturality of this framework.

\newpage

\section{(Principal) fiber bundles and the general notion of gauge}
\label{sec:PFB-gauge}

The rationale behind GR was the realization that the inertial motion has relative status, whereas the physical effects of force depend on the chosen observer. However, the laws of physics must be expressed in the form of mathematical equations, valid for any admissible observer's frame of reference. Known as the (generalized) Galilei/Einstein's \emph{principle of relativity}, it is at the core of both GR and the gauge theories of particle interactions. Once the key notions of a \emph{frame} and the \emph{transformations} between them are specified, this leads to the geometrization of all the known forces of nature. 

In a different vein, since all measurements are made relative to a choice of frame (broadly defined), and the measurement process can never be completely divorced from the aspect of the universe being measured, we are led to the concept of the bundle of reference frames, corresponding to the smooth concatenation of all admissible observers, each with its own measuring apparatus. The \emph{geometry} (broadly defined) is all about the relations, left invariant under admissible frame transformations -- in accordance with Klein's Erlangen program~\cite{Sharpe1997Diff-Geometry-Cartan,Klein1872Erlangen}.

\begin{figure}[!h]
\center{\includegraphics[width=0.5\linewidth]{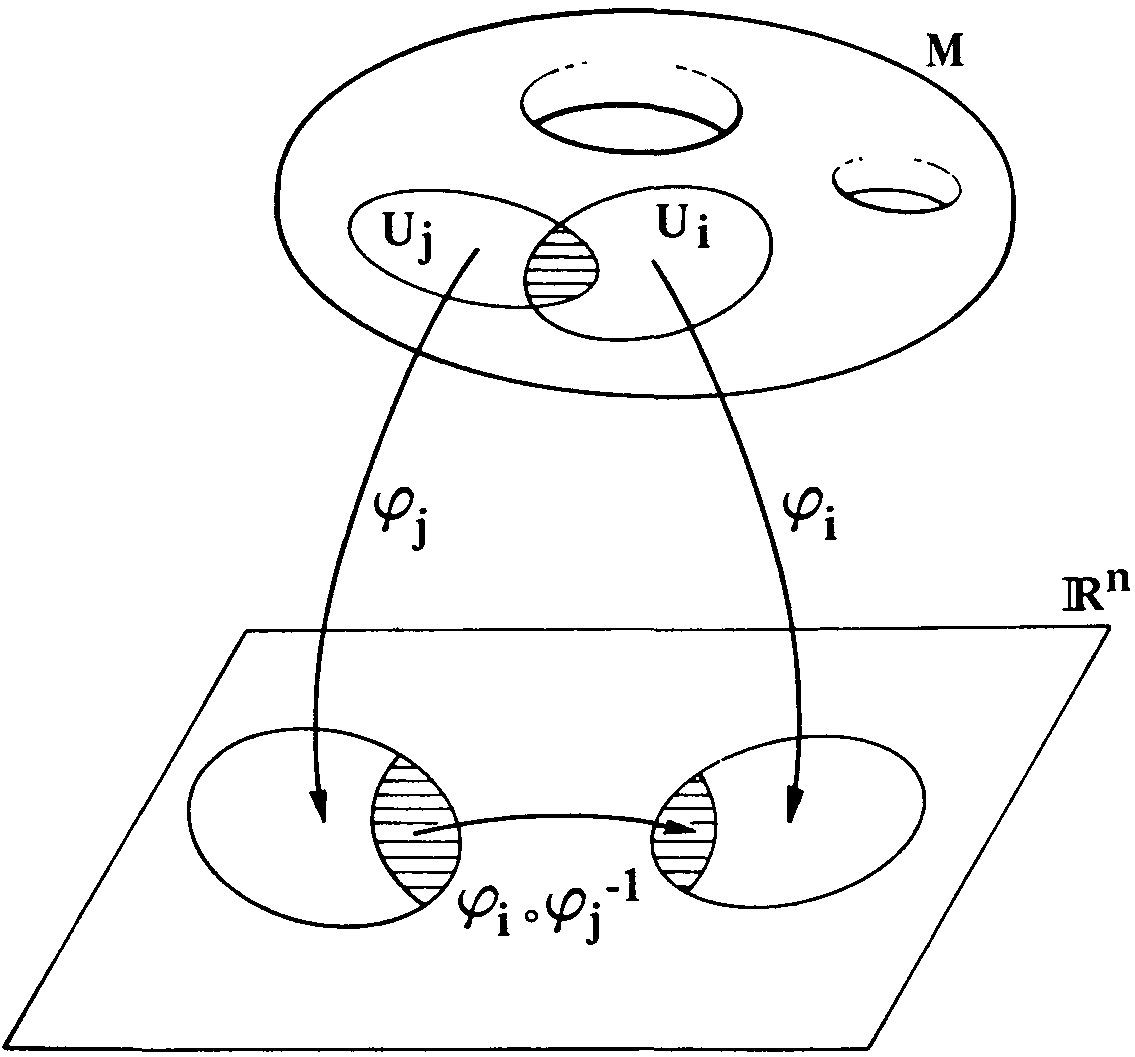}}
\caption{Schematic picture of the manifold and consistent coordinate charts (from~\cite{Bleecker1981gauge-variational-principles}).}
\label{fig:manifold}
\end{figure}
The typical example of a reference frame (and the one which underlies the formulation of GR) is provided by the smooth structure on $\mc{M}$ itself, making it into a \textbf{differentiable manifold}. Concretely, let $\{(U_i,\varphi_i)\}$ be the atlas, consisting of \emph{coordinate charts} $\varphi_i:U_i\ra\mathbb{E}^m$, covering some (paracompact, Hausdorff) $\mc{M}=\bigcup_iU_i$ by open sets $\varphi_i(U_i)\subset\mathbb{E}^m$ of the Euclidean space~\footnote{It is crucial that the \emph{affine} structure of $\mathbb{E}$ allows to canonically identify all the tangent spaces by the parallel transport of vectors to the common \emph{point} $a$, so that the notion of differentiation makes sense.}, such that the transition functions between overlapping charts $\varphi_i\circ \varphi_j^{-1}:\varphi_j(U_i\cap U_j)\ra \varphi_i(U_i\cap U_j)$ are all smooth $C^\infty$ (see~Fig.~\ref{fig:manifold}). Then the \emph{local parametrization} of $U\subset\mc{M}$ via some backwards mapping $x\equiv\varphi^{-1}$ extends to the derivative map $u\equiv x_\ast$, which establishes the linear isomorphism between the corresponding tangent spaces at the point $a\in\mathbb{E}^m$ and that of its image $x(a)\in\mc{M}$ -- this being referred to as a \textbf{frame} (at the point):
\begin{equation}
\begin{aligned}\label{eq:linear-frame}
&u \ : &  T\mathbb{E}^m \cong \mathbb{E}^m\times \mathbb{R}^m& &  &\ra &  &T\mc{M}   \\
& & (a,\mb{v})& & &\mapsto &  &X_{x(a)}\in T_{x(a)}\mc{M}. 
\end{aligned}
\end{equation}
For a time being, let us use interchangeably $x\in\mc{M}$ for the point and its parametrization. Given any standard basis $(\mb{e}_1,...,\mb{e}_m)$ in $\mathbb{R}^m$ at $a=\varphi(x)\in\mathbb{E}^m$, such a map determines a basis $(u_x(\mb{e}_1),...,u_x(\mb{e}_m))$ in $T_x\mc{M}$. Let $L_x(\mc{M})\equiv \{u_x\}$ be the set of all (linear) frames at $x$, then put
\begin{equation}\label{eq:frame-bundle}
L(\mc{M}) \ = \ \bigcup_{x\in\mc{M}} L_x(\mc{M})
\end{equation}
to be the \textbf{linear frame bundle}. It comes equipped with a natural right $\mr{GL}(m,\mathbb{R})$ action $R_{A}: L(\mc{M})\ra L(\mc{M})$, given by $R_A(u)=u\circ A$ for each $A\in \mr{GL}(m,\mathbb{R})$, which may be regarded as the change of standard basis, used by observer to describe its local neighbourhood. This is an example of a more general construction of the principal fiber bundle (PFB), recall it to be precise.
\begin{definition}\label{def:fiber-bundle}
A (locally trivial) \textbf{fiber bundle with an abstract fiber} $F$ consists of a quadruple $\xi=(E,\mc{M},\pi,F)$ of (smooth) manifolds $E$ (called the \textbf{total space}), $\mc{M}$ (called the \textbf{base}), and the (standard) \textbf{fiber} $F$, together with a submersion map $\pi : E\ra\mc{M}$ (called \textbf{projection}, sometimes), such that $\pi^{-1}(U)$ is diffeomorphic by the \textbf{local trivialization} map $\phi$ to $U\times F$ for an open set $U\subset \mc{M}$, containing $x\in\mc{M}$, and the following diagram commutes 
\begin{displaymath}
\begin{tikzcd}[column sep=small]
\pi^{-1}(U) \arrow[dr,swap, "\pi "] \arrow[rr, "\phi"] & & U\times F \arrow[dl,"\mr{proj}_1"] \\
& U & 
\end{tikzcd}.
\end{displaymath}
The pair $(U,\phi)$ is called a \textbf{chart} (or local bundle coordinate system, subordinate to the chosen coordinates on the base manifold $\mc{M}$).
\end{definition}

\begin{definition}\label{def:G-structure}
Suppose that a Lie group $G$ acts smoothly on $F$ as a group of diffeomorphisms. A~$G$ \textbf{atlas} for bundle $\xi$ is a collection of charts $\mathpzc{A}=\{(U_\al,\phi_\al)\}$, covering $\mc{M}=\bigcup_\al U_\al$, such that for each pair of charts in $\mathpzc{A}$ the map
\begin{equation}\label{eq:coordinate-change}
\phi_{\al\be} \ = \ \phi_\al^{\ph{1}}\phi_\be^{-1} : \, (U_\al\cap U_\be)\times F \ \ra \ (U_\al\cap U_\be)\times F,
\end{equation}
called a \textbf{coordinate change}, has the form $\phi_{\al\be}(x,f)=(x,s_{\al\be}(x)f)$, where the smooth maps $s_{\al\be}:(U_i\cap U_j)\ra G$ are called \textbf{transition functions} and satisfy: 
\begin{enumerate}[label={\upshape(\roman*)}, align=left, widest=iii]
\item $s_{\al\al}(x)=e$ for all $x\in U_\al$; \label{prop:transition-1}
\item $s_{\be\al}(x)=s_{\al\be}(x)^{-1}$ for all $x\in U_\al\cap U_\be$; \label{prop:transition-2}
\item $s_{\al\be}(x)s_{\be\ga}(x)s_{\ga\al}(x)=e$ for all $x\in U_\al\cap U_\be\cap U_\ga$. \label{prop:transition-3}
\end{enumerate}
The two $G$ atlases are \textbf{equivalent} if their union is also a $G$ atlas (from which a unique maximal one can always thus be formed). A $G$ \textbf{bundle} is a $\xi$, on which a $G$ \textbf{structure} is specified, that is an equivalence class of $G$ atlases.
\end{definition}

The transition functions describe how the direct products $U_\al\times F$ glue together to form the total space. Indeed, $E$ can be considered as factorspace, obtained from the disjoint union $\bigcup_\al (U_\al\times F)$ via the equivalence relation, identifying the points $(x,f)\in U_\al\times F$ and $(x,s_{\be\al}(x)f)\in U_\be\times F$. Given the covering of $\mc{M}$ by $\{U_\al\}$ and the transition functions, satisfying~\ref{prop:transition-1}-\ref{prop:transition-3}, the bundle can be reconstruted with virtually \emph{any} $G$-module space $F$ as a typical fiber.  

\begin{definition}[PFB]\label{def:PFB}
For a special type of \textbf{principal} $G$ bundle, its structure group is diffeomorphic to its standard fiber, on which $G$ acts by left translations. It can always be associated to $\xi$ by considering the fiber $P_x\equiv\pi^{-1}(x)$ to consist of the (generalized) $F$-frames $p:F\stackrel{\sim}{\ra}\pi_E^{-1}(x)$ at the point $x$, and the diffeomorphism from the standard fiber $G$ acting as a change of basis $\bar{p}(g)=p\circ g$ (cf.~\cite[Prop.1.5.4]{KobayashiNomizu1963vol-1}). Consequently, the total space $P$ is then endowed with a \textbf{smooth right action} (commuting with the left coordinate changes) $P\times G\ra P$ that is fiber preserving and acts simply transitively on each fiber (which is just the orbit $\pi^{-1}(x) = \{pg|g\in G\}$ of $G$ through $p$ above $x=\pi(p))$.
\end{definition}

In fact, the $G$ action alone is enough to fully characterize the PFB, as the following implies.

\begin{theorem}[\text{\cite[App. E]{Sharpe1997Diff-Geometry-Cartan}}]\label{th:PFB-characterization}
Let $P$ be a smooth manifold, $G$ a Lie group, and the (smooth) right action $P\times G\ra P$ be free (i.e. if $pg=p$ for some $p\in P  \Ra  g=e$) and proper (i.e. if $A$ and $B$ compact $\Ra  \{g\in G| Ag\cap B\neq\emptyset\}$ is compact). Then
\begin{enumerate}[label={\upshape(\roman*)}, align=left, widest=iii]
\item $P/G$ with a quotient topology is a topological manifold $(\mr{dim} P/G= \mr{dim}P-\mr{dim}G)$; 
\item $P/G$ has a unique smooth structure for which the canonical projection $\pi:P\ra P/G$ is a submersion; 
\item The tuple $(P,P/G,\pi,G)$ is a smooth principal right G bundle. 
\end{enumerate}
\end{theorem}

The construction in~\eqref{eq:frame-bundle} is a PFB with a structure group $\mr{GL}(m,\mathbb{R})$ and an obvious projection $\pi(u_x) =x$. Other PFBs with different groups, such as unitary $\mr{SU}(N)$, are prolific in the standard model of particle physics, where the notion of a `frame' could mean the choice of zero-phase angle or direction of axes in the (so-called `internal') isospin space~\cite{Yang-Mills1954,Utiyama1956gauge-trick,Bleecker1981gauge-variational-principles}. Since the `glueing' of fibers is solely determined by transition functions, the bundle $\xi$ over $\mc{M}$ with typical fiber $F$ can always be associated to $P$ by a smooth effective (left) action via representation $\rho:G\ra\mr{Diff}(F)$. Denoted $P\times_G^{\ph{1}}F$, this \textbf{associated} fiber bundle consists of equivalence classes, identifying the points $[pg,f]=[p,\rho(g)f]$, and can be obtained from $P\times F$ as the orbit space by factoring through the group action
\begin{equation}
g\cdot (p,f) \ = \ (pg,\rho(g^{-1})f),
\end{equation}
and the standard projection $\pi_E([p,f])=\pi_P(p)=x$.

\begin{definition}\label{def:section} 
A \textbf{local section} of a bundle $\pi:E\ra\mc{M}$ is the smooth map $\tilde{\psi}:U\ra E$ from the open subset $U\subset\mc{M}$ such that $\pi\circ\tilde{\psi}=\mr{Id}_U$. The space of (local) sections of $\xi$ over $U$ we will denote $\Ga(U\ra E)$, or simply $\Ga_U(\xi)$. When $U$ can be extended to the whole of $\mc{M}$, one speaks of \textbf{global} sections.
\end{definition}
The configurations of a (matter) field are described by the elements $\tilde{\psi}\in\Ga(\xi)$ in a coordinate-independent (geometric) manner. In [particle] physics, one usually works with the local ($F$-valued) functions $\psi_\al:U_\al\ra F$ in a region $U_\al\subset\mc{M}$, through the use of bundle charts $(x,\psi_\al(x))=(\phi_\al\circ\tilde{\psi})(x)$. Implicit in this determination is the choice of some basis in $F\cong E_x$, corresponding to certain states of the system. This \emph{fixed, but arbitrary} choice of (moving) reference frame -- known as \emph{gauge} -- is the necessary key ingredient for describing motions in a background independent fashion, referring only to immanent (and dynamical) elements of the system.

\begin{definition}\label{def:gauge}
The \textbf{gauge} is the (local) section $\tilde{\si}:U\ra P_U=\pi^{-1}(U)$ of a PFB with the group~$G$. Since $P$ has the local trivializations of the form $\phi:p\ra(\pi(p),\si(p))$, where $\si:\pi^{-1}(U)\ra G$ satisfy the right equivariance $\si(pg)=\si(p)g$, there is a natural 1-to-1 correspondence between trivializations and gauges: $\phi(\tilde{\si}(x)g)=(x,g)\leftrightarrow \tilde{\si}(x)=\phi^{-1}(x,e)$.
\end{definition}

Presumably, for the physically meaningful quantities the choice of gauge in their description should be immaterial, in accord with the relativity postulate. Hence, some additional care should be taken in order to ensure that the expressions in various patches glue coherently into well-defined objects that do not actually depend on the local trivialization. The power of PFB manifests in that it allows for a neat balance between the two perspectives above, encoding the `geometrically'-given sections in all possible frames, so to say. If $G$ acts on $F$ as $g\cdot f\equiv\rho(g)f$, one adopts the following
\begin{definition}\label{def:equiv-map} The space of ($F$-valued) $G$-equivariant maps
\begin{equation}\label{eq:equiv-map}
C(P,F) \ = \ \{\psi:P\ra F| \psi \ \text{is smooth, such that} \ \psi(pg)=g^{-1}\cdot \psi(p)\}.
\end{equation}
There is a natural bijective correspondence:
\begin{equation}\label{eq:correspondence-fields}
\psi\in C(P,F) \ \leftrightarrow \ \tilde{\psi}\in\Ga(P\times_G^{\ph{1}}F).
\end{equation}
\end{definition}

By the use of~\eqref{eq:correspondence-fields}, the field $\psi$ should be regarded as a function on $P$: if $\psi(p)$ is the value of $\psi$ relative to $p\in P_x$, then $\psi(pg)=g^{-1}\cdot \psi(p)$ is its value w.r.t. transformed frame~$pg$. The good transformation properties of $\psi$ w.r.t. frame change are generally referred to as \emph{covariance}, which will be the convention that we stick to. Given the gauge $\tilde{\si}_\al:U_\al\ra P$, we can pull-back $\psi$ down to function $\psi_\al(x)=(\tilde{\si}_\al^\ast\psi)(x)\equiv\psi(\tilde{\si}_\al(x))$ on $U_\al\subset\mc{M}$. If $\tilde{\si}_\be:U_\be\ra P$ is another gauge, then $\psi_\be(x)=\psi(\tilde{\si}_\al(x) s_{\al\be}(x))= s_{\be\al}(x)\cdot\psi_\al(x)$, showing how the locally-defined objects relate under the change of gauge. Expressing the transition functions as $s_{\al\be}(x) = \si_\al(p)\si_\be(p)^{-1}$, the property~\ref{prop:transition-3} of Def.~\ref{def:G-structure} can be readily shown.

\subsection{Differential forms and operations with them}
\label{subsec:diff-forms}

In general, for integration and differentiation purposes, the main objects of interest will also include the following
\begin{definition}[Forms] An exterior vector-valued differential $k$-form on $P$ is a smooth map (from Whitney sum) $\om:TP\oplus...\oplus TP\ra V$, whose restriction to any fiber $\om_p:T_pP\oplus...\oplus T_pP\ra V$ is $k$-times multilinear and totally skew-symmetric. If $G$ acts linearly via representation $\rho:G\ra \mr{GL}(V)$, one denotes $\La^k(P,V)$ the space of such forms, satisfying equivariance relation $R^\ast_g\om = g^{-1}\cdot\om$, i.e. $\om_{pg}(R_{g\ast}X_1,...,R_{g\ast}X_k)=\rho(g^{-1})\om_p(X_1,...,X_k)$ for $X_1,...X_k\in T_pP$. If $\dim{V}=1$, so that $V\cong\mathbb{R}$, one speaks of just $k$-forms (or invariant, scalar forms). Let $\bar{\La}^k(P,V)$ denote the space of \textbf{horizontal} forms, which are vanishing on vertical vectors: $\varphi(X_1,...,X_k)=0$ if any $\pi_{\ast}(X_i)=0$. There is a natural 1-to-1 correspondence:\label{def:forms}
\begin{equation}\label{eq:horizontal-froms}
\varphi \in \bar{\La}^k(P,F) \ \leftrightarrow \ \tilde{\varphi}\in \La^k(P/G,P\times_G^{\ph{1}}F).
\end{equation}
(That is, invariantly-defined skew-symmetric maps $\tilde{\varphi}:\bigwedge^k T\mc{M}\ra E$, akin to~\eqref{eq:correspondence-fields}. Note: unlike horizontal vectors, the horizontal forms are naturally given, since the vertical vectors are canonically generated by the group action tangent to the fiber, cf. further discussion in Sec.~\ref{sec:Cartan-connections}.)
\end{definition}

The new forms are constructed in three natural ways: 
\begin{enumerate}[label=\arabic*)]
\item The \emph{exterior multiplication} (or \emph{wedge}-product) of $\om_1\in\La^{k}(P,V_1)$ and $\om_2\in\La^{l}(P,V_2)$ is obtained via anti-symmetrization in the arguments of the combined map $\om_1\wedge\,\om_2\in\La^{k+l}(P,V_1\otimes~V_2)$ between product spaces, namely
\begin{equation}\label{eq:wedge}
(\om_1\wedge\om_2)(X_1,...,X_{k+l}) \ = \ \frac{1}{k!l!}\sum_{\varrho\in\mr{perm}(k+l)} (-1)^\varrho\, \om_1(X_{\varrho(1)},...,X_{\varrho(k)})\otimes \om_2(X_{\varrho(k+1)},...,X_{\varrho(k+l)}).
\end{equation}
\item The \emph{interior product} $\imath_X:\La^{k+1}(P,V)\ra\La^k(P,V)$ with the given vector field $X\in\Ga(TP)$ is simply an evaluation $(\imath_X\om)(X_1,...,X_k)=\om(X,X_1,...,X_k)$, and $\imath_X f:=0$ for 0-forms. It satisfies the derivation property $\imath_X(\om_1\wedge\om_2)=(\imath_X\om_1)\wedge\om_2+(-1)^{\deg(\om_1)}\om_1\wedge(\imath_X\om_2)$; 
\item The \emph{exterior derivative} $d:\La^{k}(P,V)\ra\La^{k+1}(P,V)$ is given by expression
\begin{gather}
d\om (X_0,...,X_k) \ = \ \sum_{0\leq i\leq k} (-1)^i\, X_i(\om(X_0,...,\hat{X}_i,...,X_k)) \nonumber\\
+ \sum_{0\leq i\leq j\leq k} (-1)^{i+j}\, \om([X_i,X_j],X_0,...,\hat{X}_i,...,\hat{X}_j,...,X_k)) \quad \text{(``hat'' means ``omit this entry'')},\label{eq:exterior-diff}
\end{gather}
such that it reduces to $df_p(X)=X(f)|_p = \sum f^i_\ast(X_p)\,\mb{e}_i$ on functions, and can be uniquely characterized by its properties (axiomatic definition, cf.~\cite[Lemma 1.5.13]{Sharpe1997Diff-Geometry-Cartan},~\cite[\S 4.3]{BishopGoldberg1980tensor-analysis}):
\begin{enumerate}[label={\upshape(\roman*)}, align=left, widest=iii]
\item $d(\om_1+\om_2)=d\om_1+d\om_2$ (linearity), \label{prop:d-1}
\item $d(\om_1\wedge\om_2)=d\om_1\wedge\om_2+(-1)^{\deg(\om_1)}\om_1\wedge d\om_2$ (derivation), \label{prop:d-2}
\item $d(d\om)=0$ (co-boundary). \label{prop:d-3}
\end{enumerate}
The latter property \ref{prop:d-3} $d^2=0$ -- figuring in the Poincar\'{e}'s lemma -- generalizes the local integrability condition on 1-form to be a total differential $\om=df$ (potentiality of the gradient, leading to the path-independence of the integral $\int\om$, cf.~footnote~\ref{foot:Pfaff}). 
\end{enumerate}

In the case the space of values $V$ is endowed with some multiplication $\mr{m}_\ast:V\otimes V\ra V$, the sum of (infinite dimensional) vector spaces $\La(P,V)=\bigoplus_{0\leq k\leq n}\La^k(P,V)$ is made into graded differential algebra with combined multiplication
\begin{equation}\label{eq:deifferential-algebra}
\La^k(P,V)\times \La^l(P,V) \stackrel{\wedge}{\longrightarrow} \La^{k+l}(P,V\otimes V) \stackrel{\mr{m}_\ast}{\longrightarrow} \La^{k+l}(P,V),
\end{equation}
satisfying $d(\mr{m}_\ast(\om_1\wedge\om_2))=\mr{m}_\ast(d\om_1\wedge\om_2)+(-1)^k \mr{m}_\ast(\om_1\wedge d\om_2)$.

\begin{definition}\label{def:dot-wedge}
Let $\rho_\ast:\mf{g}\ra\mf{gl}(V)$ be the Lie algebra homomorphism, induced by the group action
\begin{equation}\label{eq:algebra-action}
\mb{A}\cdot \mb{v} \ \equiv \ \rho_\ast(\mb{A})\mb{v} \ = \ \frac{d}{ds}\rho(\exp s\mb{A})\mb{v}\bigg|_{s=0}, \qquad \text{for} \quad \mb{A}\in\mf{g}, \ \mb{v}\in V.
\end{equation}
Analogously to~\eqref{eq:deifferential-algebra}, for $\om\in\La^k(P,\mf{g})$ and $\varphi\in\La^l(P,V)$, we can then define
\begin{equation}\label{eq:dot-wedge}
(\om\dot{\wedge}\,\varphi)(X_1,...,X_{k+l}) \ = \ \frac{1}{k!l!}\sum_{\varrho\in\mr{perm}(k+l)}(-1)^\varrho \, \om(X_{\varrho(1)},...,X_{\varrho(k)})\cdot \varphi(X_{\varrho(k+1)},...,X_{\varrho(k+l)}),
\end{equation}
which is just the wedge~\eqref{eq:wedge}, followed by the $\mf{g}$ action in the coefficient space. In particular, if $\mr{m}_\ast$ in~\eqref{eq:deifferential-algebra} is induced by the adjoint representation $\mr{Ad}:G\ra\mr{GL}(\mf{g})$ sending $(g,\mb{B})\mapsto\mr{Ad}(g)\mb{B}$, that is $\mr{ad}\equiv\mr{Ad}_{\ast e}:\mf{g}\mapsto\mf{gl}(\mf{g})$ sending $(\mb{A},\mb{B})\mapsto \mr{m}_\ast(\mb{A},\mb{B})\equiv\mr{ad}(\mb{A})\mb{B}=[\mb{A},\mb{B}]$, then we write: $[\om_1\wedge\om_2]=(-1)^{kl+1}[\om_2\wedge\om_1]=\sum (\om_1^\al\wedge\om_2^\be)\otimes [\mb{E}_\al,\mb{E}_\be]$ for~\eqref{eq:dot-wedge}. The graded Jacobi identity: $(-1)^{rk}[[\om_k\wedge\om_l]\wedge\om_r]+(-1)^{kl}[[\om_l\wedge\om_r]\wedge\om_k]+(-1)^{lr}[[\om_r\wedge\om_k]\wedge\om_l]=0$ is satisfied.
\end{definition}

\paragraph{Classical interpretation.} The \ul{scalar} forms constitute the natural arguments of multidimensional integration (cf.~\cite[Ch.4]{BishopGoldberg1980tensor-analysis}), when pulled-back down to $\La^k(\mc{M})$ (via gauge choice). Through their definition and properties, the following familiar operations in the Euclidean space are concisely captured and abstracted: 1) the construction via $\wedge$ of the volume-determinants built on (unspecified, or generic) argument vectors; 2) the value of the latter on the concrete infinitesimal elements of the $k$-dim surface $S_k$ is obtained via $\imath$-substitution; 3) the exterior derivative generalizes the ordinary $\mr{grad}$, $\mr{curl}$ and $\mr{div}$ of $\mathbb{R}^3$-vectors. 

Moreover, being formulated in the explicitly coordinate-free language, they are naturally applicable to non-Euclidean manifolds, which by definition represent the smooth patchwork glued from the local regions of $\mathbb{E}$. Provided some local chart and the `natural' basis $\pa_i=\pa/\pa x^i$ of vectors and dual forms $dx^i(\pa_j)=\de^i_j$ are given, the usual coordinate expressions:
\begin{align*}
\om \ & = \ \frac{1}{k!}\, \om_{i_1...i_k} \, dx^{i_1}\wedge...\wedge dx^{i_k}, \\
\om\wedge\theta \ & = \ \frac{1}{k!l!}\, \om_{i_1...i_k} \theta_{i_{k+1}...i_{k+l}} \, dx^{i_1}\wedge...\wedge dx^{i_{k+l}}, \\
\imath_X\om \ & = \ \frac{1}{k!}\, X^i\om_{ii_1...i_k} \, dx^{i_1}\wedge...\wedge dx^{i_k},\\
d\om \ & = \ \frac{1}{k!} \, d\om_{i_1...i_k}\wedge dx^{i_1}\wedge...\wedge dx^{i_k},
\end{align*}
are an easy consequence of the intrinsic definitions.

The line element of $\mathbb{E}$, however, is only applicable point-wise -- by considering each tangent space $T_x\mc{M}\approx\mathbb{R}^m$ as Euclidean. It is not until thus-induced scalar product/metric $\cg_x\in \hat{T}_x\mc{M}\otimes\hat{T}_x\mc{M}$ is specified smoothly at each point (a section of the tensor bundle over $\mc{M}$), that the interior multiplication $\imath$ could be regarded as the inner product contraction $\imath_X\hat{Y}=\hat{Y}(X)=\cg(X,Y)$ -- using the (non-canonical) identification $T_x\mc{M}\stackrel{\sim}{\ra}\hat{T}_x\mc{M}$ via $Y\mapsto \hat{Y}(.)=\cg(.,Y)$. We do not hasten to provide such an identification right away, though. The reason is that we wish to stick to the framework where the ``amorphous'' parameter space (which we tend to identify with domain $\mc{M}\approx P/H$) is clearly delineated from the value space $V$ of ``objects'' (with canonically given structure), attached to it~\footnote{Accordingly, our notation marks $X\in TP$ with ordinary typeface, while $\mb{v}\in V$ is bold.}. Arguably, this provides the cleanest strategy to isolate physical/geometrical content from possible artefacts due to coordinate choices (``the former precedes the latter'', so to say), and to ensure the re-parametrization invariance of the results. In fact, the whole method of forms and moving frames was originally developed by Cartan to tackle these subtle issues, in particular. 

In the familiar cases of fields on the fixed Galilei/or Minkowski background, there is no such puzzle, for there is a global Cartesian coordinate system valid over the entire $\mathbb{M}^m$ and canonically given inner product on forms. This can be used to construct the integrands from the physical quantities of interest (such as magnetic fluxes, or circulation of electric field), roughly speaking, by weighting the projections in the certain directions with the respective $k$-volumes. (For instance, $F=\frac{1}{2}F_{ij}dx^i\wedge dx^j=B_x dy\wedge dz+ B_y dz\wedge dx +B_z dx\wedge dy + (E_x dx+ E_y dy +E_z dz)\wedge dt $ is the typical example, where $F_{ij}$ are components of the electromagnetic stress-energy tensor.) The integration over $k$-surface is then seen as simply `summation' of the well-defined objects (scalars) in the limit of smaller and smaller divisions, connecting the `finer' description to the `coarse-grained' observables which are actually being measured.  

The physical laws (e.g., that of Faraday-Maxwell) establish the relations between different such quantities in an invariant manner, accumulating the large number of experimental facts. The bridge to the formulation in partial (1st order) differential equations, relating the numerical values of physical quantities at two infinitesimally close points, is provided by the famous \emph{Stokes formula}:
\begin{equation}\label{eq:Stokes-1}
\oint_{S_k=\pa S_{k+1}} \om \ = \ \int_{S_{k+1}} d\om.
\end{equation}
Incidentally, that also encapsulates how the new invariants, associated with the higher dimensional elements, may be constructed from the quantities at the boundary (consider the values of integrands in the decrementally small region, so that rather field or surface variations could be disregarded). 

\subsection{Motivations from QG: discretization, continuum limit}
\label{subsec:discretization-problem}

\paragraph{Regularization/cut-off in general.}
Arguably, the exact field values could never be measured pointwise (this is not even mentioning the relational character of observations), only to a sufficiently fine precision, like described above. For the practical purposes, one commonly adheres to \emph{regularization} by introducing some sort of \emph{cut-off} on physically interesting degrees of freedom. Broadly speaking, this allows to organize computations effectively, reducing the infinitude of field's characteristics down to the manageable amount (e.g., corresponding to a finite series of measurements). Examples of this include the typical energy of collisions in particle physics, the lattice spacing in the Monte-Carlo simulations, as well as the Regge discretization~\cite{Regge1961calculus} in gravity. The latter, akin to lattice gauge theory, replaces the smooth spacetime with simplicial complexes, conveying the gravitational d.o.f. (distributional) to the way how the simplexes glue together in a non-trivial manner. 

\paragraph{Continuum limit troubles.}
The characteristic of the `good' theory requires consistency of the description as one removes the cut-off, by gradually taking into account more and more degrees of freedom. This is referred to as the \emph{continuum limit}, corresponding to the infinitely fine (in principle) description of the system. The regularization should capture well the essential properties of the system, in particular, its behaviour w.r.t. symmetry transformations. This is currently not the case in the modern attempts to quantization of gravity based on the use of discrete structures for regularization, descendants of the Regge's triangulations~\footnote{LQG is somewhat exceptional, in this regard, whose (kinematical) Hilbert space allows for the continuum limit, although well-behaved only w.r.t. the symmetries of 3d hypersuface. The dynamics is viewed in terms of hypersurface deformations and problematic.}. The reason appears to be non-sufficient understanding of the diffeomorphism symmetry and its implementation, resulting in that one of the GR's main tenets is lost under discretization~\cite{Dittrich2008QG-diffeos,BahrDittrich2009broken-gauge-sym}. What is worse, it is not entirely clear how this should be realized in the first place.

\paragraph{Inspiration via analogies.}
It is not hard to notice the similarity of the above narrative to that of surrounding the whole integration process, briefly sketched around~\eqref{eq:Stokes-1}. Indeed, the increments in the principal parts of the integrals utilize some linear approximations, reminiscent of the individual simplicial cells, whose geometry is known and can be used for contraction. The refinement of the subdivision adds more contributions to the sum and the better approximation, ideally. In fact, these analogies inspired (various) geometric discretization techniques, that were exploited, e.g., for putting the Maxwell's equations on a grid (finite element method, cf.~\cite{Euler2007Maxwell-discr}), and in some TQFTs as well~\cite{Sen-etal2000geom-discretisation-Chern-Simons}. They were also incorporated into several discrete models of QG, mainly based on the topological nature of the latter in 3d.

The primary idea consists of \emph{discretizing the domain}, s.t. it is represented via the collection of basic building blocks -- called `cells' -- which topologically are closed $n$-balls, `glued together' into a \emph{cell-complex}. The details may vary, and depending on the required properties of maps, one distinguishes purely topological CW complexes, piecewise-linear, combinatorial, and simplicial ones. In the piecewise-linear context~\cite{Baez2000SF-BF}, one usually talks about point as 0-cell, a 1-cell is a closed interval, a 2-cell is a polygon, and so on. Any lower dimensional face $Y$, contained in $X$ (or equal) is called its `face' and denoted $Y\leq X$. A piecewise-linear cell-complex is a collection of cells $\mc{K}$ in some $\mathbb{R}^n$ such that: 1) if $Y\leq X\in\mc{K}$ then $Y\in\mc{K}$; 2) If $X,Y\in\mc{K}$ then $X\cap Y\in X,Y$. One writes $|\mc{K}|$ for the (set theoretic) union of cells contained in $\mc{K}$~\footnote{One notices that if the topological domain supports some non-trivial geometric structure, one should be careful that it does not run into conflict with the extra assumption on linearity of $\mathbb{R}^n$-cells, which is basically the content of our following remarks.}.

The algebraico-topological considerations then usually enter the discretization of \emph{objects on the domain} as follows. Given the simplicial cell-complex $\De$, the singular homology is constructed by considering maps of the standard $k$-simplex to a \emph{topological} space, and composing them into \emph{formal} sums. Define the $k$-chains $c_k$ as maps from (combinatorial) $k$-skeleton $\De_k=(\si_k^1,...,\si_k^{N_k})$ to the (abstract) abelian group, such that the image for each cell $-\si\in \De_k$ with reversed orientation: $c_k(-\si)=-c_k(\si)$. Taking $C_k(\De,\mathbb{R})$ to be the space of all real-valued $k$-chains, due to natural multiplicative structure (that gives the notion of `scale'), each element can be seen as the linear combination $c_k=\sum_{i=1}^{N_k}\mu_k^i \, u^i_k$ of unit chains $u_k^i(\pm\si_k^j)= \pm\de^{ij}$, and $C_k(\De,\mathbb{R})$ -- as a linear space over $\mathbb{R}$, spanned by the standard basis of elementary cells. Suppose that (the homeomorphic image of) $\De$ somehow approximates the smooth manifold $\mc{M}$. Then the $k$-forms are naturally discretized as co-chains $C^k(\De,\mathbb{R})\equiv\mr{Hom}(C^k(\De,\mathbb{R}),\mathbb{R})$ via integration over simplices: 
\begin{equation}
w(c_k) \ = \ \int_{c_k}\om \ = \ \sum_i\mu_k^i w_i, \qquad w_i \ = \ \int_{u_k^i}\om.
\end{equation}
The isomorphism of the corresponding cohomologies is induced by this map, according to de-Rham's theorem. By the Stokes formula~\eqref{eq:Stokes-1}, $d$ is the co-boundary operator dual to $\pa$. In accord with our understanding of integration, the co-chains $w$ represent densities w.r.t. measures $\mu_k$, imparted to the discretized support by $\mathbb{R}$-valued chains. 

\paragraph{Problem posing.}
As such, the purely combinatorial construction of a cell complex is well-suited to study the topological properties of $\mc{M}$, respectively, but what about geometry? Part of the geometrical information stripped away is restored by attaching signed $k$-measures (lengths, areas, volumes, etc.) to the cells of $\De$, which may be suitable for theories with background metric space. Notwithstanding, we see couple of weak points, hindering the genuine extension of this scheme to gravity. 
\begin{enumerate}[label={\upshape(\arabic*)}, align=left, widest=iii]
\item In the background independent GR, the metric itself is dynamical and subject to variation, hence it should rather enter the $\om$ `density' part.
\item The formal summation in $c_k$ makes little-to-no `geometric' sense (or at least obscure for us) prior to integration, since it is only by the field $\om$ the real values are associated to the (elementary) cells that play the role of `labels'. On the contrary, the integration itself is naturally viewed as summation.
\item At last, the scalar forms are insufficient, since the geometric information contains not only $k$-volumes of (elementary) surfaces, but also their directions (mutual~\footnote{To provide some more context: In the gauge \emph{invariant} languages, using lengths/areas/volumes as discrete (metric) variables, the insufficiency of a subset of them to characterize the configuration of polyhedral cells locally~\cite{Barrett-etal1999note-area-Regge}, led to the introduction of numerous angle-variables~\cite{Barrett19941st-order-Regge,DittrichSpeziale2008area-angle,BahrDittrich2010angle-calculus}. The latter cannot be arbitrarily assigned but have to satisfy some non-trivial constraints. (The rigid simplex is exceptional, since its shape is uniquely determined by the edge-lengths in every dimension.) In contrast, we are talking about gauge \emph{covariant} description, local to a region.}) in spacetime.
\end{enumerate}
With all the above said, we formulate the \textbf{task of integration of vector-valued forms}, attaching to the domain some sensible characteristic of geometry. One should be careful though, since the result may depend on the choice of frame. All this prompts to take a closer look on variables of the gravitational field in the background-independent GR, to analyse what constitutes its fundamental degrees of freedom in the continuum, keeping eye also on what happens in discrete QG proposals. 

\pagebreak

Let us first remind in Sec~\ref{sec:global-geometry} the very basic concepts from the ordinary (flat) Euclidean/Minkowskian spaces, permitting us to measure lengths and angles by means of a scalar product between two vectors (and, more generally, tensors). Then in Sec.~\ref{sec:Cartan-affine}, this view will be `localized' using the framework of Cartan affine connections. This strictly corresponds to the conventional way of thinking about manifolds in terms of plane tangent spaces, that are attached -- or `soldered' -- to $\mc{M}$ at the point. We return to the task in Sec.~\ref{sec:geom-sum}, where it will be addressed in the geometric context of gravity theory.

\newpage

\section{Elementary geometry}
\label{sec:global-geometry}

In this section we recount the basic notions from the geometry of the familiar Euclidean space~$\mathbb{E}$, with the straightforward generalization to Minkowski spacetime~$\mathbb{M}$, whose affine structure is carefully highlighted. The focus is on the multilinear algebra, describing plane elements of various dimensionality, where we mostly stick to~\cite[Ch.1]{Cartan1983Riemannian-geometry}. The description exploits the existence of globally defined reference frames -- including the choice of origin -- having in mind their subsequent local and approximate character, upon `gauging' and inclusion of (gravitational) interactions. 

\subsection{Affine spaces and bases} 
\label{subsec:affine-space}

Let $\mathbb{E}$ be a point set and $(V,+)$ an $m$-dimensional (real) vector space regarded as abelian group. Consider the action $\mathbb{E}+V\ra\mathbb{E}$ which is free (i.e. if $a+\mb{v}=a$ for some $a\in\mathbb{E}\ \Rightarrow \mb{v}=\mb{0}\in V$) and transitive (i.e. for any $a,b\in\mathbb{E}$ there is $\mb{v}\in V$ such that $b=a+\mb{v}$). The triple $(\mathbb{E},V,+)$ is then called an \emph{affine space}, whereas $\mathbb{E}$ -- the \emph{principal homogeneous space} for the group of translational motions $V$, the stabilizer (sub-)group of every point being trivial $\mb{0}\in V$. The dimension of an affine space is defined as the dimension of the vector space of its translations.

Translations $V$ play a twofold role for $\mathbb{E}$. First, in an obvious manner any element $\mb{v}\in V$ defines the rigid motion of space as a whole via point set correspondence (globally defined diffeomorphism) 
\begin{equation}
\begin{aligned}\label{eq:translation-global}
&\phi_\mb{v} \ : &  \mathbb{E}& &  &\ra &  &\mathbb{E}   \\
& & a & & &\mapsto &  &a+\mb{v}. 
\end{aligned}
\end{equation}
On the other side, for any fixed $o\in\mathbb{E}$ and an arbitrarily chosen linear basis $e=(\mb{e}_1,...,\mb{e}_m)$ in $V$, the bijective correspondence is established via the group action between the points of $\mathbb{E}$-space and their parametrization via $V$-elements:
\begin{equation}\label{eq:affine-basis}
a \ = \ o + \mb{e}_i^{\ph{i}}\, x^i(a).
\end{equation}
The pair $(o,e)$ forms a global \emph{affine frame} in the sense that any point $a\in\mathbb{E}$ can be written as~\eqref{eq:affine-basis} (w.r.t. some arbitrary $o\in\mathbb{E}$ and $e\in V$, of course), so that the chart $(\mathbb{E},x)$ provides the globally-defined coordinate system. 

In the affine spaces, the ``addition of points'' is not allowed. One can only pass to new points by adding elements of $V$, or meaningfully take the difference of two points to obtain the arrow-vector that they bound. In this regard, one draws the following distinction between two types of vectors. If the origin is fixed, the vector is called \emph{bound}, or \emph{sliding}, e.g. the ``radius-vector'' of some other point's position as in~\eqref{eq:affine-basis}. When only the magnitude and direction of the arrow matter, while the particular initial point is of no importance, then one speaks about \emph{free} vector. They are normally obtained as the difference of two bound vectors. The typical example of the first is the force, being dependent on the point of application, while the total angular momentum is a free vector -- indeed: $\sum_i[\mb{p}_i\times (\mb{r}_i+\mb{a})]=\sum_i[\mb{p}_i\times \mb{r}_i]$ due to $\sum_i\mb{p}_i=0$ in Newtonian mechanics~\footnote{Usually, in the  Euclidean space equipped with a choice of origin, this separation is non-relevant, since a free vector is equivalent to the bound one of the same magnitude and direction whose initial point is the origin.}.



\subsection{Group theoretic properties} 
\label{subsec:group-theory}

On the manifold $P$ of all (global) affine frames of $\mathbb{E}$ the natural right action $P\times G\ra P$ of the \emph{affine group} $G=V\rtimes\mr{GL}(V)$ is defined (`passive view'):
\begin{equation}\label{eq:affine-group-1}
(o',e') \ = \ (o,e)\,g \ = \ (o+\mb{a},eA) , \qquad g \ = \ (\mb{a},A),
\end{equation}
with the multiplication law of the semi-direct product $g_1g_2=(\mb{a}_1,A_1)(\mb{a}_2,A_2)=(\mb{a}_1+A_1\rt\mb{a}_2,A_1A_2)$. Here the components of $\mb{a}=\mb{e}_i^{\ph{i}}\,a^i$ describe the new frame's position $o'$ w.r.t. the original system. We also have used the following shorthand notation: 
\begin{equation}\label{eq:linear-group}
e' \ \equiv \ (\mb{e}_1',...,\mb{e}_m') \ = \ (\mb{e}_i^{\ph{i}}A^i_{\ph{i}1},...,\mb{e}_i^{\ph{i}}A^i_{\ph{i}m}) \ \equiv \ eA, \qquad A\in \mr{GL}(m,\mathbb{R}),
\end{equation}
defining the respective projections on the original axes of the new frame's \emph{unit} vectors.

The following short sequence of group homomorphisms
\begin{equation}\label{eq:semi-direct-seq-1}
1 \ \ra \ V \ \stackrel{\al}{\ra} \ G \ \stackrel{\be}{\ra} \ V\backslash G\cong\mr{GL}(V) \ \ra \ 1
\end{equation}
is exact in the sense that the kernel of each homomorphism is equal to the image of the preceding one. The map $\mr{GL}(V)\ra\mr{Aut}(V)$, sending $(A,\mb{a})\mapsto (A\triangleright\mb{a})\in V$, provides the so-called split of~\eqref{eq:semi-direct-seq-1} in the sense that there is a homomorphism $\ga:\mr{GL}(V)\ra G$, defined as $\ga(A)=(\mb{0},A)$, projecting to the identity: $\be\circ\ga=1$. One says that the semi-direct product $G$ is (non-trivial) extension of $\mr{GL}(V)$ by $V$. 

In a different vein (`active view'): whereas the (sub-)group of general linear transformations $H\subseteq\mr{GL}(V)$ preserves the linear structure in the sense that $A(c_1\mb{v} + c_2\mb{w}) = c_1 A(\mb{v}) + c_2 A(\mb{w})$, the affine group $G$ respects the affine structure of translations in the sense that $g(o+\mb{v})=g(o)+A(\mb{v})$. The kernel of $\be$ is the (sub-)group of translations $V$. Using this, $H$ can be identified with the stabilizer of the point, and thus the kernel of the canonical projection $\pi$ in the dual sequence: 
\begin{equation}\label{eq:semi-direct-seq-2}
1 \ \ra \ H \ \ra \ G \ \stackrel{\pi}{\ra} \ \mathbb{E}\cong G/H \ \ra \ 1.
\end{equation}
(Exact, but not homomorphisms of groups, since $\mathbb{E}$ is the homogeneous space of cosets.) The relation between two perspectives may be seen as follows. 

The set $P$ is the \emph{principal homogeneous space} for $G$, on which the latter acts freely and transitively. Hence, for an arbitrarily fixed $(o,e)$, there is a bijective correspondence between $P$ and~$G$. The group $G$ -- called the \emph{principal group} of a geometry (on $\mathbb{E}$)~\footnote{The name of the PFB has its origins in this translation of the German ``hauptgruppe'' (cf.~\cite[p.138]{Sharpe1997Diff-Geometry-Cartan}).} -- then acts on itself from the left. Decomposing $g=(\mb{a},1)\circ (\mb{0},A)$, the factor-space $G/H_o$ of (left) cosets $gH_o$, w.r.t. sub-group $H_o\approx\mr{GL}(V)$ of homogeneous transformations of the linear part of the frame, may be identified with the point set $\mathbb{E}$ itself -- in accord with Klein's Erlangen program. 

The projection map $\pi:G\ra G/H$ is a PFB, canonically associated with the arbitrary homogeneous Klein geometry, understood in the sense of a pair $(G,H)$ of (principal) group $G$ and its (closed) sub-group $H$. Had we chosen the different origin, the respective stability sub-groups are related by conjugation $(\mb{a},1)^{-1} H_o (\mb{a},1)=H_{o+\mb{a}}$, corresponding to different embeddings of $H$ in~$G$. [This base-point disappears in the infinitesimal description, thus making the choice of origin $o$ the same type of ``integration constant'' as $e$.] The splitting of the sequence~\eqref{eq:semi-direct-seq-2} (i.e. $\tilde{\ga}:G/H\ra G$, s.t. $\pi\circ\tilde{\ga}=\mr{Id}_{\mathbb{E}}$) defines the connection in the direct sense of choosing the ``horizontal'' lift, typical to the groupoid approach to this concept~\cite{Mackenzie2005groupoids-algebroids}. Both `active', and `passive' viewpoints on the transformation group are completely equivalent and allow generalization to arbitrary geometries (e.g., non-spacetime reductive, s.a. conformal or projective)~\cite{Klein1872Erlangen,Cartan1924projective-connections}; though, we have the tendency to use the language of frames as being more ``visual''. 

It is convenient to represent $\mathbb{E}^m$ -- and its tangent (free) vectors -- via an embedded $(x^0=1)$-hyperplane in $\mathbb{R}^{m+1}$, where $G$-subgroup of $\mr{GL}(m+1,\mathbb{R})$ acts linearly as follows:
\begin{equation}\label{eq:affine-linearized}
(\mb{e}_0^{\ph{i}},\{\mb{e}_j^{\ph{i}}\})
\begin{pmatrix}
1 & 0 \\
\mb{a} & A 
\end{pmatrix}
\ = \ (\mb{e}_0^{\ph{i}}+\mb{a},\{\mb{e}_i^{\ph{i}}A^i_{\ph{i}j}\}).
\end{equation}
Due to parallelizability of $\mathbb{E}$, the frames with the $\mb{e}_0^{\ph{i}}$-origin fixed  may be used, s.t. the points are represented via bound vectors $\vv{ox}\equiv\mb{X}_0^{\ph{i}}=\mb{e}_0^{\ph{i}}+\mb{x}$, having components $(1,x^i)$, while the free vectors lie in the hyperplane $\mb{X}=\mb{e}_i^{\ph{i}}X^i$ altogether and transform linearly under $H$-subgroup, correspondingly. Let us now concentrate on the basic notions of metric geometry of (free) vectors and their $k$-dim generalizations.

\subsection{Metric properties} 
\label{subsec:metric}

The affine geometry studies the properties of figures that remain unaffected under $G$ [via~\eqref{eq:affine-linearized} the latter may be viewed as a subset of projective transformations/homographies leaving the hyperplane at infinity intact]. The notions of addition of vectors and their equality may thus be established, as well as the equivalence between sliding vectors (parallelism). However, the (ratio of) lengths of two vectors is available only for those which are parallel. Indeed, the arbitrary choice of non-collinear axes $e$ [e.g. pointing to some distant stars] sets up the most general system of Cartesian coordinates: the projections of the vector on each of the axes $\mb{e}_i^{\ph{i}}$ are made by the hyperplane parallel to the linear span of $\{\mb{e}_{j\neq i}^{\ph{i}}\}$, the magnitudes being measured in terms of the individual units (rods and clocks) on each of the basis vectors of the oblique frame. 

The particular unit of length, considered as absolute, may be set up through the following
\begin{definition}\label{def:metric-vector}
The metric on $\mathbb{E}$ is  the symmetric non-degenerate $\mathbb{R}$-valued bilinear form
\begin{equation}\label{eq:metric-def}
\cg: \, V\otimes V \ \ra \ \mathbb{R}.
\end{equation}
Introducing the coefficients w.r.t. arbitrary basis $\cg_{ij}=\cg(\mb{e}_i^{\ph{i}},\mb{e}_j^{\ph{i}})$, it may be written as $\cg=\cg_{ij}\hat{\mb{e}}^i\otimes\hat{\mb{e}}^j$.
This is used to compute the squared magnitude of the vector's length and the scalar product: 
\begin{equation}\label{eq:vector-sp}
|\mb{X}|^2 \ \equiv \ \cg(\mb{X},\mb{X}) \ = \ \cg_{ij}X^iX^j, \qquad \bra\mb{X},\mb{Y}\ket \ \equiv \ \cg(\mb{X},\mb{Y}) \ = \ \cg_{ij}X^iY^j,
\end{equation}
the latter may be obtained as the $2\la$-prefactor in the expansion of $|\mb{X}+\la\mb{Y}|^2$. The (planar) angle between two vectors is then deduced to be
\begin{equation}\label{eq:angle-2d}
\cos \al \ = \ \frac{\bra\mb{X},\mb{Y}\ket}{\sqrt{|\mb{X}|^2|\mb{Y}|^2}}, \qquad \al \ \equiv \ \widehat{\mb{X}\mb{Y}}.
\end{equation}
\end{definition}
The passage from the oblique axes to the orthonormal basis $\{\mb{e}_i'\}$ can always be performed, in which the $\cg$-matrix is diagonalized and have the entries 
\begin{equation}\label{eq:orthonormal}
\cg_{ij}' \ = \ \cg_{kl}A^k_{\ph{k}i}A^l_{\ph{l}j} \ = \  \pm\de_{ij}.
\end{equation}
The number of positive and negative eigenvalues is preserved (Sylvester's law) and called the \emph{signature}. We are mainly interested in the positive-definite metric of the Euclidean space $\mathbb{E}$, and that of the Minkowski spacetime $\mathbb{M}$ with $(m-1,1)$-signature (in $m=4$)~\footnote{An intriguing and phenomenologically viable option could be the principal group $G=\mr{O}(m,1)$ of the \emph{de~Sitter space} $\mr{O}(m,1)/\mr{O}(m-1,1)$, with `translations' $V$ replaced by non-commutative $H$-module [akin to boosts $[\mf{b},\mf{b}]\subset c^{-2}\mf{o}(3)$]. This allows to accomodate for a positive cosmological constant/fundamental length scale within the modified kinematics of special relativity~\cite{Cacciatori-etal2008deSitter-rel,AldrovandiPereira2009deSitter-rel,InonuWigner1953group-contraction}.}. Denoting $\eta = \mr{diag}(+1,...,-1)$, in general, the requirement of the preservation of the scalar product $\eta=A^T\eta A$ then restricts the sub-group of homogeneous transformations to the respective orthogonal group $H=\mr{O}(\eta)$ (or special-orthogonal/proper-orthochronous, depending on allowed discrete symmetries). In the latter case, the angle in~\eqref{eq:angle-2d} should be replaced by hyperbolic $\cosh\be$ for the timelike vectors, i.e. the measure of boost in terms of additive rapidity.

The usual convention introduces the \emph{covariant} components of the vector $X_i\equiv\bra\mb{e}_i^{\ph{i}},\mb{X}\ket = \cg_{ij}X^j$, s.t. the scalar product may be concisely written as $\bra\mb{X},\mb{Y}\ket=X_iY^i=X^iY_i$, and $|\mb{X}|^2=X_iX^i$. They are inversely related to their \emph{contravariant} counterparts $X^i=\cg^{ij}X_j$ by means of $\cg^{-1}\equiv\{\cg^{ij}\}=C^T/\det\cg$, constructed from the (transpose) matrix of cofactors in the determinant expansion $\det\cg = \sum_{j=1}^m\cg_{ij}C_{ij}=\sum_{i=1}^m\cg_{ij}C_{ij}$ along any row $i$ or column $j$. Although two types of components coincide (up to sign) in the orthonormal basis and often used interchangeably, they encode the different information, and the distinction is drawn in the general case: the first give the usual projections on axes $X^{i\neq j}=0\ \Leftrightarrow \ \mb{X}\propto\mb{e}_j^{\ph{i}}$, whereas the other projections are made orthogonally using the scalar product $X_{i\neq j}=0 \ \Leftrightarrow \ \mb{X}\perp\{\mb{e}_{i\neq j}^{\ph{i}}\}$-hyperplane.

\subsection{(Simple) bivectors and multivectors} 
\label{subsec:simple-multivectors}

The (free) vector is seen as the directed line segment of certain length, i.e. it is characterized 1)~by its measure, 2)~direction, and 3)~orientation in the sense of a certain order of its endpoints (their precise position is non-important for free vectors). This may be generalized to a plane elements of higher dimensionality, leading to the construction of a `geometric algebra'. Let us examine it in some greater detail using the instructive case of a 2-plane, following~\cite[Ch.1]{Cartan1983Riemannian-geometry}. Being the first non-trivial dimension where the curvature reveals itself, it is representative of a more general situation. (Also, the bivectors constitute one of the basic constructive elements in Spin Foams and LQG theories.)

\begin{definition}\label{def:bivector-simple}
A \textbf{(simple) bivector} is defined to be the configuration formed by two vectors $\mb{X}$ and $\mb{Y}$, arranged in a certain order. To clarify its geometric meaning one formulates the conditions when two bivectors are considered to be equivalent, or equal. Namely, the two parallelograms built on the pairs of vectors as their edges should
\begin{enumerate}[label={\upshape(\roman*)}, align=left, widest=iii]
\item lie in the same (or parallel) plane, \label{prop:bivecctor-1}
\item possess the same area, and \label{prop:bivecctor-2}
\item the same orientation (circulation in the boundary). \label{prop:bivecctor-3}
\end{enumerate}
\end{definition} 
How these conditions, defining the bivector, determine its coordinate representation? Consider the operation of taking the~\emph{wedge} product, or~\emph{exterior} multiplication of vectors~\footnote{The definition of~\eqref{eq:wedge} represents the dual operation on linear functionals, applied pointwise to tangent vector spaces of an arbitrary manifold. One is straightforwardly related to the other through the concept of soldering.}, that is
\begin{enumerate}[label={\upshape(\alph*)}, align=left, widest=iii]
\item linear $\mb{X}\wedge(\al\,\mb{Y}+\be\,\mb{Z}) = \al\,\mb{X}\wedge\mb{Y}+\be\,\mb{X}\wedge\mb{Z}$; \label{prop:wedge-1}
\item associative $(\mb{X}\wedge\mb{Y})\wedge\mb{Z}= \mb{X}\wedge(\mb{Y}\wedge\mb{Z})\equiv\mb{X}\wedge\mb{Y}\wedge\mb{Z}$;\label{prop:wedge-2}
\item alternating $\mb{X}\wedge\mb{Y}=-\mb{Y}\wedge\mb{X}\quad\Leftrightarrow\quad \mb{X}\wedge\mb{X} =0$. \label{prop:wedge-3}
\end{enumerate}
Given an arbitrary basis $\{\mb{e}_i^{\ph{i}}\}$, we may decompose
\begin{equation}\label{eq:bivector-simple}
\mb{\Si} \ \equiv \ \mb{X}\wedge\mb{Y} \ = \ \frac12\Si^{ij}(\mb{e}_i^{\ph{i}}\wedge\mb{e}_j^{\ph{i}}), \qquad \Si^{ij} \ \equiv \ X^iY^j-Y^iX^j \ = \ \begin{vmatrix}
X^i & Y^i \\
X^j & Y^j 
\end{vmatrix}.
\end{equation}
(The factor of $2$ prevents overcounting of skew-symmetric components.) Using the metric for contraction, the \emph{inner} product with the vector may thus be formed:
\begin{equation}\label{eq:inner-product}
\mb{U} \ \equiv \ \iota_\mb{Z}^{\ph{i}}\mb{\Si} \ := \ \bra\mb{X},\mb{Z}\ket \mb{Y}- \bra\mb{Y},\mb{Z}\ket \mb{X}.
\end{equation}
For the vector $\mb{Z}$ to be perpendicular with the plane determined by $\{\mb{X},\mb{Y}\}$, its (covariant) components $Z_{j\neq i}$ should be linearly constrained by $\Si^{ij}Z_j=0$ for each $i$. Then for any other pair of $\{\tilde{\mb{X}},\tilde{\mb{Y}}\}$ defining the same (or parallel) plane, this leads to the proportionality of the corresponding coefficients in~\eqref{eq:bivector-simple} $\tilde{\Si}^{ij}=\la\Si^{ij}$; and similarly for the covariant components, with the same $\la$.

To express the second equivalence condition~\ref{prop:bivecctor-2} in Def.~\ref{def:bivector-simple}, one sets up the bivectors's \emph{magnitude} to be the area of the respective parallelogram, and evaluates it in terms of~\eqref{eq:bivector-simple}:
\begin{align}\label{eq:bivector-area}
|\mb{\Si}|^2 \ & \equiv \ |\mb{X}|^2|\mb{Y}|^2\sin^2(\widehat{\mb{X}\mb{Y}}) \ = \ |\mb{X}|^2|\mb{Y}|^2-\bra\mb{X},\mb{Y}\ket^2 \nonumber\\
& = \ \begin{vmatrix}
|\mb{X}|^2 & \bra\mb{X},\mb{Y}\ket \\
\bra\mb{X},\mb{Y}\ket & |\mb{Y}|^2 
\end{vmatrix} \ = \ \begin{vmatrix}
X_i & Y_i \\
X_j & Y_j 
\end{vmatrix} X^iY^j \ = \ \frac12\Si_{ij}\Si^{ij}.
\end{align}
The proportionality coefficient between components of bivectors, corresponding to the same area, is restricted by $\la^2=1\ \Rightarrow \ \la=\pm1$. The third condition~\ref{prop:bivecctor-3} of the same orientation selects the plus sign $\la=+1$, hence \emph{the two bivectors are equal iff their $m(m-1)/2$ coordinates $\Si$ coincide}. Due to this identification, we will often refer to the bivector and its associated plane somewhat interchangeably, when only the direction but not the magnitude is concerned.

An alternative derivation is possible that does not invoke the use of scalar product: the condition expressing how the vector $\mb{Z}$ may be parallel to the plane of $\{\mb{X},\mb{Y}\}$  -- i.e. linear dependent -- is concisely written as
\begin{equation}\label{eq:bivector-plane}
\mb{\Si}\wedge\mb{Z} \ = \ \frac{1}{3!}\Si^{ijk}(\mb{e}_i^{\ph{i}}\wedge\mb{e}_j^{\ph{i}}\wedge\mb{e}_k^{\ph{i}}), \qquad \Si^{ijk} \ \equiv \ \Si^{ij}Z^k+\Si^{jk}Z^i+\Si^{ki}Z^j \ = \ 0,
\end{equation}
for all $i,j,k$. This represents the vanishing of rank-3 minors of the  $m\times 3$-matrix $(\mb{X},\mb{Y},\mb{Z})$. In result, for two such bivectors satisfying~\eqref{eq:bivector-plane} the relationship follows for any subset of 3 indexes:
\begin{equation}\label{eq:bivector-proportion}
\frac{\tilde{\Si}^{ij}}{\Si^{ij}} \ = \ \frac{\tilde{\Si}^{jk}}{\Si^{jk}} \ = \ \frac{\tilde{\Si}^{ki}}{\Si^{ki}} \ \equiv \ \la,
\end{equation}
expressing the fact that the ratio is constant for the areas of two parallelograms obtained by projecting onto the respective elementary 2-planes. The bivectors will be equal if $\la=1$, as is evident by placing one of the $\mb{e}_i^{\ph{i}}\wedge\mb{e}_j^{\ph{i}}$ in the plane of $\mb{\Si}$. 

For arbitrary $\mb{Z}$ (neither parallel, nor orthogonal to $\mb{\Si}$), the geometric meaning of the inner product~\eqref{eq:inner-product} is revealed by choosing the rectangular frame with the first two axes in the $\mb{\Si}$-plane, so that in components $\mb{U}=\Si^{12}(-Z_2,Z_1,0,...,0)$. The resulting vector is obtained from the orthogonal projection of $\mb{Z}$ by rotating it through $\pi/2$ in the anti-clockwise direction parallel to the plane of the bivector, and multiplied by its magnitude. One obtains the `incidence angle' $\varsigma\equiv\widehat{\mb{\Si}\mb{Z}}$ between plane and $\mb{Z}$-direction as follows: $\cos\varsigma=|\mb{U}|/(|\mb{\Si}||\mb{Z}|)$. 

The bilinear form on $V$ naturally extends to the (irreducible) representation in the ${m(m-1)/2}$-dimensional space of (general) bivectors $\bigwedge^2 V$ via repeated use of the inner multiplication/metric contraction: 
\begin{align}\label{eq:bivectors-sp-1}
\cg(\mb{e}_i^{\ph{i}}\wedge\mb{e}_j^{\ph{i}},\mb{e}_k^{\ph{i}}\wedge\mb{e}_l^{\ph{i}})\ : & = \ \iota_{\mb{e}_j^{\ph{i}}}\circ\iota_{\mb{e}_i^{\ph{i}}}(\mb{e}_k^{\ph{i}}\wedge\mb{e}_l^{\ph{i}}) \ = \ \iota_{\mb{e}_j^{\ph{i}}}(\bra\mb{e}_i^{\ph{i}},\mb{e}_k^{\ph{i}}\ket \mb{e}_l^{\ph{i}}- \bra\mb{e}_i^{\ph{i}},\mb{e}_l^{\ph{i}}\ket \mb{e}_k^{\ph{i}}) \nonumber\\
& = \ \bra\mb{e}_i^{\ph{i}},\mb{e}_k^{\ph{i}}\ket\bra\mb{e}_j^{\ph{i}},\mb{e}_l^{\ph{i}}\ket - \bra\mb{e}_i^{\ph{i}},\mb{e}_l^{\ph{i}}\ket\bra\mb{e}_j^{\ph{i}},\mb{e}_k^{\ph{i}}\ket \ = \ \begin{vmatrix}
\cg_{ik} & \cg_{il} \\
\cg_{jk} & \cg_{jl} 
\end{vmatrix} \ \equiv \ \cg_{[ij][kl]},
\end{align}
the latter being symmetric $\cg_{[ij][kl]}=\cg_{[kl][ij]}$. The scalar product in terms of components~\eqref{eq:bivector-simple} is
\begin{equation}\label{eq:bivectors-sp-2}
\bra\mb{\Si},\mb{\Xi}\ket \ \equiv \ \cg(\mb{\Si},\mb{\Xi}) \ = \ \frac14\cg_{[ij][kl]}\Si^{ij}\Xi^{kl} \ = \ \frac12 \Si_{ij}\Xi^{ij}.
\end{equation}
One may obtain the geometric reading of contraction~\eqref{eq:bivectors-sp-2} of (simple) bivectors by choosing the first two rectangular axes in the plane of $\mb{\Si}$, thus $\bra\mb{\Si},\mb{\Xi}\ket=\Si_{12}\Xi_{12}$ equals the product of the (full) magnitude of $|\mb{\Si}|$ and that of the orthogonal projection of $\mb{\Xi}$ onto the first plane. It is zero when there is a direction parallel to the plane of one of the bivectors and perpendicular to the other. In $m=3$, the dihedral angle between two planes is deduced as $\cos\al=\bra\mb{\Si},\mb{\Xi}\ket/\sqrt{|\mb{\Si}|^2|\mb{\Xi}|^2}$.

The above constructions are straightforwardly generalized to arbitrary dimension multivectors.
\begin{definition}\label{def:multivector-simple}
A \textbf{(simple) multivector} of degree $k$, or $k$-vector, is the configuration formed by the set of $k$ vectors $\{\mb{X}_1,...,\mb{X}_k\}$ arranged in a certain order. Two $k$-vectors are said to be equal iff 
\begin{enumerate}[label={\upshape(\roman*)}, align=left, widest=iii]
\item their linear span defines the same (or parallel) $k$-plane (same `position', or `posture'), \label{prop:k-vector-1}
\item they possess the same magnitude, given by the volume that is contained in the $k$-dim parallelotope built on the elements of $k$-vector, and \label{prop:k-vector-2}
\item the orientations coincide. \label{prop:k-vector-3}
\end{enumerate}
\end{definition} 

Let $I_k=\{i_1,...,i_k\}$ shortly denote the sequence of indices, s.t. $1\leq i_1<...<i_k\leq m$ and $\mb{e}_{I_k}^{\ph{i}}\equiv \mb{e}_{i_1}^{\ph{i}}\wedge\cdots\wedge\mb{e}_{i_k}^{\ph{i}}$ be the basis elements, spanning the space $\bigwedge^k V$ of dimension ${m\choose k}=\frac{m!}{k!(m-k)!}$. The considerations similar to the above identify as equal the two $k$-vectors with the same coordinates:
\begin{equation}\label{eq:k-vector-simple}
\mb{\Si}^{(k)} \ \equiv \ \mb{X}_1\wedge\cdots\wedge\mb{X}_k \ = \ \frac{1}{k!}\Si^{I_k}\mb{e}_{I_k}^{\ph{i}}, \qquad \Si^{I_k} \ \equiv \ \Si^{i_1...i_k} \ = \ \begin{vmatrix}
X_1^{i_1} & X_2^{i_1} & \cdots & X_k^{i_1}\\
X_1^{i_2} & X_2^{i_2} & \cdots & X_k^{i_2}\\
\vdots & \vdots & \ddots & \vdots\\
X_1^{i_k} & X_2^{i_k} & \cdots & X_k^{i_k}
\end{vmatrix}.
\end{equation}
An example of trivector we have encountered in~\eqref{eq:bivector-plane}. The scalar product (and the squared magnitude, consequently) of two multivectors reads:
\begin{equation}\label{eq:k-vector-sp-1}
\bra\mb{\Si}^{(k)},\mb{\Xi}^{(k)}\ket \ = \ \frac{1}{k!} \Si_{I_k}\Xi^{I_k} \ \equiv \ \frac{1}{k!}\Si_{i_1...i_k}\Xi^{i_1...i_k},
\end{equation}
using the evaluation on elementary $k$-vectors:
\begin{equation}\label{eeq:k-vector-sp-2}
\cg_{I_kJ_k}^{\ph{i}} \ \equiv \ \cg(\mb{e}_{I_k}^{\ph{i}},\mb{e}_{J_k}^{\ph{i}}) \ = \ \begin{vmatrix}
\cg_{i_1j_1} & \cg_{i_1j_2} & \cdots & \cg_{i_1j_k}\\
\cg_{i_2j_1} & \cg_{i_2j_2} & \cdots & \cg_{i_2j_k}\\
\vdots & \vdots & \ddots & \vdots\\
\cg_{i_kj_1} & \cg_{i_kj_2} & \cdots & \cg_{i_kj_k}
\end{vmatrix}.
\end{equation}
The latter is the Gram determinant $\bra\mb{e}_{I_k}^{\ph{i}},\mb{e}_{J_k}^{\ph{i}}\ket =\det\bra\mb{e}_i^{\ph{i}},\mb{e}_j^{\ph{i}}\ket$. In particular, if $k=m=\dim V$, it gives the squared volume of the parallelepiped constructed on the unit coordinate vectors $|\mb{e}_{I_m}^{\ph{i}}|^2=\det\cg$. Since covariant components are related as $\Si_{1...m}=\det\cg\,\Si^{1...m}$, one obtains $|\mb{\Si}|=(\Si_{1...m}\Si^{1...m})^{1/2}=\sqrt{|\cg|}\det(\mb{X}_1,...,\mb{X}_m)$ for the (volume) magnitude of arbitrary $m$-vector. Indeed, the ratio will not change if we add to $\mb{X}_m$ the linear combination of $\{\mb{X}_{i\neq m}\}$, spanning the ``base'' hyperplane, and the former may be aligned with $\mb{e}_m^{\ph{i}}$, since the ``height'' $(\mb{X}_k)_\perp$ stays the same. Iteration of this procedure results in the proportinality $\Si^{1...m}/|\mb{\Si}|=(\prod_{i=1}^m X_i^i)/|\mb{\Si}|=\cg^{-1/2}$, telling us that the volumes are measured in terms of fractions of the chosen elementary one. Maybe more familiar is the dual notion of the volume element as $m$-form $\mu:=\hat{\mb{e}}'^{I_m}\equiv\hat{\mb{e}}'^1\wedge\cdots\wedge\hat{\mb{e}}'^m=|\cg|^{1/2}\hat{\mb{e}}^1\wedge\cdots\wedge\hat{\mb{e}}^m\equiv|\det A|^{-1}\hat{\mb{e}}^{I_m}\in\La^m(V)$, where the oblique and orthonormal dual bases are related inversely to that of~\eqref{eq:linear-group}, \eqref{eq:orthonormal}. The basis is called positively-oriented if $\mu(\mb{e}_1^{\ph{i}},...,\mb{e}_m^{\ph{i}})>0$~\footnote{The operations of exterior and interior product can be unified within the framework of `geometric', or Clifford algebra with the combined multiplication: $ab=\frac12(ab+ba)+\frac12(ab-ba)=\cg(a,b)+a\wedge b$; the first part of this relation is familiar from the anti-commutation law of Dirac matrices $\ga_i\ga_j+\ga_j\ga_i=2\cg_{ij}\mathbbm{1}$.}.


\subsection{General (systems of) $k$-vectors and simplicity} 
\label{subsec:general-multivector}

Not every skew-symmetric tensor is a simple $k$-vector obtained from the wedge product of $k$ vectors, also known as `blade'. A \emph{system of multivectors} is a set consisting of any number of blades. Such general $\mb{\Si}$ is completely determined by the coordinates obtained as the algebraic sums of its simple components; the scalar product extends by linearity. In other words, any (general) multivector can be expanded into sum of $\mb{e}_{I_k}^{\ph{i}}$, parallel to the $k$-planes formed by coordinate axes; however, only for the simple $\mb{\Si}^{(k)}$ the axes can be chosen in such a way that they are strictly proportional $\mb{\Si}^{(k)}\propto\mb{e}_{I_k}^{\ph{i}}$. Lets examine the case of bivector.

Consider any three equations of the system~\eqref{eq:bivector-plane} of linear constraints on $\mb{Z}$:
\begin{equation}\label{eq:Pluecker-system}
\Si^{ikl} \ = \ 0, \qquad \Si^{jkl} \ = \ 0, \qquad \Si^{ijl} \ = \ 0.
\end{equation}
Excluding $Z^k$ from the first two: $\Si^{ikl}\Si^{lj}-\Si^{jkl}\Si^{li}=0$, and using the third of~\eqref{eq:Pluecker-system}, one obtains that a simple bivector has to satisfy all the (Pl\"{u}cker) relations of the form
\begin{equation}\label{eq:Pluecker-bivector}
\Si^{ij}\Si^{kl}+\Si^{jk}\Si^{il}+\Si^{ki}\Si^{jl} \ = \ 0,
\end{equation}
so that the the system~\eqref{eq:bivector-plane} is consistent and defines the vector $\mb{Z}$ belonging to the plane. In~$m=4$, this gives the unique condition: $\Si^{12}\Si^{34}+\Si^{23}\Si^{14}+\Si^{31}\Si^{24}=0$ (Klein quadric).

In $m=3$, every bivector is simple. Taking the two independent solutions $\mb{Z}=\mb{X},\mb{Y}$ of $\Si^{12}Z^3+\Si^{23}Z^1+\Si^{31}Z^2 = 0$, one can straight away deduce that (up to overall normalization)
\begin{equation}\label{eq:bivector-3d}
\frac{X^1Y^2-Y^1X^2}{\Si^{12}} \ = \ \frac{X^2Y^3-Y^2X^3}{\Si^{23}} \ = \ \frac{X^3Y^1-Y^3X^1}{\Si^{31}} \ \equiv \ \la.
\end{equation}

From the group theoretic perspective, the set of all orthonormal $k$-frames in $V\cong\mathbb{R}^m$ (of Euclidean signature, for simplicity) forms the so-called \emph{Stiefel} manifold $\mr{St}_k(V)$. It can be viewed as a homogeneous space $\mr{St}_k(V)\cong\mr{O}(m)/\mr{O}(m-k)$, or  $\mr{SO}(m)/\mr{SO}(m-k)$ for $k<m$ (it becomes a principal homogeneous space for $k=m$, with trivial stability sub-group). There is a natural right action (free but not transitive) of the group $\mr{O}(k)$, which rotates a $k$-frame in its (orbit) space of the fiber. The natural projection $\mr{St}_k(V)\ra\mr{Gr}_k(V)$, sending a $k$-frame to the sub-space that it spans, via Th.~\ref{th:PFB-characterization} makes it a principal $\mr{O}(k)$-bundle over the \emph{Grassmanian} space $\mr{Gr}_k(V)\cong\mr{O}(m)/(\mr{O}(m-k)\times\mr{O}(k))$ of $k$-planes in $V$. The linear coordinates of (the image of a Pl\"{u}cker embedding map into a `projective variety' of) an element $\Xi\in\mr{Gr}_k(V)$ relative to a standard basis in $\wedge^k V$ satisfy the general Pl\"{u}cker relations:
\begin{equation}\label{eq:Pluecker-general}
\sum_{l=1}^{k+1}(-1)^l\, \Xi_{i_1...i_{k-1}j_l}^{\ph{i}}\Xi_{j_1...\hat{j}_l...j_{k+1}} \ = \ 0.
\end{equation}

(Note: even if the bivector is simple, the coordinates represent only the equivalence class of Def.~\ref{def:bivector-simple}. The two parallelograms are not at all guaranteed to have the same \emph{shape} -- e.g., $\mb{X}\wedge\mb{Y}=(\mb{X}-\frac12\mb{Y})\wedge(\frac12\mb{X}+\frac34\mb{Y})$. The latter is encoded in the relations between geometric forms of different order, s.a. angle between edge vectors, or their lengths, respectively. This feature seems to be directly linked with the appearance of shape-mismatch in LQG's `twisted geometries' and -- as we reveal in the second part of the thesis -- in Spin Foam amplitudes, since they both utilize face-bivectors as the fundamental variables. They may be enough for triangles though. Intuitively, however, if we want the system of simple triangle/parallelogram bivectors to describe arbitrary polygon, one needs to assure that a) they all lie in the same plane, and, moreover, b) they are composed coherently along the common boundary edges. These tasks invoke other forms, apart from sole bivectors; cf~further discussion in Ch.~\ref{ch:problem}.)

\subsection{Dual, or complementary multivectors}
\label{subsec:dual-multivectors}

The position of $k$-plane can be alternatively described via its orthogonal complement, or the normal direction(s), in other words. Therefore
\begin{definition}\label{eq:complementary-vector}
We shall call the \textbf{complementary multivector} of a given simple $k$-vector $\mb{\Si}=\mb{X}_1\wedge\cdots\wedge\mb{X}_k$, the simple $(m-k)$-vector $\star\mb{\Si}=\mb{Y}_{k+1}\wedge\cdots\wedge\mb{Y}_m$, such that
\begin{enumerate}[label={\upshape(\roman*)}, align=left, widest=iii]
\item it is completely normal to the $k$-plane of $\mb{\Si}$, that is to say $\{\mb{X}_i\}\perp\{\mb{Y}_j\}$,\label{prop:dual-vector-1}
\item both multivectors $\mb{\Si}$ and $\star\mb{\Si}$ have the same magnitude, and \label{prop:dual-vector-2}
\item the $m$-vector $\mb{\Si}\wedge\star\mb{\Si}$ is positively-oriented. \label{prop:dual-vector-3}
\end{enumerate}
(It is assumed that the space itself is oriented, and the basis unit vectors are numbered accordingly.)
\end{definition} 

Consider first the simple basis vector $\mb{e}_{I_k}^{\ph{i}}\equiv \mb{e}_{i_1}^{\ph{i}}\wedge\cdots\wedge\mb{e}_{i_k}^{\ph{i}}$ that is normalised $|\mb{e}_{I_k}^{\ph{i}}|^2=1$. From the orthogonality requirement~\ref{prop:dual-vector-1} $\bra\mb{e}_i^{\ph{i}},\mb{Y}\ket =Y_i=0$, for all $i\in I_k$, it follows that in the coordinate expansion of the complementary $\star\mb{e}_{I_k}^{\ph{i}}$ into dual basis only one term survives, with non-trivial projection:
\begin{equation}\label{eq:compl-projection}
\bra \star (\mb{e}_{i_1}^{\ph{i}}\wedge\cdots\wedge \mb{e}_{i_k}^{\ph{i}}),\mb{e}_{i_{k+1}}^{\ph{i}}\wedge\cdots\wedge \mb{e}_{i_m}^{\ph{i}}\ket \ = \ \la\,  \varepsilon_{i_1...i_ki_{k+1}...i_m}^{\ph{i}},
\end{equation}
where $\varepsilon_{i_1...i_m}^{\ph{i}}\equiv\de_{i_1...i_m}^{1...m}$ is the usual totally skew-symmetric Levi-Civita symbol. Hence
\begin{align}\label{eq:compl-basis}
\star (\mb{e}_{i_1}^{\ph{i}}\wedge\cdots\wedge \mb{e}_{i_k}^{\ph{i}}) \ & = \ \frac{\la}{(m-k)!}\, \varepsilon_{i_1...i_kj_{k+1}...j_m}^{\ph{i}}(\hat{\mb{e}}^{j_{k+1}}\wedge\cdots\wedge\hat{\mb{e}}^{j_m}) \nonumber\\
& = \ \frac{\la}{(m-k)!}\, \varepsilon_{i_1...i_k}^{\ph{i_1...i_k}j_{k+1}...j_m}(\mb{e}_{j_{k+1}}^{\ph{i}}\wedge\cdots\wedge\mb{e}_{j_m}^{\ph{i}}),
\end{align}
where the indices could be raised or lowered using the metric-induced isomorphism between $V$ and its dual. In order to determine the proportionality coefficient, choose the lexicographic order of indices, and fulfil to the $m$-vector: $\mb{\Si}^{(m)}\equiv\mb{e}_{I_k}^{\ph{i}}\wedge\star\mb{e}_{I_k}^{\ph{i}}=\la\,\mb{e}_1^{\ph{i}}\wedge\cdots\wedge\mb{e}_m^{\ph{i}}$. From conditions~\ref{prop:dual-vector-2}-\ref{prop:dual-vector-3} $\Rightarrow$ $\la=\sqrt{|\cg|}$, e.g. by applying to $k=m$ case: $\cg^2=|\mb{e}_{I_m}^{\ph{i}}|^2|\star\mb{e}_{I_m}^{\ph{i}}|^2 = |\mb{\Si}^{(m)}|^2=\la^2\cg$. 

The complementary multivector of a general (not necessary simple) $k$-vector is defined as the set of complements to the simple $k$-vectors of which the given $\mb{\Si}^{(k)}$ is the sum. That is, one extends $\star$ by linearity and obtains the general relation between components of $k$-vector and its complement: 
\begin{gather}
\star\mb{\Si}^{(k)} \ = \ \frac{1}{(m-k)!}(\star\Si)^{j_1...j_{m-k}}(\mb{e}_{j_1}^{\ph{i}}\wedge\cdots\wedge\mb{e}_{j_{m-k}}^{\ph{i}}), \nonumber\\
\begin{aligned}\label{eq:compl-components}
\text{where}\quad (\star\Si)^{j_{k+1}...j_m} \ & = \ \frac{|\cg|^{-1/2}}{k!}\,\Si_{i_1...i_k}\varepsilon^{i_1...i_kj_{k+1}...j_m}, 
\\ \text{or} \quad (\star\Si)_{j_{k+1}...j_m}^{\ph{i}} \ & = \ \frac{|\cg|^{1/2}}{k!}\,\Si^{i_1...i_k}\varepsilon_{i_1...i_kj_{k+1}...j_m}^{\ph{i}};
\end{aligned}
\end{gather}
the last covariant version represents the induced action on the dual bases of forms. Applying it twice leaves a $k$-vector unchanged, up to sign $\star^2=(-1)^\eta(-1)^{k(m-k)}$, where $(-1)^\eta\equiv\det\eta$ is the metric signature (difference).

The familiar reader will recognize the standard dualization by the Hodge star operator. In the language of forms, it establishes the linear isomorphism $\star:\La^k(V)\ra\La^{m-k}(V)$ between vector spaces of the same dimensionality ${m\choose k}={m\choose m-k}$. This may be determined through the inner product via the following general relation:
\begin{equation}\label{eq:Hodge-dual}
\al\wedge\star\be \ = \ \bra\al,\be\ket \, \mu\, ;
\end{equation}
however, we introduced it in a decidedly geometric fashion. The crucial role played by the metric is evident from~\eqref{eq:compl-components}, since the coordinates of $\star\mb{\Si}$ are basically composed of the respective orthogonal projections of $\mb{\Si}$. Notice that the Pl\"{u}cker relations~\eqref{eq:Pluecker-bivector} for bivectors (simplicity) are totally skew-symmetric and can thus be rewritten in the form of inner product with the dual $\Si^{ij}(\star\Si)_{ijk_{1}...k_{m-4}}^{\ph{i}}=0$; in $m=4$, this is the full scalar product contraction $\bra\mb{\Si},\star\mb{\Si}\ket=0$.

The duality operator satisfies number of properties. Most importantly, it commutes with the action of the group, e.g. for the bivector: $[\mb{A},\star\mb{\Si}]=\star[\mb{A},\mb{\Si}]$, where $\mb{A}\in\mf{h}=\mf{o}(\eta)$. We will also make use of the following relations (cf.~\cite{Trautman1973EC-structure}, for example):
\begin{subequations}\label{eq:Hodge-relations}
\begin{align}
\hat{\mb{e}}^i\wedge \star\mb{e}_j^{\ph{i}} \ = & \ \de^i_j \star 1, \label{eq:Hodge-1}\\
\hat{\mb{e}}^i\wedge \star(\mb{e}_j^{\ph{i}}\wedge\mb{e}_k^{\ph{i}}) \ = & \ \de^i_k \star \mb{e}_j^{\ph{i}}-\de^i_j \star \mb{e}_k^{\ph{i}}, \label{eq:Hodge-2}\\
\hat{\mb{e}}^i\wedge \star(\mb{e}_j^{\ph{i}}\wedge\mb{e}_k^{\ph{i}}\wedge\mb{e}_l^{\ph{i}}) \ = & \ \de^i_l \star(\mb{e}_j^{\ph{i}}\wedge\mb{e}_k^{\ph{i}})
+\de^i_k \star(\mb{e}_l^{\ph{i}}\wedge\mb{e}_j^{\ph{i}})
+\de^i_j \star(\mb{e}_k^{\ph{i}}\wedge\mb{e}_l^{\ph{i}}), \label{eq:Hodge-3}\\
\hat{\mb{e}}^i\wedge \star(\mb{e}_j^{\ph{i}}\wedge\mb{e}_k^{\ph{i}}\wedge\mb{e}_l^{\ph{i}}\wedge \mb{e}_m^{\ph{i}}) \ = & \ \de^i_m \star(\mb{e}_j^{\ph{i}}\wedge\mb{e}_k^{\ph{i}}\wedge \mb{e}_l^{\ph{i}})
-\de^i_l \star(\mb{e}_m^{\ph{i}}\wedge\mb{e}_j^{\ph{i}}\wedge \mb{e}_k^{\ph{i}}) \nonumber \\
+ & \
\de^i_k \star(\mb{e}_l^{\ph{i}}\wedge\mb{e}_m^{\ph{i}}\wedge \mb{e}_j^{\ph{i}})
-\de^i_j \star(\mb{e}_k^{\ph{i}}\wedge\mb{e}_l^{\ph{i}}\wedge \mb{e}_m^{\ph{i}}), \qquad \text{etc.}\label{eq:Hodge-4}
\end{align}
\end{subequations}
Finally, the useful identity gives the contraction of Levi-Civita symbols:
\begin{equation}\label{eq:Levi-Civita-contr}
\varepsilon^{i_1...i_nk_{n+1}...k_m}\,\varepsilon_{j_1...j_nk_{n+1}...k_m}^{\ph{i}} \ = \ (-1)^\eta \, m!(m-n)! \,\de^{[i_1}_{\, j_1}\cdots\de^{i_n]}_{\, j_n}.
\end{equation}

(Note: the Hodge dual becomes a little more ``tricky'', when extended to arbitrary manifolds. Introducing the ``natural'' coordinates $\{x^\mu\}$ in the local patches, one could encounter superficially two types of metric inner product: $\cg_{\mu\nu}=\cg(\pa_\mu,\pa_\nu)$ (``external'') and $\cg_{ij}$ (``internal'') w.r.t. non-holonomic basis, where it can be put constant $\eta$ by choosing rectangular axes (cf.~\cite{Peldan1993actions-GR}, for instance, and many more). Then the two seemingly different types of star-dualization may be considered, correspondingly, induced by these two options for $\cg$. We stress to evade such idiosyncrasy, since the soldering is available to relate various versions of $\cg$. If written solely in terms of forms, the `non-polynomiality' issues and `degenerate' metrics fade away (cf.~\cite[p.307]{Bleecker1990ECSK}), and one also avoids pitfalls, for instance: if applied to forms on $\mc{M}$, $\star$ changes their degree and do not commute with $d$; whereas in the space of vector values (`internal'), this just corresponds to the change of representation, and behaves perfectly fine w.r.t. taking derivatives.)

\subsection{Sliding, or bound $k$-vectors}
\label{subsec:sliding-multivector}

Finally, recall that \emph{free} vectors are distinguished from \emph{sliding} vectors, taking into account their position. This can be extended from the `geometric forms of order one' to higher dimensions -- by dropping the `parallel' requirement in the definition of a multivector configuration. Precisely, a (simple) \textbf{sliding $k$-vector} is formed by $k$ vectors \emph{situated in the same $k$-plane}, and equality for different multivectors holds only under this exact condition. 

The sliding $k$-vector can be regarded as an exterior product 
\begin{equation}\label{eq:sliding-multivector}
\mb{\Si}_{0}^{(k)} \ := \ \mb{X}_0\wedge \mb{X}_1\wedge\cdots\wedge\mb{X}_k;
\end{equation}
where we remind that $\mb{X}_0^{\ph{i}}=\mb{e}_0^{\ph{i}}+\mb{x}$ represents the point-vector with components $(1,x^i)$, in the linear $\mathbb{R}^{m+1}$ representation of the affine group $G$. One can expand into elementary basis vectors
\begin{gather}\label{eq:sliding-coordinates}
\mb{\Si}_{0}^{(k)} \ = \ \frac{1}{(k+1)!}\Si^{i_0i_1...i_k}\mb{e}_{i_0}^{\ph{i}}\wedge\mb{e}_{i_1}^{\ph{i}}\wedge\cdots\wedge\mb{e}_{i_k}^{\ph{i}}, \nonumber\\
\text{where} \quad \Si^{i_0i_1...i_k} \ \equiv \ \begin{vmatrix}
1 & 0 & 0 & \cdots & 0\\
x^{i_1} & X_1^{i_1} & X_2^{i_1} & \cdots & X_k^{i_1}\\
x^{i_2} & X_1^{i_2} & X_2^{i_2} & \cdots & X_k^{i_2}\\
\vdots & \vdots & \vdots & \ddots & \vdots\\
x^{i_k} & X_1^{i_k} & X_2^{i_k} & \cdots & X_k^{i_k}
\end{vmatrix}
\end{gather}
-- coordinates of the simple sliding $k$-vector, whose plane is passing through $\mb{x}$. The equation of such $k$-plane can be obtained by putting to zero all rank-$(k+2)$ minors of the matrix $(\mb{Z}_0,\mb{X}_0,\mb{X}_1,...,\mb{X}_k)$, where $\mb{Z}_0$ is the variable point on the plane. One obtains the relations akin to~\eqref{eq:bivector-plane} for (free) bivector. A `system of sliding $k$-vectors' is similarly defined.

Take, for instance, sliding vector obtained from the particle (of unit mass) with `radius-vector' $\mb{x}$ and velocity parallel to $\mb{X}$:
\begin{equation}\label{eq:sliding-vector}
\mb{\Si}_{0}^{(1)} \ = \ \mb{X}_0\wedge \mb{X} \ = \ X^i(\mb{e}_0^{\ph{i}}\wedge\mb{e}_i^{\ph{i}}) + \frac12(x^iX^j-x^jX^i)\, \mb{e}_i^{\ph{i}}\wedge\mb{e}_j^{\ph{i}}.
\end{equation}
One recovers within a unified framework the vector of momentum (tied to a trajectory line) and angular momentum (free) bivector, respectively. The sliding vector, whose origin lies at $\mb{x}$ and extremity at $\mb{x}'$, will be represented by $\mb{X}_0\wedge\mb{X}_0'=(x'^i-x^i)\, \mb{e}_0^{\ph{i}}\wedge\mb{e}_i^{\ph{i}} + \frac12(x^ix'^j-x^jx'^i)\,\mb{e}_i^{\ph{i}}\wedge\mb{e}_j^{\ph{i}}$ -- which describes the same thing as~\eqref{eq:sliding-vector}, taking into account $\mb{X}=\mb{X}_0'-\mb{X}_0$.

\newpage

\section{Geometric theory of connections}
\label{sec:Cartan-affine}

The \emph{homogeneous} Klein geometry of affine space admits the existence of globally defined Cartesian frames -- filling the entire space -- equivalent in an ordinary sense that their axes are parallel and origins are motionless relative to each other (i.e. the coordinates are related simply as $x'^i=x^i+a^i$, with $a^i=\mr{const}$). From more `active' perspective: the two regions can be made to coincide by the affine transformation (with constant coefficients), s.t. the point set correspondence extends to the entire space as~\eqref{eq:translation-global}. 

In Cartan's approach to Riemannian geometry (and its generalizations), the simple logic and intuition familiar from the ordinary (non-curved, model) geometry is pushed to its limit, and preserved wherever it is possible also in the non-homogeneous (curved) situation. This is achieved by first localizing the flat-space description (attach it at the point and solder the tangent spaces in the linear approximation), and then the gluing itself is also formulated in local terms of equivalence via generalized parallelism between nearby tangent spaces. \blockquote{In a masterly fashion, [Cartan] employs a Euclidean osculating space that allows an almost automatic transfer of geometric properties of curves in Euclidean space to those in a Riemannian manifold}~\cite[p.ix]{Cartan2001Orthogonal-frame}. Using his method of moving frames, many of the properties and results, which may be obtained by purely analytical tools, receive a natural geometrical interpretation. 

It is quite intricate to make the Cartan intuition precise. In Sec.~\ref{sec:Klein-bundles} we make an extensive use of the theory of $G$-structures in a principal bundle context of gauge theory, defining what we called the `Klein bundle' as the patchwork of regions which infinitesimally ``look like'' the homogeneous geometry. The Cartesian frames are not so exclusive, but any invariantly related $x'=f(x)$ (curvilinear) coordinate systems are allowed. In the bundle framework, their relation is established by means of the gauge transformation, s.t. the frame's axes can be put into motion and rotated w.r.t. infinitesimally nearby frame. Whereas on the `active' side: given the map from the (abstract) point set $Q_0: Q\hookrightarrow \mc{M}\cong P/H$, describing a configuration of the `body' that is `rigid' and `sufficiently complicated'~\footnote{By the term `rigid' is meant that for any two of configurations, there is an element $g\in G$ carrying $Q_0$ to $Q_1$ (transitivity). `Sufficiently complicated' implies that the stabilizer sub-group that fixes the body point-wise is $e\in G$ (free action). This is what Cartan originally had in mind when using the term ``moving frame'' (cf.~~\cite[p. 164]{Sharpe1997Diff-Geometry-Cartan}).}, the `deformation' is defined as the transformation of a body from a reference configuration to a current configuration.

In Sec~\ref{sec:Cartan-connections} we provide the necessary mathematical background and the general definition of `Cartan connection' in terms of infinitesimal parallelism, following the example of canonical Maurer-Cartan trivialization of the tangent bundle. After the brief overview in Sec~\ref{sec:soldering} of different approaches to the soldering and Cartan connection, in Sec~\ref{subsec:structure} we formulate the equation, defining the `osculation' between two geometries (at the point). This actually embodies the geometric picture of ``rolling'' one homogeneous space/plane on the ``lumpy'' surface of the curved manifold [cf.~Fig.~\ref{fig:hamster-ball}]. Supplied with the integrability condition of the Cartan structure equations, it fully determines the homogeneous geometry and its trivial connection. The key notion of `development', encompassing and generalizing the parallel transport, is elaborated in Sec.~\ref{subsec:parallel-transport}, and it is shown to possess the properties of a functor. In Sec~\ref{subsec:universal-der}, the covariant derivative is introduced, and one demonstrates -- using the advanced symmetry of the Klein bundle -- that it is actually the same as the Lie dragging (in the course of development). The latter directly addresses the relation between diffeomorphism transformations and local translations of the Poincar\'{e} gauge theory.

We observe two physical implications of the mathematical formalism of connection theory. In the narrative of the ``local-vs-global symmetry'' one recognizes, essentially, the gauging argument of Yang-Mills/Utiyama theories that is made tangible. On the side of relativity theory, one finds the most straightforward implementation of the Einstein's principle of equivalence between inertia/acceleration and gravitation~\cite{Cartan1986Affine-connections}. Regarding our main question, posed in introduction, one can clearly draw the line between the spurious degrees of freedom, gauge-related to the canonical/or absolute structure, and the genuine local excitations of the `physical' gravitational field. 

\subsection{Locally Klein bundles and their gauge symmetries}
\label{sec:Klein-bundles}


At our first encounter with PFB of linear frames~\eqref{eq:frame-bundle}, we agreed for a while not to draw distinction between the actual base points $\bar{u}\in \bar{U}\subset\mc{M}$ and their coordinate labels $\varphi(\bar{u})=x$ in~$\mathbb{E}^m$. Let us now delve into this subtle issue. The PFB is always a collection of admissible frames, whose relation to the group is that of a torsor~\cite{Baez2010torsors}: once the specific element is picked up, they are isomorphic (non-canonically). The gauge does perform this (arbitrary) choice for you, so that the fiber $P_{\bar{u}}=(\ker\pi)_{\bar{u}}$ at the point can be identified with the group $H$ of all admissible frame transformations. Among three possible incarnations of the gauge choice -- section $\tilde{\si}$ over $\bar{U}$ / bundle chart's local trivialization $\phi$ / and the right $H$ equivariant map $\si:P|_{\bar{U}}\ra H$, s.t. $\si(pA)=\si(p)A$ -- we prefer the latter option. 

As for the base manifold, it is covered by its own coordinate charts $\bar{\varphi}:\bar{U}\ra\ms{O}=\bar{\varphi}(U)\subset\mathbb{E}^m$, granted without specific reference to any bundle structure on top of it. The only requirement is that the transition maps between overlapping charts $\bar{\varphi}_i\circ \bar{\varphi}_j^{-1}:\bar{\varphi}_j(\bar{U}_i\cap \bar{U}_j)\ra \bar{\varphi}_i(\bar{U}_i\cap \bar{U}_j)$ be sufficiently smooth (in order to do calculus; cf.~Fig.~\ref{fig:manifold}). The two parts then become tangled, when we strive to endow locally affine $\mc{M}$ with some additional (metric) structure; also the perception of (abstract) $V$ as being tangential to $\mc{M}$ usually requires the construction of linear frame bundle from~\eqref{eq:frame-bundle}.

Lets propose some possible ramification here, via combination of two types of charts. Namely, the trivialization $\phi:p\mapsto (\pi(p),\si(p))$ may be composed with $(\bar{\varphi},\mr{Id})$ to obtain the combined map $\ka\equiv(\varphi,\si)$ -- with the pull-back $\varphi:=\pi^\ast\bar{\varphi}=\bar{\varphi}\circ\pi$ -- making the following diagram commute:
\begin{displaymath}
\begin{tikzcd}[column sep=small]
U\equiv P|_{\bar{U}}=\pi^{-1}(\bar{U})  \arrow[d,swap, "\pi "] \arrow[rr, "\ka"] & & \ms{O}\times H \arrow[d,"\mr{proj}_1"] \arrow[rr, leftrightarrow] & & G \arrow[d,"\pi_G"] \\
\bar{U}  \arrow[rr, "\bar{\varphi}"] & & \ms{O}\subset\mathbb{E}^m \arrow[rr, leftrightarrow] & & G/H 
\end{tikzcd}.
\end{displaymath}
The horizontal arrows are (local) diffeomorphisms, with the appropriate restraints imposed on the group parameters in the right column, to match the local regions in the left. We used the identification with the Klein geometry, discussed in Sec.~\ref{subsec:group-theory}. The standard projection $\pi_G:G\ra G/H$ provides the submersion map of the corresponding principal fiber bundle, \emph{canonically} associated with the base homogeneous space by the global action of the principal group. It is this pair of $(G,H)$ that defines the generic Klein geometry, e.g. in~\cite[Ch.4]{Sharpe1997Diff-Geometry-Cartan}. In analogy to the `base definition of the infinitesimal Cartan geometry' ibid~\cite[Ch.5, \S 1]{Sharpe1997Diff-Geometry-Cartan}, it thus seems to us well-motivated to \ul{suggest the following (more restrictive) bundle construction globally}.

\begin{definition}\label{def:Klein-structure}
Let $P$ be the smooth manifold and the principal $H$-bundle $\pi:P\ra P/H$. We introduce the \textbf{Klein chart} as the local diffeomorphism map $\ka:U\ra \kappa(P)\subset G$ between $U\subset P$ and open sets in $G$, satisfying right-equivariance $\ka(pg)=\ka(p)g$, and restricting to the corresponding diffeomorphism $\bar{U}=\pi(U)\stackrel{\sim}{\ra}\ms{O}=\pi_G\circ\ka(P) \subset G/H$. The opposite parametrization map $\psi=\ka^{-1}: \ka(P)\ra U\subset P$ will be called the respective \textbf{Klein gauge}. Let us introduce the \textbf{Klein $G$ atlas} on $P$ as a collection of charts $\mathpzc{K}=\{U_\al,\ka_\al\}$, covering $P=\bigcup_\al U_\al$, such that the $G$-valued transition functions between overlapping charts 
\begin{equation}\label{eq:Klein-transition}
\ka_{\al\be} \ = \ \ka_\al^{\ph{1}}\ka_\be^{-1} : \, \ka_\be(U_\al\cap U_\be) \ \ra \ \ka_\al(U_\al\cap U_\be),
\end{equation}
satisfy the following co-cycle consistency conditions (direct consequence):
\begin{enumerate}[label={\upshape(\roman*)}, align=left, widest=iii]
\item $\ka_{\al\al}(p)=e$ for all $p\in U_\al$; \label{prop:Klein-transition-1}
\item $\ka_{\be\al}(p)=\ka_{\al\be}(p)^{-1}$ for all $p\in U_\al\cap U_\be$; \label{prop:Klein-transition-2}
\item $\ka_{\al\be}(p)\ka_{\be\ga}(p)\ka_{\ga\al}(p)=e$ for all $p\in U_\al\cap U_\be\cap U_\ga$. \label{prop:Klein-transition-3}
\end{enumerate}
As usual, two Klein atlases are considered equivalent, if their union is also the Klein atlas. Respectively, the \textbf{`locally Klein $G$ bundle'} is $P$ on which the \textbf{`Klein $G$ structure'} is specified as equivalence class of atlases.
\end{definition}

The transition functions describe the details of `gluing' the local patches $\ka_\al(U_\al)$ of canonical bundle of arbitrary homogeneous Klein geometry with specified principal group action. Indeed, take the factorspace of $\bigcup_\al\ka_\al(U_\al)$ by the equivalence relation of~\eqref{eq:Klein-transition}. Essentially, we engraved the form of admissible transformations between different $\bar{\varphi}_\al(\bar{U}_\al)\sim\ms{O}\subset\mathbb{E}^m$ explicitly in the `Klein structure', and their smoothness should be apparent. This is very close to the notion of `affine manifold'~\cite[p. 224]{BishopGoldberg1980tensor-analysis}, and to Cartan's own description (according to~\cite{Marle2014fromCartan-toEhresmann}):
~\blockquote{. . . a manifold whose properties, in the neighborhood of each point, are those of an affine space, and on which there is a law for fitting together the neighborhoods of two infinitesimally nearby points: it means that if, on a neighborhood of each point, we have chosen Cartesian coordinates with that point as origin, we know the transformation formulae (of the same nature as those valid in an affine space) which allow to go from a reference frame to another reference frame with an infinitesimally nearby origin.}
However, we promoted the bundle-theoretic aspects to the forefront, in order to stay in line with the modern differential geometric approach to connections. 

The advantage brought with the new perspective is that now the right action of the principal group $G$ is defined by construction on the whole bundle space, commuting with left coordinate changes. This brings the possibility to extend the usual notion of gauge transformations to embrace all symmetries of $P$: both local Lorentz transformations/rotations (over a fixed point, or ``internal''), as well as the translations of the points themselves, corresponding to (``external'') diffeomorphisms of $\mc{M}\cong P/H$~\footnote{Infinitesimally, this was always available through the (non-canonical) isomorphism of the base-tangent and associated bundles $T\mc{M}=TP/H\approx P\times_H V$, $V\cong\mf{g}/\mf{h}$, provided by the Cartan connection form / or soldering linear isomorphism $\be_p:T_{\pi(p)}\mc{M}\stackrel{\sim}{\ra}V$ (cf.~\eqref{eq:soldering} further, and discussion in \cite{Catren2015Cartan-gauge-gravity}). We just formalized this in the global construction of the bundle space, in accord with our needs.}.

\begin{definition}\label{def:Klein-symmetry}
A (full) \textbf{gauge symmetry} of a locally Klein $G$ bundle $(P,\mathpzc{K})$ is an automorphism $f:P\ra P$ that preserves the Klein structure, such that $f(pg)=f(p)g$. We denote $\mc{G}(P,\mathpzc{K})$ the group of all gauge transformations.
\end{definition}

\begin{proposition}\label{prop:gauge-symmetry}
There is a natural bijective correspondence between the gauge group $\mc{G}(P,\mathpzc{K})$ and the space $C(P,G):=\{\tau:P\ra G|\tau(pg)=g^{-1}\tau(p)g\}$ of all $G$-valued equivariant functions, transforming in the (anti-)adjoint representation (now w.r.t. the full \emph{principal group} of the model Klein geometry). It is given by $f(p)=p\tau(p)$.
\end{proposition}

\begin{proof}
The usual argument extends straightforwardly:
\begin{itemize}
\item[($\Leftarrow$)] If $\tau\in C(P,G)$, then it follows that so-defined $f(pg)=pg\tau(pg)=pgg^{-1}\tau(p)g=p\tau(p)g=f(p)g$ is in $\mc{G}(P,\mathpzc{K})$.
\item[($\Rightarrow$)] If $f\in\mc{G}(P,\mathpzc{K})$, then conversely $pg\tau(pg)=f(pg)=f(p)g=p\tau(p)g$ for any $p\in P$, hence $\tau(pg)=g^{-1}\tau(p)g$ is in $C(P,G)$.
\item[(group)] The homomorphism property is obvious $(f\circ f')(p)=p\tau'(p)\tau(p)$, for two $f,f'\in\mc{G}(P,\mathpzc{K})$.
\end{itemize}
(Compare with Def.~\ref{def:gauge-transform} of the gauge group in Yang-Mills theory, given in Appendix.) 
\end{proof}

\begin{definition}\label{def:gauge-algebra}
We call the \textbf{gauge algebra} $\ms{G}(P,\mathpzc{K})\cong C(P,\mf{g})$ the space of Lie algebra valued functions $\zeta:P\ra\mf{g}$, such that $\zeta(pg)=\mr{Ad}(g^{-1})\zeta(p)$. The multiplication is defined by $[\zeta,\zeta'](p)=[\zeta(p),\zeta'(p)]$, s.t. the map $[\zeta,\zeta'](pg)=[\zeta(pg),\zeta'(pg)]=[\mr{Ad}(g^{-1})\zeta(p),\mr{Ad}(g^{-1})\zeta'(p)]=\mr{Ad}(g^{-1})[\zeta(p),\zeta'(p)]=\mr{Ad}(g^{-1})[\zeta,\zeta'](p)$ is again in $C(P,\mf{g})$, consequently.
\end{definition}

\begin{proposition}\label{prop:symmetry-generator}
There is a map $\mr{Exp}:C(P,\mf{g})\ra C(P,G)$ defined by $\mr{Exp}(\zeta)(p)=\exp(\zeta(p))$, such that $s\ra\mr{Exp}(s\,\zeta)$ is a one-parameter subgroup of $C(P,G)$ with 
\begin{equation}
\frac{d}{ds}\mr{Exp}(s\,\zeta)(p)\bigg|_{s=0} \ = \ \zeta(p).
\end{equation}
Moreover, for two $f,f'\in\mc{G}(P,\mathpzc{K})$ then
\begin{equation}
[\zeta,\zeta'](p) \ = \ \frac{\pa^2}{\pa s\, \pa r}\mr{Exp}(s\,\zeta)_p\mr{Exp}(r\,\zeta')_p\mr{Exp}(s\,\zeta)_p^{-1}\bigg|_{s,r=0}.
\end{equation}
Followed by the isomorphism of Prop.~\ref{prop:gauge-symmetry}, one obtains the map $\exp:\ms{G}(P,\mathpzc{K})\ra\mc{G}(P,\mathpzc{K})$, generating gauge transformations as $\exp(\zeta)(p)=p\exp(\zeta(p))$.
\end{proposition}

\begin{proof} $\mr{Exp}(\zeta)(pg)=\exp(\zeta(pg))=\exp(\mr{Ad}_{g^{-1}}\zeta(p))=\mr{Ad}(g^{-1})\exp(\zeta(p))=\mr{Ad}(g^{-1})\mr{Exp}(\zeta)(p)$, and similarly for the other statements. Note that we are simply reproducing here the results from the standard gauge theory in the context of the wider symmetry of a Klein bundle (since these may be not so widely known; cf.~\cite[\S 3.2]{Bleecker1981gauge-variational-principles}).
\end{proof}

We see no apparent contradictions in the above construction. Moreover, it is completely general w.r.t. the choice of the model homogeneous Klein geometry, and fully in line with the Cartan's twofold generalization of Euclidean geometry~\footnote{Quoting Cartan's own words on the subject~\cite{Cartan1924projective-connections} (in the approximate translation of~\cite{Marle2014fromCartan-toEhresmann}):

\blockquote{The fundamental idea stems from the notion of \emph{parallelism} introduced by M.~T.~Levi-Civita in such a fruitful way. The many authors who generalized the theory of metric spaces all started from the fundamental idea of M.~Levi-Civita, but, seemingly, without freeing it from the idea of \emph{vector}. That does not matter as long as one deals with manifolds with affine connections . . . But that seemed to forbid any hope to build an \emph{autonomous} theory of manifolds with conformal or projective connections. In fact, the main thing in M.~Levi-Civita's idea is that it allows to glue together two small, infinitesimally nearby pieces of a manifold, and it is that idea of gluing which is most fruitful.}} -- the latter is neatly captured in a diagram adapted from R.~W.~Sharpe~\cite{Sharpe1997Diff-Geometry-Cartan}, and nicely visualized in a drawings due to D.~K.~Wise~\cite{Wise2010Cartan-geometry}:

\pagebreak

\begin{figure}[h]
\center{\includegraphics[width=0.4\linewidth]{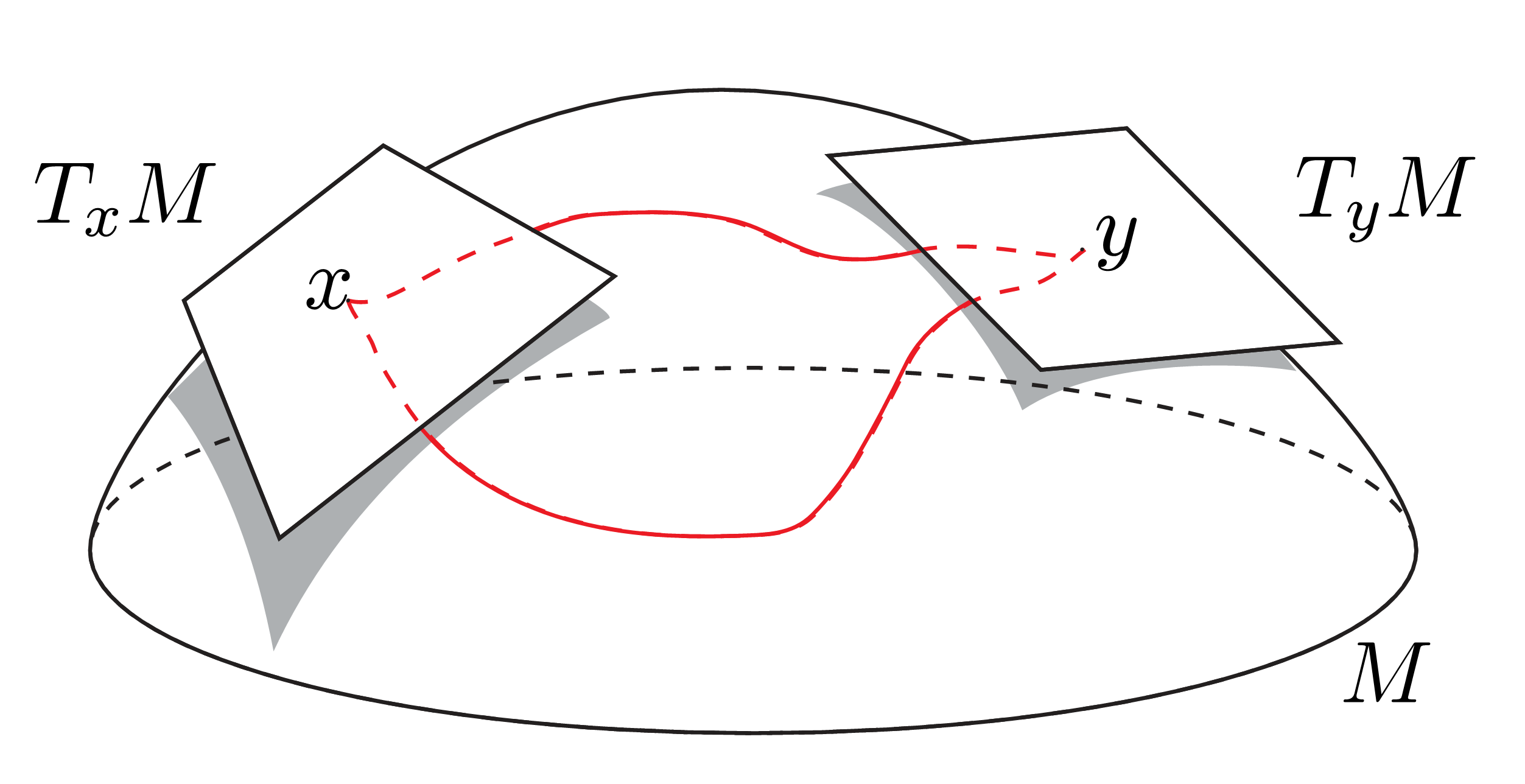}}
\end{figure}

\begin{tikzcd}[column sep = 25 ex, row sep=large]
\begin{tabular}{c}Euclidean\\ geometry\\  \end{tabular} \arrow{r}{allow\ non-homogeneity} \arrow{d}[swap]{\begin{tabular}{c}$generic$\\ $homogeneous\ spaces$\\  \end{tabular}}
& \begin{tabular}{c}Riemannian \\ geometry\\  \end{tabular} \arrow{d}{\begin{tabular}{c}$generic\ osculating$\\ $homogeneous\ spaces$\\  \end{tabular}}\\
\begin{tabular}{c}Klein\\ geometry\\  \end{tabular} \arrow{r}{allow\ non-homogeneity} & \begin{tabular}{c}Cartan\\ geometry\\  \end{tabular}
\end{tikzcd}

\begin{figure}[h]
\center{\includegraphics[width=0.4\linewidth]{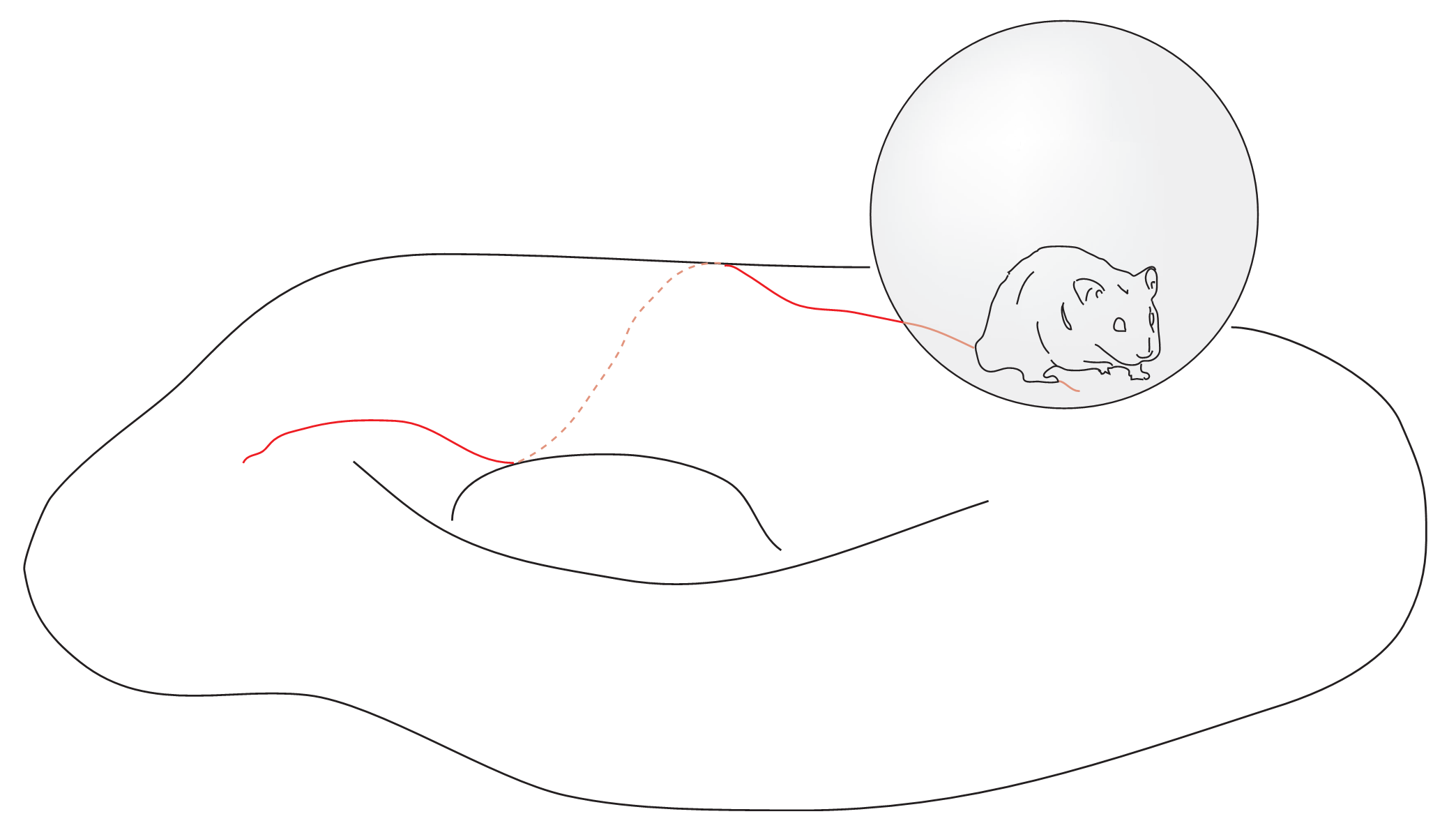}}
\caption{The general idea of Cartan geometry, and its handy visualization.}
\label{fig:hamster-ball}
\end{figure}

\newpage

\subsection{Cartan's affine connections and absolute parallelism}
\label{sec:Cartan-connections}

In order to study the general case of arbitrary geometry distorted by local perturbations (excitations of the field), an additional structure of non-trivial connection has to be specified (determining transition functions that are not constant). But before turning to the general (non-trivial) Cartan connections on $(P,\mathpzc{K})$, one should first figure out the trivial case. Suppose $\mathbb{M}^m\cong G/H$ is an affine (flat) space/time, and consider the principal homogeneous space $P$ of its (affine) frames from Sec.~\ref{subsec:group-theory}, that we regard as the locally Klein bundle, introduced above. Formally, we will need the following

\begin{definition}[\text{Parallelization~\cite[p.160]{BishopGoldberg1980tensor-analysis}}]
An $n$-dim manifold $\mc{N}$ is \textbf{parallelizable} whenever it admits a set $\{X_1,...,X_n\}$ of vector fields, globally defined on \emph{all} of $\mc{N}$, s.t. $\{X_1(p),...,X_n(p)\}$ forms a basis of $T_p\mc{N}$ for each  $p\in\mc{N}$. A manifold is parallelizable iff its tangent bundle is diffeomorphic to a trivial product $(\al,\be):T\mc{N}\stackrel{\sim}{\ra}\mc{N}\times W$, where $\al:T\mc{N}\ra\mc{N}$ is the standard projection, and $\be:T\mc{N}\ra W$ is the \textbf{trivialization} map to some vector space $W$ that induces a linear isomorphism upon restriction to each fiber $\al^{-1}(p)=T_p\mc{N}$. 
\end{definition}

An affine space is naturally parallelizable $T\mathbb{M}\cong\mathbb{M}\times V$, as well as every Lie group, in general. The trivialization map on the tautological PFB $\pi_G:G\ra G/H$ is canonically induced by the (left) group action on itself:
\begin{equation}\label{eq:Maurer-Cartan}
L_{g^{-1}\ast}:\, T_gG \ \ra \ T_eG \ \equiv \ \mf{g},
\end{equation}
defining what is called the \textbf{Maurer-Cartan form} $\om_G=L_{g^{-1}\ast}$ on $G$. It establishes the Lie algebra isomorphism $[X^\mb{A},X^\mb{B}]^{\ph{e}}_e=[\mb{A},\mb{B}]$ between $\mf{g}$ and the subspace of left-invariant vector fields on $G$, so that $L_{g\ast}X_a=X_{ga}\Leftrightarrow \om_G(X^\mb{A})=\mb{A}=\mr{const}$. (The latter may be used as the definition of the bracket on $\mf{g}$, actually.) All left-invariant vector fields are \emph{complete}~\cite[Cor.3.2.12]{Sharpe1997Diff-Geometry-Cartan},~\cite[Def.0.3.4]{Bleecker1981gauge-variational-principles} s.t. every \emph{integral curve} $\ga:I\ra G$, everywhere tangent to $X_{\ga(s)}=\dot{\ga}(s)\equiv\ga_\ast(\pa_s)$, is maximally extendable to an infinite domain $(-\infty,+\infty)$. Hence they are globally-defined, being obtained via push-forward $X^{\mb{A}}_g=L_{g\ast}\mb{A}$ of the Lie algebra basis at the identity $e\in G$ to the whole space of a group.

The intrinsic definition~\eqref{eq:Maurer-Cartan} of M.-C. may appear too much abstract, though. The extrinsic construction could be given for the linear representation of $G\hookrightarrow \mr{GL}(n,\mathbb{R})$ in a group of matrices, whose natural vector space structure can be used to calculate: $\frac{1}{s}\{L_g(a+sv)-L_g(a)\}=\frac{1}{s}L_g(sv)=gv \ \Rightarrow \ L_{g\ast}(a,v)=(ga,gv)\in TG\cong\mr{GL}(n,\mathbb{R})\times\mr{Mat}(n,\mathbb{R})$~\cite[Ch.3,~\S 1]{Sharpe1997Diff-Geometry-Cartan}. It is customary to write $\om_G=g^{-1}dg$ as the left logarithmic derivative of the identity map~$g$, representing ``the general point'' on $G$. Since $dg$ is then the identity on the tangent bundle, one has $\om_G(g,v)=g^{-1}dg(g,v)=(e,g^{-1}v)\in T_eG$. Take the abelian $G=\mathbb{R}$, for instance, and recover the standard differential. The M.-C. is also a ``grandfather'' of all left-invariant forms on $G$, such as the Haar measure $\wedge^r\om_G\in\La^r(G,\wedge^r\mf{g}\cong\mathbb{R})$, where $r=\dim(G)$. One~might advocate the omnipresence of M.-C. in the analysis of infinitesimals and groups of symmetry.

One can also characterize the Maurer-Cartan form in terms of its components. Let $\mb{E}_\al(g)=L_{g\ast} \mb{E}_{\al}$, with $\mb{E}_\al\equiv\mb{E}_\al(e)$, be a natural basis of sections $\Ga(TG)$, and $\om^\al\equiv\hat{\mb{E}}^\al$ -- the dual basis of forms: $\om^\al(\mb{E}_\be)=\de^\al_\be$. Since $\om_G(\mb{E}_\al) = \mb{E}_\al(e)$, one can write $\om_G = \sum_\al \mb{E}_\al\otimes \om^\al$. Sometimes the infinitesimal elements are also referred to as `frame' and `co-frame'. In our case, this is justified by the identification $V\cong\mf{g}/\mf{h}$ at the identity, so that the basis of the Lie algebra can be chosen as $\{\mb{E}_\al\}=\{\mb{e}_i^{\ph{i}}, \mb{e}_j^{\ph{i}}\otimes\hat{\mb{e}}^k\}$. Using the linear representation~\eqref{eq:affine-linearized} of the affine group, one writes:
\begin{align}\label{eq:agebra-split}
(\om_G)_{(\mb{a},A)}^{\ph{i}} \ & = \ 
\begin{pmatrix}
0 & 0 \\
-A^{-1}d\mb{a} & A^{-1}dA 
\end{pmatrix}
\ = \ (\mb{e}_i^{\ph{i}},\mb{e}_j^{\ph{i}}\otimes\,\hat{\mb{e}}^k)
\begin{pmatrix}
0 & 0 \\
(\om_V)^i & (\om_H)^j_{\ph{j}k} 
\end{pmatrix},  \nonumber\\ 
(\om_V)^i \ & = \ -(A^{-1})^i_{\ph{i}l} \, da^l, \qquad (\om_H)^j_{\ph{j}k}  \ = \ (A^{-1})^j_{\ph{i}l} \, dA^l_{\ph{i}k}.
\end{align}
Via natural isomorphism $T_{(\mb{a},A)}(G/H\times H)\cong T_{\mb{a}}(G/H)\times T_A (H)$, the parallelization is defined separately on each component (no-mixing). 

We mention two possible ways to read the above formula~\eqref{eq:agebra-split} in the context of the tautological bundle $P\cong G$ of the homogeneous geometry $\pi_G:G\ra G/H$. First, it could be viewed as the local expression on $G$, obtained via pull-back $\psi^\ast\om_G$ by the Klein gauge $\psi:G\ra P\cong G$ of the canonical Maurer-Cartan parallelization on $G$. (Compare this with the gauge picture of the Klein geometry in~\cite[Ch.4, \S 7]{Sharpe1997Diff-Geometry-Cartan}). Since the identity of the group is the stabilizer of the standard affine basis, one obtains simultaneously $\psi:(o,e)(\mb{0},1)\mapsto (o,e)(\mb{a},A) = (o',e')$ the position and orientation of the frame (as referred to the standard one at identity; roughly, a `deformation' of the neighbourhood). The image of derivative map $\psi_{\ast g_0}(X)=X_{\psi(g_0)}\in T_{p=g}P$ can be compared with left-invariant (constant) vector fields at $g=\psi(g_0)$. By applying $\om_G$, one effectively ``forgets the images'' and keeps only the linear part of the tangent map, essentially by ``rolling back'' to the standard frame at the identity. (This is the definition of the `Darboux derivative', see further.)


Conversely, by the use of Klein charts $(U,\ka)$ of some $(P,\mathpzc{K})$, one can induce infinitesimally a canonical parallelization on it $\om_\ka:TU\ra\mf{g}$, by pulling back the Maurer-Cartan form $\om_\ka:=\ka^\ast\om_G\equiv\om_G\circ\ka_\ast$. Since the latter carries all the basic information about the structure of $G$, this alone could be enough to characterize the neighbourhood $U$ as virtually indistinguishable from the principal homogeneous space. In other words, so-parallelizable manifold $P$ is locally identical to the group $G$ (up to some covering), but without a fixed choice of unit element. In a certain sense, one directly adopts the left-invariant vector fields from $G$.
 

Let us call the \textbf{fundamental vector field} on $(P,\mathpzc{K})$:
\begin{equation}\label{eq:fundamental-Klein-vector}
X_p^{\mb{A}} \ := \ \frac{d}{ds}(p\exp s\mb{A})\bigg|_{s=0}  \in T_pP  \qquad \mb{A}\in\mf{g}
\end{equation}
-- any vector field, generated by the right $G$ action. It is then easy to show that this is $\om_\ka$-constant:
\begin{align}\label{eq:om-const}
\om_\ka (X_p^{\mb{A}}) \ & = \ \om_G\circ\frac{d}{ds}\ka(p e^{s\mb{A}})\bigg|_{s=0} \ = \ \om_G\circ\frac{d}{ds}\ka(p) e^{s\mb{A}}\bigg|_{s=0} \nonumber\\
& = \ \om_G(X_{\ka(p)}^{\mb{A}}) \ = \ L_{\ka(p)^{-1}\ast}X_{\ka(p)}^{\mb{A}} \ = \ \mb{A} \ = \ \mr{const}
\end{align}
The behaviour of $\om_\ka$ w.r.t. right $G$ action is also standard and non-problematic -- we can prove it for fundamental fields:
\begin{align}\label{eq:om-pullback}
(R_g^\ast\om_\ka)(X_p^{\mb{A}}) \ & = \ \om_\ka (R_{g\ast}X_p^{\mb{A}}) \ = \ \om_\ka\left(\frac{d}{ds}(p g g^{-1}e^{s\mb{A}}g)\bigg|_{s=0}\right) \ = \ \om_\ka\left(\frac{d}{ds}[p g \exp(s \mr{Ad}_{g^{-1}}\mb{A})]\bigg|_{s=0}\right) \nonumber\\
& = \ \om_\ka \left(X_{pg}^{\mr{Ad}_{g^{-1}}\mb{A}}\right) \ = \ \mr{Ad}(g^{-1})\mb{A} \ = \ (\mr{Ad}(g^{-1})\om_\ka) (X_p^{\mb{A}}),
\end{align}
and then extend by linearity, since $X_p^{\mb{A}}$ form the basis of $T_pP$. Or, first pass to $G$ via local chart $\ka_\ast X_p=Y_{\ka(p)}\in T_{\ka(p)}G$, and then apply M.-C.: $(R_g^\ast\om_\ka)X_p=\om_G\circ\ka_\ast(R_{g\ast} X)_{pg}=\om_G(Y_{\ka(p)g})= L_{(\ka(p)g)^{-1}\ast}\circ R_{g\ast}(Y_{\ka(p)})= L_{g^{-1}\ast} R_{g\ast}(L_{\ka(p)^{-1}\ast}Y_{\ka(p)})=\mr{Ad}(g^{-1})\om_\ka(X_p)$, using commutation of left and right action.

Somewhat ahead the physical motivation, the above properties lead us to propose the following
\begin{definition}\label{def:Cartan-connection}
In general, we will say that the Lie algebra valued 1-form $\varpi: TP\ra\mf{g}$ on $(P,\mathpzc{K})$ determines a \textbf{Cartan $G$ connection} if:
\begin{enumerate}[label={\upshape(\roman*)}, align=left, widest=iii]
\item Upon restriction to each local chart, the trivializaion $\varpi:TU\ra\mf{g}$ determines the infinitesimal parallelization on $U$ (to the 1st order); \label{def:Cartan-1}
\item The linear isomorphism $\varpi_p:T_pP\stackrel{\sim}{\ra}\mf{g}$ over each $p\in P$ restricts to be point-wise Maurer-Cartan $\varpi(X^{\mb{A}}) =\mb{A}$ for all constant $\mb{A}\in\mf{g}$; \label{def:Cartan-2}
\item $(R_g^\ast)\varpi=\mr{Ad}(g^{-1})\varpi$ for all $g\in G$ -- $\varpi$ transforms in the (anti-)adjoint representation. \label{def:Cartan-3}
\end{enumerate} 
(Note: the provided notion of connection is \ul{not standard}. It corresponds closest to the presentation in~\cite[Def.5.3.1]{Sharpe1997Diff-Geometry-Cartan}. Although, being formulated on the Klein bundle, it enjoys wider use of symmetry w.r.t. full principal gauge group $G$. From the usual Ehresmann's vertical parallelism this differs in that $\ker\varpi=0$, however. [We comment more on their relation in the following Sec.] In a certain sense, it is combination of the two, that we hope is closer to the Cartan's original notion.) 
\end{definition}

Let us describe the behaviour of the connection w.r.t. gauge transformations.

\begin{proposition}\label{prop:connection-transform}
Let $f\in\mc{G}(P,\mathpzc{K})$ and $\tau\in C(P,G)$ be as in Prop.~\ref{prop:gauge-symmetry}. Then $f^\ast\varpi=\tau^\ast\om_G+\mr{Ad}(\tau^{-1})\varpi$ is also a connection, satisfying properties~\ref{def:Cartan-2} and~\ref{def:Cartan-3}, the first being evident. (We refer to~\cite[\S 3.2]{Bleecker1981gauge-variational-principles} for the proof in the ordinary gauge theory, which is also applicable in our case.)
\end{proposition}

For the general (non-trivial) Cartan connection $\varpi$ on the Poincar\'{e}/affine $G$ bundle, let us denote $\varpi=\theta+\om$ its components in the vector space decomposition $\mf{g}\cong\mf{p}\oplus\mf{h}$, $\mf{p}\cap\mf{h}=\emptyset$, as explicated in~\eqref{eq:agebra-split}. Such a splitting of the Lie algebra is called \emph{reductive}, in the sense that $\mf{p}=\mf{g}/\mf{h}$ is an $\mr{Ad}(H)$ (sub-)module, transforming under $H$ in the vector $V$ representation $[\mf{h},\mf{p}]\subseteq \mf{p}$. (Other reductive Klein geometries with the same $H$ include anti/de Sitter space, with non-commutative `transvections' $[\mf{p},\mf{p}]\subseteq \pm R^{-2}\mf{h}$, where $R$ is the curvature radius.) The kernels of $\om$ and $\theta$ define complementary distributions of \textbf{horizontal} and \textbf{vertical} vectors in the tangent bundle:
\begin{equation}\label{eq:vertical-horizontal}
T_pP \ = \ H_pP\oplus V_pP, \qquad HP \ := \ \ker\om , \qquad VP \ := \ \ker\theta.
\end{equation}
(Since horizontal and vertical signs always precede the bundle $P$, there should arise no confusion with the homogeneous $H$ and translational $V$ parts of the affine group. Oddly enough though, the mnemonic rule applies, associating each distribution with the kernel of corresponding parts of $\varpi$.)

The `homogeneous' $\om$ part of the affine connection constitutes the familiar \textbf{Ehresmann connection} in the principal $H$ bundle. Recall that its geometrical definition consists precisely in 1)~the smooth assignment of horizontal subspaces $H_pP$, complementary to $V_pP=\ker\pi_{\ast p}$, that 2)~satisfy invariance $R_{A\ast}H_pP=H_{pA}P$ w.r.t. $A\in H$. Equivalently, there is a 1-to-1 correspondence as above with the set of 1-forms $\om:TP\ra\mf{h}$, valued in the Lie algebra \emph{of the fiber}, such that: 1')~$\om(X_p^{\mb{A}})=\mb{A}\in\mf{h}$, and 2')~$R_A^\ast\om=\mr{Ad}(A^{-1})\om$ for all $A\in H$ (see App.~\ref{app-1}).

Usually, the horizontal sections are characterized indirectly through $X^H_p:=X_p-X_p^{\mb{A}}$ and the base directions. Namely, given a vector field $\bar{X}$ on $\mc{M}$, there is a unique `horizontal lift' $X^H$ in $P$, such that $\om(X^H)=0$ and $\pi_\ast(X_p^H)=\bar{X}_{\pi(p)}$ for all $p\in P$. Necessarily $R_{A\ast}X^H=X^H$ for all $A\in H$. In our case, horizontal vectors are also fundamental $X_p^{\mb{a}}$~\eqref{eq:fundamental-Klein-vector} -- being \emph{directly generated} by the action of translation (sub-)group $\mb{a}\in\mf{g}/\mf{h}$. This puts them on the same gauge theoretic footing as $X_p^{\mb{A}}$, with $\mb{A}\in\mf{h}$. 

Due to natural isomorphism $P/H\cong\mc{M}$ induced by the bundle projection $\pi$, there is a trivialization $\be_p:T_{\pi(p)}\mc{M}\stackrel{\sim}{\ra}\mf{g}/\mf{h}$, making the following diagram commute:
\begin{displaymath}
\begin{tikzcd}[column sep=small]
T_p(pH) \arrow[d] \arrow[rr, "\om"] & & \mf{h} \arrow[d] \\
T_p(P) \arrow[d, swap, "\pi_\ast"] \arrow[rr, "\varpi"] & & \mf{g} \arrow[d, "\pi_{G\ast}"] \\
T_{\pi(p)}(\mc{M})  \arrow[rr, "\be"] & & \mf{g}/\mf{h}
\end{tikzcd},
\end{displaymath}
such that $\be_{pA}=\mr{Ad}(A^{-1})\be_p$ (the horizontal arrows are linear isomorphisms, while vertical ones make up short exact sequences of Lie algebra homomorphisms, cf.~\cite[Th.5.3.15]{Sharpe1997Diff-Geometry-Cartan}). 

\subsection{On `soldering' and `attachment', comparison of approaches}
\label{sec:soldering}

Unlike Yang-Mills theory, the model spaces of GR have `external' geometric significance, corresponding to directions and velocities in the base. One should therefore be able to identify the tangent spaces of the ``real'' manifold $\mc{M}$ and that of some model homogeneous (flat) geometry $\mathbb{M}$ -- ``gluing'' them to the first (linear) order of approximation, so to say. The horizontal $V$-valued form $\theta$ (via induced $\be$-map) thus provides such an identification. For that reason, it is commonly known as `soldering' form, and is the key element that distinguishes gravity from the other interactions. The diagram above performs such an identification also for the respective bundles $TP\approx TG$. 

\paragraph{Formalism of (linear) frame bundles}

As a matter of fact, the bundle of linear frames -- as defined in~\eqref{eq:frame-bundle} and commonly used in GR -- is ``special'' in that it consists of linear mappings~\eqref{eq:linear-frame} from the canonical $u_x:V\ra T_{x=\pi(u)}\mc{M}$ to the tangent vector spaces (at the point). In addition to its general $H$ bundle structure, $L(\mc{M})$ or any reduction thereof (e.g. the orthonormal frame bundle $O(\mc{M})$, with the restricted structure group $H=\mr{O}(\eta)$) has a property that is not shared by other PFBs. Namely, the so-called `fundamental' $V$-valued form is defined as $\bar{\theta}:=u^{-1}\circ\pi_\ast$, s.t. it brings the values of the vector $\bar{X}_{\pi(u)}=\pi_{\ast u}(X)=\bar{\theta}^i(X)\, u(\mb{e}_i^{\ph{i}})$ in the (images in $T\mc{M}$ of the) standard basis $\{\mb{e}_i^{\ph{i}}\}\in V$, for some lift $X\in\Ga(TL(\mc{M}))$.

Conversely, the existence of the soldering form on $P$ -- i.e. $\theta\in\bar{\La}^1(P,V)$: $\mr{Ad}(H)$-equivariant, s.t. $\theta(X)= 0\ \Leftrightarrow \ \pi_\ast X = 0$ -- establishes an isomorphism (non-canonical) of corresponding bundles:
\begin{align}\label{eq:soldering}
P \ & \approx \ L(\mc{M}), & E \ = \ P\times_H V  \ & \approx \ T(\mc{M}), \\
\text{given by} \qquad p \ & \mapsto \ (\mb{e}_1^{\ph{i}}(p),...,\mb{e}_m^{\ph{i}}(p)), &  \text{where} \qquad \mb{e}_i^{\ph{i}}(p) \ \equiv\ [p,\mb{e}_i^{\ph{i}}] \ & \mapsto \ [\pi(p),\be_p^{-1}(\mb{e}_i^{\ph{i}})], \nonumber 
\end{align}
with the right $H$ action defined as $(\be_p^{-1}(\mb{e}_1^{\ph{i}}),...,\be_p^{-1}(\mb{e}_m^{\ph{i}}))\circ A=(\be_{pA}^{-1}(\mb{e}_1^{\ph{i}}),...,\be_{pA}^{-1}(\mb{e}_m^{\ph{i}}))$. In accord with the correspondence~\eqref{eq:horizontal-froms}, the isomorphism could be invariantly described by the map $\tilde{\theta}:T\mc{M}\ra E$, whereas the soldering gives the coordinates $\theta_p(X)$ of $\tilde{\theta}(\pi_\ast X)$ in the frame $p$. 

The fundamental form $\bar{\theta}$ is induced on $L(\mc{M})$ via pull-back of $\theta$ from $P$, using the above isomorphism~\eqref{eq:soldering}. It is sometimes called `canonical' on $L(\mc{M})$, since its coordinate-invariant version~\eqref{eq:horizontal-froms} gives the `tautological' identification of $T(\mc{M})$ with itself~\footnote{Concretely, if $\bar{X}=u(\mb{X})=[u,\mb{X}]$ and $X$ is its lift in $P$, then $\tilde{\bar{\theta}}(\bar{X})=[u,\bar{\theta}(X)]=[u,u^{-1}\circ\bar{X}]=[u,\mb{X}]=\bar{X}$.}. This simply attests the fact that the bundles $L(\mc{M})$ and its associate, far from being arbitrary, are related to the geometry of $\mc{M}$. If one forgets that the bundles~\eqref{eq:soldering} are actually soldered by means of $\theta$ itself -- which may be non-canonical -- both terms `fundamental' and `soldering' for $\theta$ are sometimes used interchangeably, due to their identical properties (cf. discussions in~\cite{Catren2015Cartan-gauge-gravity,Petti2006translational-symmetries} of this issue). Since one of our tasks is to delineate the ``background structure'' from the ``field'' on it (in relativistic physics parlance), we will avoid such misuse. Summarizing, one can state

\begin{proposition}[$\varpi \leftrightarrow (\theta,\om)$, cf.~\text{\cite[App. A.2]{Sharpe1997Diff-Geometry-Cartan}}]\label{prop:equiv-1} 
There is a bijective correspondence between Cartan's affine geometries $(P,\mathpzc{K},\varpi)$ as given by Def.~\ref{def:Cartan-connection} and the principal $H$ frame bundles~\eqref{eq:frame-bundle} with fundamental form $\theta$ and Ehresmann connection $\om$.
\end{proposition}

\paragraph{Ehresmann's standpoint}

The process of identification between (linear) tangent spaces of two manifolds by means of soldering -- as described above -- somewhat conceals (within $p\in P$) the `point of attachment'. This aspect becomes crucial when two such (affine) spaces, attached at different points, are put into correspondence by means of connection/parallel transport. Because an affine space $(\mathbb{M},o)$ with a marked point is indistinguishable from its vector tangent space $T_o\mathbb{M}\cong V$ at that point, the confusion may arise about the nature of admissible frame transformations, establishing their equivalence at different $x\in\mc{M}$. In brief: the affine parallel transport defined by rolling does not fix the origin; the linear parallel transport can be recovered by applying a translation.

This issue of `attaching map' is more pronounced in the Ehresmann's formalization of the affine connections of Cartan (cf.~\cite{Marle2014fromCartan-toEhresmann}, \cite[Ch.3, \S 3]{KobayashiNomizu1963vol-1}; initially, the concept was rather vaguely defined). More in detail, the original construction by Ehresmann starts not with~$P$, but the principal $G$ bundle~$P'$. This could also be viewed as associated bundle $P'\simeq P\times_H G$ of affine frames over $\mc{M}$, on which the action of $G$ is defined. Respectively, $B:=P'/H\simeq P'\times_G F$ is the associated bundle of affine spaces $F\cong G/H$ (not vector $V$-fibers). One then says that $B$ is \emph{``soud\'{e}''} to $\mc{M}$ if $\dim F=\dim\mc{M}$, and the following conditions are satisfied:
\begin{enumerate}[label={\upshape(\roman*)}, align=left, widest=iii]
\item The structure group $G$ of $B$ can be reduced to $H$. This is equivalent to the existence of the global section $\tilde{\si}:\mc{M}\ra B$, selecting in a smooth manner the point of attachment in each homogeneous space $B_x$ over $x\in\mc{M}$. The reduced principal $H$ bundle can then be obtained via pull-back construction:  \label{def:soude-2}
\begin{displaymath}
\begin{tikzcd}[column sep=small]
P =\tilde{\si}^\ast P' \arrow[d,swap, "\pi"] \arrow[rr, hookrightarrow] & & P' \arrow[d,"\pi'"] \\
\mc{M}  \arrow[rr, "\tilde{\sigma}"] & & B=P'\times_G G/H
\end{tikzcd}.
\end{displaymath}
The element  $p\in P$ is considered as a [linear] mapping of the standard fiber $F$ onto $B_x$, such that $p(o)=\tilde{\si}(x)$, where $o$ is the point of $F\cong G/H$ which defined the isotropic (sub-)group~$H$.
\item There is an isomorphism of vector bundles $T\mc{M}\approx \tilde{\si}^\ast VB$, where the latter is the bundle of (vertical) tangent vectors $V_{\tilde{\si}(x)}B$ to $B_x$ at $\tilde{\si}(x)$ (as $x$ running through $\mc{M}$). It can then be shown that this is equivalent to existence of the form of \emph{``soudure''} $\theta\in\bar{\La}^1(P,V)$ (cf.~\cite{Kobayashi1956Cartan-connections}).
\label{def:soude-3}
\end{enumerate} 

\begin{proposition}[$\varpi \leftrightarrow \ker\varpi'|_P=0$, cf.~\text{\cite[App. A.3]{Sharpe1997Diff-Geometry-Cartan}}]\label{prop:equiv-2} 
The Ehresmann connection $\varpi':TP'\ra \mf{g}$ defines the absolute parallelism on $P$ in terms of a pair $(\theta,\om)$, iff the restriction $\varpi'|_P=\theta+\om$ has no kernel; by transitivity, it is then equivalent to the Cartan connection of Def~\ref{def:Cartan-connection}.
\end{proposition}

\paragraph{Gauge theoretic interpretation.}

The Ehresmann connections are considered more general in the sense that they could violate some of the above requirements (e.g. allow ``slipping/ or twisting''), or even dispence with the soldering and still maintain the horizontal lifts with the notion of (linear) parallel transport over certain base manifold (``externally specified''; cf.~App.~\ref{app-1})~\footnote{The dissatisfaction with non-intuitiveness and extreme generality of the construction is vividly expressed in the Sharpe's book~\cite{Sharpe1997Diff-Geometry-Cartan}.}, as in Yang-Mills. The drawback, in our viewpoint, is that the action of translation (sub-)group -- although defined on $P'$ -- is lost upon reduction, or rather hidden in the \emph{non-canonical} choice of attaching section $\tilde{\si}$. The latter is interpreted as a (partial) ``gauge fixing'' in~\cite{Catren2015Cartan-gauge-gravity}, and the possibility is indicated that~\blockquote{we can change the attaching section by rolling the [local homogeneous model] along diffeomorphisms of $\mc{M}$ . . . Therefore, the invariance of a theory under $\mr{Diff}(\mc{M})$ guarantees its invariance under transformations of the attaching section, i.e. under changes of the partial gauge fixing that defines the attachment.}


We share the similar attitude towards the issue of diffeomorphisms. Yet the concept of the Klein bundle may provide several improvements. Observe that the reduction/gauge fixing mechanism is, alternatively, induced by the right-equivariant map $\si:P'\ra F\cong G/H$, such that $P_0=\si^{-1}([g_0H])$~\cite[Prop.4.2.14]{Sharpe1997Diff-Geometry-Cartan}. If $P_1$ is any other reduction, corresponding to the element $g_1=g_0g$, then the two are related by $P_1=P_0g$ (as sub-bundles; each of them appears as the image of the corresponding section $\tilde{\si}(x)=[p,\si(p)^{-1}],\, \pi(p)=x$.) Such maps are in-built ab initio in the notion of the local Klein chart as $\pi_G\circ\ka$, and their opposites $\psi=\ka^{-1}$ perform the attachment. The `change of attachment' is now simply generated by (local) gauge transformations $\mc{G}(P,\mathpzc{K})$ of Def~\ref{def:Klein-symmetry}. 

We notice that the labels $x\in\mc{M}$ play no substantial role in the Ehresmann's soud\'{e}-construction, so that the base could effectively be collapsed to a single point. The framework $(P,\mathpzc{K},\varpi)$ of Cartan geometry is thus virtually indistinguishable from the principal (Ehresmann) connection in the $G$ bundle over trivial manifold, consisting of one point. At the same time, the Klein bundle could be viewed as comprising both principal $P$ and soldered $B$ bundles, in a sense. Its elements are regarded not as linear but affine maps $\ka^{-1}$ (or `sections' $\tilde{\ka}=[p,\ka(p)^{-1}]$) from the model space $\mathbb{M}$ to $P/H\sim \mathbb{M}_{\pi(p)}$: the first component attaches the space, while the other one orients it in a certain manner. This should explain the apparent conundrum of group seemingly acting in the base. There is simply no base as such, only the point set with affine spaces attached to them, and glued by the equivalence relation of $G$ structure. The construction is thus genuinely \emph{`background-independent'}.


\subsection{Rigidity of the Klein geometry, `osculation', and equations of structure}
\label{subsec:structure}

The general concept of the Cartan geometry (see Fig.~\ref{fig:hamster-ball}), and the construction of the Klein bundle, foster us to \ul{put forward the following equation}:
\begin{equation}\label{eq:osculation}
(\varpi - \ka^\ast\om_G)_p^{\ph{i}} \ = \ 0, \qquad \text{or, equivalently} \qquad (\psi^\ast\varpi - \om_G)_{\ka(p)}^{\ph{i}} \ = \ 0,
\end{equation}
\ul{defining the \textbf{`osculation'}} of the homogeneous Klein geometry to that of Cartan's $(P,\mathpzc{K},\varpi)$. Namely, their respective fundamental fields are identified (at the point $p\in P$), in some local chart~$\ka$/gauge~$\psi$. The next reasoning is adapted from the justification in~\cite[Ch.3, \S 6]{Sharpe1997Diff-Geometry-Cartan} of the theorem to follow.

The idea is to construct, locally, the graph in $(M\times G)$ of what would be the function $f:M\ra G$ from some manifold, if it were to define the isomorphism of geometric structures $\om_G\sim\varpi$ (see~Fig.~\ref{fig:geom-solder}). Then a point $(p,f(p))$, in order to remain on a graph, could move only in certain directions, so that a tangent vector $(v,w)$ must satisfy the condition $f_\ast(v)=w$. Equivalently, by post-composing with canonical Maurer-Cartan trivialization $\om_G\circ f_\ast=f^\ast\om_G\equiv\om_f:TP\ra TG\ra\mf{g}$, the resulting Lie algebra valued form defines what is called the \textbf{Darboux derivative} of $f:M\ra G$ -- if such map exists. Its effect is to ``forget'' the underlying images of $f$, and keep only linear tangential part.
\begin{figure}[h]
\center{\includegraphics[width=0.4\linewidth]{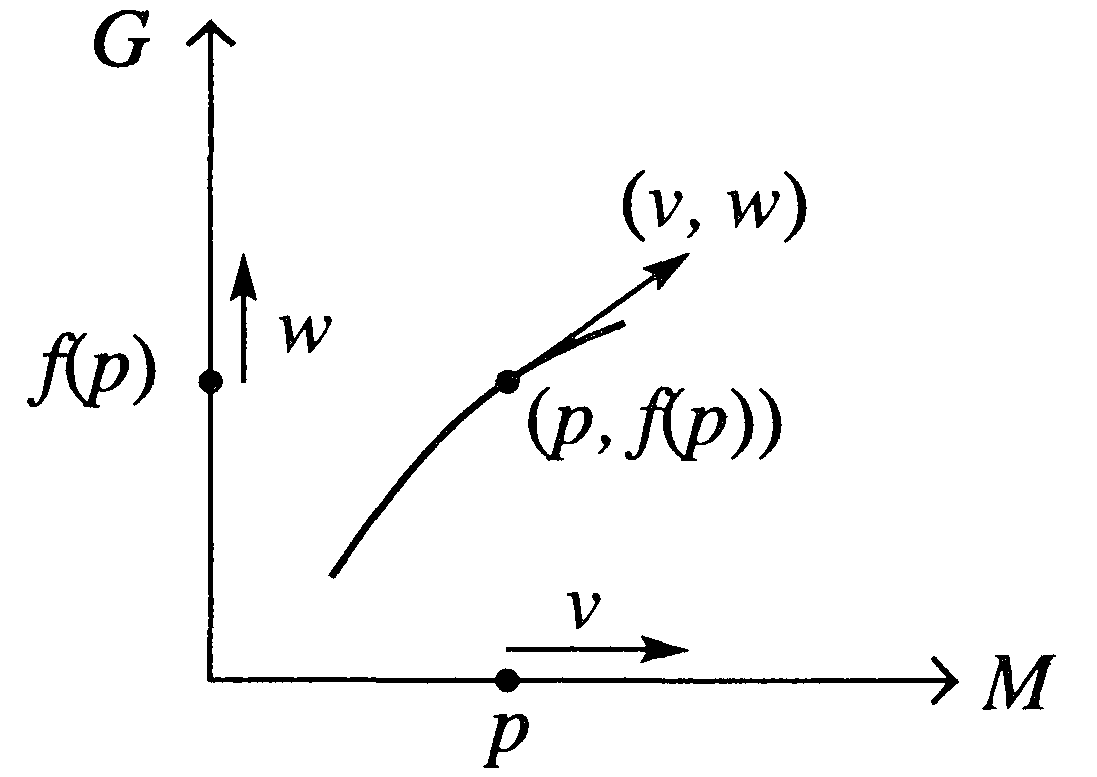}}
\caption{The geometric idea of `osculation'~\eqref{eq:osculation} (pic. from~\cite[p.117]{Sharpe1997Diff-Geometry-Cartan}).}
\label{fig:geom-solder}
\end{figure}

Now, if we are merely given a form $\varpi:TP\ra\mf{g}$, one can still apply $\om_G\circ\ka_\ast$ to $P$, defining the distribution~\eqref{eq:osculation}, which may be non-integrable. In order to characterize certain form as the Darboux derivative of the possible $f$, which appears as the leaf of involutive distribution, we appeal here to the following integrability conditions.

\begin{theorem}[(Non-abelian version of) the fundamental theorem of calculus]\label{th:fund-calculus}
Let $G$ be a Lie group with Lie algebra $\mf{g}$. Let $\om$ be a $\mf{g}$-valued 1-form on the (generic) smooth manifold $\mc{N}$, satisfying the \textbf{structural equation of Maurer-Cartan}
\begin{equation}\label{eq:structural-eq}
d\om+\frac12 [\om\wedge\om] \ = \ 0.
\end{equation}
Then
\begin{enumerate}[label={\upshape(\alph*)}, align=left, widest=iii]
\item For each point $p\in\mc{N}$, there \emph{exists} a neighbourhood $U$ of $p$ and a smooth map $f:U\ra G$ such that $\om|_U=\om_f$~\cite[Th.3.6.1]{Sharpe1997Diff-Geometry-Cartan}.
\item Moreover, the integral is \emph{unique} up to left translation: if $f_1,f_2:\mc{N}\ra G$ are two maps (called primitives) for which $\om_{f_1}=\om_{f_2}$, then there is an element $g\in G$ (constant of integration), such that $f_2(p)=g f_1(p)$~\cite[Th.3.5.2]{Sharpe1997Diff-Geometry-Cartan}. 
\item The map $f:\mc{N}\ra G$ is \emph{globally} defined, if the additional topological requirement is met, namely, the triviality of the monodromy representation $\Phi_\om:\pi_1(\mc{N},b)\ra G$, with the period group $\Ga\equiv\mr{Im}(\Phi_\om)$ being $e\in G$~\cite[Th.3.7.14]{Sharpe1997Diff-Geometry-Cartan} (see the next Sec. for definitions).
\end{enumerate}
\end{theorem}

The necessity of the M.-C. condition is a simple exercise. Using the definition of exterior derivative: $d\om_G(X,Y)=X(\om_G(Y))-Y(\om_G(X))-\om_G([X,Y])= -\om_G([X,Y])$, if applied to left-invariant fields. Since the latter form the basis, the result follows by linearity for any $X,Y\in TG$. It is inherited by $\om_f$ due to naturality property $df_\ast=f_\ast d$ of the derivative. The sufficiency proof is less trivial, and was delivered by Cartan~\footnote{As long as the purely analytical aspects are concerned, this is the special case of the (Pfaffian) system of linear differential equations $\theta_\al=0, \ \al=1...r$ in $n+r$ variables. The prospective solution allows a geometric reading in terms of $n$-dim integral (sub-)manifold $\mc{N}$, passing through any generic point of space, and everywhere tangent to the given distribution of $n$-vectors $\{\pa_i\}_{i=1...n}$. The system then appears to define the distribution in dual terms of (vector-valued) forms as $T_p\mc{N}\stackrel{?}{=}\ker\theta_p$, s.t. the solution $\mc{N}$ exists (locally) iff the condition of complete integrability (involutivity of the distribution, Frobenius) is satisfied: $d\theta_\al =0\  \mr{mod}\, \theta_1,...,\theta_r$. This implies also the existence of $r$ 1st integrals $u^\al=F^\al(x^1,...,x^n;z^1,...,z^r)$, constant on the solutions $z^\al= f^\al(x^1,...,x^n;z_0^1,...,z_0^r),\,z_0^\al=z^\al(x_0^i)$, s.t. the system is algebraically equivalent to $dF^\al=0$.\label{foot:Pfaff}}.

One will say that the Maurer-Cartan form, being canonical, determines in this way the \emph{trivial connection} on $(P,\mathpzc{K})$. In accord with the `Lie group $\leftrightarrow$ Lie algebra' correspondence, the previous theorem can be used to show that such form fully characterizes (up to some covering) the manifold $P$ as principal homogeneous space, that is identical to the group $G$, but without a fixed choice of unit element. (As the slogan says: ``a torsor is a group that has forgotten its neutral element''~\cite{Baez2010torsors}.) 

\begin{theorem}[Rigidity of the Klein geometry \text{\cite[Th.3.8.7]{Sharpe1997Diff-Geometry-Cartan}}]\label{th:G-characterization}
Let $\mc{N}$ be a connected smooth manifold, and $\om:T\mc{N}\ra\mf{g}$ be a Lie algebra valued 1-form, that satisfies the conditions:
\begin{enumerate}[label={\upshape(\alph*)}, align=left, widest=iii]
\item $d\om+\frac12 [\om\wedge\om] = 0$;
\item $\om_p:T_p\mc{N}\ra\mf{g}$ is an isomorphism on each fiber;
\item All the $\om(X)=\mr{const}$ vector fields are complete~\footnote{In the note on the second printing of his book~\cite{Sharpe1997Diff-Geometry-Cartan}, Sharpe points out the minor subtlety in his proof, which requires a little bit stronger condition of completeness for the vector-valued isomorphism $\om_p:T_p\mc{N}\stackrel{\sim}{\ra}V$. Precisely, the vector field $(\om^{-1}(f(s)),\pa_s)$ on $\mc{N}\times\mathbb{R}$ has to be complete for every smooth function $f:\mathbb{R}\ra V$. As a counterexample of completeness, one can mention the celebrated Hawking-Penrose theorem~\cite[Ch.8]{HawkingEllis1973large-scale-structure}: provided the certain plausible conditions on energy-momentum are satisfied, it states the existence of geodesic paths, inextendible beyond the certain critical point, where they meet the ``singularity'' in a finite proper time (which cannot happen for a smooth vector field~\cite[Ex.2.1.18]{Sharpe1997Diff-Geometry-Cartan}).}.
\end{enumerate}
Then
\begin{enumerate}[label={\upshape(\alph*)}, align=left, widest=iii]
\item The universal cover $\check{\pi}:G\ra\mc{N}$, for an arbitrary choice $e\in G$, has the structure of a Lie group with Lie algebra $\mf{g}$, such that $\check{\pi}^\ast\om$ is its Maurer-Cartan form.
\item The discrete (sub-)group of periods $\Ga\subset G$ acts by left multiplication on $G$ as the group of covering transformations, i.e. $\mathpzc{Gal}(G/\mc{N}):=\{T\in\mr{Diff}(G)|\check{\pi}\circ T=\check{\pi}\}$.
\end{enumerate}
\end{theorem}

One sees that the Def.~\ref{def:Klein-structure} of the Klein bundle naturally encompasses (but is not exhausted by) the notion of `locally Klein geometry' from~\cite[Ch.4, \S 3]{Sharpe1997Diff-Geometry-Cartan}. This is globally a group space $G$ that may be not simply-connected though, having a discrete period subgroup $\Ga\subset G$. In a context of affine (reductive) group, it corresponds to Cartan's `locally Euclidean Riemannian manifold'~\cite[Ch.11]{Cartan2001Orthogonal-frame}, such that in a sufficiently small (but finite) neighbourhood of any of its points it can be \emph{exactly} mapped/or `developed' piece-by-piece on the region of the Euclidean space with the same line element $ds^2$. However, one cannot be certain a priori 
\begin{enumerate}[label={\upshape(\arabic*)}, align=left, widest=iii]
\item That each point of the Euclidean space could be obtained in the  
development, and 
\item That a point of the Euclidean space obtained in the development of the 
manifold could not be obtained more than once.
\end{enumerate}
Hence, globally the correspondence may be not bijective~\footnote{This is the closest classical geometric analogue that we could find to the so-called `moduli spaces of flat connections', studied in TQFT and QG. The proposed modification to the BF-theory of Ch.~\ref{ch:prop-2} is largely based on this observed similarity. This is also in line with the proposed `new vacuum' in canonical LQG~\cite{DittrichGeiller2015new-vacuum,DittrichGeiller2015flux-formulation-classic}, corresponding to zero curvature of the topological BF model.}. The typical examples include: circular cylinder, torus, cone (excluding singular vertex), and any developable surface in general.

Concerning our discussion -- leaving the topological issues aside -- the above results clearly indicate that the trivial connections take part of the canonical background structure, corresponding to some absolute/ideal model space (homogeneous Klein geometry, in general). However, it is not perceived as a fixed spacetime inert arena, on which events happen, but rather some reference device to relate to non-canonical structures.

Equation~\eqref{eq:osculation} for $\varpi$ seemingly has the form of Darboux derivative. However, we explicated the meaning associated with it, which should be more-or-less clear. We view it as performing identification between two geometries point-wise, and defining non-integrable distribution, in general. This is repercussion of the property~\ref{def:Cartan-2} in Def.~\ref{def:Cartan-connection}. (Recall that $(P,\mathpzc{K},\varpi)$ is virtually the same as Ehresmann principal connection over a pure point manifold, and has the curvature vanishing in directions of the fiber. The idea of connection here is not horizontality per se but parallelism given by $\varpi$, and the equation~\eqref{eq:osculation} simply attest its difference from the canonical one.)

For the affine group, one obtains what Cartan called `the Euclidean space, osculating at the point'. [Initially it is required that the metrics of two spaces ``osculate'' at a point -- having the same values of the scalar product $\cg_{ij}$ and the first derivatives $\pa_i\cg_{jk}$~\cite[Ch.12-13]{Cartan2001Orthogonal-frame}, modulo general coordinate transformations. That is, the Levi-Civita connection of the given line element $ds^2$ is adopted in Riemannian geometry; however, one quickly dispenses with that limitation in the most general scenario.]
In Cartan's language of moving frames, equation~\eqref{eq:osculation} corresponds to %
\begin{subequations}\label{eq:frame-displacement}
  \begin{empheq}[left=\empheqlbrace]{align}
    d\mb{m}^{\ph{i}}  &= \ \mb{e}_i^{\ph{i}} \, \theta^i, \label{eq:frame-displacement-1} \\
    d\mb{e}_j^{\ph{i}}  &= \ \mb{e}_i^{\ph{i}}\, \om^i_{\ph{i}j}, \label{eq:frame-displacement-2}
  \end{empheq}
\end{subequations}
describing the frame's displacement, \emph{matching} the affine spaces tangential to infinitesimally nearby points $\mb{m}$ and $\mb{m}'=\mb{m}+d\mb{m}$~\cite{Cartan1986Affine-connections}:~\blockquote{Here, the coefficients $\theta$ and $\om$ are linear combinations of the differentials $\{du^\mu,dv^\al\}$ of the parameters, characterizing the configuration of the frame. These [$m+r$, where $m=\dim\mathbb{M}$, $r=\dim H$] Pfaffian forms enable one, in effect, to fix the frame at $\mb{m}+d\mb{m}$ in terms of a given frame at $\mb{m}$. That is, they define the small affine displacement relating the two frames.}
In quite a casual manner, through the equations~\eqref{eq:frame-displacement}, Cartan equips $\mc{M}$ with an ``affine connection'':~\blockquote{. . . if a law relating affine spaces associated with any two \emph{infinitesimally} close points $\mb{m}$ and $\mb{m}'$ is specified. The choice of this law is quite arbitrary; it only has to enable us to say that such and such point in the affine space associated with $\mb{m}$ corresponds to such and such point in the affine space of $\mb{m}'$, and that such and such vector in the first space is parallel or equal to such and such vector in the second. In particular, $\mb{m}'$ itself can be specified with respect to the affine space originating at $\mb{m}$.~\cite[p.60]{Cartan1986Affine-connections}}

The coefficients $(\theta,\om)$ are not arbitrary for integrable connections. Indeed, differentiating the Pfaffian system~\eqref{eq:frame-displacement}, the vanishing of $\mb{e}_i^{\ph{i}}\, d\theta^i+ d\mb{e}_i^{\ph{i}}\wedge\theta^i=0$, and $\mb{e}_i^{\ph{i}}\, d\om^i_{\ph{i}j}+ d\mb{e}_i^{\ph{i}}\wedge\om^i_{\ph{i}j}=0$ should be the algebraic consequence of the system itself~\cite{Cartan2001Orthogonal-frame}. Substituting~\eqref{eq:frame-displacement}, one obtains the 1st and the 2nd structure equations: 
\begin{subequations}\label{eq:structure-0}
  \begin{empheq}[left=\empheqlbrace]{align}
    d\theta^i+ \om^i_{\ph{i}j}\wedge\theta^j \ & = \ 0, \\
    d\om^i_{\ph{i}j}+ \om^i_{\ph{i}k}\wedge\om^k_{\ph{k}j} \ & = \ 0,
  \end{empheq}
\end{subequations}
arising in the decomposition of~\eqref{eq:structural-eq} for the reductive algebra $\mf{g}=\mf{g}/\mf{h}\oplus\mf{h}$ of the affine group. 

For the general Cartan (affine) connection, the non-trivial quantities
\begin{equation}\label{eq:torsion-curvature}
\Theta \ := \ d\theta +[\om\wedge\theta],  \qquad \Om \ := \ d\om+\frac12 [\om\wedge\om],
\end{equation}
determine the \textbf{torsion $\Theta\equiv D^\om\theta$} and the \textbf{curvature $\Om\equiv D^\om\om$} of $\varpi=\theta+\om$, using Def.~\ref{def:covariant-differential} of the covariant differential. Thus both of them are horizontal, and the latter happens to coincide with the curvature~\eqref{eq:cuvature-Ehresmann} of the Ehresmann connection, characterizing non-integrability of the distribution $HP=\ker\om$. The torsion, however, is strictly related to the first line of~\eqref{eq:frame-displacement} and not to $VP=\ker\theta$; the vertical fibers of $TP$ are completely integrable to the leaves $\sim H$, simply because of horizontality of $\Theta\neq 0$. Quantities~\eqref{eq:torsion-curvature} can be combined into a single object of `Cartan $\mf{g}$-valued curvature' $\Pi\equiv(\Theta,\Om)$. Their meaning is revealed (partially) in the next Sec.~\ref{subsec:parallel-transport} in terms of frame displacements, while the physical/mechanical significance is exposed in Ch.~\ref{ch:Einstein-Cartan} in terms of stresses and tensions [defects] in the `medium'.




\subsection{Path-integrability, development and parallel transport}
\label{subsec:parallel-transport}

In general, it is impossible to satisfy the Maurer-Cartan equation~\eqref{eq:structural-eq} in any extended region exactly, identifying the piece of $P$ with $G$. One can only establish equivalence infinitesimally for two frames next to each other. What is the situation for arbitrary points $\mb{m}_0$ and $\mb{m}_1$, connected by a finite path $\bar{\ga}$ in $\mc{M}\cong P/H$? To answer the question, one chooses at each of the intermediate points an affine frame and performs matching step by step.

\begin{definition}\label{def:development}
Let $\ga:(I,s_i,s_j)\ra (P,p_i,p_j)$ be a piecewise smooth path, $\dot{\ga}(s)=\ga_\ast(\pa_s)\in T_{\ga(s)}P$, and $\varpi:TP\ra\mf{g}$ a smooth 1-form of Cartan connection on $(P,\mathpzc{K})$. Pull back the osculation equation~\eqref{eq:osculation} down to $I$ by $\ga$. The structural condition~\eqref{eq:structural-eq} is then identically satisfied on 1-dim interval, as involving 2-forms, hence the fundamental theorem applies. The unique smooth map 
\begin{align}
(\ka\ga): \ (I,s_0) & \ra (G,g_0), \\
s \ & \mapsto (\mb{a},A)(s), \nonumber
\end{align}
satisfying the system of ordinary differential equations
\begin{equation}\label{eq:development}
\ga^\ast\varpi(\pa_s) \ = \ (\ka\ga)^\ast\om_G(\pa_s) \qquad \Leftrightarrow \qquad \left\{\begin{aligned}
   -A^{-1}\frac{d\mb{a}}{ds} \ &= \ \theta(\dot{\ga}(s)),  \\
   A^{-1}\frac{dA}{ds} \ &= \ \om(\dot{\ga}(s)),
  \end{aligned}\right.
\end{equation}
is called the \textbf{development of $\varpi$ on $G$ along $\ga$} starting at $g_0$. For the affine connections, the homogeneous second part of~\eqref{eq:development} is known as the \textbf{holonomy} equation and gives the (linear) parallel transport of \emph{free} vectors. The vectorial first part is less familiar: it applies to radius vector (and, consequently, to all bound multivectors), s.t. paths in $\mc{M}\cong P/H$ are developed onto $\mathbb{M}$.
\end{definition}

(Remark: It actually suffices that $\varpi$ be merely piecewise-smooth on $I$ for there to be a unique development. Upon performing partitions $I=\bigcup_{i=1}^n[s_{i-1},s_i]$, it will be continuous everywhere and smooth where $\varpi$ is smooth. If the notion of development is to be used for regularization in Quantum Gravity, the caveats of this sort may be significant for the issues of subsequent refinement and continuum limit.)

One calls a \emph{path} an equivalence class $[\ga_{ij}]$ of (parameterized) curves, connecting two given points $p_i,p_j\in P$, modulo a finite number of retracings and reparametrizations of the curve (in a suitable category, s.a. piecewise-/smooth, or analytic~\footnote{Let us mention that for a smooth $P$, every continuous $\la:(I,a,b)\ra (P,p,q)$ is homotopic to a smooth path. Any smooth $\la_1,\la_2:(I,a,b)\ra(P,p,q)$ that are continuously homotopic are also smoothly homotopic~\cite[p.119]{Sharpe1997Diff-Geometry-Cartan}.}). The natural operation of composition $[\ga_{ij}]\circ [\ga_{jk}]$ is then defined, being associative $([\ga_{ij}]\circ [\ga_{jk}])\circ[\ga_{kl}] = [\ga_{ij}]\circ([\ga_{jk}]\circ[\ga_{kl}])$, as well as the reverse path $[\ga_{ij}]^{-1}=[\ga_{ji}]$, s.t. $[\ga_{ij}]\circ[\ga_{ij}]^{-1}=\mr{Id}_i$. The set $\{P, \Upsilon\}$ of `points-objects' $p\in P$, together with `paths-morphisms' $[\ga]\in \Upsilon$, then forms a \emph{groupoid}, where each map $[\ga_{ij}]\in\mr{hom}(p_i,p_j)$ is an isomorphism. One then calls a `path' any of its representative curves (as was implied in Def.~\ref{def:development}).

\begin{proposition}
The following properties of development characterize $\ka$ as a \emph{functor} mapping from $\{P, \Upsilon\}$ to the principal gauge group $G$, s.t. the above algebraic structure is preserved~\cite[p.120]{Sharpe1997Diff-Geometry-Cartan}.
\begin{enumerate}[label={\upshape(\alph*)}, align=left, widest=iii]
\item Let $\ga:(I,s_i,s_j)\ra (P,p_i,p_j)$ have the development $(\ka\ga):(I,s_i,s_j)\ra(G,g_i,g_j)$. Then $g(\ka\ga):(I,s_i,s_j)\ra(G,gg_i,gg_j)$ is also a development of $\ga$ (i.e. initial conditions in $G$ may be freely chosen).
\item The inverse $(\ka\ga)^{-1}:(I,s_j,s_i)\ra(G,g_j,g_i)$ is the development along the reversed path $\ga^{-1}:(I,s_j,s_i)\ra (P,p_j,p_i)$.
\item If the second path $\la:(J,s_j,s_k)\ra (P,p_j,p_k)$ has the development $(\ka\la):(J,s_j,s_k)\ra (P,g_j,g_k)$, then the joint path $(\ga\circ\la):(I\cup J,s_i,s_k)\ra (P,p_i,p_k)$ is developed onto $\ka(\ga\circ\la)\equiv(\ka\ga)\circ(\ka\la):(I\cup J,s_i,s_k)\ra (P,g_i,g_k)$.
\end{enumerate}
\end{proposition}

The utility of the notion of Klein gauge is clearly visible. It plays the role of a ``proxy''/intermediary agent while rolling the homogeneous Klein space on the lumpy surface of the Cartan geometry (see. Fig.~\ref{fig:hamster-ball}). During this process, traced out curves and other figures leave an imprint, that might be examined for deviations from homogeneity. Supplying the results of development, $\ka$ thus \emph{contains all the (finite) information about the Cartan connection and geometry} of $\mc{M}\cong P/H$. \ul{We suggest it could be used for regularization} in discrete and quantum gravity: both of the connection $\om$ and metric $\theta$ degrees of freedom. (We elaborate further on this issue in Sec.~\ref{sec:geom-sum} and~\ref{subsec:Cartan-quntization}, respectively.)

In the case of the affine group, the result of the development gives the `Euclidean space of conjugacy along a line'~\cite[Ch.14]{Cartan2001Orthogonal-frame}. To each point $\bar{\ga}(s)\in\mc{M}$ there corresponds a point $\mb{m}\in\mathbb{M}$, with the linear frame attached to it. Since the integral curves $c(t,x,\mb{v})$ of $\be^{-1}(\mb{v})\in T_x\mc{M}$ for all $\mb{v}\in V$, determine the exponential parametrization $\exp_{\be}:T\mc{M}\ra \mc{M}\times\mc{M}$ as the local diffeomorphism $(x,X)\mapsto (x,c(1,x,\be(\mb{v}))$ from the neighbourhood of the zero-section to the diagonal $(x,x)\in\mc{M}\times\mc{M}$, one actually obtains a development in $\mathbb{M}$ not only of a given line, but also of the neighbourhood of this line in $\mc{M}$. 

\pagebreak

If the development is required to preserve scalar product $\cg_{ij}=\cg(\mb{e}_i^{\ph{i}},\mb{e}_j^{\ph{i}})$, the connection coefficients have to satisfy skew-symmetry:
\begin{align}\label{eq:orthogonal-displacement}
d\cg_{ij} \ & = \ \bra d\mb{e}_i^{\ph{i}},\mb{e}_j^{\ph{i}}\ket+ \bra \mb{e}_i^{\ph{i}},d\mb{e}_j^{\ph{i}}\ket \nonumber\\
& = \  \om_{ij}^{\ph{i}} + \om_{ji}^{\ph{i}} \ = \ 0 , \qquad \text{where} \qquad \om_{ij}^{\ph{i}} \ \equiv \ \cg_{ik} \om^k_{\ph{k}j}, 
\end{align}
upon substitution from~\eqref{eq:frame-displacement}. That is, Lie algebra is restricted to $\mf{h}=\mf{o}(\eta)$ of the corresponding orthogonal (sub-)group $H$.  Assuming the first equation~\eqref{eq:frame-displacement} is integrable~\footnote{The 1st structure equation (zero torsion) implies that the so-called ``natural frame'' can be chosen, s.t. the forms $\theta^i=dx^i$ are exact. It then follows from $\om^i_{\ph{i}j}\wedge dx^j=0$ that coefficients of $\om^i_{\ph{i}j}=\Ga^i_{\ph{i}jk}dx^k$ are symmetric $\Ga^i_{\ph{i}jk}=\Ga^i_{\ph{i}kj}$, corresponding to the Christoffel symbols of the Levi-Civita connection: $d\cg_{ij}= (\cg_{il}\Ga^l_{\ph{l}jk}+\cg_{jl}\Ga^l_{\ph{l}ik})\, dx^k$.}, let the forms $\theta^i$ be given, determining the linear element $ds^2\equiv\bra d\mb{m},d\mb{m}\ket=\eta_{ij}\, \theta^i\otimes\theta^j$ in the orthonormal frame. It is then possible to find $m(m-1)/2$ skew-symmetric forms $\om_{ij}^{\ph{i}}=-\om_{ji}^{\ph{i}}$ as the unique solution of the system $d\theta^i+\om^i_{\ph{i}j}\wedge\theta^j=0$. (Cartan's lemma~\cite[\S 79]{Cartan2001Orthogonal-frame}~\cite[Lem.6.3.4]{Sharpe1997Diff-Geometry-Cartan}; for general linear $H=\mr{GL}(V)$, the uniqueness is guaranteed by the Weyl's fundamental theorem of metric geometry~\cite[\S 32]{Cartan2001Orthogonal-frame}. We refer to the latter for the geometric interpretation of zero-torsion in terms of `deformations'.) 

It thus should be clear that

\begin{proposition}[$(P,\mathpzc{K},\varpi,\Theta=0)\leftrightarrow (L(\mc{M}),\cg)$, \text{cf.~\cite[Ch.6, \S 3]{Sharpe1997Diff-Geometry-Cartan}}] The Riemannian geometry on $\mc{M}$, as given by the non-degenerate smooth metric in the calculus on the (linear) frame bundle picture is equivalent (up to overall scale) to the torsion-free Cartan connection on $(P,\mathpzc{K})$.
\end{proposition}


The development provides tools for studying also the global properties of topology on~$\mc{M}$. However, the following result shows that the \emph{trivial} connections are appropriate to cope with the task.

\begin{proposition}
If the connection $\varpi$ satisfies structural equation, the result of the development does not depend on the choice of the path $\ga$, but only on the homotopy class $\llbracket \ga\rrbracket$~\cite[Th.3.7.7]{Sharpe1997Diff-Geometry-Cartan}. 
\end{proposition}

\begin{definition}
The monodromy representation $\Phi_\varpi:\pi_1(P,p)\ra G$ of $\varpi$ is a group homomorphism, obtained via development, starting at $e\in G$, of the trivial connection along all possible loops $\la:(I,\pa I)\ra (P,p)$, based at the fixed frame $p\in P$. ($\pi_1$ is the Poincar\'{e}'s fundamental group of classes of homotopically equivalent loops, i.e. related by a continuous/smooth deformation.) The image of $\Phi_\varpi$ is called the period/ or monodromy group. It is defined only up to conjugacy~\cite[Th.3.7.11]{Sharpe1997Diff-Geometry-Cartan}.
\end{definition}

\pagebreak

Returning to the arbitrary connections, the following property demonstrates the relationship between vector-horizontal and homogeneous parts of the development.

\begin{proposition}
If $(\ka\ga)$ is the development of $\varpi$ along the curve $\ga$, starting at $g_0=\ka\ga(s_0)\in G$, and $h:I\ra H$, then the development along $\ga h$ starts at $g_0 h(s_0)$, and related as $(\ka\ga h)=(\ka\ga)h$ to the first one. As a consequence, the \textbf{development of the path $\bar{\ga}:I\ra\mc{M}\cong P/H$} on $\mathbb{M}\cong G/H$ is defined as $\overline{\ka\ga}\equiv\pi_G(\ka\ga)$, being independent on the particular lift $\ga\subset P$~\cite[Ch.6, \S 4]{Sharpe1997Diff-Geometry-Cartan}.
\end{proposition}
This reaffirms once again that the concept of the Klein gauge is a viable one, and was chosen correctly. It allows one to map almost effortlessly between the curves on two spaces, and transmit their properties. Among all possible curves, the most `simple' ones are of special interest.

\begin{definition}
In an arbitrary Cartan geometry $(P,\mathpzc{K},\varpi)$, a \textbf{generalized circle} on $\mc{M}\cong P/H$ is defined as the $\pi$-projection of an integral curve of the $\varpi$-constant vector field on $P$. In the homogeneous Klein geometry this corresponds to the $\pi_G$-projections of the left $G$ translates of the one-parameter subgroups of $G$. (The definition extends the notion of geodesics beyond the case of the affine/reductive Lie algebra $\mf{g}=\mf{g}/\mf{h}\oplus\mf{h}$, whose straight lines are generated by $\mf{p}=\mf{g}/\mf{h}$.)
\end{definition}

\begin{proposition}
For the torsion-free Levi-Civita $\om$, the `straightest' paths of Cartan connection coincide with the `shortest' geodesics in Riemannian geometry, extremizing the arc-length $\de\int_I\theta =0$ w.r.t. path variations. This extends to the cases with $\Theta\neq0$, whenever (the so-called irreducible `vector' component of) torsion satisfies $\Theta^i_{\ph{i}ij}=0,\ j=1,...,m$. In~other words, $\Theta^i$ has to be orthogonal to the line segment; cf.~\cite[Ch.IV]{Cartan1986Affine-connections}, also the proof for Riemannian case in~\cite[\S 116]{Cartan2001Orthogonal-frame} could be easily adapted.
\end{proposition}

\subsection{Generalized tensors. Universal covariant derivative}
\label{subsec:universal-der}

Apart from the frames and their relations of parallelism, one needs the space of objects, which are ``being framed''. The experience with Riemannian geometry and GR brings the following notion of a tensor as a set of numbers, transforming by the certain law, and corresponding to some physical/geometric `entity'. In all possible frames they are encoded via equivariant maps $C(L(\mc{M}),T^{(p,q)})\equiv\{f:L(\mc{M})\ra T^{(p,q)}|f(u\circ A)=A^{-1}\cdot f(u) \ \text{for all} \ u\in L(\mc{M}), \, A\in \mr{GL}(m,\mathbb{R})\}$, where $T^{(p,q)}$ is the space of multilinear functions $f:\hat{\mathbb{R}}^m\times\stackrel{p}{\cdots}\times\hat{\mathbb{R}}^m\times\mathbb{R}^m\times\stackrel{q}{\cdots}\times\mathbb{R}^m\ra \mathbb{R}$, on which the representation is defined as $(A\cdot f)(\hat{v}_1,...,\hat{v}_p,v_1,...,v_q)=f(\hat{v}_1\circ A,...,\hat{v}_p\circ A,A^{-1}v_1,...,A^{-1}v_q)$. (The usual language of tensor calculus in index notation may be retrieved from here via evaluation on basis elements.) This is equivalent to the coordinate-independent characterization as sections of the tensor bundle $T^{(p,q)}(\mc{M})\equiv\bigcup_{x\in\mc{M}}T^{(p,q)}(T_x(\mc{M}))$. On the one side, this seems too broad for our purposes -- the notion of equivalence by general coordinate transformations being too wide. On the other hand, not all representations could be obtained this way (e.g., if we want to incorporate fermions, or use infinite dimensional Hilbert spaces, at some stage). 

We therefore adopt the following
\begin{definition}\label{def:G-tensors}
Let the principal group $G$ act on $W$ via certain representation $\rho:G\ra\mr{Diff}(W)$ [most often, one can linearize $\mr{GL}(W)$]. \textbf{Tensors of type $(W,\rho)$} on the affine Klein bundle $(P,\mathpzc{K})$, are taken to be equivariant functions and forms 
\begin{equation}\label{eq:G-tensors}
\La^k(P,W)\ \equiv\ \{f:P\ra W|f(pg)=\rho(g^{-1})f(p)\}.
\end{equation}
\end{definition}

\begin{proposition}
As usual, an alternative invariant characterization is possible in terms of equivalence classes $[pg,f]=[p,\rho(g)f]$. We posit the existence of the following natural isomorphisms: 
\begin{subequations}
\begin{equation}
\Ga(TP/G) \ \cong \ \Ga^G(TP),
\end{equation}
to the space of the $G$-invariant vector fields, i.e. sections $X$ of $TP\ra  P$, s.t. $X_{pg}=R_{g\ast}X_p$; and
\begin{equation}
\La^k(P,W) \ \cong \ C(P,W \otimes \wedge^k \mf{g}^\ast),
\end{equation}
to the space of equivariant functions, where $\mf{g}^\ast$ transforms in the anti-adjoint representation $\mr{Ad}^\ast(G)$; then
\begin{equation}
\bar{\La}^k(P,W) \ \cong \ C(P,W\otimes\wedge^k \mf{p}^\ast)
\end{equation}
characterizes horizontal forms (spanned by $\theta$), for which the arguments in $\mf{p}=\mf{g}/\mf{h}$ can be ``fed''.
\end{subequations}
(These correspondences have been excerpted by us from the geometric gauge theory~\cite[p.195]{Sharpe1997Diff-Geometry-Cartan}, and groupoid approach~\cite[p.86]{Mackenzie2005groupoids-algebroids}, and extended to the $G$ symmetry of $(P,\mathpzc{K})$. We leave them without proofs, which we guess, should be more-or-less straightforward generalization of the existing ones for the $H$-bundle. Together with~\eqref{eq:soldering}, they should verify the equivalence of different approaches.)
\end{proposition}

In the ordinary gauge theory, the expressions in the local/bundle charts $\tilde{\si}^\ast f:\bar{U}\ra W$ are obtained using trivializing sections $\tilde{\si}:\bar{U}\ra P$. In the case of $(P,\mathpzc{K})$, tensors can be pulled back by the Klein gauge $\psi:\ka(P)\ra P$ to the local expressions/or images $f_\ka\equiv\psi^\ast f$ on $G$ in the chart $(U,\ka)$. Essentially, one can map back and forth between $P$ and $G$ using $\ka$ (at the point and in the local neighbourhood). We will therefore always assume that for all practical purposes (like integration and differentiation), the appropriate local gauge has been chosen. In other words, we can safely put $p=(\mb{m},\mb{e}_1^{\ph{i}},...,\mb{e}_m^{\ph{i}})$ to be seen as `exact' expression (pointwise), framing the images of vectors and multivevector valued forms. For instance, $f(p)= \mb{X}(\mb{m},e)=\mb{e}_i^{\ph{i}}X^i$ would give the coordinates of a free vector in the frame centered at $\mb{m}$, while $f(p)=\mb{X}_0(\mb{m},e)=\mb{m}+\mb{e}_i^{\ph{i}}x^i$ brings the radius-vector of a point. The latter will serve the origin of the new frame $p'=(\mb{m}',e')$. In the most of circumstances, the arguments $p$ may be omitted, being clear from the presentation~\footnote{It is said in~\cite{Willmore1959Diff-Geometry} that \blockquote{E. Cartan regarded such concepts as infinitesimal elements, small quantities of the second order, etc., as quite precise entities, and his geometrical insight was amply powerful to prevent his falling into the numerous pitfalls which normally accompany such concepts. For example, he would regard~\eqref{eq:frame-displacement} as established by treating the geometrical interpretation as fundamental, and interpreting this geometrical fact symbolically.}

It appears that such symbolic calculus is more well-founded by placing the tangent elements in the value space, instead of the base manifold, through the notion of Klein gauge and using the `osculation'~\eqref{eq:osculation}, in particular.}.

In the spirit of the gauge theory and Levi-Civita's (linear) parallelism of GR, one could also define the notion of \textbf{universal covariant derivative} as
\begin{equation}\label{eq:universal-der}
\nabla_{\mb{X}} f \ := \ \varpi^{-1}(\mb{X})f \ \equiv \ X(f),
\end{equation}
that is generated by the fundamental vector fields $\varpi(X)=\mb{X}=(\mb{q},\mb{Q})\in\mf{g}=\mf{p}\oplus\mf{h}$ of a given Cartan connection. The calculation gives the absolute variation of the field as
\begin{align}\label{eq:abs-var}
\nabla_{\mb{X}} f(p) \ & = \ \frac{d}{ds}f(p\exp(s\,\mb{X}))\bigg|_{s=0}  \nonumber \\
& = \ \frac{d}{ds}\rho(\exp(-s\,\mb{X}))f(p)\bigg|_{s=0} \ = \ -\rho_\ast(\mb{X})f(p).
\end{align}
It embraces the definition of~\cite[p.194]{Sharpe1997Diff-Geometry-Cartan}, and modifies it in two regards:
\begin{enumerate}[label={\upshape(\arabic*)}, align=left, widest=iii]
\item The horizontal part of the derivative is also generated by the group action, acquiring the form of a gauge-transformation.
\item It is able to move points and change positions of bound tensors, \ul{providing the gauge-theoretic description of the (local) diffeomorphism group}.
\end{enumerate}
(\emph{There is no difference between covariant and Lie derivatives in that picture!})
This looks unusual, so lets provide more details. There are three principal cases to consider: 
\begin{enumerate}[label={\upshape(\alph*)}, align=left, widest=iii]
\item Rotation/Lorentz transformation. For instance, applying $X=\om^{-1}(\mb{Q})$ to the free vector $f(p)=\mb{e}_i^{\ph{i}}Y^i\equiv\mb{Y}$, one obtains
\begin{equation*}
\nabla_{\mb{Q}}\mb{Y} \ = \ -\mb{Q}\cdot \mb{Y} \ = \ - \mb{e}_i^{\ph{i}}\, Q^i_{\ph{i}j}Y^j.
\end{equation*}
For the orthogonal group $H=\mr{O}(\eta)$, the bivector representation is isomorphic to the adjoint $(\mf{h},\mr{ad}(H))$, so that
\begin{equation*}
\nabla_{\mb{Q}}(\mb{Y}_1\wedge\mb{Y}_2) \ = \ -[\mb{Q},\mb{Y}_1\wedge\mb{Y}_2] \ = \ - \mb{e}_i^{\ph{i}}\wedge \mb{e}_j^{\ph{i}}\, (Q^i_{\ph{i}k}Y_1^kY_2^j + Q^j_{\ph{i}k}Y_1^iY_2^k)
\end{equation*}
and similarly for higher multivectors. These are just infinitesimal variations, generated by the linear transformation of the frame at the fixed point.
\item The points themselves are shifted by applying translations. Take, for instance, $f(p)=\mb{m}+\mb{0}$ to be the zero radius vector, representing origin of the frame $p=(\mb{m},e)$. Consider the curve $\ga:(I,0)\ra (P,p)$, s.t. $\dot{\ga}(s)=\ga_\ast(\pa_s)=X^{\mb{q}}_{\ga(s)}=\theta^{-1}_{\ga(s)}(\mb{q})$ at $p=\ga(0)$. Then we have
\begin{align*}
\nabla_{\mb{q}}\mb{m}(p) \ & \equiv \ df(X^{\mb{q}}_p) \ = \ \frac{d}{ds}f(pe^{s\mb{q}})\bigg|_{s=0} \nonumber \\ 
& = \ \frac{d}{ds}(\mb{m}+s\mb{q}+...)\bigg|_{s=0} \ = \ \theta(X^{\mb{q}}_p),
\end{align*}
which coincides with the first equation in~\eqref{eq:frame-displacement}. 

\item Analogously, for the translation of a free vector, one takes the linear frame itself $f(p)=\mb{e}_j^{\ph{i}}$ (having coordinates uniformly constant and equal 1). Applying the same procedure as above:
\begin{equation*}
\nabla_{\dot{\ga}}\mb{e}_j^{\ph{i}}(p) \ = \ \frac{d}{ds}\mb{e}_j^{\ph{i}}\, e^{s\om(\dot{\ga})}\bigg|_{s=0} \ = \ \mb{e}_i^{\ph{i}} \, \om_p(\dot{\ga})^i_{\ph{i}j},
\end{equation*}
the second equation of~\eqref{eq:frame-displacement} is reproduced, if the frame does not stay parallel. (Notice that in order to describe how vector changes from point to point, it is not enough to give just two nearby points of $\mc{M}$, one must provide two nearby frames in $P$. The above change could be corrected by performing the gauge transform $\mc{G}(P,\mathpzc{K})$, making two frames, or rather their images in $G$, parallel.)
\end{enumerate}

The coherency of the framework is observed, providing justification for the frame's infinitesimal displacement formulas. The expressions in vector notation (``active'') could be supplemented by equivalent formulas, relating coordinates of a point in the affine space at $\mb{m}$ to its corresponding coordinates in $\mb{m}+d\mb{m}$ (describing the same point in the ``passive'' view of group transformations):
\begin{subequations}
\begin{align}
\nabla_{\dot{\ga}} \mb{X}_0(p) \ & = \ \mb{e}_i^{\ph{i}} \left(\frac{dx^i}{ds}+\theta^i(\dot{\ga})+\om(\dot{\ga})^i_{\ph{i}j}x^j\right) \ = \ 0, \\
\nabla_{\dot{\ga}} \mb{X}(p) \ & = \ \mb{e}_i^{\ph{i}} \left(\frac{dX^i}{ds}+\om(\dot{\ga})^i_{\ph{i}j}X^j\right) \ = \ 0.
\end{align}
\end{subequations}
Assuming parallel frames, one could also note that $df(\mb{m}+s\mb{q},e)/ds=-\rho_\ast(\mb{q})f(\mb{m},e)$ -- the action of translations is generated by the usual shift operator $\rho_\ast(\mb{e}_i^{\ph{i}})=-\pa_i$. The above formulas are just repercussions of the (infinitesimal) equation of development. 


The 2nd expression corresponds to the usual covariant derivative, and the notion of parallel (covariantly constant) vector fields on $\mc{M}$~\footnote{As have been said in\cite[p. 111]{Bleecker1981gauge-variational-principles}:~\blockquote{Since everyone has a different definition of covariant differentiation, we leave it to you to show that this is correct for your definition.}}. The 1st equation describes in coordinates the radius-vector undergoing a Cartan displacement. The resulting change around a small contour is given by $\Delta x^i+\Theta^i+\Om^i_{\ph{i}j}x^j=0$. Any point of $\mc{M}$ has a neighbourhood admitting $m$ linearly independent constant vector fields iff $\Om=0$. In this case, then $\Theta=0$ is equivalent to the existence of a radius-vector field defined over some neighbourhood of every point~\cite{Trautman1970fibre-bundles}.

\chapter{The Einstein-Cartan theory of relativity}
\label{ch:Einstein-Cartan}

In theoretical physics, one most often encounters two distinct notions of `geometry', roughly corresponding to the dualism between the modern theories of fundamental interactions. 
\begin{enumerate}[label={\upshape(\Alph*)}, align=left, widest=iii]
\item In the Yang-Mills type gauge theories, the state of the particle $\psi\in\Ga(\xi)$ is typically viewed as taking values in some `internal' vector space~$F$, attached at every point of base manifold~$\mc{M}$; whereas the `connective' geometry prescribes (infinitesimally, via gauge potential 1-forms $A$) how these various fibers are related. This is the situation in the Standard Model (SM) of particle physics, where some `external' (absolute) background is required in addition, to govern the locations and motion in $\mc{M}$. \label{descr:SM}
\item The latter is precisely the task of the second type of geometry, endowing $\mc{M}$ with some metric properties (typically in the Riemannian formulation). The flat spacetime picture could be considered exact -- if the gravitational interactions (weak at the typical scales of SM) are disregarded; however, it is viewed, instead, as only approximate description -- valid in the local neighbourhood.  In general, the geometry may be non-trivially `curved' -- according to GR, whose field equations govern the dynamics of the background itself. \label{descr:GR}
\end{enumerate}


The attempts have been undertaken to conceptualize two theories in each other's terms. First, the space in GR is curved by the certain Levi-Civita connection $\om$ of the first type~\ref{descr:SM}. So the main challenge is the role of metric $\cg$ and the relation between `internal-vs-external' symmetries~\cite{Petti2006translational-symmetries,Catren2015Cartan-gauge-gravity}. The preference given to $\om$ or $\cg$ may lead to various gauge theories of gravity $(\mr{B})\subset(\mr{A})$~\cite{BlagojevicHehl2013,GotzesHirshfeld1990MacDowellMansouri-geometry,StelleWest1980broken-deSitter-holonomy,Wise2010Cartan-geometry}, or the Kaluza-Klein models $(\mr{A})\subset(\mr{B})$~\cite[Ch.9]{Bleecker1981gauge-variational-principles}, correspondingly. 

We prefer more conservative/geometric approach to connections, fully exposed in the previous Ch.~\ref{sec:Cartan-affine}. This section is intended to further highlight (and remind) how the original Einstein-Cartan gravity~\cite{Cartan1986Affine-connections} is already a perfect gauge-theoretic framework for GR in its own right~\footnote{It is also sometimes referred to as Einstein-Cartan-Sciama-Kibble theory, after the re-discovery of the corresponding lagrangian in the works~\cite{Kibble1961gauge-trick-ECSK,Sciama1964ECSK}, following the Utiyama's extrapolation~\cite{Utiyama1956gauge-trick,Trautman1970fibre-bundles,Bleecker1981gauge-variational-principles} of the Yang-Mills hypothesis~\cite{Yang-Mills1954} to all interactions. We only briefly mention these developments here, since this route goes somewhat against our own thinking about gauge theory in terms of ``generalized relativity''. Historically, one can trace the motivations for the general idea of connection back to GR and Levi-Civita's parallelism~\cite{Iurato2016Levi-Civita-history}, from where it was abstracted by Weyl~\cite{Weyl1952Space-Time-Matter} and Cartan~\cite{Cartan1986Affine-connections}, among others, and later formalized in Ehresmann's works~\cite{Marle2014fromCartan-toEhresmann}. It was probably not until the Yang-Mills paper, soon thereafter, that the wider community of [particle] physicists became closely acquainted with these abstract concepts -- which, nevertheless, ``were not dreamed up. They were natural and real.'' (cf. discussion in~\cite{Yang1977gauge-geometry}).}. It goes hand-in-hand with the Einstein's intuition and his original insights on the nature of inertia and acceleration. 

The preference given to this model of gravity is not meant to stress the certain form of the action functional, or the specific modification, which allows the coupling of non-trivial torsion with angular momentum. Rather, our goal is to demonstrate that the formalism of the theory due to Cartan is structurally rigid and geometrically clean, leading unequivocally to both the Einstein equations and the aforementioned generalization. As for the prospects of quantization -- using Loop or Foam approaches -- in our view, the provided techniques are capable of supplying mathematically more sound footing for the discretization and diffeomorphism symmetry, as well as the notion of observables.

\section{Energy-momentum and spin- tensors}
\label{subsec:energy-momentum}

Consider the ordinary Minkowski space-time ($\mathbb{M}^m$,$\eta$) of special relativity with $m=4$. It is endowed with the trivial affine connection, satisfying structure equations~\eqref{eq:structure-0}. The notion of equivalence of frames have an absolute meaning, the development being independent on the particular path. Suppose that all $e\equiv\{\mb{e}_i^{\ph{i}}\}$ have been put parallel to a particular frame $p=(\mb{m}_0,e)$ at the point $d\mb{m}_0=0$, such that $d\mb{e}_i^{\ph{i}}=0$, and the forms $\om$ everywhere vanish. From the first structure equation $d\theta=0$, the forms $\theta$ are then exact differentials. Hence, without loss of generality, one can set $d\mb{m}=dx^i\, \mb{e}_i^{\ph{i}}$, $x$~being an affine coordinates for the entire manifold.

In a different vein, a spacetime vector has components $\mb{X}=(0,X^i)=(0,y^i-x^i)$, obtained by subtracting the coordinates of two events at its origin $\mb{x}$ and extremity $\mb{y}$, respectively. The coordinates $X^i$ will be identical in all parallel frames $\mb{e}'^{\ph{i}}_i =\mb{e}_i^{\ph{i}}$ that are simply translates of each other $o'=o+\mb{a}$. This establishes the notion of equivalence (reflexive $\mb{X}\sim \mb{X}$, symmetric $\mb{X}\sim \mb{Y} \ \Rightarrow \ \mb{Y}\sim \mb{X}$, and transitive $\mb{X}\sim \mb{Y}, \, \mb{Y}\sim \mb{Z} \ \Rightarrow \ \mb{X}\sim \mb{Z}$) that constitutes an affine structure. When the frame is rotated/boosted/or rescaled, the components $X'^i=A^i_{\ph{i}j}X^j$ are related by $A\in\mr{GL}(m,\mathbb{R})$, s.t. the linear structure $[\la\mb{X}]\approx \la[\mb{X}]$ and $[\mb{X}+\mb{Y}]\approx [\mb{X}]+[\mb{Y}]$ is preserved. In~addition, one requires the interval $|\mb{X}'|^2=|\mb{X}|^2$ to be invariant under $H=\mr{O}(\eta)$, reducing the admissible frame transformations to the Poincar\'{e} (sub-)group $G=V\rtimes H$. It is often convenient to choose the orthonormal linear basis $\bra\mb{e}_i^{\ph{i}},\mb{e}_j^{\ph{i}}\ket = \eta_{ij}$, which is always possible.


Now, consider the point particle of the mass $\mu$, moving along $\ga:I\ra\mathbb{M}\cong P/H$ with instantaneous velocity $\dot{\ga}$, s.t. its vector of energy-momentum is 
\begin{equation}\label{eq:momentum-vector}
\mb{P} \ \equiv \ \mu \, d\mb{M}(\dot{\ga}) = \mu \, \frac{dx^i}{ds}\, \mb{e}_i^{\ph{i}}.
\end{equation}
(To conform with the notation of Sec.~\ref{sec:global-geometry}, let $\mb{M}=\mb{m}_0+\mb{m}\equiv \mb{m}_0+\mb{e}_i^{\ph{i}}\, x^i$ denote the position $\mb{m}$ of the particle w.r.t. some fixed-but-arbitrary frame $p=(\mb{m}_0,e)$; we allow some sloppiness w.r.t. `big/small' letter for a radius in these few sections. Parameter $s$ is normally a proper time.) The `principle of inertia' then states that it remains constant in direction and magnitude in the absence of interactions:
\begin{equation}\label{eq:acc-0}
\dot{\mb{P}} \ \equiv \ \mu \, \frac{d^2x^i}{ds^2} \, \mb{e}_i^{\ph{i}} \ = \ 0,
\end{equation} 
(valid modulo global frame transformations with constant coefficients).
In other words, due to acceleration being zero, the velocity stays parallel to itself, and the particle's trajectory is a straight line. This preservation of momentum is trivially the consequence of the homogeneity of space w.r.t. translations.

There is another known conservation law of angular momentum, due to isotropy w.r.t. rotations/Lorentz transformations. They can be concisely put together using the formalism of multivectors from Ch.~\ref{sec:global-geometry}. Namely, consider a \emph{sliding} vector
\begin{equation}\label{eq:sliding-momentum}
\mb{P}_0^{(1)} \ := \ \mb{M}\wedge \mb{P} \ = \ P^i(\mb{m}_0^{\ph{i}}\wedge\mb{e}_i^{\ph{i}}) + \frac12(x^iP^j-x^jP^i)\, \mb{e}_i^{\ph{i}}\wedge\mb{e}_j^{\ph{i}}, \qquad P^i \ = \ \mu\, dx^i/ds,
\end{equation}
consisting of the (bound) linear momentum, and the bivector of the respective angular momentum (orbital). It is actually a Lie algebra $\mf{g}=\mf{p}\oplus\mf{h}$ valued quantity, and the preservation $\dot{\mb{P}}_0^{(1)}=0$ is a consequence of the symmetry w.r.t. (global) $G$ action.

Similarly, the dynamics of continuous media can be formulated via vector-valued 3-form:
\begin{equation}\label{eq:medium-momentum}
\tilde{\mb{P}}^{(3)} \ = \ \mb{e}_i^{\ph{i}} \, \tilde{P}^i, \qquad \text{where} \qquad \tilde{P}^i \ = \ \tilde{P}^{ij}\, \frac{1}{3!}\varepsilon_{jklm}^{\ph{i}}\, dx^k\wedge dx^l\wedge dx^m,
\end{equation}
and $\tilde{P}^{ij}=\mu \,  U^i U^j+p^{ij}$ is the usual energy-momentum tensor of an effective hydrodynamic description. Here $\mu$ should be understood as the matter density per unit volume, $\mb{U}=\mb{e}_i^{\ph{i}}\, dx^i/ds$ giving the (averaged) velocity of each element. Whereas the pressure $p$ is considered to be the flux of momentum resulting from irregularities of molecular velocities (``internal stresses''), s.t. $U^ip_{ij}=0$. It is constructed from the complementary (orthogonal) vector of a hypersurface trivector:
\begin{equation}\label{eq:hypersurface-dual}
\star\mb{\Si}^{(3)} \ = \ \star (d\mb{m}\wedge d\mb{m}\wedge d\mb{m}) \ = \ \frac{1}{3!}\Si^{ijk}\varepsilon_{ijkl}^{\ph{i}}\, \hat{\mb{e}}^l, \qquad \Si^{ijk} \ = \ dx^i\wedge dx^j\wedge dx^k,
\end{equation}
in such a way that the scalar product
\begin{align}\label{eq:total-mass}
\bra \tilde{\mb{P}}^{(3)},d\mb{m}\ket \ & = \ \mu \, \bra \star\mb{\Si}^{(3)},d\mb{m} \ket \ = \ \mu \star (\mb{\Si}^{(3)} \wedge d\mb{m}) \nonumber \\
& = \ \mu \star \mb{\Si}^{(4)}  \ = \ \mu \, \frac{1}{4!}\varepsilon_{ijkl}^{\ph{i}}\, dx^i\wedge dx^j\wedge dx^k\wedge dx^l
\end{align}
gives the total mass of the matter (in its rest frame; recall that $d\mb{m}=dx^i\mb{e}_i^{\ph{i}}$, $ds^2=\bra d\mb{m},d\mb{m}\ket$). 

The conservation law for the energy-momentum [flowing in a world tube] is then simply stated:
\begin{equation}\label{eq:energy-conserv}
d\tilde{\mb{P}}^{(3)} \ = \ 0,
\end{equation}
leading to the Euler's equations of the fluid dynamics. Analogously to the particle case, there is additionally an angular momentum preservation:
\begin{equation}\label{eq:angular-conserv}
d(\mb{M}\wedge\tilde{\mb{P}}^{(3)}) \ = \ \frac12 (dx^i\wedge \tilde{P}^j) \,  \mb{e}_i^{\ph{i}}\wedge \mb{e}_j^{\ph{i}} \ = \ 0,
\end{equation}
taking into account~\eqref{eq:energy-conserv}. In particular, it leads to the symmetry $p^{ij}=p^{ji}$ of energy-momentum tensor.

(Remark: It might appear that too many wedges are involved, or the expressions seem ambiguous. We find it convenient to consider the wedge-product of the form components (playing the role of coordinates) as `induced' by the skew-symmetric multiplication of vectors. For instance, $d\mb{m}\wedge d\mb{m}=\mb{e}_i^{\ph{i}}\wedge\mb{e}_j^{\ph{i}}\, dx^i\otimes dx^j=\mb{e}_i^{\ph{i}}\wedge\mb{e}_j^{\ph{i}}\frac12(dx^i\otimes dx^j-dx^j\otimes dx^i)\equiv \mb{e}_i^{\ph{i}}\wedge\mb{e}_j^{\ph{i}} \frac 12 (dx^i\wedge dx^j)$, by Def.~\eqref{eq:wedge}. Every next such multiplication $d\mb{m}\wedge (.)$ brings the required symmetry factor. [Such vector notations may be not widespread, but they allow to always ``keep track'' of the frame chosen, and the nature of the objects may seem clearer.]) 

\paragraph{A quick reminder} The key element of what follows is the Stokes integration~\eqref{eq:Stokes-1}, being dual to a differential (co-boundary) operator. In order to conform with the usual terminology, let us succinctly recap how it works, in the geometric context. Consider: 1) the closed contour $\ga=\pa S$, bounding the 2-dim parameterized surface $S=S(\al,\be)$; 2) the position on the contour is specified by a radius vector $\mb{m}=\mb{m}_0+\mb{e}_i^{\ph{i}}\,x^i$, and the tangential direction -- via $d\mb{m}=\mb{e}_i^{\ph{i}}\,dx^i$ (evaluated on~$\dot{\ga}$); 3) the vector $\mb{Q}=\mb{Q}(\mb{m})$ of the certain physical/geometrical quantity is bound and projected to $\ga$, obtaining the scalar 1-form $\varphi:=\bra \mb{Q},d\mb{m}\ket\equiv Q_i(x)\,dx^i$. 

One decomposes the entire region into curvilinear parallelograms defined by intersection of coordinate lines. When performing summation, the integrals along internal boundary lines are being passed twice in the opposite directions. Hence they cancel in pairs, and it suffices to restrict ourselves with the domain of an elementary parallelogram, corresponding to infinitesimal parameters $(d\al,d\be)$:
\begin{gather*}
\oint\varphi \ = \ \frac12 (\pa_i Q_j-\pa_j Q_i)\, J^{ij}\, d\al \, d\be, \\ 
\text{where} \qquad J^{ij} \ \equiv \ \frac{\pa(x^i,x^j)}{\pa(\al,\be)} \ = \ dx^i\wedge dx^j(\pa_\al,\pa_\be) \ = \ \begin{vmatrix}
\frac{\pa x^i}{\pa\al} & \frac{\pa x^i}{\pa\be} \\
\frac{\pa x^j}{\pa\al} & \frac{\pa x^j}{\pa\be} ,
\end{vmatrix}
\end{gather*}
-- the determinant of the Jacobian matrix~\footnote{Indeed, for the parallelogram $ABCD$ passed in alphabetic order, let the differentiation $d_1=d\al\, \pa_\al$ be made along lines $AB$ and $DC$, and $d_2=d\be\, \pa_\be$ along $AD$ and $BC$, respectively. One has then: $\int_A^B=\varphi(d_1)$, $\int_B^C=\varphi(d_2)+d_1\varphi(d_2)$, $\int_C^D=-\varphi(d_1)-d_2\varphi(d_1)$, $\int_C^D=-\varphi(d_2)$. The sum gives the right answer~\cite[p.21]{Cartan2001Orthogonal-frame}.}. 

This is just a familiar integration of scalar forms, independent on the choice of axes. Consider instead the sliding vector tangent to the contour $\dot{\ga}$, ``sitting at the tip'' of the radius $\mb{m}$ and swapping $\ga$ in a circular manner: 
\begin{equation}\label{eq:bivector-surface}
\oint \mb{m}\wedge d\mb{m} \ = \ \mb{e}_i^{\ph{i}}\wedge\mb{e}_j^{\ph{i}} \, \frac12J^{ij}\, d\al \, d\be.
\end{equation}
One obtains the (simple, free) bivector~\eqref{eq:bivector-simple} of a surface element $\mb{\Si}^{(2)}$. The factor of 2 provides the correct normalization~\eqref{eq:bivector-area}. The result for the scalar form can be seen as contraction of two bivectors, giving projection of the spatial variation of $Q$ across the surface~$S$. In $m=3$ this becomes just the curl in the direction of normal $\bra[\nabla\times\mb{Q}],\star\mb{\Si}^{(2)}\ket$ (using duality). 

The whole integral over $S$ is then seen as a \emph{geometric sum} of the (system of) bivectors, for which the simplicity is not satisfied (i.e. $S$ is not planar, in general). When the boundary $\ga=\pa S$ shrinks to a point, the total sum over such `cycle' will vanish, since bivector is an exact differential. One obtains the closure condition $\oint\mb{\Si}^{(2)}=0$, as the geometric representation of $\pa S=0$.


\newpage

\section{Equivalence principle and the force of gravity}
\label{sec:energy-momentum}

Einstein's `principle of relativity' states that the physical laws should be formulated invariantly w.r.t. all admissible frames of reference. The formalization of this concept includes the so-called `general-covariance' postulate, s.t. the equations are usually expressed using tensor fields on $L(\mc{M})$, made equivalent under arbitrary smooth reparametrizations of $\mc{M}$. The essential idea is that coordinates do not exist a priori in nature, but are only artifices used in describing nature, and hence should play no role.

In the Cartan's version of GR, the notions of coordinate- and frame- transformations are carefully delineated. Roughly speaking, it is `coordinate-invariant', but `frame-covariant' [in our personal opinion]. The above formulation of dynamics (in flat space-time) relies on the choice of frame (`gauge-fixing'), but the coordinates on the domain $(P,\mathpzc{K})$ are concealed within the neutral language of forms~\footnote{To overcame the intrusive impression of coordinate-dependence, it maybe worth stressing that the parameters $x^i$ appear as 1st integrals, labelling the \emph{solutions} of the first structure equation of zero-torsion (cf. footnote~\ref{foot:Pfaff} in Sec.~\ref{subsec:structure}). So that in Riemannian case they are easily mixed with the free parameters of the frame's position in the local chart on~$\mc{M}$. We must say, this is actually puzzling [sic].}. In order to extend the previous formulas to any variable choice of (local) frames, the gauge transformations $\mc{G}(P,\mathpzc{K})$ are performed.

One thus puts at every occurrence $dx\ra d\mb{m}=\theta$, so that
\begin{equation}\label{eq:mometum-arb}
\mb{P} \ = \ \mu \, \frac{d\mb{m}}{ds} \ = \ \theta(\dot{\ga}), \qquad \mb{P}_0^{(1)} \ = \ \mb{M}\wedge \mb{P}.
\end{equation}
The conservation laws will be made invariant by replacing the derivations of fields to the covariant $\nabla$ from~\eqref{eq:universal-der}, and the exterior differentiation of forms is modified using~\eqref{eq:frame-displacement}, such that~\footnote{In analogy to the covariant differential $D^\om$ of the Ehresmann (part of the) connection, one could in principle introduce the notation $D^\varpi$ to stress the action both on free and bound tensor forms. We restrain, for  the sake of economy of notation, and will assume everywhere that the differential $d$ acts on the frame as in~\eqref{eq:frame-displacement}.}
\begin{equation}\label{eq:energy-conserv-arb}
d(\mb{M}\wedge\tilde{\mb{P}}^{(3)}) \ = \ (\mb{M}\wedge \mb{e}_i^{\ph{i}}) \, [d\tilde{P}^i+\om^i_{\ph{i}j}\wedge \tilde{P}^j] + (\mb{e}_i^{\ph{i}}\wedge \mb{e}_j^{\ph{i}}) [ \theta^i\wedge \tilde{P}^j] \ = \ 0.
\end{equation}

Even more precisely, using the proof for geodesics in~\cite[Prop.6.2.7]{Sharpe1997Diff-Geometry-Cartan}, the curve will develop into a straight line, if it satisfies
\begin{equation}\label{eq:geodesic}
\dot{\theta}^i(\dot{\ga})+\om^i_{\ph{i}j}(\dot{\ga})\, \theta^j(\dot{\ga}) \ = \ \la\, \theta^i(\dot{\ga}),
\end{equation}
namely, if the velocity and acceleration are linearly dependent. (Here $\la$ is an affine parameter; the right hand side will vanish for a geodesic parametrized by arclength $s$.)

\pagebreak

Recall that the space-time is Minkowski $\mathbb{M}$, and connection is trivial, just written in an arbitrary gauge. It is then a simple matter to incorporate acceleration due to (non-trivial) gravitational field of forces. Indeed, the mechanical law tells that the (covariant) derivative of the energy-momentum equals the `hyperforce' vector:
\begin{equation}\label{eq:mech-law}
\nabla_{\dot{\ga}}\mb{P} \ = \ \mb{G}.
\end{equation}
Since the interval $ds^2=\bra d\mb{m},d\mb{m}\ket$ is preserved, one has $\bra \mb{G},d\mb{m}\ket=0$, telling that the infinitesimal work $G^t dt=G^xdx+G^ydy+G^zdz$ done by the force equals the energy change (regarding the infinitesimal displacement as $dx^i\propto\theta^i(\dot{\ga})$, in fact).

At any particular instant one can compensate for $\mb{G}$ via gauge transformation (to the particle's rest frame), and actually for any finite path $\ga$ one can construct such a frame, which will be `freely-falling' in the gravitational field. Naturally, the equations for the affine/flat spacetime will preserve their form for the new notion of the parallel transport, given by $\mb{G}$ -- \emph{but with modified coefficients} $\varpi=\theta+\om$, determining non-trivial/non-integrable connection~\cite{Cartan1986Affine-connections}:~\blockquote{If the definition of equivalence of two nearby Galilean frames -- or, equivalently, the definition of equality of two 4-vectors whose origins are infinitesimally close -- is so chosen as to cancel the gravitational forces, the equations of dynamics would reduce to [$\nabla_{\dot{\ga}}\mb{P}=0$].}~\footnote{In analogy to Riemannian GR, one could call this ``derivative-coupling'' scheme.}

In other words, the motion of the particle agrees with the principle of inertia (constancy of velocity/or momentum preservation) if the equivalence of successive inertial frames is `dictated' by the acceleration due to gravitational field, through the new $\varpi$. The framework of Cartan connections thus implements the Einstein's `principle of equivalence' between gravitation and acceleration directly at the level of (modified) geodesics. In this way, we have demonstrated how the independent excitations of the physical field of gravity are naturally encoded in the geometric language of (non-trivial) connections on $(P,\mathpzc{K})$. 

The formula~\eqref{eq:mech-law} for the (linear) energy-momentum should be complemented with the corresponding law for angular momentum. Or, in terms of analogies from Newtonian mechanics (which are quite precise): the rate of change of the `moment of inertia' $\sim[\mb{r}\times\mb{P}]$ is given by the `moment of force' $\sim[\mb{r}\times\mb{G}]$ (and possibly some non-trivial torque $\mb{H}$). Assuming flat Minkowski space-time $\mathbb{M}$, the (covariant) derivative $\nabla_{\dot{\ga}}\mb{P}^{(1)}_0$ of the \emph{sliding} vector of momentum should be put equal to the corresponding sliding vector of ``force/torque''. One can write it concisely:
\begin{subequations}\label{eq:momentum-law}
\begin{equation}\label{eq:momentum-law-1}
\mb{M}\wedge(\nabla_{\dot{\ga}}\mb{P}) +(\nabla_{\dot{\ga}}\mb{M})\wedge\mb{P} \ = \ \mb{M}\wedge \mb{G} + \mb{H},
\end{equation}
or, in terms of components:
\begin{gather}\label{eq:momentum-law-2}
 (\mb{m}_0\wedge\mb{e}_i^{\ph{i}})\, \{\dot{\theta}^i+\om^i_{\ph{i}j}\, \theta^j\}
+ (\mb{e}_i^{\ph{i}}\wedge\mb{e}_j^{\ph{i}})\, \{x^i\,[\dot{\theta}^j+\om^j_{\ph{i}k}\, \theta^k] + [\dot{x}^i+\om^i_{\ph{i}k}\,x^k]\, \theta^j\} \nonumber \\ 
 = \ (\mb{m}_0\wedge\mb{e}_i^{\ph{i}})\, G^i + (\mb{e}_i^{\ph{i}}\wedge\mb{e}_j^{\ph{i}})\, [ x^i\, G^j + H^{ij}],
\end{gather}
\end{subequations}
where all the forms are evaluated at $\dot{\ga}\in TP$ of the particle's trajectory. We omitted the coefficient of mass, putting it $\mu=1$, s.t. expressions on the right are normalized accordingly (i.e. $\mb{G}$ is acceleration). If~one incorporates both $\mb{G}$ and $\mb{H}$ into the modified connection $\varpi=\theta+\om$ with non-trivial $\Theta,\Om\neq 0$, the motion will be `inertial' w.r.t. the new definition, that is satisfying the conservation of (linear and angular) momentum: $\nabla_{\dot{\ga}}\mb{P}^{(1)}_0=0$.

Analogously, the geometric difference $d(\mb{M}\wedge\tilde{\mb{P}}^{(3)})$ between `incoming' and `outgoing' flows of (effective) energy-momentum -- in the world tube across its 3-dim boundary -- will be given by the pressures and stresses $\mb{M}\wedge \tilde{\mb{G}} + \tilde{\mb{H}}\equiv [\mb{M}\wedge \mb{G} + \mb{H}](\star\mb{\Si}^{(4)})$, due to gravitational forces and/or torques, acting in continuous medium. This change is also accommodated by the modified connection coefficients $\varpi$, such that~\eqref{eq:energy-conserv-arb} is satisfied. Using $\theta^i\wedge\frac{1}{3!}\varepsilon_{jklm}^{\ph{i}}\, \theta^k\wedge \theta^l\wedge \theta^m = \de^i_j\star\mb{\Si}^{(4)}$, one can see that, in general, the energy momentum tensor will not be symmetric: $p^{[ij]}=H^{ij}$ -- namely, if the torsion is non-vanishing.

The above formulas~\eqref{eq:momentum-law} are written in the most explicit form, being omitted/skipped in the Cartan's original derivation~\cite{Cartan1986Affine-connections}. One could provide more context for them in terms of the \emph{kinematics} of curves in an arbitrary spacetime $\mc{M}$. Indeed, the covariant derivative simply gives the (instantaneous) variation of the developed line in $\mathbb{M}$, where the usual notion of parallelism applies. So that the expression for the curve's (free) unit tangent vector $\mb{T}:=\mb{P}/|\mb{P}|$ can be seen as the ordinary limit of the variation $\De\mb{M}/s=(A(s)\cdot\mb{M}_0-\mb{M}_0)/s\xrightarrow[s \to 0]{}\dot{\mb{M}}\equiv\nabla_{\dot{\ga}}\mb{M}=\mb{T}$~\footnote{The translation $\mb{a}$ in the law of development of radius vector does not appear: it contracts in the difference of the two radiuses, when developed to a single frame. The latter could be either some arbitrary third reference frame, or the one attached to the curve (kinematical), s.t. the radius of the point coincides with the free vector (and essentially with translation $\mb{a}$).}. Roughly speaking, the development ``rolls back'' the tangents and other vectors to the single frame at the instant, s.t. one can make sense of addition/subtraction of vector- and tensor quantities in general. The parametrization is chosen in terms of arclength, s.t. $\nabla_{\dot{\ga}}|\mb{P}|=0$.

In other words, the theory of curves can be extended from the affine/flat space-time $\mathbb{M}$ to the Riemann-Cartan space $\mc{M}$, which it osculates at the point. Assuming for simplicity the Euclidean model $\mathbb{E}^3$, the Frenet-Serret frames $\{\mb{I}_i\}\equiv \{\mb{T},\mb{N},\mb{B}\}$ associated to the corresponding curves in two spaces will be the same, consisting of the unit \emph{tangent} $\mb{T}$, \emph{principal normal} $\mb{N}$, and \emph{binormal} $\mb{B}=[\mb{T}\times\mb{N}]$ vectors. They are required to be orthonormal $\bra\mb{T},\mb{N}\ket=\bra\mb{N},\mb{B}\ket=\bra\mb{B},\mb{T}\ket=0$, $|\mb{T}|=|\mb{N}|=|\mb{B}|=1$, and their derivatives be expressed in terms of each other. 

\pagebreak

Using the covariant derivative $\nabla_{\dot{\ga}}$ and expanding w.r.t. $\{\mb{I}_i\}$, one obtains without modification the \textbf{generalized Frenet-Serret formulas}:
\begin{subequations}\label{eq:Frenet-Serret}
  \begin{empheq}[left=\empheqlbrace]{align}
    \dot{\mb{M}} \ &= \ \mb{T}, \\
    \dot{\mb{T}} \ &= \ \ka \, \mb{N}, \\
    \dot{\mb{N}} \ &= \ -\ka \, \mb{T}+\tau\, \mb{N}, \\
    \dot{\mb{B}} \ &= \ -\tau \, \mb{N},
  \end{empheq}
\end{subequations}
where $(\theta^1,\theta^2,\theta^3)=(ds,0,0)$, $(\om^2_{\ph{2}1},\om^3_{\ph{3}2},\om^1_{\ph{2}3})=(\ka,\tau,0)\, ds$, and the quantities $(\ka,\tau)$ are respectively the \emph{curvature} and \emph{torsion} of the curve $\ga$. The `straight' Riemannian geodesics are the lines of zero curvature $\ka=0$. As for the curves of zero torsion, one can characterize them to be planar, i.e. the osculating plane element $\{\mb{T},\mb{N}\}$ is transported parallel to itself iff the torsion vanishes $\tau=0$. In an $m$-dim manifold, one obtains a series of equations, in each of which a new parameter and a new vector are introduced (2nd, 3rd, etc. curvatures $\ka_i$ and normals $\mb{N}_i$; we refer to~\cite[Chs.13,22]{Cartan2001Orthogonal-frame} for further details).

\newpage

\section{On the geometric summation viewed as `coarse-graining'}
\label{sec:geom-sum}

From the differentiation of vector fields we pass to the problem of integration of forms. This  appears quite often as a regularization technique in the modern approaches to Quantum Gravity. In Sec.~\ref{subsec:discretization-problem}, we briefly mentioned the measure theory and algebraic topology among many facets that usually accompany this subject. Having in mind the motivations from gravity, however, we decided to focus on the geometric aspects. In particular, we were puzzled by two things: 
\begin{enumerate}[label={\upshape(\arabic*)}, align=left, widest=iii]
\item The nature and place of the metric degrees of freedom $\cg\sim\theta$. It plays part both of the `native' measure on $\mc{M}$, as well as the dynamical field that should be `smeared'.
\item The role of the frame choice and the gauge-dependence for the vector quantities, which appear among main objects of interest.
\end{enumerate}

It is legitimate to view the integration process as the natural `coarse-graining' operation, performing summation of the elementary/infinitesimal contributions to obtain an effective description. In Sec.~\ref{subsec:energy-momentum}, we recapitulated the basics of Stokes integration in flat space, associating the simple bivector to the element of a surface. The integral is then understood as a geometric sum of multi-vectors, or, equivalently, the algebraic sum of components, in the standard basis~\footnote{We note, in passing, that the appearance of `non-simplicity' is linked with coarse-graining in that paradigm. Trivially, the composition of elementary 2-cells does not span a 2-plane, in general.}.

Now, once the gauge theoretic structure of the Cartan geometry has been revealed, we can supply an accurate handling to the issue of frame dependence. One actually finds this task to be directly addressed in~\cite[p.126]{Cartan2001Orthogonal-frame}, in conjunction with the problem of (absolute) exterior differentiation. The notion of development occupies the central place and appears to be paramount. Let us therefore formulate a few simple rules, that allow an (almost) algorithmic construction of the new geometric forms, associated to the higher dimensional elements, using the generalized Stokes integration.

In an arbitrary geometry $(P,\mathpzc{K},\varpi)$, one cannot add vectors with different initial points, if these points are not connected by certain paths. Consider a frame $p=(\mb{m},e)$ at the point inside the domain in question, and take a family of paths $\ga$, connecting it with different points of the domain. Then we can calculate integrals over surfaces, volumes, etc. Lets figure out first the case of a \emph{free} tensor, for brevity called a vector in $W$ with the basis $\mb{e}_i^{\ph{i}}$ in an appropriate representation $\rho:G\ra\mr{GL}(W)$. 

Suppose that such a vector is associated to the infinitesimal element of $k$-dim surface $S_k\subset\mc{M}$ in the frame $p'=(\mb{m}',e')$ at some other point by means of a differential form $\varphi=\mb{e}_i^{\ph{i}}\,\varphi^i\in\bar{\La}^k(P,W)$~\footnote{We assume the form is horizontal in order to possess good transformation properties w.r.t. gauge transformations $f^\ast\varphi=\tau^{-1}\cdot\varphi$.}. One performs the parallel transport to the single $p$ via pull-back $R^\ast_{A^{-1}}\varphi=A\cdot\varphi$, using the homogeneous part $A\in H$ of the development $g=(\ka\ga)\in G$, relating two frames $p=p'g$. After all frames have been synchronized in this way, one can perform addition of vectors in the usual sense. The sum is represented by the integral, expressed in a $p$-frame as
\begin{equation}\label{eq:integr-free-1}
\int_{S_k} A\cdot\varphi \ = \ \int_{S_k} \mb{e}_i^{\ph{i}}\, A^i_{\ph{i}j}\varphi^j,
\end{equation}

This gives us the sort of a `coarse-grained vector quantity' $\varphi_{S_k}(p)$, associated with the surface by means of the form integration. Naturally, the result will depend on the family of paths taken in its definition. If the surface is a boundary $S_k=\pa S_{k+1}$ of the region, application of the Stokes formula~\eqref{eq:Stokes-1} leads to
\begin{equation}\label{eq:integr-free-2}
\int_{S_{k+1}} \mb{e}_i^{\ph{i}} \, d (A^i_{\ph{i}j}\varphi^j) \ = \ \int_{S_{k+1}} \mb{e}_i^{\ph{i}} \, (A^i_{\ph{i}j}d\varphi^j +dA^i_{\ph{i}j}\wedge \varphi^j ) .
\end{equation}

If we, however, restrict ourselves to a domain that is sufficiently small in order to be able to consider only an element of the integral, -- then the application of the holonomy equation of~\eqref{eq:development} leads to the independence of the result on the choice of paths. Indeed, the $k$-form is just the suggestive name for a determinant built on the small displacements, tangent to the surface, so that the evaluation on one of such vectors leads to $dA^i_{\ph{i}j}=A^i_{\ph{i}l}\,\om^l_{\ph{i}j}=\om^i_{\ph{i}j}$, taking into account the initial condition $A=1$ at $p$. All in all, we are led to the general expression for a vector valued form:
\begin{equation}\label{eq:Stokes-2}
\oint_{S_k=\pa S_{k+1}} \mb{e}_i^{\ph{i}} \, \varphi^i \ = \ \int_{S_{k+1}} \mb{e}_i^{\ph{i}} \, (d\varphi^i +\om^i_{\ph{i}j}\wedge \varphi^j),
\end{equation}
This provides another justification for~\eqref{eq:frame-displacement-2}, and the convenient mnemonic rule follows for the Stokes integration: apply the differentiation (modified by the connection) to the components $\varphi^i$ as well as the frame vectors. 

As for the bound tensors, it will suffice to provide the formula for the radius vector, representing `a point of attachment'. The minus sign is a little bit ``tricky''. One may notice that in result of the development, the (initial) point $\mb{m}=\mb{m}'+\mb{a}$ acquires coordinates $a^i$ w.r.t. $ p'=(\mb{m}',e')$. Hence in order to ``roll back'', one applies an inverse (or, calculate the difference $\De\mb{m}=\mb{m}'-\mb{m}=-\mb{a}$ (free vector) w.r.t. an arbitrary frame). If the frames are sufficiently close, so that the path $\ga\sim I$ corresponds to an elementary edge bounded by $\pa I=\{s,s'\}$, the differentiation brings $d\mb{m}=-d\mb{a}=A(0)\cdot\theta=\mb{e}_i^{\ph{i}}\,\theta^i$, thus ``erecting an arrow'' at the point, corresponding to velocity $\dot{\ga}$. Roughly speaking, the Stokes formula could be extended to the case of a $0$-dim boundary of an interval, where it almost ``merges'' with the derivative $\nabla$ of a radius, described in the previous Sec. All in all, the rule follows that one simply applies~\eqref{eq:frame-displacement-1} for the Stokes summation/integration of bound multi-vectors. 

\subsection{Cartan-geometric treatment of curvature and torsion}

Let us re-trace how the non-triviality of the connection is manifested in the violation of the Poincar\'{e}'s lemma ($d^2=0$ for the ordinary differential -- dual to the ``boundary of a boundary is zero'' $\pa^2=0$). Take the free vector $\mb{X}=\mb{e}_i^{\ph{i}}\, X^i$, corresponding to some velocity/momentum, and calculate similarly its variation along some (unspecified) path segment:
\begin{equation}\label{eq:vector-variation}
\oint_{\pa I} \mb{X} \ = \ \mb{e}_i^{\ph{i}}\,(d X^i+\om^i_{\ph{i}j}X^j) \ \equiv \ D^\om\mb{X}.
\end{equation}
Together with the radius variation, they can be combined into the single sliding $\De(\mb{m}\wedge\mb{X}):=(\mb{m}\wedge\mb{e}_i^{\ph{i}})\,(dX^i+\om^i_{\ph{i}j}X^j)+  (\mb{e}_i^{\ph{i}}\wedge\mb{e}_j^{\ph{i}})\,\theta^iX^j$.  Take the contour integral, giving the sum of such changes across the small cycle $\ga=\pa S$, bounding a surface element:
\begin{align}\label{eq:cycle-displacement}
\oint_\ga \De(\mb{m}\wedge\mb{X}) \ & = \ (\mb{m}\wedge\mb{e}_i^{\ph{i}})\,\Om^i_{\ph{i}j}X^j + (\mb{e}_i^{\ph{i}}\wedge\mb{e}_j^{\ph{i}})\,\Theta^i X^j \nonumber \\
& \equiv \ \mb{m}\wedge (\mb{\Om}\cdot\mb{X}) +(\mb{\Theta}\wedge\mb{X}).
\end{align}
The differentials $dX$ disappear in the computation, and the result depends only on the value of the field at the point. 

One obtains the infinitesimal displacement associated with the cycle $\ga$ (or plaquette~$S$). When the contour is developed onto $\mathbb{M}$, the resulting frame at the target will not coincide with the one at the source, in general (see Fig.~\ref{fig:Cartan-cycle}). Displacement gives the small translation $-\Theta$ and rotation $-\Om$ that have to be applied in order to match the two frames at the endpoints of the developed cycle. 

The natural isomorphism $\mf{h}\cong\bigwedge^2 V$ of $H$-representations is used, in order to identify the Lie algebra element of the curvature $\mb{\Om}=(\mb{e}_i^{\ph{i}}\wedge\mb{e}_j^{\ph{i}})\,\frac12 \Om^{ij}$ with the system of simple bivectors. They provide the decomposition of the total rotation into $m(m-1)/2$ elementary ones in the planes of the axes (corresponding to generators of a Lie algebra $\mf{h}$). The values of components measure the magnitudes of the respective small angles of elementary rotations.

\begin{figure}[!h]
\center{\includegraphics[width=0.3\linewidth]{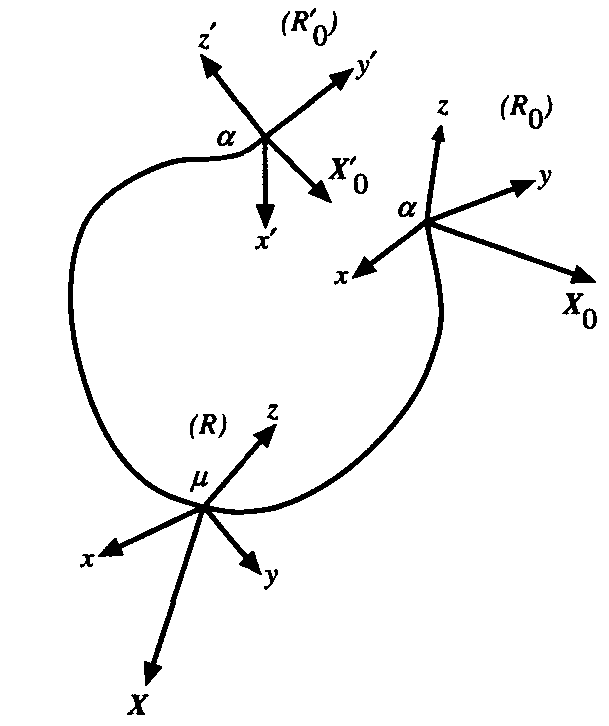}}
\caption{The developed cycle and the associated displacement (from~\cite[p.129]{Cartan2001Orthogonal-frame}).}
\label{fig:Cartan-cycle}
\end{figure}



\paragraph{Preservation of curvature and torsion of a surface.} There are important 2-dim integral invariants, constructed from $\Theta,\Om$, that satisfy the `closure' conditions. For the sake of generality, let us consider, for a moment, the system of \emph{general linear} transformations, associated by means of the 2-form $\Om\in\mf{gl}(V)$ to the elements of the surface $S_2 =\pa S_3$, bounding an infinitesimal region of space. The definition of the dual basis $\hat{\mb{e}}^i_{\ph{j}}(\mb{e}_j^{\ph{i}})=\de^i_j$ and the formula~\eqref{eq:frame-displacement-2} together provide the corresponding displacement of the dual frame $d\hat{\mb{e}}^i_{\ph{j}} =- \hat{\mb{e}}^j_{\ph{i}}\, \om^i_{\ph{i}j}$. Application of the generalized Stokes integration~\eqref{eq:Stokes-2} leads to
\begin{equation}\label{eq:Bianchi-2-meaning}
\oint_{S_2} \mb{\Om} \ \equiv \ \oint_{S_2}(\mb{e}_i^{\ph{i}}\otimes\hat{\mb{e}}^j_{\ph{i}})\, \Om^i_{\ph{i}j} \ = \ (\mb{e}_i^{\ph{i}}\otimes\hat{\mb{e}}^j_{\ph{i}}) ( d\Om^i_{\ph{i}j} +\om^i_{\ph{i}l}\wedge\Om^l_{\ph{i}j} - \Om^i_{\ph{i}l}\wedge\om^l_{\ph{i}j}),
\end{equation}
for an infinitesimal $3$-dim element (parallelepiped).

The vanishing of the integrand can be readily shown by substitution $\Om^i_{\ph{i}j}=d\om^i_{\ph{i}j}+ \om^i_{\ph{i}k}\wedge\om^k_{\ph{k}j}$, and constitutes the well-known \textbf{2nd group of Bianchi identities}:
\begin{equation}\label{eq:Bianchi-2}
d\Om \ = \ [\Om\wedge\om] \qquad \Leftrightarrow \qquad D^\om\Om \ = \ 0,
\end{equation}
for the (Ehresmann) curvature of the connection $\om$, corresponding to the general linear group of frame transformations. The above formula~\eqref{eq:Bianchi-2-meaning} clarifies their geometric meaning.

Returning to the orthogonal case $H=\mr{O}(\eta)$, consider the system of (dual) vectors, and the `moment of a pair' (according to~\cite[p.131]{Cartan2001Orthogonal-frame}), associated with a surface:
\begin{align}\label{eq:Bianchi-1-meaning}
\oint_{S_2} \star\mb{\Theta}-\mb{m}\wedge(\star\mb{\Om}) \ & \equiv \
\oint_{S_2} (\star\mb{e}_i^{\ph{i}})\, \Theta^i - \mb{m}\wedge\star(\mb{e}_i^{\ph{i}}\wedge\mb{e}_j^{\ph{i}})\, \frac12\Om^{ij} \nonumber \\
& = \ (\star\mb{e}_i^{\ph{i}})\, [d\Theta^i+\om^i_{\ph{i}j}\wedge\Theta^j] -(\eta_{jk}^{\ph{i}}\star\mb{e}_i^{\ph{i}}-\eta_{ik}^{\ph{i}} \star\mb{e}_j^{\ph{i}})\, \frac12 [\theta^k\wedge \Om^{ij}] \nonumber \\
& = \ (\star\mb{e}_i^{\ph{i}})\, [d\Theta^i+\om^i_{\ph{i}j}\wedge\Theta^j - \theta^j\wedge \Om^i_{\ph{i}j}].
\end{align}
(We have used the relation~\eqref{eq:Hodge-2}, and due to commutation of $\star$ with $\mf{h}$, the covariant derivative applied to $\Om$ eliminates by the previous~\eqref{eq:Bianchi-2}.) The quantity on the right again vanishes for any $\Theta^i=d\theta^i+\om^i_{\ph{i}j}\wedge\theta^j$, and one obtains the \textbf{1st group of Bianchi identities}:
\begin{equation}\label{eq:Bianchi-1}
D^\om\Theta \ = \ \Om\,\dot{\wedge}\,\theta,
\end{equation}
expressing the preservation of the integral in~\eqref{eq:Bianchi-1-meaning}. 

This clarifies the Cartan's example in $m=3$. In this form, the result is valid in arbitrary dimensionality, and seems to be new. The above conservation laws are strictly related with gauge symmetries, and as a matter of fact, are the direct consequences of description in terms of connections~\footnote{An analytic proof of the cycle displacement and identities in terms of sub-dominant errors and their bounds may be found in~\cite[sec.34-37]{Cartan1986Affine-connections}}.

\subsection{Surface defects and non-closure associated with torsion}
\label{subsec:non-closure}

The non-vanishing torsion implies that the infinitesimal parallelograms do not close during development~\footnote{Strictly speaking, they never do in non-flat geometries, but the defect is decrementally small for $\Theta=0$ (cf.~example of a sphere in~\cite[\S 101]{Cartan1983Riemannian-geometry}).}. The 1st Bianchi identities provide an extension of the non-closure to the higher-dimensional elements of surface/volume/etc. 

In the example~\eqref{eq:bivector-surface} from flat space, the bivector representing an element of a surface was associated to it by means of a Stokes integration of a sliding vector $\mb{\Si}^{(1)}_0=\mb{m}\wedge\boldsymbol{\theta}$ around a closed contour $\ga=\pa S$. In the general case of arbitrary geometry $(P,\mathpzc{K},\varpi)$, we have:
\begin{align}\label{eq:bivector-defect}
\tilde{\mb{\Si}}^{(2)} \ : & = \ \oint_\ga \mb{\Si}^{(1)}_0 \ = \ d[(\mb{m}\wedge\mb{e}_i^{\ph{i}})\,\theta^i] \nonumber \\
& = \ (\mb{e}_i^{\ph{i}}\wedge\mb{e}_j^{\ph{i}})\, \frac12\theta^i\wedge\theta^j + (\mb{m}\wedge\mb{e}_i^{\ph{i}})\, D^\om\theta^i \ \equiv \ \mb{\Si}^{(2)} + \De \mb{\Si}^{(2)}
\end{align}
where we have introduced the notation $\mb{\Si}^{(2)} :=\boldsymbol{\theta}\wedge\boldsymbol{\theta}$ for the surface (free) bivector, and $\De \mb{\Si}^{(2)}:=\mb{m}\wedge\mb{\Theta}$ denotes the `defect' arising due to presence of torsion.

When summing the elementary contributions over the boundary of a small three dimensional region $S_3$, the \emph{closure condition will be violated} (since $\theta$ is not an exact form):
\begin{align}\label{eq:closure-defect}
\oint_{S_2=\pa S_3} \tilde{\mb{\Si}}^{(2)} \ & = \ (\mb{e}_i^{\ph{i}}\wedge\mb{e}_j^{\ph{i}})\,\left[\frac12 (D^\om\theta^i\wedge\theta^j-\theta^i\wedge D^\om\theta^j)+\theta^i\wedge D^\om\theta^j\right] + (\mb{m}\wedge\mb{e}_i^{\ph{i}})\, D^\om\Theta^i \nonumber \\
& = \ (\mb{m}\wedge\mb{e}_i^{\ph{i}})\, \Om^i_{\ph{i}j}\wedge \theta^j \ \equiv \ \mb{m}\wedge\mb{\Om}\,\dot{\wedge}\,\boldsymbol{\theta}.
\end{align}
The contributions from non-integrable $\theta$ in $\mb{\Si}^{(2)}$ contract with those from  the torsion defect $\De \mb{\Si}^{(2)}$, appearing for $d\mb{m}=\theta$. Hence, one obtains via summation the net/integrated `torsion-defect' (bound to some point inside the small volume), given by the 1st Bianchi identities~\eqref{eq:Bianchi-1}. This should add to their geometric reading as well.

The situation is different for \emph{bound} vectors. Pick up the frame at the point inside the small volume $S_3$, and associate the sliding (multi-)vector with the infinitesimal element of its boundary surface $S_2$. For the elementary bivectors $\mb{\Si}^{(2)}$ of the surface itself, the formula~\eqref{eq:bivector-defect} generalizes to the ``$2\ra3$'' case:
\begin{align}\label{eq:trivector-defect}
\tilde{\mb{\Si}}^{(3)} \ : & = \ \oint_{S_2=\pa S_3} \mb{m}\wedge \mb{\Si}^{(2)} \ = \ d\left[(\mb{m}\wedge\mb{e}_i^{\ph{i}}\wedge\mb{e}_j^{\ph{i}})\,\frac12\theta^i\wedge\theta^j\right] \nonumber \\
& = \ (\mb{e}_i^{\ph{i}}\wedge\mb{e}_j^{\ph{i}}\wedge\mb{e}_k^{\ph{i}})\, \frac{1}{3!}\theta^i\wedge\theta^j\wedge\theta^j + (\mb{m}\wedge\mb{e}_i^{\ph{i}}\wedge\mb{e}_j^{\ph{i}})\, D^\om\theta^i\wedge\theta^j \ \equiv \ \mb{\Si}^{(3)} + \De \mb{\Si}^{(3)}.
\end{align}
One obtains the hypersurface trivector $\mb{\Si}^{(3)} :=\boldsymbol{\theta}\wedge\boldsymbol{\theta}\wedge\boldsymbol{\theta}$, corresponding to the volume of the bounded region, together with some `torsion defect' of the sum $\De \mb{\Si}^{(3)}:=\mb{m}\wedge\mb{\Theta}\wedge\boldsymbol{\theta}$.

The process can be iterated in higher dimensions. Let us restrict to $m=4$, and perform the ``$3\ra 4$'' integration step. If one takes~\eqref{eq:trivector-defect} to represent the hypersurface boundary of the (small) spacetime volume, then its failure to close is provided by
\begin{align}\label{eq:hypersurface-closure}
\oint_{S_3=\pa S_4} \tilde{\mb{\Si}}^{(3)} \ & = \ (\mb{m}\wedge\mb{e}_i^{\ph{i}}\wedge\mb{e}_j^{\ph{i}})\, \left[D^\om\Theta^i\wedge\theta^j +D^\om\theta^i\wedge D^\om\theta^j\right]  \nonumber \\
 & = \ (\mb{m}\wedge\mb{e}_i^{\ph{i}}\wedge\mb{e}_j^{\ph{i}})\, \Om^i_{\ph{i}k}\wedge\theta^k\wedge\theta^j \ \equiv \ \mb{m}\wedge(\mb{\Om}\,\dot{\wedge}\,\boldsymbol{\theta})\wedge\boldsymbol{\theta}.
\end{align}
The 4-vector, spanning the `chunk of' spacetime $S_4$, is calculated analogously to~\eqref{eq:trivector-defect}:
\begin{align}\label{eq:4-volume-defect}
\tilde{\mb{\Si}}^{(4)} \ : & = \ \oint_{S_3=\pa S_4} \mb{m}\wedge \mb{\Si}^{(3)} \ = \ d\left[(\mb{m}\wedge \mb{e}_i^{\ph{i}}\wedge\mb{e}_j^{\ph{i}}\wedge\mb{e}_k^{\ph{i}})\,\frac{1}{3!}\theta^i\wedge\theta^j\wedge\theta^k\right] \nonumber \\
& = \ (\mb{e}_i^{\ph{i}}\wedge\mb{e}_j^{\ph{i}}\wedge\mb{e}_k^{\ph{i}}\wedge \mb{e}_l^{\ph{i}})\,\frac{1}{4!}\theta^i\wedge\theta^j\wedge\theta^k\wedge\theta^l + (\mb{m}\wedge\mb{e}_i^{\ph{i}}\wedge\mb{e}_j^{\ph{i}}\wedge\mb{e}_k^{\ph{i}})\, \frac12 \theta^i\wedge\theta^j\wedge D^\om\theta^k \nonumber \\
& \equiv \ \mb{\Si}^{(4)} + \De\mb{\Si}^{(4)},
\end{align}
giving the 4-volume of spacetime region $|\mb{\Si}^{(4)}|=\star\mb{\Si}^{(4)}=\star(\boldsymbol{\theta}\wedge\boldsymbol{\theta}\wedge\boldsymbol{\theta}\wedge\boldsymbol{\theta})$, and the corresponding `defect' $\De\mb{\Si}^{(4)}:=\mb{m}\wedge\mb{\Theta}\wedge\boldsymbol{\theta}\wedge\boldsymbol{\theta}$. 

It is interesting to consider what impact does the `dualization' of forms have on the result of the integration. For instance, in the trivial 1-dimensional case, take the dual $(m-1)$-vector of the contour element, using~\eqref{eq:Hodge-1}:
\begin{align}\label{eq:dual-contour-sum}
\oint_{\ga=\pa S} \mb{m}\wedge\star\boldsymbol{\theta} \ & = \ d[(\mb{m}\wedge\star\mb{e}_i^{\ph{i}})\,\theta^i] \nonumber \\
& = \ (\mb{e}_i^{\ph{i}}\wedge\star\mb{e}_j^{\ph{i}})\, \frac12\theta^i\wedge\theta^j + (\mb{m}\wedge\star\mb{e}_i^{\ph{i}})\, D^\om\theta^i \ \equiv \ \mb{m}\wedge\star\mb{\Theta},
\end{align}
Hence, \emph{the result does not depend whether to sum over free or bound dual elements}.

One can go higher in dimensions, consecutively obtaining the respective defects due to non-zero torsion. For the orthogonal complement of the surface bivector $\mb{\Si}^{(2)}$, it is given by
\begin{align}\label{eq:dual-area-sum}
\oint_{S_2=\pa S_3} \mb{m}\wedge \star \mb{\Si}^{(2)} \ & = \ d\left[\mb{m}\wedge\star(\boldsymbol{\theta}\wedge\boldsymbol{\theta})\right] \ = \ \left[\mb{e}_i^{\ph{i}}\wedge \star(\mb{e}_j^{\ph{i}}\wedge\mb{e}_k^{\ph{i}})\right]\,\frac{1}{3!}\theta^i\wedge\theta^j\wedge\theta^k \nonumber \\
& \ph{=} \ \ + \mb{m}\wedge\star(\mb{e}_i^{\ph{i}}\wedge\mb{e}_j^{\ph{i}})\, \frac12 \left[D^\om\theta^i\wedge\theta^j-\theta^i\wedge D^\om\theta^j\right] \nonumber \\
& = \ \mb{m}\wedge\star(\mb{e}_i^{\ph{i}}\wedge\mb{e}_j^{\ph{i}})\, D^\om\theta^i\wedge\theta^j \ \equiv \ \mb{m}\wedge\star(\mb{\Theta}\wedge\boldsymbol{\theta}),
\end{align}
where the 2nd formula~\eqref{eq:Hodge-2} has been used to eliminate the contribution in the first line. 

Correspondingly, the ``$3\ra 4$'' passage gives the result:
\begin{align}\label{eq:dual-hypers-sum}
\oint_{S_3=\pa S_4} \mb{m}\wedge \star \mb{\Si}^{(3)} \ & = \ d\left[\mb{m}\wedge \star(\mb{e}_i^{\ph{i}}\wedge\mb{e}_j^{\ph{i}}\wedge\mb{e}_k^{\ph{i}})\,\frac{1}{3!}\theta^i\wedge\theta^j\wedge\theta^k\right] \nonumber \\
& = \ \mb{m}\wedge\star(\mb{e}_i^{\ph{i}}\wedge\mb{e}_j^{\ph{i}}\wedge\mb{e}_k^{\ph{i}})\, \frac12 \theta^i\wedge\theta^j\wedge D^\om\theta^k \ \equiv \ \mb{m}\wedge\star(\mb{\Theta}\wedge\boldsymbol{\theta}\wedge\boldsymbol{\theta}),
\end{align}
where the first (free vector) part of the integral is zero by the same reason as in~\eqref{eq:dual-area-sum}, using~\eqref{eq:Hodge-3}.



\subsection{Curvature tensors and scalars}
\label{subsec:curvature}

Applying the same procedure, one can construct the \emph{trivector} of the curvature associated to the element of 3-dim region as the geometric sum of bivectors of plaquette curvatures in the boundary:
\begin{align}\label{eq:trivector-curvature}
\mb{\Om}^{(3)} \ : & = \ \oint_{S_2=\pa S_3} \mb{m}\wedge \mb{\Om} \ = \ d\left[(\mb{m}\wedge\mb{e}_i^{\ph{i}}\wedge\mb{e}_j^{\ph{i}})\,\frac12\Om^{ij}\right] \nonumber \\
& = \ (\mb{e}_i^{\ph{i}}\wedge\mb{e}_j^{\ph{i}}\wedge\mb{e}_k^{\ph{i}})\, \frac{1}{3!}\Om^{ijk} \ \equiv \ \boldsymbol{\theta}\wedge\mb{\Om}, \qquad  \Om^{ijk}\ := \ \theta^i\wedge\Om^{jk} + \theta^j\wedge\Om^{ki}+\theta^k\wedge\Om^{ij}.
\end{align}
By the 1st Bianchi identities~\eqref{eq:Bianchi-1-meaning}, the bound \emph{dual} curvatures will not sum up to zero, but instead:
\begin{equation}\label{eq:bound-dual-curv}
\oint\mb{m}\wedge \star\mb{\Om}=\star(\mb{\Om}\,\dot{\wedge}\,\boldsymbol{\theta}).
\end{equation}
Unlike the 2nd Bianchi identities~\eqref{eq:Bianchi-2-meaning}, trivectors~\eqref{eq:trivector-curvature} of the hypersurface curvature (in $m=4$) do not satisfy 4-dim closure condition:
\begin{equation}\label{eq:non-closure-curv}
\oint_{S_3=\pa S_4} \mb{\Om}^{(3)} \ = \ (\mb{e}_i^{\ph{i}}\wedge\mb{e}_j^{\ph{i}}\wedge\mb{e}_k^{\ph{i}})\, \frac12 \Om^{ij}\wedge D^\om\theta^k \ \equiv \ \mb{\Om}\wedge\mb{\Theta},
\end{equation}
if only torsion is non-vanishing. The 4th curvature (multi-)vector will be given by 
\begin{align}\label{eq:4-curvature}
\oint_{S_3=\pa S_4} \mb{m}\wedge \mb{\Om}^{(3)} \ = & \ (\mb{m}\wedge\mb{e}_i^{\ph{i}}\wedge\mb{e}_j^{\ph{i}}\wedge\mb{e}_k^{\ph{i}})\, \frac12 \Om^{ij}\wedge D^\om\theta^k \nonumber\\  + & \ (\mb{e}_i^{\ph{i}}\wedge\mb{e}_j^{\ph{i}}\wedge\mb{e}_k^{\ph{i}}\wedge\mb{e}_l^{\ph{i}})\, \frac{1}{4!}\Om^{ijkl} \ \equiv \ \mb{m}\wedge\mb{\Om}\wedge\mb{\Theta} +\mb{\Om}^{(4)}, \nonumber \\
\Om^{ijkl}\ := & \ \theta^i\wedge\Om^{jkl} - \theta^j\wedge\Om^{ikl}+\theta^k\wedge\Om^{ijl}-\theta^l\wedge\Om^{ijk},
\end{align}
up to the same correction/boundary term; and so on (cf.~\cite[\S 193]{Cartan1983Riemannian-geometry}).

Introduce the coefficient tensors $\Theta^i = \frac12\Theta^i_{\ph{i}kl}\,\theta^k\wedge\theta^l$ of torsion, and curvature $\Om^i_{\ph{i}j} = \frac12\Om^i_{\ph{i}j,kl}\,\theta^k\wedge\theta^l$, respectively. One can write the differentials as 
\begin{subequations}
\begin{align}\label{eq:der-coord}
D^\om\Theta^i \ & \equiv \ \frac12 D^\om\Theta^i_{\ph{i}kl}\wedge\theta^k\wedge\theta^l \ = \ \frac12\Theta^i_{\ph{i}[kl|h]}\, \theta^h\wedge\theta^k\wedge\theta^l, \\
D^\om\Om^i_{\ph{i}j}\ & \equiv \ \frac12D^\om\Om^i_{\ph{i}j,kl}\wedge\theta^k\wedge\theta^l \ = \ \frac12\Om^i_{\ph{i}[j,kl|h]}\, \theta^h\wedge\theta^k\wedge\theta^l,
\end{align}
\end{subequations}
where ``dash'' corresponds to the covariant derivative $\nabla$ (generated by $\mb{e}_h^{\ph{i}}$). The 1st and the 2nd Bianchi identities then acquire the form:
\begin{subequations}\label{eq:Bianchi-coord}
\begin{align}
\Om^i_{\ph{i}j,kl}+\Om^i_{\ph{i}k,lj}+\Om^i_{\ph{i}l,jk} \ & = \ \Theta^i_{\ph{i}jk|l}+\Theta^i_{\ph{i}kl|j}+\Theta^i_{\ph{i}lj|k}, \label{eq:Bianchi-coord-1} \\
\Om^i_{\ph{i}j,kl|h}+\Om^i_{\ph{i}j,lh|k}+\Om^i_{\ph{i}j,hk|l} \ & = \ 0. \label{eq:Bianchi-coord-2}
\end{align}
\end{subequations}

Assuming the vanishing torsion $\Theta=0$, the first~\eqref{eq:Bianchi-coord-1} reduces to the familiar \emph{algebraic} Bianchi identities for the \emph{Riemann tensor} $R_{ij,kl}\equiv \cg_{ih}\Om^h_{\ph{h}j,kl}$. They are used to derive the symmetry w.r.t. permutations of the (pairs of) indices $R_{ij,kl}=R_{kl,ij}$, s.t. the number of (algebraically) independent components counts as:
\begin{gather*}
\frac{m^2(m^2-1)}{12} \ = \ \frac{m(m-1)}{2}+\frac{m(m-1)(m-2)}{2}+\frac{2m(m-1)(m-2)(m-3)}{4!}\\
\equiv \ [R_{ij,ij}]+[R_{ij,ik}]+[R_{ij,kl}] .
\end{gather*}
These symmetries will be lost in the most general scenario $\Theta\neq0$, as well as the restriction on the number of components.

Just like torsion manifests itself in surface and closure defects, the curvature leads to the failure of parallel transport of vectors~\eqref{eq:vector-variation}~\footnote{From the physics standpoint: Since one cannot detect a uniform distribution of force field (trivial connection) by means of mechanical experiments carried out within a given system which is embedded in the field, it follows that one can perceive only the \emph{curvature} which manifests itself physically in the \emph{variation} of the field (cf.~\cite[p.126]{Cartan1986Affine-connections}).}. Consider the bivector with components $\Si^{ij}= X^iY^j-Y^iX^j\equiv\theta^i\wedge\theta^j(\pa_\al,\pa_\be)$, representing the small element of the surface $S$, in the frame $p=(\mb{m},e)$ chosen in its interior. The vector $\mb{X}$, carried by parallelism around enclosing contour $\ga$, acquires the variation
\begin{equation}\label{eq:rotation-bivector}
\oint_{\ga=\pa S}D^\om\mb{X} \ = \ \mb{e}_i^{\ph{i}}\, \frac12 \Om^i_{\ph{i}j,kl}\Si^{kl}X^j.
\end{equation}
The displacement associated with the cycle is given by the inverse rotation $Q^i_{\ph{i}j}=-\frac12\Om^i_{\ph{i}j,kl}\Si^{kl}$. (Note that equivalence class of bivectors does not restrict to just parallelogram shapes.)

\begin{definition}
The \textbf{sectional} curvature in the planar direction is the (normalised) scalar product
\begin{equation}\label{eq:sec-curv}
K^{(2)} \ := \ \frac{\Si^{ij}Q_{ij}}{\Si^{kl}\Si_{kl}} \ = \ \frac{\bra\mb{\Om}(\mb{X},\mb{Y})\cdot\mb{X},\mb{Y}\ket}{|\mb{X}|^2|\mb{Y}|^2-\bra\mb{X},\mb{Y}\ket^2},
\end{equation}
defining the projection of the rotation associated with an infinitesimal surface element on the direction of the plaquette itself. (E.g., $K^{(2)}=R^{-2}=\mr{const}$ for the 2-dim sphere, due to isotropy.)

Similarly, the \textbf{mixed} curvature in two different planar directions can be given by
\begin{equation}\label{eq:mix-curv}
\tilde{K}^{(2)} \ := \ \frac{\Si'^{ij}Q_{ij}}{\sqrt{\sum|\Si'^{kl}|^2}\sqrt{\sum|\Si^{mn}|^2}} \ = \ -\frac{\bra\mb{\Si}',\mb{\Om}\ket}{|\mb{\Si}'||\mb{\Si}|},
\end{equation}
for which~\eqref{eq:sec-curv} is the special case when two directions coincide.
\end{definition}

The sectional curvature is actually defined on the Grassmanian $\mr{Gr}_2(V)$, i.e. depends only on the direction of the 2-plane at the point. It allows to fully characterize (the Riemann part of) the curvature tensor, as the following result implies. 

\begin{proposition}
If $K^{(2)}$ is given in all planar directions at the point, then all components $R_{ij,kl}$ of the Riemann-Christoffel tensor can be found. (Note that the proof essentially uses the symmetry of $R$, and may not hold if the torsion is present; cf.~\cite[\S 98]{Cartan2001Orthogonal-frame}). 

Closely related to this infinitesimal result, the \textbf{Ambrose-Singer theorem} states that the group of all loop holonomies $\mr{Hol}_p(\om)$ at the point is generated by Lie algebra elements of curvature $\Om_q(X,Y)\in\mf{h}$, as $q$ ranges over all points which can be joined to $p$ by a horizontal curve, s.t. $\om(X)=\om(Y)=0$ are horizontal vectors $X,Y\in TP$ (cf~\cite{AmbroseSinger1953holonomy-theorem}).
\end{proposition}

The \emph{hypersurface} curvature associated to the trivector with components $\Si^{ijk} = \theta^i\wedge\theta^j\wedge\theta^k$ (evaluated on tangents to $S_3$), is analogously defined by projecting~\eqref{eq:trivector-curvature} on the corresponding directions:
\begin{equation}\label{eq:hyper-curvature}
K^{(3)} \ = \ -\frac{\bra \mb{\Si}^{(3)},\mb{\Om}^{(3)}\ket}{|\mb{\Si}^{(3)}|^2} \ \propto \ R_{ij,kl}\Si^{ijh}\Si_h^{\ph{h}kl} ,
\end{equation}
as well as the 4-dim curvature: $K^{(4)}\propto R_{ij,kl}\Si^{ijhr}\Si_{hr}^{\ph{hr}kl}$. Expression~\eqref{eq:hyper-curvature} does not change if we `dualize' under the scalar product. So let us put $\mb{\Si}^{(3)}\equiv\star\mb{\Si}^{(1)}_\perp$ to be complementary to the hypersurface normal, with components $\Si^{ijk}=\varepsilon^{ijkl}d^3\Si_l$. And for the dual representation of hypersurface curvature:
\begin{align}\label{eq:Einstein-tensor}
\star(\boldsymbol{\theta}\wedge\mb{\Om}) \ & = \ \star(\mb{e}_i^{\ph{i}}\wedge\mb{e}_j^{\ph{i}}\wedge\mb{e}_h^{\ph{i}})\,\frac14 \Om^{ij}_{\ph{ij}kl}\, \theta^h\wedge\theta^k\wedge\theta^l \nonumber \\
& = \ \hat{\mb{e}}^r\, \frac14 \varepsilon_{ijkr}^{\ph{i}}\, \Om^{ij}_{\ph{ij}lm}\, \varepsilon^{klmn}\, d^3\Si_n \nonumber \\
& = \ \hat{\mb{e}}^j \,(-1)^\eta \left[\frac12 \Om^{kl}_{\ph{kl}kl}\de^i_j-\Om^{ih}_{\ph{ih}jh}\right]d^3\Si_i \ \equiv \ \hat{\mb{e}}^j \, E^i_{\ph{i}j} \, d^3\Si_i
\end{align}
one obtains the \textbf{Einstein tensor} $E_{ij}=\cg_{ih}E^h_{\ph{l}j} = R_{ij}-\frac12 R\, \cg_{ij}$, where $R^i_{\ph{i}j}:=R^{ih}_{\ph{il}jh}$, $R:=R^h_{\ph{i}h}$ are the \textbf{Ricci cuvature} tensor and scalar, respectively (we have put $(-1)^\eta=-1$ for Minkowski signature). The dual of 4-curvature is naturally a scalar $\star\mb{\Om}^{(4)}=(-1)^\eta R\,\det\theta$.

In Einstein-Cartan theory, the object~\eqref{eq:Einstein-tensor} is regarded as the \emph{`geometric representation of the physical vector of energy-momentum'}, or rather~\cite[p.118]{Cartan1986Affine-connections}:~\blockquote{the energy momentum density is the physical manifestation of a vector of geometric origin} (depending on preferences). In the Riemannian case when torsion vanishes, the closure of the geometric sum of dual curvature vectors (whether free or bound) represents the law of conservation for the `quantity of motion'. The (covariant) derivative being zero then leads to the analytic expression in terms of divergence $E^i_{\ph{i}j|i}=0$ for the Einstein tensor (also known as `contracted Bianchi identities').

Geometric interpretation of~\eqref{eq:Einstein-tensor} is provided in terms of `projection of a rotation on a hyperplane'~\cite[p.119]{Cartan1986Affine-connections}. Consider a 3-dim volume element enclosed by a surface $S_2=\pa S_3$, and take a point $\mb{a}$ inside. With each of the small elements of a surface, centered at $\mb{m}$, there is associated a rotation, and the vector $(\mb{m}-\mb{a})$ connects the surface to the interior. One then projects each rotation onto hyperplane orthogonal to $(\mb{m}-\mb{a})$ (could be regarded as `time normal'), where it can be represented by the dual vector (essentially the vector of `angular momentum' $\mb{J}$ for the little $\mr{O}(3)$-group of hypersurface rotations.) Multiplied by the length of $|\mb{m}-\mb{a}|$, the geometric sum of such vectors will give the `energy momentum' contained in the volume element under consideration.

\paragraph{On possible discretizations and mechanical interpretation.}

The role of Cartan's angular momentum~\eqref{eq:Einstein-tensor} and its conservation was vividly promoted by J. A. Wheeler in the form of a `boundary-of-a-boundary' topological argument~\cite[Ch.15]{MTU1973Gravitation},~\cite{Kheyfets1986boundary-principle}. It can be shown that the Regge's triangulations~\cite{Regge1961calculus} in terms of edge length variables could be re-derived from this `geometrodynamical' description~\cite{Miller1986geometrodynamic-Regge}, under certain assumptions on the lattice and its dual skeleton. The tight link between (non-)conservation of~\eqref{eq:Einstein-tensor} and `contracted Bianchi identities' can be very advantageous for the analysis of symmetries~\cite{Williams2012contracted-Bianchi-Regge,Gentle-etal2009Kirchoff-like-conservation-Regge,FreidelLouapre2003SF-diffeos}, given the convoluted nature of diff-invariance in discrete GR~\cite{Dittrich2008QG-diffeos,BahrDittrich2009broken-gauge-sym}.

The framework of Cartan connections is a versatile tool for regularization of theories without `background' metric structure. The curvature is ordinarily quantified in terms of (the product of) holonomies around plaquettes, similar to lattice gauge theories. One can mention several attempts to interpret Regge gravity in Poincar\'{e} gauge theoretic terms~\cite{Caselle-etal1989Poincare-calculus,Gionti2005discrete-Poincare-gravity}. In this way, the relations to quantum Spin Foam models have been established~\cite{Gionti2006discrete-Poincare-to-spinfoams}. Nevertheless, we strive to contribute to the above picture couple of (relatively) new ingredients: 
\begin{enumerate}[label={\upshape(\arabic*)}, align=left, widest=iii]
\item The natural `coarse-graining' flow of tensors. As we have demonstrated, the discretization by means of form-integration could be thought of in terms of usual geometric summation of vectors, transported via parallelism. This supplies a number of recipes for constructing the higher dimensional objects from the elementary ones, as well as for effectively `smoothing' the description -- in terms of the collective behaviour of components.
\item Background-independence does not allow any pre-existing lattice, apart from that is determined by the dynamical elements of the theory themselves. The affine translations are well-suited, s.t. in the enlarged configuration space, the role of lattice is reduced to purely topological cell-complex. The latter corresponds to the path-connectedness of regions, while it is conceivable that the notion of the `dual/or complementary' lattice is also part of dynamics.
\item Gauge-covariance, rather than gauge-invariance~\cite[p.54]{Cartan1986Affine-connections}:~\blockquote{To measure quantities of physical interest, one has available a local reference frame which plays the role of a true Galilean frame in a patch of space-time immediately surrounding the observer.} For instance, the Riemannian length of the geodesic segment $\int ds\sqrt{\cg_{ij}\dot{\ga}^i\dot{\ga}^j}$ will be the same as that of the straight line $\int\theta$ of developed path. In addition, the vectors possess the `sense of direction', and the linear spaces are much easier to handle, in general.
\end{enumerate} 

\begin{figure}[!h]
\center{\includegraphics[width=0.3\linewidth]{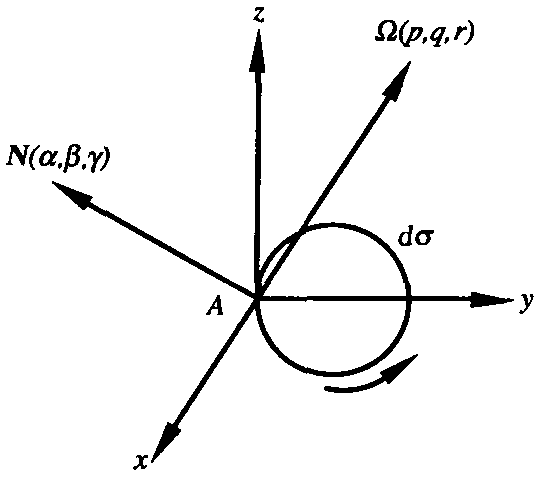}}
\caption{On mechanical interpretation (from~\cite[p.140]{Cartan2001Orthogonal-frame}).}
\label{fig:pressure}
\end{figure}

\pagebreak

We conclude the section with an analogy. Consider $m=3$ Riemannian space, where the conservation can acquire a noteworthy mechanical form. Suppose that a frame is connected with a point $A$ and let $d\boldsymbol{\si}=\star\mb{\Si}=(\al,\be,\ga)d\si$ be an oriented cycle of an element of the surface $d\si$, where $\al,\be,\ga$ denote the direction cosines of its normal $\mb{N}$. If $K_{ij}=\frac14\varepsilon_{ikl}^{\ph{i}}\,\Om^{kl}_{\ph{kl}mn}\,\varepsilon_j^{\ph{j}mn}$ gives the `double dual' of the curvature, then the dual vector of associated rotation $\star\mb{\Om}=(p,q,r)d\si$ has components of the form:
\begin{align*}
p & = K_{11}\al + K_{12}\be + K_{13}\ga, \\
q & = K_{21}\al + K_{22}\be + K_{23}\ga, \\
r & = K_{31}\al + K_{32}\be + K_{33}\ga. 
\end{align*}

These relations are identical to the formulas of elasticity, which define the liquid pressure on a surface element in continuous medium. Taking into account that $K_{ij|j}=0$ by the Bianchi identities~\eqref{eq:Bianchi-coord-2}, it then follows that~\cite[\S 198]{Cartan1983Riemannian-geometry}:
\blockquote{If a 3-dimensional Riemannian space is imagined to be a continuous medium such that the elastic pressure which acts on each element of the surface is equal to the vector representing the Riemannian curvature of the element, then the medium is in equilibrium under the action of the forces of elasticity.}
It could be inferred that the discretization of Cartan gravity might have deep similarities to that of the elasticity theory~\cite{Yavari2008discr-elasticity}, the physical equations taking the form of a `balance law'.
\begin{figure}[!h]
\center{\includegraphics[width=0.3\linewidth]{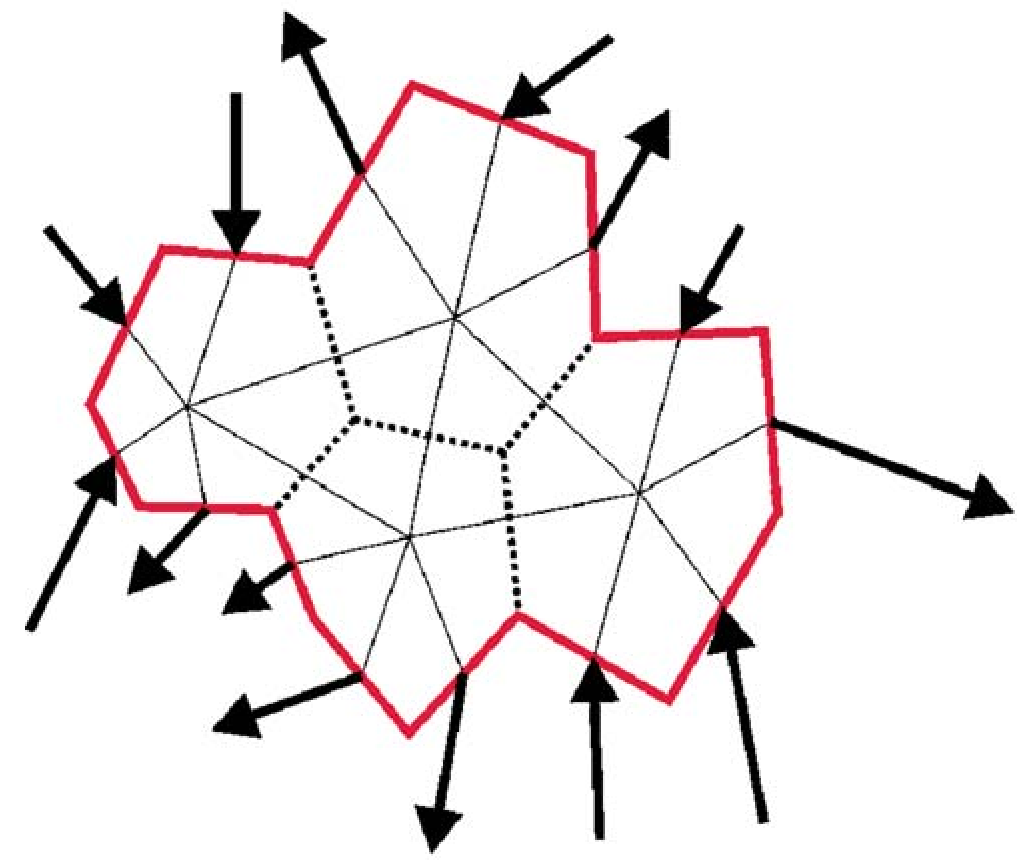}}
\caption{Rough analogy (from~\cite{Yavari2008discr-elasticity}).}
\label{fig:elast}
\end{figure}

\subsection{Action functional and variational principle(s)}

The equations of relativity theory can be obtained from the Hamilton's principle of least action. In general, the variational problem is defined by the choice of independent variables, as well as the form of lagrangian $S=\int\mc{L}$. Irrespectively of the latter, we can list three cases that can be considered in relation to GR (the fourth will be given in the next chapter).

\begin{enumerate}[label={\upshape(\arabic*)}, align=left, widest=iii]
\item The Hilbert variational principle, appearing in metric $\cg$ formulations. Since there is a unique metric-compatible torsion-free Levi-Civita connection, the above conditions are solved ab initio, expressing $\om=\om(\theta,d\theta)$ in $\mc{L}$ in terms of a given co-frame and its derivatives.
\item An equivalent Palatini constrained variation, in which both metric $\cg\sim\theta$ and connection $\om$ are varied, under condition of torsion-freeness (and metric-compatibility); the latter could be imposed via Lagrange multipliers~\cite{Tsamparlis1978onPalatini}, which play the role of `constraint reaction forces~\cite{Kichenassamy1986Lagrangian-multipliers-grav}~\footnote{The usual custom is to refer as ``Palatini'' to \emph{any} variation of $(\theta,\om)$ treated as independent variables, like our third item in the list. However, in \cite{Ferraris-etal1982Palatini-history} this method is traced back to Einstein himself, which was later (somewhat erroneously) attributed to Palatini's very similar techniques.}.
\item By `Einstein-Cartan' we will understand any variation of $\varpi=(\theta,\om)$ as independent variables.
\end{enumerate} 

The simplest natural lagrangian for relativity is given by the scalar curvature $\mc{L}=\star(\boldsymbol{\theta}\wedge\boldsymbol{\theta}\wedge\mb{\Om})$. If coupled to some matter field, such action functional will be stationary w.r.t. the E.-C. variation (with compact support, s.t. $\de\varpi$ vanishes on the boundary) if and only if
\begin{subequations}
\begin{align}\label{eq:EC-equations}
\star(\boldsymbol{\theta}\wedge\mb{\Om}) \ & = \ \mb{P}, \\
\star(\boldsymbol{\theta}\wedge\mb{\Theta}) \ & = \ \mb{S},
\end{align}
\end{subequations}
where $\mb{P}$ is the energy-momentum and $\mb{S}$ the spin-tensor (up to normalization), respectively. Since the map $\mb{\Theta}\mapsto \underbrace{\boldsymbol{\theta}\wedge...\wedge\boldsymbol{\theta}}_{m-3}\wedge\mb{\Theta}$ essentially establishes the linear isomorphism of representations $\bar{\La}^2_p(P,V)\cong\bar{\La}^{m-1}_p(P,\wedge^{m-2}V)$, having the same dimensionality ${m\choose 2}\times m={m\choose m-1}\times{m\choose m-2}$, it can be shown that $\mb{S}=0 \ \Rightarrow \ \mb{\Theta}=0$ (cf.~\cite{Bleecker1990ECSK}, for instance). Geometrically: ``torsion vanishes in the region of space-time where the energy momentum density can be represented by a simple vector'' -- e.g., in vacuum~\cite[p.124]{Cartan1986Affine-connections} (i.e. sliding vector, s.t. the angular momentum bivector is purely `orbital').

We articulate two possible viewpoints on the action functional of the Einstein-Cartan theory (in $m=4$), reflecting two ways of constructing the scalar curvature (gauge-invariant) from covariant tensors. They correspond to different approaches to regularization and quantization of GR, respectively.
\begin{enumerate}[label={\upshape(\arabic*)}, align=left, widest=iii]
\item Introducing the complementary (simple) bivector $\star\mb{B}:=2\,\mb{\Si}=(\mb{e}_i^{\ph{i}}\wedge\mb{e}_j^{\ph{i}})\, \theta^i\wedge\theta^j$, one can form the scalar product contraction
\begin{align}\label{eq:Regge-action}
\bra \mb{B},\mb{\Om}\ket \ & = \ \bra \star (\mb{e}_i^{\ph{i}}\wedge\mb{e}_j^{\ph{i}}),\mb{e}_k^{\ph{i}}\wedge\mb{e}_l^{\ph{i}}\ket \, \frac{(-1)^\eta}{2} \theta^i\wedge\theta^j\wedge\Om^{kl} \nonumber \\
& = \ \bra \mb{e}_i^{\ph{i}}\wedge\mb{e}_j^{\ph{i}},\mb{e}_k^{\ph{i}}\wedge\mb{e}_l^{\ph{i}}\ket \, \frac{(-1)^\eta}{4} \varepsilon_{ijkl}^{\ph{i}}\,\theta^i\wedge\theta^j\wedge\Om^{kl} \nonumber \\
& = \ \frac{(-1)^\eta}{2} \varepsilon_{ijkl}^{\ph{i}}\,\theta^i\wedge\theta^j\wedge\Om^{kl} \ = \ R\det\theta \ \equiv \ -2|\mb{\Si}| |\mb{Q}|_\perp^{\ph{i}},
\end{align}
where $|\mb{Q}|_\perp^{\ph{i}}$ is the projection in the planar direction $\mb{B}$ of the rotation $\mb{Q}$, associated with the cycle. This is essentially a sectional curvature~\eqref{eq:sec-curv}, and the expression on the right corresponds to the \emph{deficit angle} of Regge triangulations~\cite{Regge1961calculus,FriedbergLee1984Regge-derivation,Cheeger-etal1984Regge-convergence,Barrett19941st-order-Regge,Miller1997Hilbert-Regge,McDonaldMiller2008scalar-Regge}. (At least in the Riemannian case $\Theta=0$; the possibility of discretizing torsion along these lines was discussed in~\cite{Drummond1986Regge-Palatini}.) The left hand side of~\eqref{eq:Regge-action} can be also read in terms of invariant Killing form $\mr{K}(\mb{A},\mb{B})=\tr(\mr{ad}(\mb{A})\mr{ad}(\mb{B}))$ on the algebra $\mb{A},\mb{B}\in\mf{h}$ (bilinear and symmetric). It is the basis for discretization in Spin Foam models, discussed in Sec.~\ref{sec:Spin-Foams}. 
\item The scalar product of vector quantities can be directly formed as
\begin{align}\label{eq:Poincare-Cartan}
\bra \mb{P},\boldsymbol{\theta}\ket \ & \equiv \ \star(\mb{\Om}\wedge\boldsymbol{\theta})\,\dot{\wedge}\, \boldsymbol{\theta} \ : = \ 
\bra\star (\mb{e}_i^{\ph{i}}\wedge\mb{e}_j^{\ph{i}}\wedge\mb{e}_k^{\ph{i}}),\mb{e}_l^{\ph{i}}\ket \, \frac12 \Om^{ij}\wedge\theta^k\wedge\theta^l \nonumber \\
& = \ \bra \mb{e}_h^{\ph{i}},\mb{e}_l^{\ph{i}} \ket \, \frac12 \varepsilon_{ijk}^{\ph{ijk}h}\,\Om^{ij}\wedge\theta^k\wedge\theta^l  \ = \ \frac12 \varepsilon_{ijkl}^{\ph{i}}\,\Om^{ij}\wedge\theta^k\wedge\theta^l,
\end{align}
obtaining the same action functional. Written this way, it has the appearance of the so-called \emph{Poincare-Cartan integral invariant}~\cite{Cartan1922integral-invariants}: $I_{PC} = \oint_\ga p_\mu dx^\mu = \oint_\ga p_i dx^i - \ms{H} dt$ (schematically; recall that $\boldsymbol{\theta}=d\mb{m}$). This is the fundamental object in analytic mechanics: from the condition of independence on the contour $\ga$ around the closed tube of trajectories (in the `enlarged' phase space $(p,q,t)$), the entire dynamics of the system can be (re-)derived. 

There is (almost) one step from here to the formulation in the language of symplectic geometry -- namely, the pair of $(p,q)$ defines the (relativistic) phase space, with pre-canonical closed 2-form given by $\nu=dp\wedge dq$ (cf.~discussion in~\cite[Ch.3]{Rovelli2004QG}). The symplectic isometries/or canonical transformations are the diffeomorphisms that preserve $\nu$. They are generated by the Hamiltonian vector fields $\hat{\chi}$/or corresponding functions $\chi=\chi(p,q)$, satisfying $\imath_{\hat{\chi}}\nu=d\chi$, s.t. the commutator $[\hat{\chi},\hat{\chi}']$ is given by the Poisson bracket as $\{\chi,\chi'\}=\nu(\hat{\chi},\hat{\chi}')$. In the field-theoretic generalization~\cite{Kijowski1973canon-fin-dim,KijowskiTulczyjew1979symplectic-framework}, the time parameter $t\in\mathbb{R}$ is replaced by the points/labels $x\in\mc{M}$ on the space-time manifold, while the configuration space is normally a bundle $Q\ra\mc{M}$. The elements of the corresponding phase space need to be integrated over the boundaries of a domain, and thus given by the $T^\ast Q$-valued $(m-1)$-densities. 

The dynamical trajectories are replaced by the `Lagrangian submanifolds' $\ms{N}$, required to satisfy $\nu|_{\ms{N}}^{\ph{i}}=0$ for the generalized symplectic $(m+1)$-form (the latter may be degenerate due to gauge symmetries, leading to constraints). The action/lagrangian provides the `special symplectic structure' -- i.e. potential function $p\,dq=dS:Q\ra T^\ast Q$, generating the dynamics, s.t. $p=\pa S/\pa q$. (Depending on the chosen `control mode', there might be several such special symplectic structures -- like in thermodynamics -- e.g., any particular hamiltonian $\ms{H}$ corresponds to a certain choice of time flow $\pa/\pa t$ in the foliation picture.) It is interesting to investigate how the connection formulation on the Klein bundle $(P,\mathpzc{K},\varpi)$ fits into this picture (cf. recent~\cite{CattaneoSchiavina2019symplectic-EC}), which we leave for the prospective future research.
\end{enumerate}

\paragraph{Symmetries and Noether identities}

The invariances of the action w.r.t. the local transformations of the group result in the set of identities, that are in 1-to-1 correspondence with the gauge symmetries, which is the content of the Noether's 2nd theorem. (The 1st theorem deals with the global symmetries and the conserved charges/divergence laws, accordingly~\cite{Kosmann-Schwarzbach2011Noether}.)

\begin{definition}
Given two Cartan geometries $(P_i,\mathpzc{K}_{\ph{i}i},\varpi_i)$, $i=1,2$, and the diffeomorphism $\bar{f}:\mc{M}_1\ra\mc{M}_2$ of their base manifolds $\mc{M}_i\cong P_i/H$, the \emph{unique} cover map $f:P_1\ra P_2$, satisfying $f^\ast\varpi_2=\varpi_1$, defines the \textbf{geometric isomorphism} (cf.~\cite[p.185]{Sharpe1997Diff-Geometry-Cartan} for uniqueness.)
\end{definition}

The gauge symmetry transformations form the group $\mc{G}(P,\mathpzc{K})$ of the bundle automorphisms $f:P\ra P$, which preserve the $G$-action of the principal Poincar\'{e} group $f(pg)=f(p)g$. Given their parametrization in terms of functions $\tau\in C(P,G)$ of Prop.~\ref{prop:gauge-symmetry}, the two connections $\varpi=(\theta,\om)$, which are `gauge-equivalent' via pull-back $f^\ast\varpi=\tau^\ast\om_G+\mr{Ad}(\tau^{-1})\varpi$ of Prop.~\ref{prop:connection-transform}, specify the isomorphism of geometries in the above sense. The corresponding transformation of the action functional is
\begin{equation}
\int_{\bar{U}}\mc{L}(f^\ast\varpi) \ = \ \int_{\bar{U}}\bar{f}^\ast \mc{L}(\varpi) \ = \ \int_{\bar{f}(\bar{U})}\mc{L}(\varpi),
\end{equation}
where $\bar{f}:\mc{M}\ra\mc{M}$ is induced by the projection $\bar{f}(\pi(p))=\pi(f(p))$. (It is implied that appropriate parametrization in the local Klein chart $\bar{\ka}(\bar{U})=\bar{\ka}\circ\pi(U)=\pi_G\circ \ka(U)$ has been chosen; cf. discussions on relation between `active and passive diffeomorphisms' in~\cite[\S 2.2.4]{Rovelli2004QG} and~\cite[\S 19.1.2]{Thiemann2007ModCanQuantGR}.)

The one-parameter sub-groups of transformations are generated by the gauge algebra as $\tau_s(p)=\exp(s\,\zeta(p))$, $\zeta\in C(P,\mf{g})$, following Prop.~\ref{prop:symmetry-generator}. The corresponding infinitesimal changes of $\varpi$ are, equivalently, given by the Lie derivative $\de_X^{\ph{i}}\varpi = \ms{L}_X^{\ph{i}}\varpi =(d\circ \imath_X^{\ph{i}}+\imath_X^{\ph{i}}\circ d)\varpi$ w.r.t. fundamental vector fields, s.t. $\varpi(X)=\mb{X}=(\mb{q},\mb{Q})\in\mf{g}=\mf{p}\oplus\mf{h}$ (recall~\eqref{eq:universal-der}). The result acquires the form of `covariant Lie derivative'~\cite{Gronwald1998covariant-translations}, if written in vector notation and with greater detail:
\begin{align}\label{eq:diff-gauge-conn}
\de_X^{\ph{i}}\boldsymbol{\varpi} \ & \equiv \ \left(D^\varpi\circ \imath_X^{\ph{i}}+\imath_X^{\ph{i}}\circ D^\varpi\right)\boldsymbol{\varpi} \nonumber \\
& = \ d\imath_X^{\ph{i}}\left(\boldsymbol{\theta}+\boldsymbol{\om}\right) +\imath_X^{\ph{i}}\left(D^\om\boldsymbol{\theta}-[\boldsymbol{\om}\wedge\boldsymbol{\theta}] + D^\om\boldsymbol{\om} - \frac12 [\boldsymbol{\om}\wedge\boldsymbol{\om}]\right) \nonumber \\
& = \ d\left(\boldsymbol{\theta}(X)+\boldsymbol{\om}(X)\right) - \left([\boldsymbol{\om}(X),\boldsymbol{\theta}]-[\boldsymbol{\om},\boldsymbol{\theta}(X)]\right)-\frac12\left([\boldsymbol{\om}(X),\boldsymbol{\om}]-[\boldsymbol{\om},\boldsymbol{\om}(X)]\right) + \imath_X^{\ph{i}}\left(\mb{\Theta}+\mb{\Om}\right) \nonumber \\
& = \ d(\mb{q}+\mb{Q}) +[\boldsymbol{\om}+\boldsymbol{\theta},\mb{q}+\mb{Q}] + \imath_X^{\ph{i}}(\mb{\Theta}+\mb{\Om}) \qquad \text{(taking into acount $[\boldsymbol{\theta},\mb{q}]=0$)} \nonumber \\
& = \ D^\om\mb{q}-\mb{Q}\cdot\boldsymbol{\theta} + D^\om\mb{Q} + \left\{\mb{e}_i^{\ph{i}}\,\Theta^i_{\ph{i}kl}+(\mb{e}_i^{\ph{i}}\wedge\mb{e}_j^{\ph{i}}) \,\frac12\Om^{ij}_{\ph{ij}kl}\right\}\frac12(\theta^k(X)\theta^l-\theta^k\theta^l(X))\nonumber \\
& = \ \mb{e}_i^{\ph{i}} \left\{D^\om q^i+q^k \, \Theta^i_{\ph{i}kl}\theta^l -Q^i_{\ph{i}j}\theta^j\right\} +(\mb{e}_i^{\ph{i}}\wedge\mb{e}_j^{\ph{i}}) \, \frac12 \left\{D^\om Q^{ij} + q^k\, \Om^{ij}_{\ph{ij}kl}\theta^l\right\} 
\end{align}

The expression in the 4th line demonstrates the validity of the shorthand symbolic notation of the first. One easily recognizes in $\de_Q^{\ph{i}}\theta$ and $\de_Q^{\ph{i}}\om$ the standard Lorentz gauge-transformations in the fiber, over the fixed point. The other part, depending on $q$, encodes the action of base-diffeomorphisms in the gauge-theoretic (`internal') language of Cartan connections. The same technique if applied to torsion and curvature tensors leads to their Lie derivatives (cf.~\cite{Catren2015Cartan-gauge-gravity}), consisting of:
\begin{subequations}\label{eq:diff-gauge-curv-tor}
\begin{align}
\de_Q^{\ph{i}}\mb{\Theta} \ & = \ [\mb{\Theta},\mb{Q}], & \de_q^{\ph{i}}\mb{\Theta} \ & = \ D^\om\imath_X^{\ph{i}}\mb{\Theta} +\imath_X^{\ph{i}} [\mb{\Om},\boldsymbol{\theta}], \label{eq:diff-gauge-tor} \\
\de_Q^{\ph{i}}\mb{\Om} \ & = \ [\mb{\Om},\mb{Q}], & \de_q^{\ph{i}}\mb{\Om} \ & = \ D^\om\imath_X^{\ph{i}}\mb{\Om}, \label{eq:diff-gauge-curv}
\end{align}
\end{subequations}
where in the second column $X=\theta^{-1}(\mb{q})$. 

The variation of the Einstein-Cartan action functional can be written as
\begin{align}\label{eq:EC-variation}
\de \int \bra \star(\mb{\Om}\wedge\boldsymbol{\theta}),\boldsymbol{\theta}\ket \ & = \ 2 \, \int  \bra\star(\mb{\Om}\wedge\boldsymbol{\theta}),\de\boldsymbol{\theta}\ket + \bra\star(\mb{\Theta}\wedge\boldsymbol{\theta}),\de\boldsymbol{\om}\ket \nonumber\\
& = \ - 2 \, \int \bra \de\boldsymbol{\theta}, \star(\boldsymbol{\theta}\wedge\mb{\Om})\ket + \bra \de\boldsymbol{\om}, \star(\boldsymbol{\theta}\wedge\mb{\Theta})\ket .
\end{align}
The invariance under~\eqref{eq:diff-gauge-conn} leads, upon substitution, to the identities:
\begin{subequations}\label{eq:EC-Noether}
\begin{align}
D^\om\left[\star(\boldsymbol{\theta}\wedge\mb{\Om})\right] \ & = \ \star(\mb{\Theta}\wedge\mb{\Om}) \label{eq:EC-Noether-1} \\
D^\om\left[\star(\boldsymbol{\theta}\wedge\mb{\Theta})\right] \ & = \ \boldsymbol{\theta}\wedge\star(\boldsymbol{\theta}\wedge\mb{\Om}), \label{eq:EC-Noether-2}
\end{align}
\end{subequations}
for independent $q$ and $Q$ parameters, correspondingly. (Up to exact differential/boundary terms of the form $\bra\mb{q},\mb{P}\ket +\bra\mb{Q},\mb{S}\ket$.) The first equation~\eqref{eq:EC-Noether-1} relates \emph{non-vanishing divergence} of Einstein tensor $E$ to non-trivial torsion, while its \emph{asymmetric part} $E^{[ij]}$ is given by the derivative of the spin-tensor in the second equation~\eqref{eq:EC-Noether-2}~\footnote{Result~\eqref{eq:EC-Noether-1} is straightforward due to commutation of $\star$ and $D^\om$. Writing 1st Bianchi identities~\eqref{eq:Bianchi-1}, alternatively, as $D^\om\mb{\Theta}=\star(\boldsymbol{\theta}\wedge\star\mb{\Om})$, the purely algebraic identities follow from~\eqref{eq:EC-Noether-2}: $\boldsymbol{\theta}\wedge\star(\boldsymbol{\theta}\wedge\mb{\Om})=-\star(\boldsymbol{\theta}\wedge\star(\boldsymbol{\theta}\wedge\star\mb{\Om}))$ (cf.~\cite{Bleecker1990ECSK}).}.

The relations~\eqref{eq:EC-Noether} are identities, in the sense that they are satisfied by direct calculation, without using the field equations (i.e. they are valid \emph{``off-shell''}). If one denotes $P_l^{\ph{i}}:=\frac12\varepsilon_{ijkl}^{\ph{i}}\, \theta^i\wedge\Om^{jk}$, $S_{ij}^{\ph{i}}:=\frac12\varepsilon_{ijkl}^{\ph{i}}\,\theta^k\wedge\Theta^l$, following~\eqref{eq:EC-equations}, the (modified) conservation laws follow immediately:
\begin{subequations}\label{eq:modified-conserv}
\begin{align}
D^\om P_i^{\ph{i}} \ & = \ \theta^j\wedge\left\{\Theta^k_{\ph{k}ij}P_k^{\ph{i}}+\frac12\Om^{kl}_{\ph{kl}ij}S_{kl}^{\ph{i}}\right\}, \label{eq:modified-conserv-1} \\
D^\om S_{ij}^{\ph{i}} \ & = \ \theta_i^{\ph{i}}\wedge P_j^{\ph{i}}-\theta_j^{\ph{i}}\wedge P_i^{\ph{i}}, \label{eq:modified-conserv-2}
\end{align}
\end{subequations}
acquiring in components the form, suggested by A. Trautman (cf.~\cite{Trautman1973EC-structure} and preface to~\cite{Cartan1986Affine-connections}).

We conclude the Chapter with the following passage due to B. S. DeWitt on the role of symmetry group in relation to observables and locality of the variational equations of motion~\cite[p.13]{DeWitt1965groups-fields}:

\blockquote{Since the class of all physical observables has only group invariants as members one may alternatively be tempted to suppose that the group is really irrelevant. This is equally wrong. Although it is true that in all cases which have ever been considered it has always been found possible to give a complete \emph{intrinsic} dynamical description of the system (i.e., a way to write a complete set of dynamical equations in terms of observables only), the intrinsic equations have nevertheless in every case proved to be (a) \emph{nonlocal} and (b) unobtainable from a stationary action principle in which only observables are varied.* This therefore suggests that the role of the invariance group is twofold: Firstly, it provides the means by which a stationary action principle can be set up for the system. Secondly, it provides the means by which the fundamental locality of the dynamical equations may be displayed. 

This view of the significance of the invariance group is not as narrow as it 
first appears. The existence of a stationary action principle is essential to  
quantization, and locality of the dynamical equations is essential to causality. Therefore, whatever the physical principle involved in the existence of the invariance group may be, it is at least as fundamental as quantum theory and the concept of causality.}

\newpage

\thispagestyle{empty}
\chapter{Modern theories of Quantum Gravity}
\label{ch:quantization}

After we have figured out the geometric and gauge theoretic content of the classical theory of relativity due to Einstein-Cartan in the first part of the thesis, let us turn to the modern non-perturbtive approaches to quantization of gravity. We will provide only a brief account of two of them, where connection plays the crucial part, namely, Loop Quantum Gravity (LQG) and Spin Foams (SF).

There exist many viewpoints on what constitutes the `quantization' of a given `classical' theory. We will concentrate mainly on two elements which are (likely) to be most essential (referring for details to~\cite{Dirac1930QMprinciples,vonNeumann1932QM-foundations}, and not delving too much into interpretations):
\begin{enumerate}[label={\upshape(\arabic*)}, align=left, widest=iii]
\item Superposition principle, realized through the construction of (linear, or projective) Hilbert space of quantum states $|\psi\ket \in \mc{H}$.
\item Description of dynamics and/or observables in terms of probability amplitudes $\mc{A}:|\psi\ket \ra \mathbb{C}$ (also satisfy linear addition, unlike classical summation of probabilities themselves). 
\end{enumerate}

Concerning quantization of gravity/geometry, the first aspect is thoroughly developed within (canonical) LQG. Though, the latter struggles to incorporate dynamics, which is the focus of a complementary path-integral viewpoint of Spin Foams. Let us elaborate on essential elements of both constructions, making some broad comparisons with the perspective provided by Cartan connection formulation~\footnote{In the rest part of the thesis, we use more traditional coordinate-based approach.}.

\pagebreak

\section{Loop Quantum Gravity perspective}

Loop Quantum Gravity is the canonical quantization of General Relativity in terms of the so-called Ashtekar connection variables and conjugate fluxes of momenta (cf.~standard references~\cite{Rovelli2004QG,Thiemann2007ModCanQuantGR,AshtekarLewandowski2004LQG-status,Nicolai-etal2005LQG-outside,HanHuangMa2007LQG-structure}). The \textbf{classical prerequisites} consist of the following ingredients:
\begin{itemize}
\item The theory is not formulated directly in spacetime, but introduces first the foliation by hypersurfaces $\mc{S}_3\subset\mc{M}$, in order to put the theory into hamiltonian/canonical form. Accordingly, the basic variables $(\theta,\om)$ are split into: $3$-dimensional $(A,B)$ belonging to the spatial slice, `non-dynamical' lapses and shifts $N,N_a$, and potentials $A_0$. The action then takes the schematic form
\begin{equation}
S \ = \ \int dt\int d^3x (\tilde{B}^a_i\dot{A}^i_a -\ms{H}), \qquad \ms{H} \ = \ N C + N^a C_a + A^i_0 C_i,
\end{equation}
where the Hamiltonian $\ms{H}$ is the combination of first-class constraints: Gauss $C_i\approx 0$, spatial diffemorphisms $C_a\approx 0$, and (dynamical) Hamiltonian constraint $C\approx 0$. (The indices $i=1,2,3$ are w.r.t. the generators of a Lie algebra $\mf{so}(3)\cong\mf{su}(2)$ of hypersurface gauge-rotations, whereas spatial $a=1,2,3$ are `leftovers' of coordinate parameterization of $\mc{M}$; however, the components $(N,N_a)$ of the so-called ADM decomposition do not form the space-time vector~\cite{KiriushchevaKuzmin2011canonical-GR-myths}.) There are also the second-class constraints, which do not appear explicitly in the `time-gauge', or one `solves' them ab initio as in~\cite{e-Sa2001Holst-Hamiltonian}. (Elaborating the full constraint algebra together with Dirac brackets may be a formidable task, cf.~\cite{Alexandrov2000CovariantLQG}.)
\item The Ashtekar-Barbero (AB) connection variables correspond to a certain combination ${}^{(\ga)}A^i \ = \ \frac{1}{2}\eps^{0i}_{\phantom{0i}jk}\om^{jk} + \ga \om^{0i}$ of spatial Levi-Civita connection and the extrinsic curvature, the latter characterizes the embedding of a hypersurface $\mc{S}_3$ in the ambient $\mc{M}$. (Essentially, the `lucky' isomorphism $\mf{su}(2)\cong\mathbb{R}^3$ of vector spaces is being exploited for the `boost' part of the Lorentz connection.) At the level of the action functional, one does not start directly with the Einsten-Cartan lagrangian, but instead its modification by the Holst term $\mc{L}_{EC}+\ga^{-1} \bra\boldsymbol{\theta}\wedge\boldsymbol{\theta},\mb{\Om}\ket$ is being used \cite{Holst1996action}. (This may roughly correspond to the `mixed' curvature~\eqref{eq:mix-curv}). In the absence of spin-current, the term disappears for zero-torsion, and the Einstein equations of GR are reproduced classically.
\item The conjugate fluxes $\tilde{B}^a_i=\frac12\eps_{ijk}^{\ph{i}}\theta^j_b \theta^k_c\eps^{abc}=\det(\,\vec{\theta}\,)\,\theta^a_i$ are densitized inverse spatial triads (corresponding to the dual bivector $B=\star(\vec{\theta}\wedge\vec{\theta})$, restricted to hypersurface). The Gauss constraint $C_i=\mc{D}_a\tilde{B}^a_i\approx 0$ has the form of the surface closure condition, the covariant derivative is taken w.r.t. Ashtekar-connection. It generates spatial rotations, and satisfy the algebra of angular momentum. We omit the details of the other constraints and their closure.
\end{itemize}

\subsection{Kinematical Hilbert space of LQG}

The quantization in the connection representation is akin to the Shroedinger wave theory, where $(x,p)\sim (A,B)$ are quantized like multiplication and derivative operators on the space of functionals over connections $\mc{C}$, schematically:
\begin{equation}
\hat{A}(\mb{x})\Psi[A] \  = \ A(\mb{x})\Psi[A], \qquad
\hat{B}(\mb{y})\Psi[A] \  = \ \frac{\hbar}{i}\frac{\de}{\de A(\mb{y})}\Psi[A],
\end{equation}
having the canonical commutation relations (c.c.r.) $[\hat{A}(\mb{x}),\hat{B}(\mb{y})]=i\hbar\,\de^3(\mb{x},\mb{y})$. The Gauss constraint $\hat{C}_i\Psi=\mc{D}_a\de\Psi/\de A^i_a\approx0$ then expresses the invariance of the wave functional w.r.t. infinitesimal gauge transformations $\mc{G}$ of the connection $\Psi[A^g]=\Psi[A]$, leading to the $\mr{SU}(2)$-invariant subspace $\mc{C}/\mc{G}$. (Analogously for other constraints.)

The obstacle to the construction of the (invariant) inner product/measure on such a configuration space of connections $\mc{C}/\mc{G}$, which is both non-linear and infinite-dimensional, is bypassed by adopting techniques from lattice gauge theory (and also $C^\ast$-algebras of axiomatic QFT). There, reconstruction of gauge potentials is available, knowing the holonomies (cf.~\cite{Giles1981gauge-Wilson-reconstruction} for Wilson loop traces, satisfying certain Mandelstam identities). Let $\Ga$ be a \emph{graph}, consisting of finite number $L$ of `links' $\ell$, joining at the `nodes' $n$ of the total amount $N$. One thus defines the graph's \textbf{partial configuration spaces} as $\mc{U}_\Ga=\mr{SU}(2)^{L\subset\Ga}$, spanned by the collection of holonomies $h=h[A]$, associated to the links of $\Ga$.

Given two oriented graphs ordered by inclusion $\Ga\leq\Ga'$, s.t. the former is the subset of the latter, there exists a `projection map'
\begin{equation}
\pi_{\Ga'\Ga}  : \, \mc{U}_{\Ga'} \ \ra \ \mc{U}_\Ga , \qquad \pi_{\Ga'\Ga}(\{h_{\ell'}\})_\ell \ = \ \overleftarrow{\prod_{\ell'\subset\ell}} h_{\ell'}^{[\ell,\ell']},
\end{equation}
relating the `finer' description to the `coarser' one. ($[\ell,\ell']=\pm1$ is the relative orientation of $\ell$ and $\ell'$.) With the `partially-ordered set' of graphs there are associated operations of 
\begin{itemize}
\item addition: $\pi_{\Ga'\Ga}(h_1\cdots h_{L+1}) = (h_1\cdots h_{L+1})$;
\item subdivision $\pi_{\Ga'\Ga}(\cdots,h_i, h_j,\cdots)=(\cdots,h_i\cdot h_j,\cdots)$;
\item inversion  $\pi_{\Ga'\Ga}(\cdots,h_i,\cdots)=(\cdots,h_i^{-1},\cdots)$,
\end{itemize}
corresponding to the natural properties of connection as defining the morphism of groupoid structures. The associativity property for $\Ga\leq\Ga'\leq\Ga''$ is satisfied: $\pi_{\Ga'\Ga}\circ \pi_{\Ga''\Ga}=\pi_{\Ga''\Ga}$.

By means of $\pi_{\Ga'\Ga}$, one can `glue' all the finite-dimensional spaces $\mc{U}_{\Ga}$, using the construction of a `projective limit' (not analytic one):
\begin{equation}
\bar{\mc{U}} \ \equiv \ \lim_{\Ga\leftarrow} \mc{U}_{\Ga} \ := \ \left\{\{a_\Ga\}\big| a_\Ga\in  \mc{U}_{\Ga}, \pi_{\Ga'\Ga}a_{\Ga'} = a_{\Ga} \ \forall \, \Ga\leq\Ga'\right\}.
\end{equation}
This technically requires an extension of the configuration space to the closure $\bar{\mc{C}}=\mr{Hom}(\Upsilon,\mr{SU}(2))$, consisting of the \emph{generalized connections} (distributional, of which continuous mappings constitute the measure zero sub-set). Essentially, the continuum limit configuration is given by the collection of all its partial representatives $a_\Ga$. The `finite' projector can be defined as $\Pi_\Ga(\{a_\Ga\}_\Ga):=a_\Ga$, satisfying $\pi_{\Ga'\Ga}\Pi_\Ga=\Pi_\Ga$.

\paragraph{Partial Hilbert spaces} are spanned by `cylindrical functions' over a finite number of link holonomies, which can probe the connection only `smeared' along one-dimensional structures. More precisely, a function $f:\bar{\mc{U}}\ra\mathbb{C}$ is called \emph{cylindrical} over a graph $\Ga$ if there exists $f_\Ga:\mc{U}_{\Ga}\ra \mathbb{C}$, such that:
\begin{equation}
f(\{a_\Ga\}_\Ga) \ = \ f_\Ga(a_\Ga),
\end{equation}
which is equivalent to say $f=f_\Ga\Pi_\Ga$. The function cylindrical over $\Ga$ will be also cylindrical over all the finer graphs $\Ga'\geq\Ga$, s.t. the following relation holds $f_{\Ga'}=f_{\Ga}\pi_{\Ga'\Ga}$.

The superposition of various states can be considered, since the set $\mr{Cyl}$ of cylindrical functions over a partially ordered collection of graphs forms a \emph{vector space}. In fact, for the two cylindrical states $f\in\mr{Cyl}_\Ga$, $f'\in\mr{Cyl}_{\Ga'}$, there exist their common refinements on the graph $\Ga''\geq\Ga,\Ga'$, s.t. $f+f'\in\mr{Cyl}_{\Ga''}$.

Since the cylindrical functionals of the connection $\Psi_{\Ga,f}[A]\equiv f_\Ga(h_{\ell_1}[A],\cdots,h_{\ell_L}[A])\in\mr{Cyl}_\Ga$ are simply functions of $L$ elements of $\mr{SU}(2)$ group, the natural Haar measure $d\mu$ on the latter is used to define the Hilbert space \emph{inner product}:
\begin{equation}
\bra\Psi_{\Ga,f}|\Psi_{\Ga,f'}\ket \ = \ \int_{\mr{SU}(2)^L} d\mu^{\otimes L} \overline{f_{\Ga}(...)}f'_{\Ga}(...),
\end{equation}
which gives the square-integrability $\mc{H}_\Ga=\mr{L}^2(\mc{U}_\Ga,d\mu_\Ga)$, $d\mu_\Ga\equiv d\mu^{\otimes L}$. The collection of partial measures $\{d\mu_\Ga\}_\Ga$ satisfies the \emph{cylindrical consistency condition} $(\pi_{\Ga'\Ga})_\ast d\mu_\Ga = d\mu_{\Ga'}$, and comes from the unique (under certain conditions) Ashtekar-Lewandowski $\bra \Psi_1|\Psi_2\ket = \int_{\bar{\mc{C}}}\mc{D}\mu[A]\overline{\Psi_1[A]}\Psi_2[A]$ measure on $\bar{\mc{U}}$ (continuum).

In order to relate different Hilbert spaces on different graphs $\Ga\leq\Ga'$, the \emph{embedding maps} are defined $\iota_{\Ga\Ga'}:\mc{H}_\Ga\ra\mc{H}_{\Ga'}$ (isometric w.r.t. the measure $\mc{D}\mu$), such that
\begin{equation}
(\iota_{\Ga\Ga'}\Psi_\Ga)(a_{\Ga'}) \ \equiv \ \Psi_\Ga (\pi_{\Ga'\Ga}(a_{\Ga'})),
\end{equation}
and the consistency is satisfied $\iota_{\Ga'\Ga''}\circ\iota_{\Ga\Ga'}=\iota_{\Ga\Ga''}$, for consecutive $\Ga\leq\Ga'\leq\Ga''$.

The whole kinematical Hilbert space if then spanned by such graph states via `inductive limit' construction
\begin{equation}
\mc{H}_{\mr{kin}} \ := \ \bigsqcup_{\Ga\subset\mc{S}_3}\mc{H}_\Ga \big/ \sim \ = \ \mr{L}^2(\bar{\mc{U}},\mc{D}\mu),
\end{equation}
modulo the equivalence relation $\Psi_\Ga\sim\Psi'_{\Ga'}$, if there exists $\Ga''\geq \Ga,\Ga'$, s.t. $\iota_{\Ga\Ga''}\Psi_\Ga=\iota_{\Ga'\Ga''}\Psi'_{\Ga'}$ (i.e. two states can be refined to the same state).

\subsection{Spin-networks}

Analogously to the plane waves being the Fourier basis for the functions on $\mathbb{R}$ (abelian group), the Peter-Weyl theorem gives the decomposition of the functions on (locally compact, semisiple) group $H$ into elementary `harmonics' of irreducible representations. For the $\mr{SU}(2)$, the orthonormal basis (w.r.t. $d\mu$) is given by the Wigner rotation matrices $D^j_{mn}(h)\equiv\bra j,m|D^j(h)|j,n\ket$ of the angular momentum, in representation space $\mc{H}^j$ of the spin $j$. Any generic cylindrical state can thus be expanded into
\begin{equation}
\Psi_\Ga \ = \ \sum_{\substack{j_i,m_i,n_i \\ \forall i=1,...,L}} C^{j_1,...,j_L}_{m_1,...,m_L,n_1,...,n_L} \bigotimes D^{j_i}_{m_i,n_i}(h_{\ell_i}),
\end{equation}
where the `Fourier coefficients' are given by $C^j_{mn}=(2j+1)\int_{\mr{SU}(2)} d\mu(h) f(h)\overline{D^j_{mn}(h)}$ (in the case of a single link).

The solution to the Gauss constraints $\hat{C}_i|\Psi\ket\approx 0$, at each node, selects the gauge-invariant subspace. This can be enforced through the group averaging, using the standard projector 
\begin{equation}
P^n_{\mr{inv}}: \, \bigotimes_\ell \mc{H}^{j_\ell} \ \ra \ \mr{Inv}_{\mr{SU}(2)}^{\ph{i}} \bigotimes_\ell \mc{H}^{j_\ell}, \qquad  P^n_{\mr{inv}} \ = \ \int_{\mr{SU}(2)} d\mu(u) \bigotimes_{\ell\supset n} D^{j_\ell}(u),
\end{equation}
at each node, where the product is over all links $\ell$ meeting at $n$. The invariant singlet states, arising in the expansion $\bigotimes_\ell \mc{H}^{j_\ell} =\bigoplus_{J} (\mc{H}^J)^{k_J}$ into irreducible spins are called \emph{intertwiners} $|\iota\ket\in \mr{Inv}_{\mr{SU}(2)}^{\ph{i}} \bigotimes_\ell \mc{H}^{j_\ell}$, spanning the invariant subspace of the product at the node. (There is a single unique intertwiner in the three-valent case.) One can use the basis of intertwiners to rewrite the invariant projector in terms of identity resolution:
\begin{equation}\label{eq:intertwiner-resolution-node}
P^n_{\mr{inv}} \ = \ \sum |\iota\ket\bra\iota|.
\end{equation}
Inserting~\eqref{eq:intertwiner-resolution-node} at each node, the expansion $\Psi_\Ga = \sum_{j_\ell} C^{j_\ell} \Psi_{\Ga,j_\ell,\iota_n}$ will be given in terms of the orthonormal basis of \textbf{invariant spin-network} states~\cite{Baez1996spin-networks-QG}:
\begin{equation}\label{eq:spin-network}
\Psi_{\Ga,j_\ell,\iota_n} \ = \ \bigotimes_n \iota_n \bigotimes_\ell D^{j_\ell}(h_\ell),
\end{equation}
where contraction/saturation of indices is implied for all $D$-matrices matching the basis element $\iota_n$ at the node.

\paragraph{The geometric picture} assoiated with $\mc{H}_\Ga$ comes both from the semi-classical and phase space considerations. For example, the \emph{Livine-Speziale} (overcomplete) basis of coherent intertwiners~\cite{LivineSpeziale2007new-vertex} is obtained by the group averaging
\begin{equation}\label{eq:Livine-Speziale}
|\{j_\ell,\mb{n}_\ell\}\ket \ = \ \int_{\mr{SU(2)}} d\mu(u) \ u\rt\bigotimes_\ell|j_\ell,\mathbf{n}_\ell\ket,
\end{equation}
applied to the `stack' of Bloch spin-coherent states of minimal uncertainty. Apart from representation spins $j$, they are labelled by the unit normals $\mb{n}\in S^2\cong \mr{SU}(2)/\mr{U}(1)$ of the homogeneous space of a sphere, arising in the Perelomov's construction $|j,\mb{n}\ket=D^j(\tilde{\si}(\mb{n}))|j,\pm j\ket$ (where $\tilde{\si}: S^2\ra \mr{SU}(2)$ is the chosen section of the Hopf fibration $\xi$ over classical phase space of a sphere, cf.~\cite{Perelomov1986Generalized-CS,AliAntoineGazeau2014Generalized-CS}). Such $|\iota\ket$ are peaked on the closed configurations $\sum_\ell j_\ell\mb{n}_\ell=0$ in the limit $j\ra\infty$, which endows the intertwiner states with a geometric interpretation in terms of (semi-classical) polyhedra~\cite{BianchiDonaSpeziale2011Polyhedra}. 

The partial Hilbert spaces $\mc{H}_\Ga$ can also be viewed more directly in terms of the quantization of the space of shapes of the collection of such (fuzzy) polyhedra, glued in non-trivial manner~\cite{ConradyFreidel2009geometry-reduction,FreidelSpeziale2010Twisted-geometries,FreidelSpeziale2010Twistors-to-twisted-geometries,LivineTambornino2012LQG-spinor-representation,FreidelGeillerZiprick2013cont-LQG-phase-space}. In particular, only the areas are matched at their intersection, allowing the discrepancy/discontinuity of shapes. Such collective configurations bear the name of \emph{`twisted geometries'}. Their interpretation is problematic in terms of classical (discrete) geometry~\cite{RovelliSpeziale2010LQG-on-a-graph}. In our analysis of the (hyper-)cuboid example in Ch.~\ref{ch:problem}, the related issues are brought up in the context of the analogous `shape-mismatch', arising in Spin Foam models.

\paragraph{The quantization of geometric operators} is one of the major achievements of LQG. Dual to the link holonomies, there are momenta-fluxes $\hat{B}$ which are naturally smeared over complementary two-dimensional surfaces $\hat{B}(S)$. Quantized as the left-invariant vector fields on a Lie group, they act on holonomies with a sort of `grasping' operator, picking up Lie algebra element at every point, where the surface is pierced by a link $\ell=\ell_1\circ\ell_2$ orthogonally:
\begin{equation}
\hat{B}_i(S)h_\ell[A] \ = \ -i\hbar \int_S d^2\si\,  n_a \frac{\de h_\ell[A]}{\de A^i_a(x(\si))} \ = \ \pm \hbar\, h_{\ell_1}[A]J_i h_{\ell_2}[A].
\end{equation}
For the spin-network state $\Psi_\Ga$, one obtains $\hat{B}_i(S)|\Psi\ket=\hat{J}^{(j)}_i|\Psi\ket$ the generator in the representation $j$ of the link. The \emph{gauge-invariant} operator of elementary area (squared) $\hat{B}^2(S_\ell)\propto C_\ell=j_\ell(j_\ell+1)$ has thus a discrete spectrum of the Casimir operator. (This feature can be traced back to the compactness of the group $\mr{SU}(2)$ being used.) The spin-network states are seen from here as the quantum states of `geometric excitations of the space' itself.

\newpage

\subsection{On the `Loop-like' quantization of Cartan gravity?}
\label{subsec:Cartan-quntization}

We allow in this section a few \underline{speculations} about the relation that the Cartan gauge gravity might have to the LQG quantization, by comparison.

\begin{itemize}
\item The first thing to notice is the enlarged configuration space, associated with $\varpi$: although roles of $\theta$ and $\om$ are drastically different, they both are treated on the same gauge-theoretic footing as parts of the unified structure~\footnote{Let us mention that D. Wise had in mind, together with M. Barenz~\cite{Barenz2011Cartan}, to unravel the Cartan structure within LQG phase space variables. This may also provide a new perspective on the relations between dual connection and (non-commutative) flux representations~\cite{Baratin-etal2011non-comm-flux,Dittrich-etal2013generalized-fluxes}.}
\item Technically, the space of (generalized) Cartan connections is similarly realized: the notion of development supplies the covariant functor mapping from the groupoid of paths to the gauge group $G$. It, however, encompasses much more than the Ehresmann's holonomy, since $G$ is the full \emph{principal} group of geometry. The elements of the Poincar\'{e} group associate the parallel transport, as well as the vector of certain length, to elementary edges (e.g. could be taken as `exact' for geodesic segments.) The construction of the graph Hilbert spaces from `cylindrical' functions is applicable, in principal (leaving possible non-compactness issues aside).
\item It is manifestly a spacetime connection, both Lorentz and `diffeomorphism'-covariant, while the graphs are not necessarily restricted to the hypersurface~\footnote{For the discussion of covariance aspects of LQG, see~\cite{Samuel2000spacetime-connection,Livine2006Towards-CovarintLQG,AlexandrovRoche2011CovariantLQG-critique,RovelliSpeziale2011Loretz-covariance}.}. One mentioned that the foliation picture might be not so pivotal for symplectic geometry. However, the phase space structure requires further studies, as well as the notion of observables.
\item In LQG, the diffeomorphism constraint operator is hard to define at the point by its very nature. The `averaging' procedure is used instead, taking the equivalence classes of `s-knots' irrespectively of the precise location in $\mc{S}_3$. It maybe interesting to compare with the picture of Cartan gravity, where the translations are generated by usual shift operators.
\item It may be too restrictive to attach physical meaning only to invariant quantities (magnitudes), but the tensor operators are also quantized successfully (the simplest example is the angular momentum). Such quantities could then be of composite nature.
\end{itemize}

To provide an example for the latter, one can make a curious observation, in the context of discussion. Considering the Einstein tensor, it is observed that this formally has the structure of the Pauli-Lubanski vector of the Poincar\'{e} algebra:
\begin{subequations}
\begin{align}
i\,[\mc{J}^{\ph i}_{ij},\mc{J}^{\ph i}_{kl}] \ & = \ \mc{J}^{\ph i}_{k[i}\eta^{\ph i}_{j]l} - \mc{J}^{\ph i}_{l[i}\eta^{\ph i}_{j]k}, \nonumber \\
i\, [\mc{J}_{ij}^{\ph{i}},\mc{P}_k^{\ph{i}}] \ & = \ \mc{P}_{[i}^{\ph{i}}\eta_{j]k}^{\ph{i}}.
\end{align}
\end{subequations}

So, if one were to quantize vector-like tetrads as $\boldsymbol{\theta}\mapsto \mc{P}$, and expand the cycle holonomy to obtain the Lie algebra element of curvature $\mb{\Om}\mapsto\mc{J}$, then the Einstein tensor/`energy-momentum vector' would acquire the form:
\begin{equation}\label{eq:Pauli-Lubanski}
\star(\mb{\Om}\wedge\boldsymbol{\theta}) \ \mapsto \ \mc{W}= \star(\mc{J}\wedge{\mc{P}}), \qquad \mc{W}_i \ := \ \frac12 \varepsilon_{ijkl}^{\ph{i}}\mc{J}^{jk}\mc{P}^l
\end{equation}
This has the algebraic properties (up to normalization):
\begin{subequations}
\begin{align}
\mc{P}^i\mc{W}_i^{\ph{i}} \ & = \ 0, &  [\mc{P}_i^{\ph{i}},\mc{W}_j^{\ph{i}}] \ & = \ 0, \\
i\,[\mc{J}_{ij}^{\ph{i}},\mc{W}_k^{\ph{i}}] \ & = \ -\mc{W}_{[i}^{\ph{i}}\eta_{j]k}^{\ph{i}}, &  i\,[\mc{W}_i^{\ph{i}},\mc{W}_j^{\ph{i}}] \ & = \ \varepsilon_{ijkl}^{\ph{i}}\mc{W}^k\mc{P}^l.
\end{align}
\end{subequations}
(Note that in 3d vector notation $W_0=\vec{J}\cdot\vec{P}$, $\vec{W}=E\vec{J}-\vec{P}\times\vec{K}$.) The first line expresses the `transversality' and `conservation' w.r.t. translations, while the second -- behaviour under Lorentz transformations, and the covariant version of the 3d angular momentum algebra.

The unitary irreducible representations of the Poincar\'{e} group (massive case) are labelled by the Casimir operators $|\mc{P}|^2=\mu^2$ (rest mass) and $|\mc{W}|^2 = -\mu^2 \si(\si+1)$ ($\si$ -- spin quantum number). Have we been ``quantizing gravity'', or just re-invented the spin?

\section{Path-integral approach of Spin Foams}
\label{sec:Spin-Foams}

The Spin Foam program is the discrete path-integral/state-sum approach that emerged from the confluence of researches in Loop gravity and the broad generalization of a Topological Quantum Field Theory (TQFT).  

\subsection{Sum over histories, general boundary field theory}
\label{subsec:general-boundary}

As a simple motivating example, consider the time evolution map $U(t_f-t_i)$ of the particle from the (eigenstate of the) position $x_i$ in initial configuration, to that of the final $x_f$. The `transition probability amplitude' can be re-written concisely
\begin{equation}
\mc{A}(x_f,t_f;x_i,t_i) \ = \ \bra x_f,U(t_f-t_i)x_i\ket  \ = \ \int_{\substack{x(t_i)=x_i \\ x(t_f)=x_f}} \exp \left(\frac{i}{\hbar} S[x]\right)\mc{D}x,
\end{equation}
in the form of the Feynman's integral along interfering paths, where every single history of a particle contributes to the sum with an oscillatory phase, given by the action. 

In case of a field theory, one instead integrates the configurations $\phi(x)$ over the $m$-dim region $\mc{R}\subset\mc{M}$, bounded by the hypersurfaces $\mc{S}_t$ at initial and final time steps~\cite{Engle2013SF-intro}
\begin{equation}\label{eq:general-boundary}
\mc{A}(|\varphi\ket,\pa \mc{R}) \ \equiv \ \mc{A}(|\varphi_f\ket\otimes|\varphi_i\ket,\mc{S}_{t_f}\cup \mc{S}_{t_i}) \ = \ \int_{\substack{\phi|_{t_i}=\varphi_i \\ \phi|_{t_f}=\varphi_f}} e^{iS[\phi]}\mc{D}\phi \ \equiv \ \int_{\phi|_{\pa\mc{R}}=\varphi } e^{iS[\phi]}\mc{D}\phi,
\end{equation}
where the field has simultaneous eigenstates $|\varphi_t\ket\in\mc{H}_{\mc{S}_t}^{\ph{i}}$, in the corresponding Hilbert spaces.

In the \emph{general boundary formulation}, $\mc{R}$ is allowed to be any finite space-time region~\cite{Oeckl2003general-bndry}. This permits purely local calculations (irrespectively of the asymptotic behaviour at infinity), consistent with the locality of the measuring apparatus one would actually use. Furthermore, the lack of an a priori fixed notion of which space-time regions may be used is compatible with the spirit of background independence. The usual situation is reproduced for the disjoint boundaries like in~\eqref{eq:general-boundary}, s.t. the Hilbert space is the tensor product $\mc{H}_{\pa\mc{R}}^{\ph{i}}=\mc{H}_{\mc{S}_{t_f}}^\ast\otimes\mc{H}_{\mc{S}_{t_i}}^{\ph{i}}$. 

It would be suggestive to denote the elementary contributions with the `amplitude' $\mc{A}\propto e^{iS[\phi]}$ (heuristically), while the extended sum will give the `partition function' $Z_{\mc{R}}^{\ph{i}}[\varphi]=\int_{\phi|_{\pa\mc{R}}=\varphi}\mc{A}\,\mc{D}\phi$. In such a form, it can be viewed as an extension of the axiomatics of TQFT (Atiyah), whose data -- consisting of (finite dimensional) state spaces and partition  functions -- have to satisfy a number of conditions~\cite{Barrett1995QG-TQFT}: 1) The rules of quantum  mechanics for composing amplitudes, and 2) Functoriality, or the correct behaviour under diffeomorphisms of  manifolds (s.a. cobordisms). The latter may be viewed as the appropriate gluing conditions:
\begin{equation}
Z_{\mc{R}_2}^{\ph{i}}\circ Z_{\mc{R}_1}^{\ph{i}} \ = \ Z^{\ph{i}}_{\mc{R}_2\cup_{\mc{S}}^{\ph{i}}\mc{R}_1},
\end{equation}
where $\pa\mc{R}_1=\bar{\mc{S}}_2\cup\mc{S}_1$, $\pa\mc{R}_2=\bar{\mc{S}}_3\cup\mc{S}_2$.

The amplitude/partition function acts as a linear functional $Z_{\mc{R}}^{\ph{i}}: \, \mc{H}_{\pa\mc{R}}^{\ph{i}} \ \ra \ \mathbb{C}$. For a given state $|\Psi\ket \in \mc{H}_{\pa\mc{R}}^{\ph{i}}$, the result may be written using the path-integral representation as
\begin{equation}
Z_{\mc{R}}^{\ph{i}}(\Psi) \ \equiv \ \bra Z_{\mc{R}}^{\ph{i}}|\Psi\ket \ = \ \int_{\pa\mc{R}} \mc{D}\varphi \, Z_{\mc{R}}^{\ph{i}}[\varphi]\Psi[\varphi].
\end{equation}

\subsection{Topological BF theory and its quantization} 

An example of the implementation of this program is provided by a close TQFT relative of gravity, called BF-theory~\cite{Baez2000SF-BF}. (The name stems from the tradition to denote the curvature/`force field' of the connection as $F=F[A]$.) It is determined by the action functional $\int\bra\mb{B},\mb{\Om}\ket$, having the appearance of~\eqref{eq:Regge-action}, however the Lie algebra $\mf{h}$-valued $(m-2)$-form $\mb{B}$ is an arbitrary variable field. The field equations are then trivially
\begin{equation}\label{eq:BF-eom}
D^\om\mb{B} \ = \ 0, \qquad \mb{\Om} \ = \ 0.
\end{equation}

Apart from the usual gauge rotations/Lorentz transformations, the action is invariant w.r.t. enlarged symmetry of `shifts'
\begin{equation}\label{eq:topological-symmetry}
\de\mb{B} \ = \ D^\om \tilde{\mb{Q}}, \qquad \de\boldsymbol{\om} \ = \ 0,
\end{equation}
generated by \emph{$(m-3)$-form} $\tilde{\mb{Q}}\in\La^{m-3}(P,\mf{h})$, which makes the theory topological/finite-dimensional (i.e. devoid of non-trivial local excitations of geometry). To establish the relation of~\eqref{eq:topological-symmetry} with diffeomorphisms, the equations of motion are usually involved (i.e. ``on-shell'' considerations, cf.~\cite{Buffenoir-etal2004Hamiltonian-Plebanski}). We notice, however, that $\mb{B}$ may still be regarded as the general system of bivectors (non-simple, but spanned by~$\theta$), transforming as $\de\mb{B}=D^\om\imath_X^{\ph{i}}\mb{B} + \imath_X^{\ph{i}}D^\om\mb{B}$ under diffeomorphisms; whereas for the curvature of the connection: $\de\mb{\Om}=D^\om\imath_X^{\ph{i}}\mb{\Om}$.  Then the symmetry is just the consequence of diffeo-invariance of the action $\de\bra\mb{\Om},\mb{B}\ket = \imath_X^{\ph{i}}\bra \mb{\Om}, D^\om\mb{B}\ket + d\bra\imath_X^{\ph{i}}\mb{\Om},\mb{B}\ket+ \bra \mb{\Om}, D^\om\imath_X^{\ph{i}}\mb{B}\ket$ (up to divergence). The first two terms do not contribute, leaving the `remnant' $(m-3)$-form $\tilde{\mb{Q}}=\imath_X^{\ph{i}}\mb{B}$ of~\eqref{eq:topological-symmetry}.

Such a theory may be successfully regularized and quantized. For the \textbf{discretization} of the field-theoretic d.o.f., one considers (2-skeleton of) a cell-complex $\mc{K}$, consisting of 0-dimensional `vertices'~$v$, 1-dimensional `edges'~$e$, and 2-dimensional `faces'~$f$. Typically, $\mc{K}$ appears as a combinatorial dual of another cell-complex $\De$ of \emph{simplicial} decomposition of $\mc{M}$ (triangulation). The connection is then substituted for holonomies, `path-ordered' along the edges:
\begin{equation}\label{eq:discrete-holonomies}
h_e[\om] \ \equiv \ \int_e \overleftarrow{\exp} \left(\om\cdot\mc J\right).
\end{equation}
They form the finite product of parallel transport around a closed cycle of edges surrounding a face:
\begin{equation}
H_f \ := \ \overleftarrow{\prod_{e\subset f}} \, h_e,
\end{equation}
which play the role of discrete curvature operator. The $B$-field is naturally `smeared' as complementary bivector over the (dual) $(m-2)$-dimensional cells $S_f$ of a $\De$-triangulation~\footnote{Mention that the integrand should be acted upon by holonomies, referring it to the single frame at the source, in order to ensure the correct transformation properties of the Lie algebra element under the gauge rotations (cf.~\cite{FreidelGeillerZiprick2013cont-LQG-phase-space}).}:
\begin{equation}\label{eq:discrete-B}
B_f \ := \ \int_{S_f} B \, \in \bigwedge^2\mathbb{R}^m.
\end{equation}
which we label bijectively with $f$. The action functional, regularized this way over over $\mc{K}$, takes the form
\begin{equation}
S_{BF}^{\mc{K}} \ := \ \sum_f \tr (B_f\cdot H_f).
\end{equation}
The first equation of motion~\eqref{eq:BF-eom} is naturally seen as the discrete `closure condition':
\begin{equation}\label{eq:closure-3d}
\sum_{f\supset e} B_f \ = \ 0 \qquad \forall \, e,
\end{equation}
for the 2-dimensional boundary of the 3-dimensional region (tetrahedron $\tau_e$), dual to an edge $e$. In the canonical picture, this corresponds to the Gauss law constraint. The local flatness condition on curvature then becomes simpy $H_f=0$ for each $f\in \mc{K}$. 

On the boundary of a 2-complex, the graph $\pa\mc{K}=\Ga$ is induced: to each edge $e$ there uniquely corresponds a node $n=\pa e$, while the face $f$ cuts a link $\ell\subset\pa f$. Correspondingly, the discrete variables for the `kinematic' picture of the boundary graph  Hilbert spaces $\tilde{\mc{H}}_\Ga$ comprise the induced set of holonomies and bivectors $\{h_\ell, B_\ell\}_{\ell\subset\Ga}$. (Note that they take values in the full Lorentz group/algebra of spacetime, accordingly.)

\paragraph{The `sum over states/histories' quantization} proceeds as follows, using tools from the group representation theory. The discretized partition function
\begin{equation}\label{eq:partition-1}
Z_{BF}^{\mc{K}} \ := \ \int \prod_e dh_e \int \prod_f dB_f \, e^{i\sum_f \tr(B_f\cdot H_f)}
\end{equation}
contains the ordinary integrals that can be performed exactly:
\begin{equation}\label{eq:partition-2}
Z_{BF}^{\mc{K}} \ = \ \int \prod_e dh_e\, \prod_f \left( \de(H_f) +\de (-H_f)\right).
\end{equation}
By means of the Peter-Weyl theorem for a compact group $H$, for simplicity, the delta-distribution can be re-written:
\begin{equation}\label{eq:Peter-Weyl}
\de(h) \ = \ \sum_\rho d_\rho \chi^\rho(h) , \qquad \chi^\rho(h) \ = \ \mr{Tr}(D^{\rho}(h)),
\end{equation}
in terms of a character traces of face-holonomies in the group's unitary representation $\rho$~\footnote{The appearance of the second ``delta'' in~\eqref{eq:partition-2} is not very plausible. It is often omitted or `corrected', we will follow this practice as well.}. Upon substitution into~\eqref{eq:partition-2}, one obtains:
\begin{equation}
Z_{BF}^{\mc{K}} \ = \ \int \prod_e dh_e\, \prod_f \, d_{\rho_f}\, \mr{Tr} \left[D^{\rho_f}\left(\prod_{e\subset f} \, h_e\right)\right],
\end{equation}
for the set of faces $f$ meeting at an edge $e$ (forming the boundary of a polyhedron $\tau_e$ in the dual cellular decomposition $\De$ picture).

One now recalls that the invariant projectors for the `stack' of representation spaces $\mc{H}_e=\bigotimes_{f\supset e}\mc{H}^{\rho_f}$, appearing at the edges:
\begin{equation}\label{eq:edge-invariance}
P^e_{\mr{inv}}: \, \mc{H}_e \ \ra \ \mr{Inv}_H^{\ph{i}} \mc{H}_e, \qquad  P^n_{\mr{inv}} \ = \ \int_H dh_e \, \bigotimes_{f\supset e} D^{\rho_f}(h_e),
\end{equation}
can be equivalently rewritten using the identity resolution in terms of basis of intertwiners (invariant tensors): 
\begin{equation}\label{eq:intertwiner-resolution-edge}
P^e_{\mr{inv}} \ = \ \sum_{\iota_e} |\iota_e\ket\bra\iota_e| \ \equiv \ \mathbbm{1}^{\ph{i}}_{\mr{Inv}_H^{\ph{i}} \mc{H}_e}.
\end{equation}

This allows to re-group the expression for partition function in terms of representations (`spins') and intertwiners (could be also labelled by `intermediate spins' in the recoupling scheme). One notices that the set of edge-intertwiners $\iota_e$ can be re-collected at the vertex; whereas their indices, belonging to the same face-representation $\rho_f$, be contracted according to the pattern of a 2-complex~$\mc{K}$, one obtains an expression for the \textbf{vertex amplitude}:
\begin{equation}\label{eq:vertex-1}
\mc{A}_v \ := \ ``\mr{Tr}" \left(\bigotimes_{e\rightarrow v} |\iota_e\ket \, \bigotimes_{e\leftarrow v} \bra\iota_e|\right).
\end{equation}
(The ``trace'' here just to represent saturation of indices by contraction, of the same type as in spin-networks~\eqref{eq:spin-network}; the dualization is applied for the different orientation of `in-going' and `out-going' direction of edges.) At the same time, the spin-network state $|\Psi_{\Ga_v}\ket\in \tilde{\mc{H}}_{\Ga_v}$ is induced on the boundary graph, cut by a sphere around each vertex (see Fig.~\ref{fig:boundary-vertex}), so that the amplitude acts as
\begin{equation}\label{eq:vertex-2}
\mc{A}_v :\, |\Psi_{\Ga_v}\ket \ \ra \ \mathbb{C},
\end{equation}
thus justifying its name.

\begin{figure}[t]
\center{\includegraphics[width=0.3\linewidth]{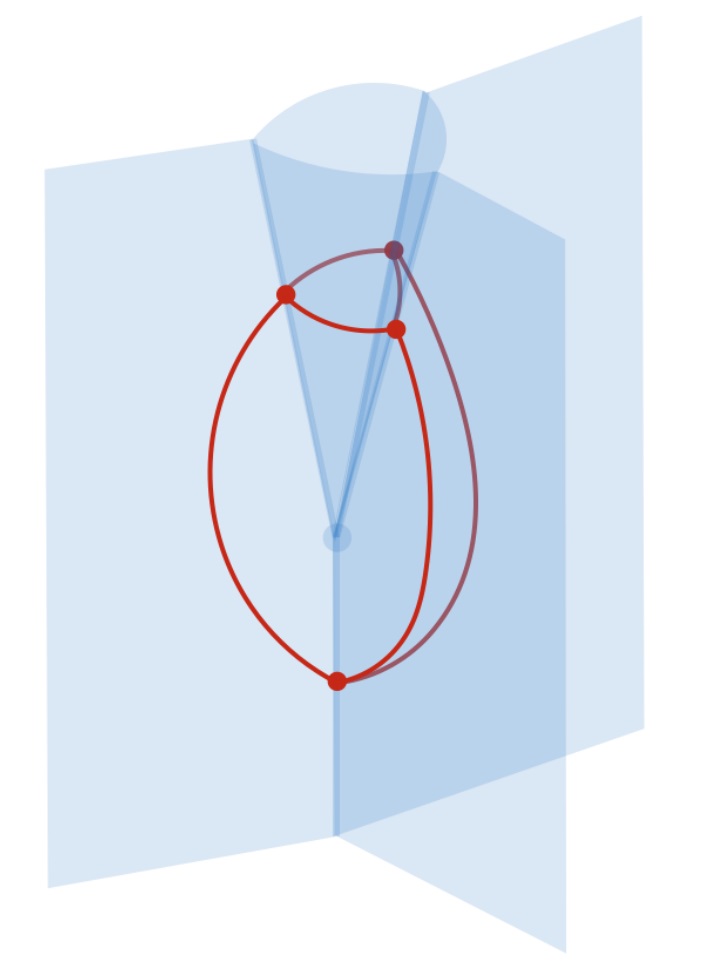}}
\caption{The boundary graph of a 2-complex $\mc{K}$ cut at the vertex (scematic 3d view from~\cite{Rabuffo2018SpinFoam-RGflow}).}
\label{fig:boundary-vertex}
\end{figure}

The tuple $(\mc{K},\{\rho_f,\iota_e\})$ of a 2-complex $\mc{K}$, `coloured' with the group representation labels of faces and edges, meeting at vertices, bears the name of a \textbf{spinfoam}. It is regarded as a `history' of a discretized/quantum spacetime geometry (flat for BF theory). The obtained expression for partition function 
\begin{equation}\label{eq:partition-3}
Z_{BF}^{\mc{K}} \ = \ \sum_{\rho_f,\iota_e} \prod_f d_{\rho_f} \, \prod_v \mc{A}_v(\{\rho_f\}_v,\{\iota_e\}_v),
\end{equation}
performs the summation over contributions from all such spinfoams, for a given 2-complex~$\mc{K}$.

\paragraph{Gravity in 3d and Ponzano-Regge model.} This picture becomes an exact quantization for the Einstein gravity in $m=3$. Indeed, since all bivectors are simple there, the classical continuum theories coincide. Considering the case of triangulations, there is a unique intertwiner for the 3-valent node, given by Wigner's $3j$-symbol $\iota^{m_1m_2m_3}=\left(\begin{smallmatrix}j_1 & j_2 & j_3 \\
m_1 & m_2 & m_3
  \end{smallmatrix}\right)$. The dual complex to a triangulation is also simplicial, so that at each vertex intertwiners are contracted along the tetrahedron pattern, obtaining the Wigner's invariant $\left(\begin{smallmatrix}j_1 & j_2 & j_3 \\
j_4 & j_5 & j_6
  \end{smallmatrix}\right)=\{6j\}_v$-symbol for the vertex amplitude of \emph{Ponzano-Regge} model of Euclidean 3d gravity~\cite{BarrettNaish-Guzman2009Ponzano-Regge}~\footnote{This could be compared with another succesful quntization of topological 3d gravity by Witten~\cite{Witten1988TQFT-3d-grav,Witten1989amplitudes-3d-grav}, using $G$-structures and ultimately based on the reformulation in terms of a Chern-Cimons theory for certain Cartan connection~\cite{Carlip1993geom-structures,Wise2009symmetric-spaces-Cartan}.}. 
  
It is reassuring that in the scaling limit $j\ra\la j$ of large spins $\la\ra\infty$, this quantity in the asymptotic regime leads to the exponential with the oscillatory phase given by Regge's deficit angle of discrete gravity. The relation to the canonical LQG can be established~\cite{Rovelli1993LQG-SF-3d,NouiPerez2005LQG-SF-3d,BonzomFreidel2011LQG-SF-3d}. This turned out to be the major inspiration for the Spin Foam program, trying to extend the above results to realistic case Lorentz group in four dimensional spacetime. For the `kinematical' picture of LQG states on the hypersurface this would allow to supply them with `dynamics', which could be viewed in terms of varying foliations and `evolving spin-networks', correspondingly~\cite{ReisenbergerRovelli1997SFfromLQG-sum,Markopoulou1997spin-network-evolution}.

\subsection{Pleba\'{n}ski constrained formulation and bivector geometries}
\label{subsec:Plebanski}

\paragraph{The setting.} The classical backdrop behind the Spin Foam proposals in 4d is the observation due to Pleba\'{n}ski~\cite{Plebanski1977,Celada-etal2016BF-review} that the Einstein-Cartan action can be recast as a constrained BF theory:
\begin{equation}\label{eq:BF+cnstr}
S[B,\om, \la] \ = \ \int \frac12 B_{AB}^{\phantom{AB}}\wedge \Om^{AB}[\om] \ + \ \la^\al  C_\al[B].
\end{equation}
(Here $A,B=0,1,2,3$ -- (internal, or anholonomic) indices in the defining vector representation of the homogeneous Lorentz group $H=SO(3,1)$; the indices of the coordinate labels on $\mc{M}$ will be lowercase latin $a,b=0,1,2,3$, respectively.) The 2nd term represents constraints $C_\al[B]=0$ on $B$-field, enforced by the Lagrange multipliers $\la^\al$ ($\al$ -- multi-index); they effectively reduce the number of independent $B$-components, so that the bivector is given by the (dual) simple product $2B=\star \theta \wedge \theta\equiv\star\Si$ of (some) tetrad co-frame $\theta$. Hence, on the constraint surface, the theory acquires the form familiar Einstein-Cartan form:
\begin{equation}\label{eq:Einstein-Cartan}
S_{\mr{EC}}^{\ph{i}}[\theta,\om] \ = \ \int \frac12 \varepsilon_{ABCD}^{\phantom{ABCD}}\,\theta^A\wedge \theta^B\wedge \Om^{CD}[\om].
\end{equation}

\paragraph{The strategy} in the majority of Spin Foam approaches is to \textit{first quantize and then constrain}, according to the following route:
\begin{enumerate}
\item discretize the classical theory on a piecewise-flat partition of the spacetime $\mc M$ (most commonly, simplicial);
\item quantize the topological BF part of the discretized theory;
\item impose (a version of) simplicity constraints $C_\al[B]\approx 0$ directly at the quantum level.
\end{enumerate}
The core non-trivial part in constructing SF models for gravity is the third step of the implementation of a quantum version of the simplicity constraints at the level of state-sum for BF theory:
\begin{equation}\label{eq:simplicity-quant}
\widehat{C_\al[\Si]} \ \approx \ 0.
\end{equation}
Depending on the first/second-class nature of the set \eqref{eq:simplicity-quant}, they should either annihilate the state functionals (\textit{\`{a} la} Dirac), or to be imposed weakly on matrix elements (\textit{\`{a} la} Gupta-Bleuler). This usually leads to restrictions on spin labels $j_f$ and/or intertwiners $\iota_e$ of the boundary Hilbert spaces. 

The most widely known and well studied is the Pleba\'{n}ski's quadratic set of constraints, existing in 2 versions:
\begin{equation}\label{eq:Plebanski-cnstr}
\text{(a)} \qquad B^{AB}\wedge B^{CD} \ = \ V \, \varepsilon^{ABCD} \quad\qquad \stackrel{\tilde{V}\neq 0}{\Longleftrightarrow} \qquad\quad \text{(b)} \qquad \varepsilon_{ABCD}^{\phantom{ABCD}}B^{AB}_{ab} B^{CD}_{cd} \ = \ \tilde{V} \, \varepsilon_{abcd}^{\phantom{ABCD}}.
\end{equation}
They are equivalent, provided the quantity $V$ (resp. $\tilde{V}$) -- which is defined by~\eqref{eq:Plebanski-cnstr} through contraction with $\varepsilon$ -- is non-vanishing~\cite{DePietriFreidel1999Plebanski}. In that case, there are two non-degenerate sectors of solutions:
\begin{equation}\label{eq:Plebanski-sectors}
I^\pm: \quad B^{AB} \ = \ \pm \theta^A\wedge \theta^B,  \qquad\qquad II^\pm: \quad B^{AB} \ = \ \pm \varepsilon^{AB}_{\phantom{AB}CD}\theta^C\wedge \theta^D,
\end{equation}
and $V=\pm\frac{1}{4!}\eps^{\phantom{A}}_{ABCD}\theta^A\wedge \theta^B\wedge \theta^C\wedge \theta^D=\tilde{V}\,d^4x$ acquires an interpretation of spacetime 4-volume.
The sectors $II^\pm$ reproduce \eqref{eq:Einstein-Cartan} up to the discrete sign ambiguity, while $I^\pm$-sectors give the topological Holst term. The treatment of degenerate case $\tilde{V}=0$, and relations between sectors can be found in \cite{Reisenberger1999Plebanski}.

\paragraph{The bivector constraints and the discretization} of classically equivalent forms of constraints \eqref{eq:Plebanski-cnstr} lead to two, \textit{a priori} different, SF models. The (a)-case gives the version of the Reisenberger state-sum proposal \cite{Reisenberger1997state-sum} (corresponding to a self-dual formulation), whereas the case (b) is the most prevailing and leads to the Barrett-Crane (BC)~\cite{BarrettCrane1998Euclid,BarrettCrane2000Lorentz} model.

Suppose, our 2-complex $\mc{K}$ is dual to a triangulation $\De$. Then, depending on the relative position of triangles in a 4-simplex, the constraints fall into 3 types:
\begin{enumerate}[label={\upshape(\roman*)}]
\item $\displaystyle{\star B_f \cdot B_f=0}$ for each triangle/face $f$ -- diagonal (or `face') simplicity; \label{diag-simplicity}
\item $\displaystyle{\star B_f\cdot B_{f'}=0}$ if two faces share an edge $f\cap f'=e$ -- cross-simplicity (or `tetrahedral' constraint); \label{cross-simplicity}
\item $\displaystyle{\star B_f\cdot B_{f'}=:\tilde{V}_v(f,f')}$ for any pair of faces $f,f'$ meeting at the vertex~$v$ and spanning 4-simplex volume -- the so-called volume (or `4-simplex') constraint.\label{volume-quadr}
\end{enumerate}

Each of these constraints have different status and are treated accordingly. In particular, they are implemented, respectively, at the level of faces/tetrahedra/4-simplices. The simplicity constraints formulate the necessary and sufficient conditions for the general system of bivectors to correspond to the faces of the discrete cell-complex of the single 4-simplex (given by the metric $\cg\sim\theta$, discretized as vectors $E=\int \theta$ along 1-dimensional lines). The 1st condition implies that the bivector is \textit{simple}, i.e. given by the wedge product of two vectors:
\begin{equation*}
\text{\ref{diag-simplicity}} \qquad \Rightarrow \qquad B^{AB}_{f_1} \ = \  E^{[A}_2E^{B]}_3 \qquad \text{or} \qquad \star B^{AB}_{f_1} \ = \  E^{[A}_2E^{B]}_3
\end{equation*}
If two triangles share a common edge, then the sum of the corresponding two bivectors is also simple:
\begin{equation*}
\text{\ref{cross-simplicity}} \qquad \Rightarrow \qquad B^{AB}_{f_2} \ = \  E^{[A}_3E^{B]}_1 \qquad \text{or} \qquad \star B^{AB}_{f_2} \ = \  E^{[A}_3E^{B]}_1, \qquad \text{and cyclically }\ \forall f\supset e.
\end{equation*}
The closure condition \eqref{eq:closure-3d} states that the geometry of the tetrahedron $\tau_e$, built on vectors $E_1,E_2,E_3$ (or its dual), has boundary surface which is closed. This condition allows a generalization to arbitrary valence and is sufficient to uniquely specify the geometry of a flat polyhedron~\cite{BianchiDonaSpeziale2011Polyhedra}. 

Regarding the kinematical Hilbert spaces induced on a boundary, we set up the following
\begin{definition}
We call a \textbf{bivector geometry} the graph $\Ga\subset S^3$, together with the set of bivectors $\{B_\ell\}_{\ell\in\Ga}$ associated with its oriented links as induced from the faces at the vertex $v\in\mc{K}$. We also demand the bivector geometry to satisfy the closure~\eqref{eq:closure-3d}, as well as the first two of (diagonal and cross-) simplicity constraints.
\end{definition}
The underlying reason, which allows to pass from the description in terms of edge vectors to face bivectors, is contained in the Barrett's \emph{reconstruction theorem}~\cite{BarrettCrane1998Euclid}. It states that the flat 4-simplex is uniquely associated (up to the orientation, translations and inversions) with the set of ten bivectors, satisfying the above three conditions (supplemented with the orientation reversion $B_{AB}=-B_{BA}$ + some non-degeneracy requirements). 

The role of \ref{volume-quadr} is to ensure that the geometries of the tetrahedra fit together to form consistently a 4d geometry, in particular, that the volume of a 4-simplex is invariably defined. The volume constraint~\ref{volume-quadr} is not the part of conditions, defining the bivector geometry, because it is implied by the constraints on the tetrahedral level and the closure. The derivation roughly goes as follows~\cite{EPRL-FK2008flipped2,EPRL-FK2008LS}. Label the five tetrahedra with $e=1,...,5$; the triangle $\triangle_{12}$ is shared by two respective tetrahedra. Using the closure \eqref{eq:closure-3d}, say for tetrahedron $1$, and contracting it with all the other bivectors, one can freely swap between triangles, for instance:
\begin{equation}\label{eq:volume-simplex}
\star B(\triangle_{12})\cdot B(\triangle_{45})+\star B(\triangle_{13})\cdot B(\triangle_{45})\ = \ -\star B(\triangle_{14})\cdot B(\triangle_{45})-\star B(\triangle_{15})\cdot B(\triangle_{45})\ = \ 0,
\end{equation}
s.t. the r.h.s. eliminates on the surface of the simplicity constraints \ref{cross-simplicity} $\Rightarrow$ hence \ref{volume-quadr} follows. (In the canonical picture parlance, \ref{volume-quadr} is interpreted as  a ``secondary'' constraint, which ensures the dynamical conservation of the simplicity constraints \ref{cross-simplicity} across the 4-simplex~\footnote{This is, however, not the statement of the Hamiltonian analysis of the underlying action~\cite{Buffenoir-etal2004Hamiltonian-Plebanski} in the Bergmann's terminology.}.) 

\subsection{Linear formulation with normals}
\label{subsec:linear-formulation}

Replacement of~\ref{volume-quadr} by~\eqref{eq:closure-3d} is particularly beneficial for the quantum theory, since the linear in $B$ and local in each tetrahedron closure constraint is much more easier to deal with. The major ingredient in the new EPRL-FK-KKL Spin Foam proposals is the \textit{linearization (partial) of the rest of simplicity constraints}. It follows directly from the geometric meaning of conditions \ref{diag-simplicity},\ref{cross-simplicity}, which basically state that four triangles, belonging to the same tetrahedron $\tau_e$ and described by the area bivectors $\Sigma^{AB}_f$, lie in one hyperplane. We now discuss briefly some details of this re-formulation, as they appear in the literature. 

In the original construction~\cite{EPRL-FK2008flipped2,EPRL-FK2008finiteImmirzi,EPRL-FK2008FK,EPRL-FK2008LS}, one associates \emph{auxiliary normal} 4d-vectors $\mc V^A_e,\ e=1,...,5$, to each of the five tetrahedra in the boundary of a 4-simplex (assume they are all timelike $\mc V_e\in \mathbb{H}^3_+\cong SL(2,\mathbb{C})/SU(2)$). The quadratic diagonal \ref{diag-simplicity} and cross-simplicity~\ref{cross-simplicity} could then be incorporated into the single \emph{`linear cross-simplicity'} constraint, by requiring the orthogonality: 
\begin{equation}\label{eq:linear-simplicity}
\text{(ii')}\qquad \forall f\supset e \ : \qquad \Si_f^{AB} \mc V_{Be}^{\phantom{B}} \ = \ 0 \qquad \Leftrightarrow  \qquad I_B(\mc V_e)\rt \Si_f \ = \ 0.
\end{equation}
The projector on the r.h.s. separates the ``boost'' components of $\Si$, which are co-aligned with $\mc V$:
\begin{equation}\label{eq:BR-projectors}
I_B(\mc V)^{AB,CD}\ :=\ \pm 2 \mc V^{[B}\eta^{A][C}\mc V^{D]}, \qquad I_R(\mc V)^{AB,CD}\ :=\ \eta^{A[C}\eta^{D]B} \mp 2 \mc V^{[B}\eta^{A][C}\mc V^{D]},
\end{equation}
from the ``rotational'' part, generating the isotropic $H_{\mc V} = h(\mc V)\rt SU(2)$ subgroup, which leaves $\mc V$ invariant. The gain in this new form of constraints is that it excludes the undesired $I^{\pm}$-solutions in~\eqref{eq:Plebanski-sectors}, leaving just the mix of the gravitational $II^{\pm}$-sectors, as well as degenerate one~\cite{Engle2011Plebanski-sectors} (which we do not consider here). It thus imposes stronger conditions than quadratic~\ref{diag-simplicity} and~\ref{cross-simplicity}, which then automatically follow.

We are making two remarks, calling the attention to normals $\mc V_e$, the prime interest of the present work. The first observation is that the constraints now depend on the enlarged set of configuration variables. In the symbolic notation of~\eqref{eq:simplicity-quant}, it is to be replaced with
\begin{equation}\label{eq:simplicity-quant+normal}
\widehat{C_\al[\Si,\mc V]} \ \approx \ 0,
\end{equation}
in order to reflect the introduction of a new geometric objects, in addition to bivectors. They are now both may be \emph{a priori} subject to quantization. 

The second comment concerns the relation of object on the constraint surface, namely: 
\begin{equation}\label{eq:linear-simplicity-solution}
\Si^{AB}_f \ = \ E^{[A}_1E^{B]}_2 \ = \ \frac{3!}{h_f^e} \star \left(\mc N^{\phantom{A}}_{f\supset e}\wedge\mc V^{\phantom{B}}_{e\phantom{f}}\right)^{AB} \qquad \text{or, conversely} \qquad  \mc V^A_e \ = \ \frac13 h_f^e \left(\star \Si^{AB}_f\right) \mc N^{\phantom{A}}_{B f\supset e} \, . 
\end{equation}
Here the edge vectors  $E_{1,2}\perp\mc N,\mc V$ of triangular face $S_f$ are orthogonal to the (spacelike, $\mc N^2=+1$) surface normal $\mc N$, lying within tetrahedron $\tau_e$: $\mc N\cdot\mc V=0$. In order for $|\mc V|^2\equiv\mc V^A\mc V_A$ to be the 3d volume (squared) of $\tau_e$, the proportionality coefficient is ought to be the height $h_f^e=(E_3\cdot\mc N)$ from the base $S_f$ to the apex. It appears from the above relations as if neither of $\Si,\mc V$ could be considered more ``fundamental'', since one can be expressed through the other and vice versa. The resolution of conundrum ultimately lies, not surprisingly, in the composite nature of quantities, both comprised of tetrad $\theta$ d.o.f. This standpoint will pave the way for our extension of field space in Sec.~\ref{sec:Poincare-BF}, and reformulation of~\eqref{eq:simplicity-quant+normal} in Sec.~\ref{sec:continuum-linear}.


\subsection{On the quantization in new models}\label{sec:quantization}
In general, the expression for the spinfoam partition function (associated with the 2-complex~$\mc{K}$) factorizes as
\begin{equation}\label{eq:partition-Z}
Z_{\mc{K}} \ = \ \sum_{j_f,\iota_e}\prod_{f}\mc A_f\prod_{e}\mc A_e\prod_{v}\mc A_v,
\end{equation}
where the choice of elementary amplitudes for the vertices, edges and faces should be specified, based on the physical/geometric considerations of the studied theory. One could say that the construction of the EPRL-FK models is highly inspired by the desire to meet the kinematical Hilbert spaces of LQG induced on its boundary graph. (Not so much for its predecessor of the purely spacetime bivector Barrett-Crane model.)

The bivectors, whether dual or not, are quantized as the Lie algebra elements, or left invariant vector fields on the group. Given that the LQG is based on the Holst action functional with non-trivial Barbero-Immirzi parameter $\ga$~\cite{Holst1996action}, one usually quantizes the linear combination, appearing in the kinematical part of the lagrangian (roughly corresponding to the conjugate momenta-fluxes in the Poisson symplectic structure). Agreeing for a while to specify this generic bivector with the notation $B$ of the BF theory action, upon quantization:
\begin{equation}\label{eq:Immirzi-flipped-map}
B_f \ = \ \star\Si_f +\frac{1}{\ga} \Si_f \ \mapsto \ \hat{B}_f \qquad \stackrel{\ga^2\neq \pm 1}{\Longleftrightarrow} \qquad \Si_f \ \mapsto \ \left(\frac{1}{\ga}-(-1)^\eta\ga\right)^{-1}\left(\hat{B}_f-\ga \star \hat{B}_f\right).
\end{equation}
In other words, there are two independent invariant bilinear forms on $\mathfrak{so}(3,1)$, resulting in the Holst action, which is classically equivalent to the vacuum EC theory (on-shell), but may differ quantum mechanically~\footnote{Since~\eqref{eq:Plebanski-cnstr} is invariant w.r.t. $\star$ but not $P_{\ga}=1+\frac{1}{\ga}\star$, it is somewhat perplexing that for finite $\ga$ both solution sectors (for $\Si$) lead to the Holst action for gravity with different effective parameters. This is not really the issue here, since the linear simplicity isolates the sectors in a more efficient way, irregardless of the value of $\ga$}.


\paragraph{The EPRL map.} Up to this point, the Barbero-Immirzi parameter $\ga$ did not partake in the formulation of constraints and should be irrelevant for their geometric content. It, however, plays somewhat mysterious role in quantization and essential for comparison with canonical LQG theory. Let us briefly recap on the basic features of the quantum vertex amplitude which arise in this way, without delving too much into details though

The unitary irreducible representations $\rho_f$ of the $\mr{rank}$-2 Lorentz/rotation algebra in $m=4$ are labelled by the eigenvalues of two Casimir operators $C^{(1)}_H(\rho_f)=\mc{J}\cdot\mc{J}$ and $C^{(2)}_H(\rho_f)=\star\mc{J}\cdot\mc{J}$, where $\mc{J}\in\mf{h}$ are generators. One should therefore relate them with the the Casimir $C_{\mc V}^{\ph{i}}(j_f)=\vec{J}\cdot\vec{J}$ of the corresponding little group, leaving the normal invariant.
\begin{itemize}
\item The linear cross-simplicity~\eqref{eq:linear-simplicity} is imposed weakly in the gauge-fixed setting, i.e. for the standard normals -- either $\mc V^A_0=\de^A_0$ for spacelike (or $\mc V^A_3=\de^A_3$ for tetrahedra of mixed signature), characterizing the canonical embedding of $H_0\cong\mr{SU}(2)$ (or $H_3\cong\mr{SU}(1,1)$) into $H\cong\mr{SL}(2,\mathbb{C})$. All the various techniques~\footnote{Such as vanishing matrix elements of~\eqref{eq:linear-simplicity}~\cite{EPRL-FK2008flipped2} $\sim$ master constraint~\cite{EPRL-FK2008finiteImmirzi,Conrady2010timelike2} $\sim$ restriction of coherent state basis~\eqref{eq:Livine-Speziale} to those with the simple expectation values in the semi-classical limit~\cite{EPRL-FK2008FK,ConradyHnybida2010timelike1}} lead to the relation between the Casimir elements of the Lorentz group and corresponding little subgroup of hypersurface rotations:
\begin{equation}\label{eq:EPRL-cross-simpl}
\frac{C_H^{(2)}(\rho_f)}{2C_{\mc V}^{\ph{i}}(j_{f\supset e})} \ \simeq \ \ga.
\end{equation}
This defines the embedding map for `spins' $j$ into decomposition of irreps $\mc{H}^{\rho_f}$ w.r.t. spin states $\mc{H}^j$. 

One nice feature of \eqref{eq:EPRL-cross-simpl} is that its exact implementation \cite{Alexandrov2010newSF-from-CovariantSU2} projects the spin-connection $\om$ in the holonomies~\eqref{eq:discrete-holonomies} to the (covariant lift of) Ashtekar-Barbero connection of LQG:
\begin{equation}\label{eq:projected-AB}
\pi^{(j)}\left(\om^{AB}_a \mc J^{(\chi)}_{AB} \right)\pi^{(j)} \ = \ {}^{(\ga)}A_a^I L_I^{(j)}, \qquad {}^{(\ga)}A^I \ = \ \frac{1}{2}\varepsilon^{0I}_{\phantom{0I}JK}\om^{JK} + \ga \om^{0I},
\end{equation}
here $\pi^{(j)}$ projects on the $j$-irrep of the $\mr{SU}(2)$ subgroup, and $L^I=\frac12\varepsilon^{0IJK}\mc J_{JK}$ is the canonical generator of rotations in the corresponding representation.

\item The part of the linear simplicity~\eqref{eq:linear-simplicity} is first-class and imposed strongly. Taking into account~\eqref{eq:EPRL-cross-simpl}, it is equivalent to the (quadratic) diagonal simplicity~\eqref{diag-simplicity}, if the Barbero-Immirzi parameter is included. It relates the $\mr{SL}(2,\mathbb{C})$ Casimirs:
\begin{equation}\label{eq:EPRL-diag-simpl}
\qquad \left(1 \pm\ga^2\right) C_H^{(2)}(\rho_f)-2\ga C_H^{(1)}(\rho_f) \ \simeq \ 0, 
\end{equation}
and puts restrictions on allowed `simple' irreps $\rho_f$.

\item The closure condition~\eqref{eq:closure-3d} translates into the requirement of the $H$-invariance of the amplitude and is ordinarily implemented through the group integration. It obviously encompasses the invariance w.r.t. the little group $H_0$ of the embedded $j$-states within the tensor product of simple representations, stacked at the tetrahedron $\tau_e$ bounded by the faces $S_f$. Thereby the EPRL embedding map is established:
\begin{equation}\label{eq:EPRL-map}
\Phi_\ga \ : \quad \mathrm{Inv}_{\mr{SU}(2)}^{\ph{1}} \,  \bigotimes_{f\supset e} \, \mc{H}^{j_f} \ \longrightarrow \ \mathrm{Inv}_{\mr{SL}(2,\mathbb{C})}^{\ph{1}}  \, \bigotimes_{f\supset e}\, \mc{H}^{\rho_f},
\end{equation}
Hence the states in the boundary space are labelled by $\mr{SU}(2)$ intertwiners glued into spin-networks. The last portion of the integration over the homogeneous space $H/H_0$ can be vied as summing over all possible gauge choices for the normals $\mc V \in H \rt \mc V_0$, and thus restoring the full Lorentz invariance at the vertex in the gauge-fixed model.
\end{itemize}

One clearly sees the subsidiary role of $\mc V$'s: in the construction of the model they are treated as ``unphysical'' gauge choice, which one can specify freely, and later ``erase'' this information. In effect, $\mc V$ allows one to reduce the problem of constraint imposition to the level of little group $H_0$, instead of operating directly on the covariant level of the full Lorentz group $H$. The few next comments discuss the possible role of these normals, and the issues of covariance in general: 

\begin{itemize}

\item For instance, one knows that the relative of the time-normal field explicitly appears as non-trivial lapse/shift components in the Lorentz-covariant canonical quantization of the 1st order action~\eqref{eq:Einstein-Cartan} with the Holst term. The boundary states of SF models are spanned by the projected spin networks of this `covariant LQG'~\cite{LivineDupuis2010lifting-SU2}:
\begin{equation}\label{eq:projected-spin-network}
\Psi_{\Ga=\pa\mc{K}}^{\phantom{\Ga}}\big(\big\{h_{\ell}\big\},\big\{\mc V_n\big\}\big)  \ = \  \bigotimes_{n\subset\pa e} \iota^{(n)}_{\mc V} \bigotimes_{\ell\subset\pa f} \left( \pi^{(j_{t(\ell)})}_{\,\mc V} \, D^{\rho_{\ell}}_H (h_{\ell}) \, \pi^{(j_{s(\ell)})}_{\,\mc V}\right),
\end{equation}
where the normals are discretized naturally over nodes $n$. The state functionals are invariant w.r.t. the covariant Lorentz group action on both sets of variables:
\begin{equation}\label{eq:projected-transformation}
\Psi \big(\big\{h_{\ell}\big\},\big\{\mc V_n\big\}\big) \ = \ \Psi \big(\big\{U_{t(\ell)}^{-1}h_{\ell}^{\ph{1}}U_{s(\ell)}^{\ph{1}}\big\},\big\{U_n\rt\mc V_n\big\}\big), \qquad \forall \, U_n \in H.
\end{equation}

\item Historically, one of the incentives, which led to FK model~\cite{EPRL-FK2008FK}, was to solve the so-called ``ultra-locality'' problem with the BC amplitude -- namely, the apparent shortage in intertwiner d.o.f., which signified about the limited nature of correlations between neighbouring 4-simplices' geometries. 
On a more technical level, the resolution of identity, associated with the invariant vector space $\mc{H}_e:=\mathrm{Inv}_H^{\phantom{1}}\big[\bigotimes_{f\supset e} \mc{H}^{\rho_f}\big]$ at each edge of initial BF spin foam, rewritten in terms of coherent intertwiners (for a moment, $H=\mr{Spin}(4)\cong \mr{SU}(2)\otimes \mr{SU}(2)$):
\begin{equation}\label{eq:identity-resolution}
\mathbbm{1}_{\mc{H}_e}^{\phantom{1}} \ = \ \bigotimes_\pm\int \prod_{f\supset e} d^2\mathbf{n}^\pm_{ef} \, d_{j^\pm_f} \int dh_{ve}^\pm\int dh_{v'e}^\pm \ h_{ve}^\pm\big|j^\pm_f,\mathbf{n}^\pm_{ef}\big\ket\big\bra j^\pm_f,\mathbf{n}^\pm_{ef}\big|\left(h_{v'e}^\pm\right)^\dagger ,
\end{equation}
is replaced by a projector, where summation is only over those states in the `simple' representations $j^+=j^-$, which solve the quantum cross-simplicity~\eqref{eq:linear-simplicity} (as expectation values). Specifically, the existence of a common $u_e\in \mr{SU}(2)$ group element is inferred, representing 4d normal $\mc V_e$, which establishes the relation $\mathbf{n}^-=-u_e\rt\mathbf{n}^+$. 

The gluing of two 4-simplices -- via identifying first the geometries, corresponding to their common tetrahedron~$\tau_e$, and only then performing an integration -- takes into account the missing correlations between neighbouring vertices, sharing an edge. Whereas the unique Barrett-Crane intertwiner is obtained if one integrates separately at each vertex over (then decoupled) geometries. Arguably, the latter identification concerned only an `internal' 3d geometry of $\tau_e$, encoded in the spins and 3d normals $\{j,\mathbf{n}\}$, corresponding to the (canonically embedded) little group $H_0=\mr{SU}(2)$.

The key point of the present work to treat normal $\mc{V}$ as truly independent geometric variable, characterizing the placement of 3d faces in 4d, creates some tension with the implementation of gauge invariance in the EPRL-FK vertex amplitude, if the above logic is extended by analogy to $\mc V$. Indeed, the dependence on the subsidiary variable $u_e$ is ``eaten'' by the follow-up $H$-group integration, performed independently at each vertex. The situation is quite similar to the BC intertwiner, so there still may be some d.o.f. left uncorrelated (even though if gauge).
 
\item A similar type of arguments have been put forward on the basis of the Lorentz-covariant canonical quantization endeavour~\cite{Alexandrov2008SF-from-CovariantLQG}. It has been argued that allowing an additional variable $\mc V$ remain unintegrated, the covariant transformation properties~\eqref{eq:projected-transformation} necessitate a relaxation of the closure condition~\cite{Alexandrov2008SF-cnstr-revisited}. The closure of the discrete bivectors is a too restrictive Gauss law, because the gauge transformations should act on the vector variables as well. In the preliminary Hamiltonian analysis of~\cite{GielenOriti2010Plebanski-linear}, the very same reason led authors to artificially enlarge the phase space by the fictitious momenta, corresponding to $\mc V$. In Ch.~\ref{ch:prop-2} we will see how our modification responds to both these objectives in quite a natural manner.


\end{itemize}

\chapter{Problem with the EPRL construction}
\label{ch:problem}

The kinematical Hilbert space of LQG contains all possible graphs of higher-valent nodes, in general. In order to match this in the boundary of the EPRL-FK model, the latter was extended in the KKL proposal~\cite{KKL2010AllLQG,KKL2010correctedEPRL}. The EPRL embedding map $\Phi_\ga$ of the original construction was generalized straightforwardly to edges, where $L\geq 4$ faces meet, corresponding to arbitrary polyhedra not restricted to triangulations. In the asymptotic analysis of KKL amplitude~\cite{BahrSteinhaus2016Cuboidal-EPRL,Dona-etal2017KKL-asympt}, there were observed that certain ``non-geometric'' configurations contribute non-trivially into path-integral state-sum. They are called like this because some shape-mismatch is allowed, as well as torsion, so that the clear interpretation in terms of Riemannian metric geometry is absent, for semi-classic states~\footnote{They are SF analogues of LQG's `twisted' geometries~\cite{FreidelSpeziale2010Twisted-geometries} (discontinuous over flat faces), or torsionful  `spinning' (continuous over arbitrarily curved faces) piecewise-flat geometries~\cite{FreidelGeillerZiprick2013cont-LQG-phase-space,FreidelZiprick2014Spinning-geometry}; whilst Regge configurations appear only as a constrained subset~\cite{DittrichRyan2011simplicial-phase-space}. Note that the non-zero torsion generically presents in LQG phase space by the Lemma 2 in~\cite{FreidelGeillerZiprick2013cont-LQG-phase-space}, since the Ashtekar-Barbero connection mixes up extrinsic curvature -- residing over the edges of the cellular decomposition -- whose impact is governed by the Immirzi parameter $\ga$. It has also been argued on the basis of more involved examples, such as $n$-point correlation functions and extended triangulations with the bulk curvature, that the `double scaling limit' with the flipped $\ga^{-1}\ra\infty$ in front of the Holst term is required \blockquote{in order to reduce the SF dynamical variables to the configurations compatible with the metric geometry of the triangulation} (see~\cite{MagliaroPerini2013asympt-Regge} and references therein).}.


In this section, by scrutinizing an instructive case of hypercuboidal vertex and boundary graph, we demonstrate that the appearance of `non-geometricity' can be traced back to the way how the simplicity constraints are imposed in the classical theory. In particular, the possibility to neglect the 4-volume constraint in the simplex does not hold for more complicated polytopes, so that \ul{Barrett's reconstruction is not applicable}. We then proceed with the application to the same system of the fully linear treatment due to Gielen and Oriti~\cite{GielenOriti2010Plebanski-linear}, with independent normals. The workings of their `linear volume' constraint prompt to switch from the normals (3-forms) directly to edge lengths (tetrads) as new independent variables, using the Hodge duality. In the rest of the thesis we study the implications of that change on the classical continuum theory.

\section{The `volume' constraint is not implemented}
\label{sec:problem-quadratic}

As discussed in Sec.~\ref{subsec:Plebanski}, the discretization of (quadratic) volume constraint employs several tetrahedra of the 4-simplex, hence it is usually thought of as consistency condition on time evolution (``secondary'' constraint). Indeed, \eqref{eq:volume-simplex} shows that if the cross-simplicity together with the 3d closure holds true for all tetrahedra, it does not matter which of the face bivectors are used to calculate the volume of the 4-simplex. Thus, it is not imposed explicitly in the quantum theory, once the former two are implemented. The same proof using the cable-wire diagrammatic representation of the 4-simplex amplitude shows that this indeed holds in  the quantum theory as well~\cite{Perez2013SF-review}, at least semi-classically.

We notice that the argument heavily relies on the combinatorics of the 4-simplex and does not necessarily extend to the generic case of arbitrary 2-complex. Explicitly, this appears already in the simple case of (hyper)cuboidal graph~\cite{BahrSteinhaus2016Cuboidal-EPRL}. There the plain ansatz, using the semi-classical substitute for the exact vertex amplitude, has been studied for the flat (no curvature) rectangular lattice. The expression for the amplitude is a straightforward KKL generalization of the Euclidean $\text{EPRL}_{\ga<1}$ model, in the FK representation using coherent states:
\begin{subequations}
\begin{gather}\label{eqs:cub-amplitude}
\mc A_v^\pm \ = \ \int \prod_e dh_{ve} \, e^{S^\pm[h_{ve}]} \ \stackrel{j\ra \infty}{\sim} \ \sqrt{\frac{(2\pi)^{21}}{\mathrm{det}\left(-\pa^2 S\big(\vec{h}_{\mathrm{c}}\big)\right)}} e^{S(\vec{h}_{\mathrm{c}})} ,
\\
S^\pm[h_{ve}] \ = \ \frac{1\pm\ga}{2}\sum_{(ee')} 2j_{(ee')} \ln\bra-\mathbf{n}_{ee'}|h_{ve}^{-1}h_{ve'}^{\phantom1}|\mathbf{n}_{e'e}\ket.
\end{gather}
\end{subequations}
Here the summation goes over the (ordered) pairs $(ee')=(f\cap\pa\mc T_v)$ -- the (directed) links of a 6-valent combinatorial hypercuboidal boundary graph, and the data $\{j,\mathbf{n}\}$ in this symmetry reduced setting was chosen to represent (semi-classically) $\mathbb{R}^3$-cuboids: 
\begin{equation}\label{eq:cuboid-intertwiner}
|\iota\ket \ = \ \int du \ u\rt\bigotimes_{i=1}^3|j_i,\mathbf{n}_i\ket|j_i',\mathbf{n}'_i\ket, \qquad j_i' \, \mathbf{n}_i' = - j_i \, \mathbf{n}_i,
\end{equation}
glued along their faces. $\pa^2 S$ denotes the Hessian matrix, evaluated at the critical point $\pa S\big(\vec{h}_{\mathrm{c}}\big)=\Re\, S\big(\vec{h}_{\mathrm{c}}\big)=0$.

\begin{figure}[!ht]
\center{\includegraphics[width=0.5\linewidth]{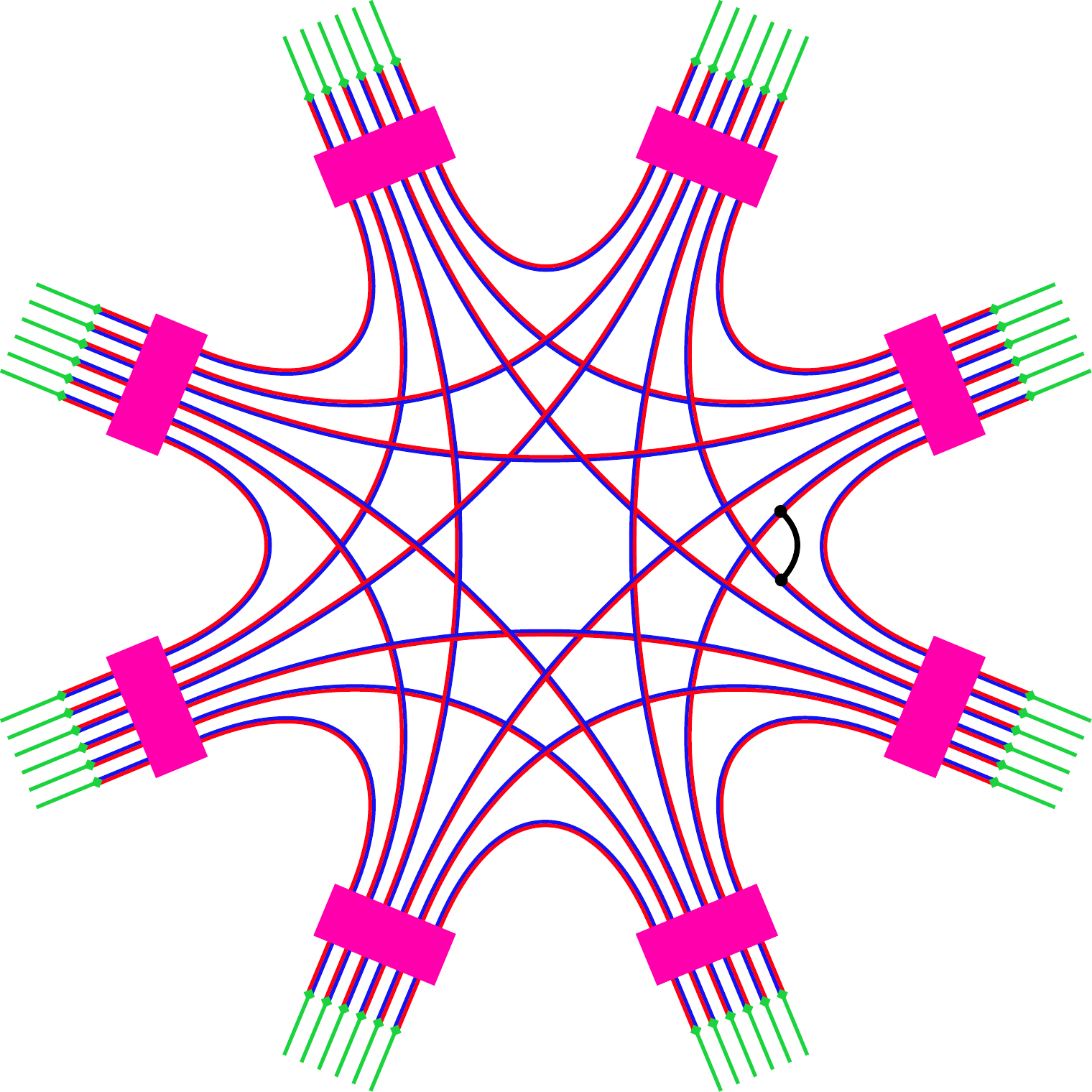}}
\caption{Diagrammatic representation of the hypercuboidal amplitude. The notation is as in~\cite{Perez2013SF-review}: lines (`wires') depict the $SL(2,\mathbb{C})$ representation matrices, and boxes (`cables') -- invariant projectors (group integrations).}
\label{fig:cable-wire-cub}
\end{figure}
It turns out that the 4-volume of a flat hypercuboid cannot be unambiguously ascribed to a vertex, using the prescription akin to~\ref{volume-quadr} for 4-simplex, where its consistency is guaranteed by~\ref{eq:volume-simplex}. If the rectangular lattice is \textit{geometric} (i.e. we are on the solution to simplicity constraints), it is characterized entirely in terms of its edge lengths $E_i,\, i=t,x,y,z$, and the unique geometric 4-volume can be computed irregardless of the faces chosen $\tilde{V}_v:=E_tE_xE_yE_z=\Si_{tx}\Si_{yz}=\Si_{xy}\Si_{zt}=\Si_{xz}\Si_{yt}$, where each area is simply given by the product of the cooresponding edge lengths, e.g. $\Si_{xy}=E_xE_y$, etc. Instead, we get 6 arbitrary areas/spins~$j_{(ij)}$ which do not necessarily satisfy the latter 2 conditions. Indeed, if we try to proceed like in~\eqref{eq:volume-simplex}, starting with the expression $j_{xy}j_{zt}$ (depicted by a `grasping' on Fig.~\ref{fig:cable-wire-cub}) and applying the 3d closure for the spatial cuboid $\tau_t$, we end up with a tautological result: the contributions from parallel faces (bounding $\tau_t$ and the 2 adjacent anti-podal cuboids $\tau_i$, $\tau_{-i}$) enter with equal areas/spins but opposite signs $\star\Si_{ij}\cdot\Si_{kt}=-\star\Si_{ji}\cdot\Si_{kt}, \ i,j,k=x,y,z$, thus contracting each other in the sum~\footnote{Stronger, $\star\Si_{xz}\cdot\Si_{zt}=\star\Si_{yz}\cdot\Si_{zt}=0$ by the cross-simplicity.}, so we arrive at the dull equality $j_{xy}j_{zt}=j_{yx}j_{zt}$.

\ul{The essential ingredient of the EPRL construction, namely, that one could effectively replace the `volume' part of the simplicity by the 3d closure, is not valid for a higher valence. We encounter the problem that the model is not constrained enough to complete the reduction from BF to gravitational theory.} The measure of deviation is captured by the `non-geometricity' parameter, in this case:
\begin{equation}\label{eq:non-geometricity}
\varsigma \ = \ \begin{pmatrix}
           j_{xy}j_{zt}-j_{xz}j_{yt} \\
           j_{xz}j_{yt}-j_{xt}j_{yz} \\
           j_{xy}j_{zt}-j_{xt}j_{yz}
         \end{pmatrix}.
\end{equation}
The numerical studies of~\cite{BahrSteinhaus2016Cuboidal-EPRL} show that the non-geometric configurations with $\varsigma\neq0$ \textit{do} generically contribute to the path-integral, although their impact might be exponentially suppressed. The dumping is controlled by the width of the Gaussian -- the effective ``mass'' term $m^2_\varsigma(\al) \approx 2\al -1>0$ for $\al\gtrsim 0.5$, which depends crucially on the parameter $\al$ in the choice of the face amplitude $\mc A_f^{(\al)} = \big((2j_f^++1)(2j_f^-+1)\big)^\al$. Reassuringly, in the same range of~$\al$ indications were given for the tentative continuum limit in the form of a phase transition, with the restoration of the (remnant) diff-invariance. This lead authors to speculate that the freedom in the face amplitude might be restricted on physical grounds, for one should definitely obtain geometric states in the classical limit.

\section{Fully linear treatment of the hypercuboid}
\label{sec:problem-linear}

Naturally, the 2 missing constraints to impose in this elucidating example are $\varsigma=0$, however, it is unevident how to proceed in the most general case. The simplicial `4-volume constraint' makes little sense here, unless appended with some additional requirements (as we tentatively propose in~\cite{BahrBelov2017VolumeSimplicity}, based on the certain type of graph invariants). Being the part of Pleba\'{n}ski's \emph{quadratic} formulation, it is also inorganic to the model built on linear constraints. An alternative fully \textit{linear} formulation was put forward in~\cite{GielenOriti2010Plebanski-linear}, providing both the continuum version of the cross-simplicity~\eqref{eq:linear-simplicity}, as well as the linearized counterpart for the `volume' constraints of the form~\eqref{eq:simplicity-quant+normal}. It introduces the basis of 3-forms $\vartheta^A$, whose discretization naturally associates 4d normal vectors
\begin{equation}\label{eq:discrete-normals}
\mc V^A_e \ = \ \int_{\tau_e} \vartheta^A
\end{equation}
to tetrahedra~$\tau_e$. Although, strictly speaking, neither Pleba\'{n}ski, nor Gielen-Oriti's linear version were ever formulated beyond triangulations, let us extrapolate the latter to our cuboidal setting, like we did with 4-volume. 


\begin{figure}[!hb]
\center{\includegraphics[width=0.7\linewidth]{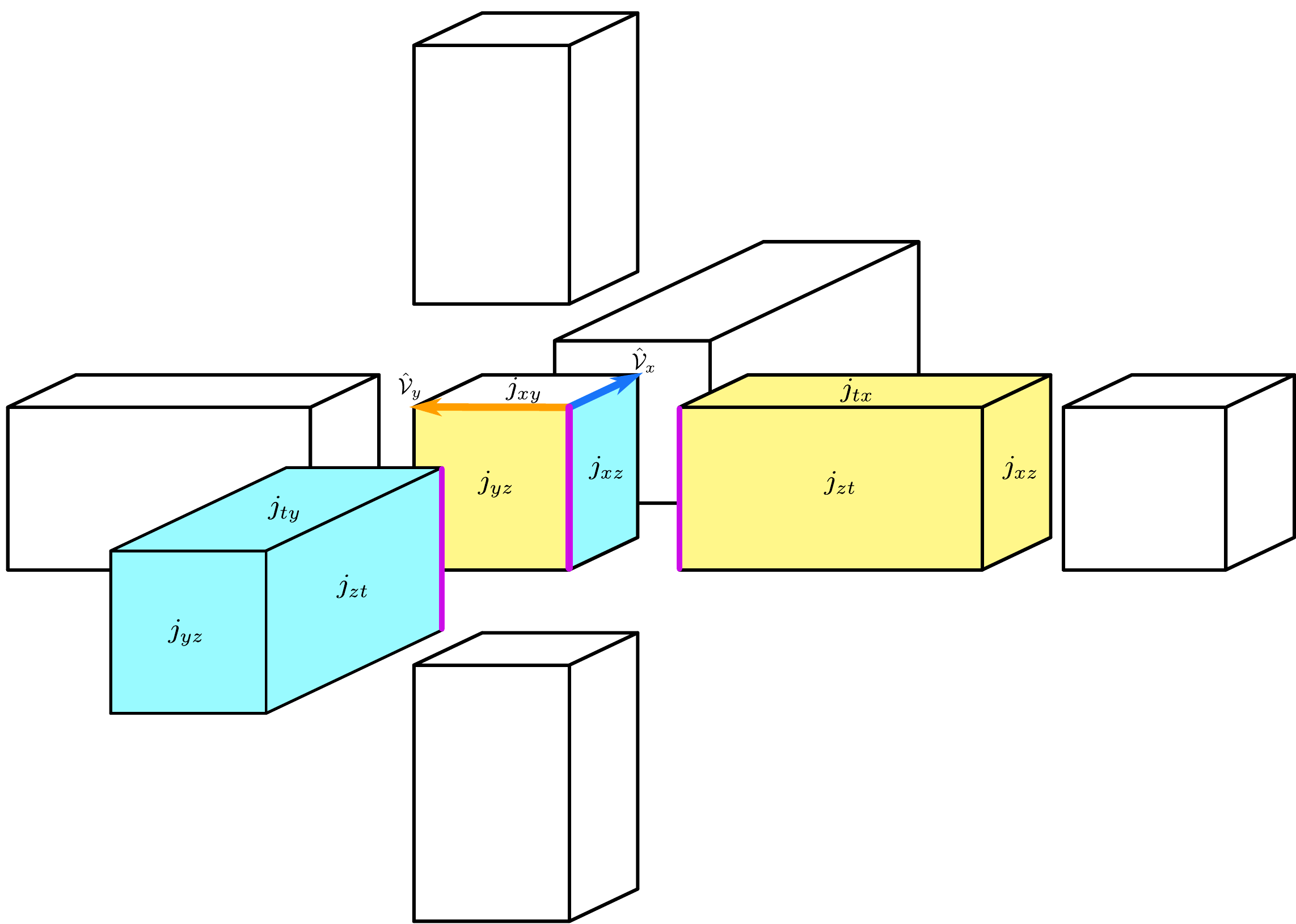}}
\caption{Elements of linear volume constraint: edge $E_z$ (purple) is shared by exactly 3 faces and 3 cuboids, orthogonal to $x$ (blue) and $y$ (yellow) directions. Colour scheme corresponds to $\Si\cdot\mc V$ pairing.}
\label{fig:hypercube}
\end{figure}

\newpage

Analogously to 4-simplex, every vertex is the source of four edge vectors, each of which is shared by exactly three cuboids, intersecting along three faces, respectively. Their directions can be identified with that of the eight cuboidal normals $\hat{\mc V}_i$ (e.g. aligned with $\hat{x}_i$ and $\hat{x}_i' = -\hat{x}_i$ of standard cartesian grid). However, their norms are considered as free parameters (like surface areas $j$). The bivector data~\footnote{We exclude the Immirzi parameter from consideration and put $B=\star\Si$. From here on all the following discussion is purely classical.} is restricted to satisfy
\begin{equation}\label{eq:hypercube-bivectors}
(B_{ij}+B_{i'j})\cdot\hat{\mc V}_j \ = \ 0,
\end{equation}
that is opposite faces of 3d cuboid are required to be parallel and of equal area. Indeed, this is precisely the boundary data for coherent intertwiners~\eqref{eq:cuboid-intertwiner}, if 3d normals are defined as $(j\hat{\mc N})_i= B_{ij}\cdot\hat{\mc V}_j$. Pick up some vertex, and fix the edge $\hat{\mc V}_l$, ejecting from it. For every such pair, via straightforward transfer from~\cite{GielenOriti2010Plebanski-linear}, one writes down the two linearised volume (vector) constraints:
\begin{equation}\label{eq:linear-volume-simplicity}
\mc V_i\cdot \Si_{il} \ = \ \mc V_j\cdot \Si_{jl} \ = \ \mc V_k\cdot \Si_{kl} \qquad \forall \ i,j,k\neq l.
\end{equation}
We have one equality per each 3d cuboid in the cycle, sharing the given edge as a hinge, at which the two consecutive 2d faces intersect. In other words, the volume constraint is discretized at the edge-cuboid pairing (see figure~\ref{fig:hypercube}). Because of orthogonality, all the terms are proportional to $\hat{\mc V}_l$. Let us regroup it at each vertex via formation of sums over edges, orthogonal to one particular direction $\hat{\mc V}_k$:
\begin{equation}\label{eq:reduced-simplicity}
\sum_{\{i,j\}\neq k} \mc V_{i}\cdot (\star B_{jk}) \ = \ 0 \qquad \forall \, k.
\end{equation}
(Writing $\star B_{zt}=\Si_{xy}$, etc., is just the convenient relabelling of face by its two orthogonal directions, corresponding to cuboids adjacent along $S_f$.) Each individual term in the sum is proportional to $\hat{\mc V}_{l\neq k}$, and by linear independence in 3d subspace $\perp\hat{\mc V}_k$ we can rewrite pre-factors in the form of proportion:
\begin{equation}\label{eq:proportion}
\frac{|\mc V_i|}{|\mc V_j|} \ = \ \frac{|B_{ik}|}{|B_{jk}|} \qquad \forall \, i,j\neq k,
\end{equation}
the third equality being the consequence of the other two. Taking another such identity for $l\neq k$ and the same $i,j$ we arrive at $|B_{ik}||B_{jl}|=|B_{il}||B_{jk}|$ -- the familiar geometricity conditions.

\section{Normals vs. frames}
\label{sec:normals-frames}

Our proof, however, showed something more than that. The volumes $|\mc V_i|$ of 3 cuboids in the cycle around edge $\hat{\mc V}_k$ relate as areas of their bases $|B_{ik}|$. The proportionality coefficient $|E_k| = |\mc V_i|/|B_{ik}|$ is independent of $i$ and we infer the existence of heights/edge lengths, invariantly defined in terms of (independent) $|\mc V_i|, |B_{ik}|$ from any of 3 cuboids, sharing this edge. The shapes of faces are obviously matching. We essentially exploited self-reciprocal nature of orthogonal lattice, allowing an identification of edges with $\mc V_i$ up to constants. Had we allowed fluctuations in directions $\hat{\mc V}_i$, the two notions would separate, and multiple angles between the lattice and its dual would intervene the formulas.

\paragraph{Closure of 4d normals.} If we agree to call ``non-local'' or ``dynamical'' either the quantities or relations, where objects from several elements of 3d boundary are involved, then the constraint~\eqref{eq:reduced-simplicity} is necessarily non-local. Indeed, it collects volumes of 3 cuboids at their intersection edge, multiplied by the areas that they cut at the beginning of the edge. Gielen and Oriti demonstrated in the case of the 4-simplex that these non-trivial equations may be replaced by another non-local condition of \emph{4d closure} of normals: 
\begin{equation}\label{eq:closure-4d}
\sum_{e\supset v} \mc V^A_e \ = \ 0 \qquad \forall \, v,
\end{equation}
together with the usual cross-simplicity~\eqref{eq:linear-simplicity} and 3d closure of bivectors~\eqref{eq:closure-3d} at each tetrahedron $\tau_e$ (i.e. ``local'', or ``kinematical'' constraints). 

This result reminds of the \textit{Minkowski's theorem}, extended to 4d. Recall that the latter asserts the existence of the unique (up to congruence) flat convex polytope, if the set of its face normals $\{\mc V_e\}$ is given, satisfying~\eqref{eq:closure-4d} (cf.~\cite{BianchiDonaSpeziale2011Polyhedra} for the 3d case). The non-trivial part is to demonstrate that it is compatible with characterization via simplicity constraints. We did not find the way to prove the analogous statement for hypercuboid. The complication, arising in this case, we relate to the presence of well-separated boundary elements (e.g. ``past'' and ``future'' cuboids in the foliation picture). Whereas in 4-simplex, every tetrahedron shares at least one common point with any other.


Returning to constrained-BF framework of Spin Foam models, our inspection highlights the following fact. In order to complete the reduction from the topological theory to gravity, using the linear formulation~\eqref{eq:simplicity-quant+normal} with independent normals, some additional requirements have to be met. This may be either `linear volume' constraint, or its equivalent. The obvious rewriting of the norms $|\mc V_l|$, using~\eqref{eq:proportion}, as
\begin{equation}\label{eq:volume-unique}
|E_i| |B_{il}| \ = \ |E_j| |B_{jl}| \ = \ |E_k| |B_{kl}| \qquad \forall \, i,j,k\neq l,
\end{equation}
leads to useful (re)interpretation of constraints~\eqref{eq:linear-volume-simplicity}. Namely, they \emph{indirectly} characterize edge lengths through compatibility both with $B$'s \emph{and} $\mc V$'s, such that the unique (metric) volume, equal to $|\mc V_i|$, can be ascribed to each cuboid. After this identification the normals have to satisfy closure~\eqref{eq:closure-4d}.

We find it convenient to switch to a simpler variables, leading to a more manageable set of relations. Instead of volumes $\mc V$, we therefore propose to use directly the edge vectors $E$, as suggested by~\eqref{eq:volume-unique}. Unlike the original Gielen-Oriti's~\eqref{eq:linear-volume-simplicity}, the new version can be stated locally at the level of 3d polyhedra. We study the continuum formulation behind it in Sec.~\ref{sec:continuum-linear}, proving the equivalence with the usual simplicity constraints $B=\star \theta\wedge \theta$.

As another option, one could prefer instead to impose the closure on 4d normals $\mc V$, at least in the case of 4-simplex. Ordinarily, the 3d closure~\eqref{eq:closure-3d} of bivectors can be thought of as the result of integration of the continuous Gauss law $\nabla^{(\om)}_{[c}\Si^{AB}_{ab]}=0$ over a 3-ball triangulation (using the Stokes' theorem and the flatness of connection).
There is, however, the problem with such field-theoretic interpretation of~\eqref{eq:closure-4d}, because the corresponding dynamical law $\nabla^{(\om)}_{[a}\vartheta^A_{bcd]}=0$ for independent 3-forms is absent in the continuum theory. This flaw is fixed in Ch.~\ref{ch:prop-2}. In accord with our choice of variables, we first establish the link with the condition of vanishing torsion for tetrads $\theta$, and then implement it via Lagrange multipliers, prior to any imposition of simplicity constraints.

\newpage

\thispagestyle{empty}
\chapter{Proposition \textnumero1: Quantum 4-volume}
\label{ch:prop-1}

In this Chapter, we address an imbroglio with the volume which our analysis in Sec.~\ref{sec:problem-quadratic} revealed. One keeps the context of bivector geometry and Plebanski quadratic constraint formulation. We are trying to find the way how the polytopes more general than the rigid simplex could be characterized purely in terms of their $2$-dimensional substructures (a variant of `holography'). That is the inverse task consists in extending the Barrett's reconstruction theorem, starting with bivectors and identifying the missing consistency requirement. We establish the relation with the certain knot-invariants associated with bivector geometry on a graph, coming from the geometry of a polytope. The tentative solution is proposed for the Spin Foam amplitude, and potential implications discussed. (In particular, we notice the similarity with another proposal to include cosmological constant.)

\section{4-volume of the polytope in terms of bivectors}
\label{sec:Hopf-link}

While every polytope $P$ embedded in $\mathbb{R}^4$ uniquely determines the set of bivectors $B_\ell$ (or, equivalently, their dual version) associated to its $2$-dimensional faces (corresponding to oriented links $\ell$ in the boundary graph $\Ga\subset S^3$), it is an unsolved problem to reconstruct a polytope from its face bivectors. This construction lies at the heart of the simplicity constraints in the spin foam models for quantum gravity. The reason for this is that the Hilbert space vectors in the theory are the $\mr{SU}(2)\times \mr{SU}(2)$ spin-network functions arising from the quantization of $BF$ theory, and the $B$-field of this theory is precisely what assigns bivectors $B_\ell\in\bigwedge^2\mathbb{R}^4$ to oriented links $\ell$. The $B_\ell$ act naturally as operators on states $\Psi$, and thus the simplicity constraints in terms of the bivector operators are used to select a certain subset of $\mr{SU}(2)\times \mr{SU}(2)$ spin-networks, corresponding to solutions to the simplicity constraints on the quantum level.

In the case of the $4$-simplex, the (classical, discrete) simplicity constraints \ref{diag-simplicity}~$\displaystyle{\star B_\ell \cdot B_\ell=0}$, \ref{cross-simplicity}~$\displaystyle{\star B_\ell\cdot B_{\ell'}=0}$, and closure \eqref{eq:closure-3d}~$\sum_{\ell\supset n} [n,\ell] B_\ell= 0$ are (given the non-degeneracy) precisely those constraints on a set of bivectors $B_\ell$ associated to a boundary graph, which ensure that there exists a geometric $4$-simplex where the induced bivectors are the given ones.

For a higher valence graph state the situation is more complicated, since in general spin foam models there might not even be a polytope from which to take the bivectors. The only input we have is an oriented graph $\Ga\subset S^3$, together with bivectors $B_\ell\in\bigwedge^2\mathbb{R}^4$ associated to the links $\ell$, which we defined to be a \emph{bivector geometry}, satisfying the three conditions above~\footnote{It is understood that any such geometry where the orientation of a link $\ell$ is reversed and the corresponding $B_\ell$ is replaced by $-B_\ell$, defines that same bivector geometry.}. (We use somewhat interchangeably the ``links'' and elementary ``edge'' in this paragraph, dealing mostly with a graph states.)

\subsection{Hopf link volumes}

We define a set of numbers which rely on the embedding of $\Gamma\subset S^3$. In particular, we equip $S^3\subset\mathbb{R}^4$ with its standard orientation. By stereographic projection (with respect to a point which can be on $\Gamma$ itself), $\Gamma$ can be projected to $\mathbb{R}^3$, and from there to $\mathbb{R}^2$, with crossings.  Given the orientation of edges $\ell$ in $\Gamma$, there are two types of crossings, called \emph{positive} and \emph{negative}, depicted in Figure~\ref{Fig:TwoKnottingCases}. To these crossings $\mathpzc{C}$ we assign crossing numbers $\sigma(\mathpzc{C})=\pm 1$, depending on their type.

\begin{figure}[hb]
\begin{center}
\includegraphics[scale=1.0]{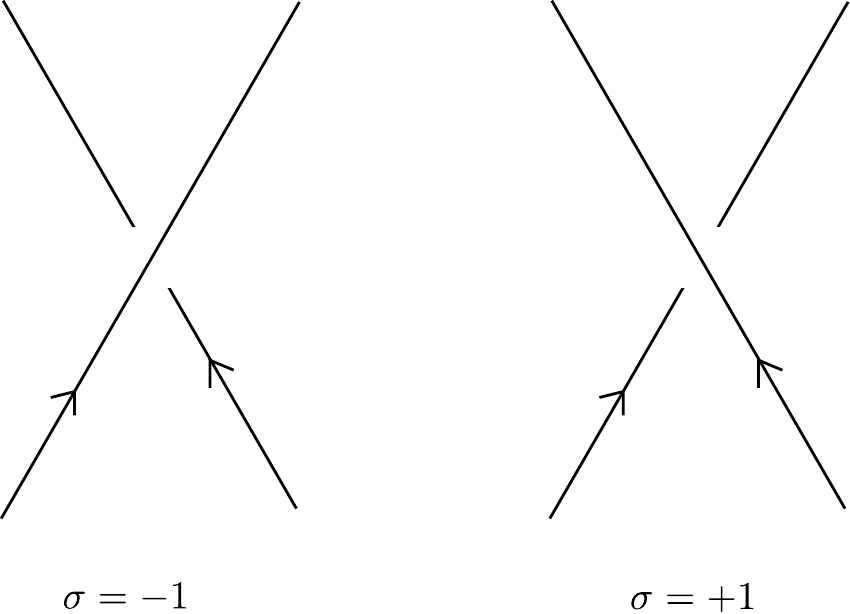}
\end{center}
\caption{With orientations, there are two types of crossings $\mathpzc{C}$ of links in a graph $\Gamma$. They receive a crossing number $\sigma(\mathpzc{C})=\pm 1$. }\label{Fig:TwoKnottingCases}
\end{figure}

For each crossing $\mathpzc{C}$, we define the corresponding crossing volume $V_\mathpzc{C}$ to be
\begin{eqnarray}\label{Eq:DefinitionCrossingVolume}
V_\mathpzc{C}\;:=\;\sigma(\mathpzc{C})\;\star\Big(B_1\wedge B_2\Big),
\end{eqnarray} 

\noindent where $B_{1,2}$ are the bivectors associated to the two links within the crossing, and $\star:\bigwedge^4\mathbb{R}^4\to
\mathbb{R}$ is the Hodge operator. We define the total $4$-volume of the bivector geometry to be
\begin{eqnarray}\label{Eq:DefinitionTotalVolume}
V\;:=\;\frac{1}{6}\sum_{\mathpzc{C}} V_{\mathpzc{C}}.
\end{eqnarray}

\noindent Indeed, one can show that, in case the bivector geometry does come from a polytope $P\subset \mathbb{R}^4$, $V$ is precisely the volume of $P$~\cite{Bahr2018polytope-volume,Bahr2018non-convex-polytope}. Furthermore, we define \emph{Hopf link volumes} the following way: A Hopf link $\mathpzc{H}$ in $\Gamma$ is defined to be a subset of edges in $\Gamma$ which form two linked (i.e.~non-intersecting) cycles  when embedded in $S^3$, see Fig.~\ref{Fig:HopfLink}. Furthermore, they are not allowed to have any other crossing with any edge not in the Hopf link.

For a crossing $\mathpzc{C}$, we write $\mathpzc{C}\between \mathpzc{H}$, if it is between two edges of $\mathpzc{H}$. The Hopf-link volume $V_\mathpzc{H}$ associated to $\mathpzc{H}$ is then defined as
\begin{eqnarray}\label{Eq:HopfLinkVolume}
V_\mathpzc{H}\;:=\;\frac{1}{6}\sum_{\mathpzc{C}\between \mathpzc{H}}V_\mathpzc{C}.
\end{eqnarray}

\begin{figure}
\begin{center}
\includegraphics[scale=0.5]{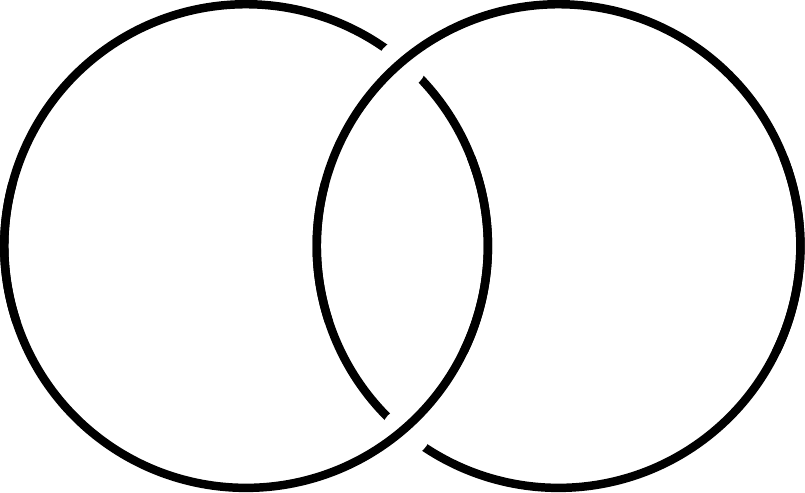}
\end{center}
\caption{A Hopf link, i.e.~two closed, linked curves, embedded in $S^3$.  }\label{Fig:HopfLink}
\end{figure}

\noindent It is not difficult so see that both $V$ and $V_\mathpzc{H}$ are independent of the embedding of $\Gamma$ to the plane with crossings. In other words, both are properties of the bivector geometry (and a choice of $\mathpzc{H}$) only. Also, both are clearly invariant under change of graph edge orientations. (See~\cite{Bahr2018polytope-volume,Bahr2018non-convex-polytope} for details.)

\begin{proposition}
We say that the bivector geometry satisfies the \textbf{Hopf link volume-simplicity constraint}, if $V_\mathpzc{H}$ is independent of the choice of Hopf link $\mathpzc{H}$ in $\Gamma$. In the following example, we will see that this is precisely the condition allowing for a reconstruction of the $4$-dimensional polytope from the bivector geometry.
\end{proposition}

\subsection{Case: the hypercuboid}\label{Sec:Hypercuboid}

In this section we consider the example of a bivector geometry which also occurs in the asymptotic analysis of the EPRL spin foam model~\cite{BahrSteinhaus2016Cuboidal-EPRL}. It also plays a prominent role in renormalization computation of the model \cite{BahrSteinhaus2016phase-transition,BahrSteinhaus2017cub-renormalization}. It is the prime example for a geometry in which diagonal- and cross-simplicity constraints are satisfied, but not necessarily volume simplicity. In particular, there does not, in general, exist a $4d$ polytope associated to it. However, the existence of such a $4d$ geometry can be formulated in terms of the Hopf volumes $V_\mathpzc{H}$, as we will show. 

\begin{figure}
\begin{center}
\includegraphics[scale=0.25]{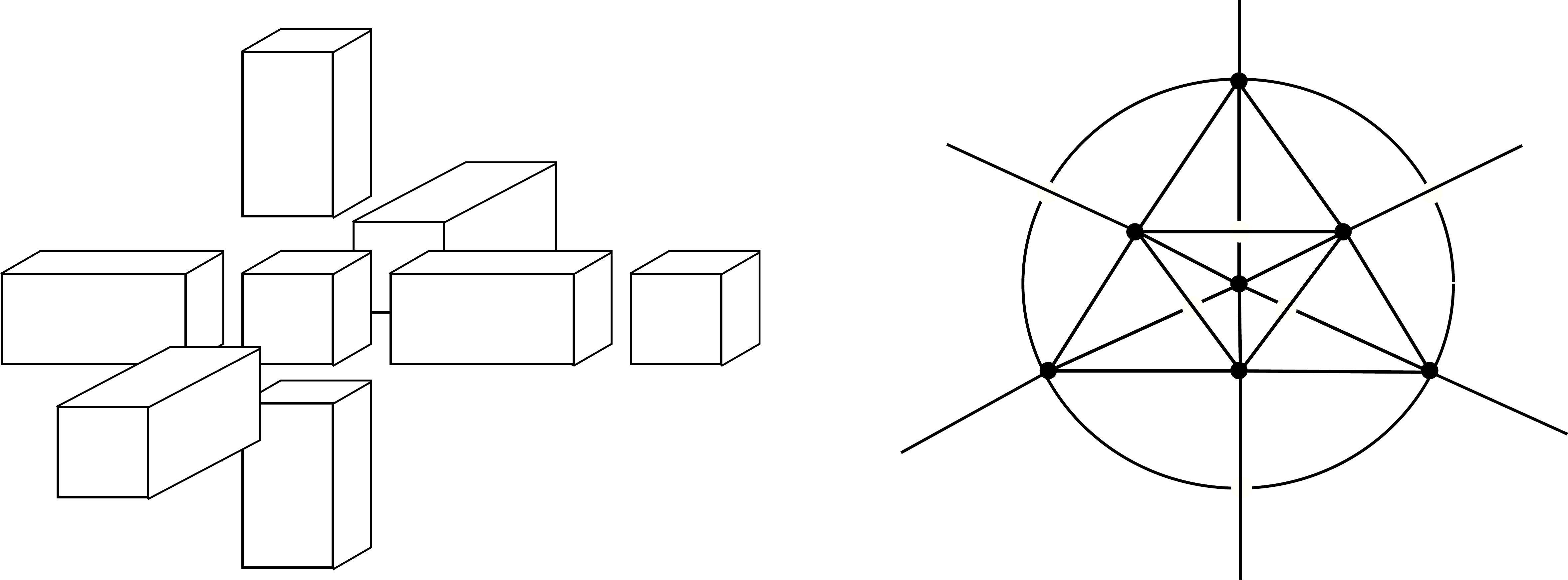}
\end{center}
\caption{The boundary of a hypercuboid, where every $3$-dimensional polyhedron is a cuboid. Its boundary graph, where each of the cuboids corresponds to one vertex. Vertex number 8 sits at the infinitely far point in $S^3$.}\label{Fig:HypercuboidGraph}
\end{figure}

The underlying graph $\Gamma$ is the dual boundary graph of a $4d$ hypercuboid (see Fig.~\ref{Fig:HypercuboidGraph}). Note that, at this point, we only take the graph itself, but do not consider the hypercuboid as geometric, $4$-dimensional polytope. Rather, we consider a bivector geometry $\{B_\ell\}_\ell$ on $\Gamma$ such that, for each vertex of the graph, the incident bivectors form a $3d$ cuboid embedded in $\mathbb{R}^4$. Bivectors of this form are such that all $B_\ell$ lying on a great circle in $\Gamma$ coincide. Since there are six independent great circles in this graph (see Fig.~\ref{Fig:HopfLinkColour}), this leaves us with the freedom of choosing six bivectors. They are as follows: 
\begin{eqnarray}\label{Eq:BivectorGeometryHypercuboid}
B_1\;&=&\;A_1\,\star\big(e_y\wedge e_z\big),\quad B_2\;=\;A_2\,\star\big(e_z\wedge e_x\big),\quad B_3\;=\;A_3\,\star\big(e_x\wedge e_y\big),\\[5pt]\nonumber
B_4\;&=&\;A_4\,\star\big(e_z\wedge e_t\big),\quad B_5\;=\;A_5\star\,\big(e_t\wedge e_y\big),\quad B_6\;=\;A_6\,\star\big(e_x\wedge e_t\big),
\end{eqnarray}

\noindent with areas $A_i>0$ and unit normal vectors $e_k\in\mathbb{R}^4$. It is straightforward to check that, for each vertex in the graph, the incident bivectors satisfy closure~\eqref{eq:closure-3d}, as well as linear simplicity~\eqref{eq:linear-simplicity}. In particular, they also satisfy diagonal-, cross-simplicity, as well as the non-degeneracy conditions specified in~\cite{BarrettCrane1998Euclid}.

\begin{figure}
\begin{center}
\includegraphics[scale=0.5]{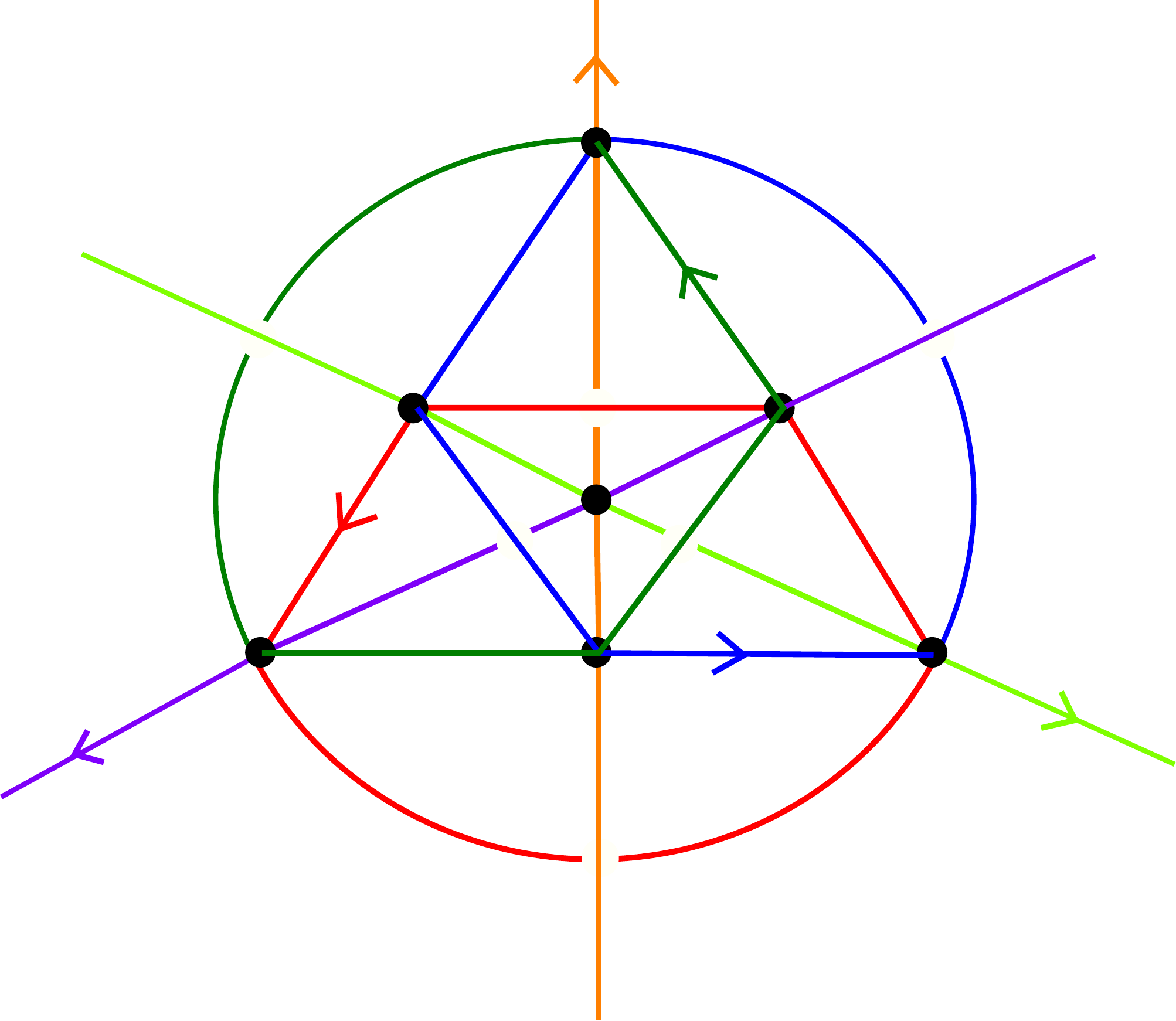}
\end{center}
\caption{The six independent loops in the hypercuboid graph $\Gamma$, colour-coded. Each such loop corresponds to one great circle on $S^3$. Two pairs form, respectively, one of the three independent Hopf links $\mathpzc{H}_{\,1}$, $\mathpzc{H}_{\,2}$, $\mathpzc{H}_{\,3}$ in $\Gamma$. The orientation of the edges in $\Gamma$ are such that each of the six crossings $\mathpzc{C}$ has $\sigma(\mathpzc{C})=+1$.}\label{Fig:HopfLinkColour}
\end{figure}

However, the bivector geometry~\eqref{Eq:BivectorGeometryHypercuboid} does, in general, not follow from a $4d$ geometric polytope. While the restriction of~\eqref{Eq:BivectorGeometryHypercuboid} to each single vertex describes the  geometry of a $3d$ cuboid, these cuboids do not fit together in $\mathbb{R}^4$. Whereas their touching faces are parallel and describe rectangles with coinciding areas, their shapes do not match (i.e.~they are rectangles with differing side lengths).

While the shape-matching problem is well-known in the $4$-simplex case, there the boundary states describing a non-shape-matching geometry are suppressed in the semiclassical asymptotics of the vertex amplitude~\cite{Barrett-etal2009asympt4simplex-euclid,Barrett-etal2010asympt4simplex-lorentz}. In the hypercuboid case (and more genral polytopes), however, geometries as depicted above are not suppressed in that limit, as has been observed in \cite{BahrSteinhaus2016Cuboidal-EPRL,Dona-etal2017KKL-asympt}. In the previous Ch.~\ref{ch:problem} we connected this to the fact that, in case one does not deal with a $4$-simplex, linear simplicity and closure \emph{do not} imply volume simplicity~\cite{Belov2018Poincare-Plebanski}. Hence, the set of allowed bivector geometries is not constrained enough, and also allows non-geometric configurations. 

However, the situation is different when we demand the bivector geometry to additionally satisfy the Hopf link volume-simplicity constraint, i.e.~if we demand that $V_\mathpzc{H}$ is independent of $\mathpzc{H}$. In the case of the hypercuboid, there are essentially three different Hopf links (depicted in Fig.~\ref{Fig:HopfLinkColour}). The respective Hopf-volumes are easily computed to be
\begin{eqnarray}
V_{\mathpzc{H}_{\,1}}\;=\;\frac{1}{3}A_1A_6,\qquad 
V_{\mathpzc{H}_{\,2}}\;=\;\frac{1}{3}A_2A_5,\qquad 
V_{\mathpzc{H}_{\,3}}\;=\;\frac{1}{3}A_3A_4.
\end{eqnarray}

\noindent Therefore, this version of the volume simplicity constraint reads 
\begin{eqnarray}\label{Eq:VolumeSimplicityHypercuboid}
A_1A_6 \ = \ A_2A_5 \ = \ A_3A_4.
\end{eqnarray}

It should be noted that these two conditions are precisely those that reduce the $6$-dimensional space of hypercuboidal bivector geometries~\eqref{Eq:BivectorGeometryHypercuboid} (given by six $A_i$) to the $4$-dimensional subset of hypercuboidal geometries (given by four lengths of the edges of the hypercuboid): using~\eqref{Eq:VolumeSimplicityHypercuboid}, one can easily show that the edge lengths of the $3d$ cuboids meeting at a common face derived from~\eqref{Eq:BivectorGeometryHypercuboid} agree, hence define four edge lengths, which in turn define a geometric hypercuboid in four dimensions. Thus, in the hypercuboid case, the Hopf link volume-simplicity constraint is precisely the right condition to allow for a reconstruction of the $4$-dimensional hypercuboid from the bivectors.

\pagebreak

\section{Quantum constraint in terms of Hopf-link invariance}
\label{sec:volume-quant}

So far we have considered classical bivector geometries. In the following we will define a quantum version of the total volume and the Hopf link volumes, as well as discuss a version of the volume simplicity constraint derived from it.

The Hilbert space we use is the one from discrete BF theory, for gauge group $H=\mr{SU}(2)\times \mr{SU}(2)$. There is one Hilbert space $\mathcal{H}_\Gamma$ associated to the boundary of a $4d$ polytope with oriented boundary graph $\Gamma$. 
\begin{eqnarray}\label{Eq:BoundaryHilbertSpace}
\mathcal{H}_\Gamma\;=\;\mr{L}^2\Big(H^L,d\mu_{\text{Haar}}^{\ph{i}}\Big)^{\text{inv}}\;\cong\;\mr{L}^2\Big(H^L/H^N,d\mu\Big),
\end{eqnarray}

\noindent where $L$ is the number of links/edges in $\Gamma$, and $N$ the number of nodes. The invariance is with respect to the action of $k\in H^N$ on $h\in H^L$, i.e.~$h_\ell\,\to\,k_{s(\ell)}^{-1}h_\ell k_{t(\ell)}$, where $s(\ell)$ and $t(\ell)$ are source- and target vertex of the edge $\ell$. 

Due to the split of the gauge group $H=\mr{SU}(2)\times \mr{SU}(2)$ into right and left $\mr{SU}(2)$, each state in~\eqref{Eq:BoundaryHilbertSpace} can be written as a linear combination of tensor products of two $\mr{SU}(2)$ spin network functions $\Psi=\psi^+\otimes \psi^-$. The field $B_\ell=(\vec{b}_\ell^+,\vec{b}_\ell^-)$ acts as the respective left-invariant vector fields $(X^{+,I},X^{-,I})$, $I=1,2 ,3$, on $\mr{SU}(2)$, i.e.
\begin{eqnarray}
\hat{b}_\ell^{+,\,I}\Psi\;=\;\Big(X_\ell^{+,I}\psi^+\Big)\otimes\psi^-,\qquad 
\hat{b}_\ell^{-,\,I}\Psi\;=\;\psi^+\otimes\Big(X_\ell^{-,I}\psi^-\Big).
\end{eqnarray}

\noindent Using the relation~\eqref{eq:Immirzi-flipped-map} between the $B$-field and $\Sigma$ in the Euclidean case $(-1)^\eta=+1$, we get
\begin{eqnarray}
\vec{\sigma}_\ell^\pm\;=\;\frac{\gamma}{1\pm\gamma}\vec{b}_\ell^\pm.
\end{eqnarray}

\noindent Since, for two bivectors $\Sigma_i\sim(\vec{\sigma}_i^+,\vec{\sigma}_i^-)$, with $i=1,2$, one has that
\begin{equation}
\star(\Sigma_1\wedge \Sigma_2)\;=\;2\big(\vec{\sigma}_1^+\cdot\vec{\sigma}_2^+\,-\,\vec{\sigma}_1^-\cdot\vec{\sigma}_2^-\big),
\end{equation}

\noindent the crossing volume $V_\mathpzc{C}$~\eqref{Eq:DefinitionCrossingVolume} of the two edges $e_{1,2}$, is quantized as
\begin{eqnarray}\label{Eq:CrossinVolumeOperator}
\hat{V}_\mathpzc{C}\;=\;2\sigma(\mathpzc{C})\gamma^2\sum_{I=1}^3\left(\frac{1}{(1+\gamma)^2}X_{\ell_1}^{+,I}\otimes X_{\ell_2}^{+,I}\,-\,
\frac{1}{(1-\gamma)^2}X_{\ell_1}^{-,I}\otimes X_{\ell_2}^{-,I}
\right).
\end{eqnarray}

\noindent The expression~\eqref{Eq:CrossinVolumeOperator} has, up to factors, also appeared in~\cite{Han2011cosm-const-vertex}, where it was introduced to construct an ad-hoc deformation of the EPRL spin foam model to include a cosmological constant (cf. also~\cite{BahrRabuffo2018EPRL-deformation}). Using this, we similarly to~\eqref{Eq:DefinitionTotalVolume},~\eqref{Eq:HopfLinkVolume}, define
\begin{eqnarray}
\hat{V}\;:=\;\frac{1}{6}\sum_\mathpzc{C}\hat{V}_\mathpzc{C},\qquad \hat{V}_\mathpzc{H}\;:=\;\frac{1}{6}\sum_{\mathpzc{C}\between \mathpzc{H}}\hat{V}_\mathpzc{C}.
\end{eqnarray}

\noindent While it is easy to define operator analogues for total volume and Hopf link volume, the corresponding constraints on the quantum level are a bit tricky. There are two reasons for this. 

Firstly, it seems unlikely that the correct way of imposing the constraints is to impose it as strong or weak constraints on the state space. The reason  for this is that the volume constraint is  conceptually different from diagonal- and cross-simplicity (or linear simplicity, for that matter): The other constraints are decidedly \emph{kinematical}, since they involve conditions that can be formulated entirely in terms of single intertwiners of the $3d$ boundary data. Indeed, diagonal- and cross-simplicity (as well as closure) are used to guarantee that the intertwiners of the boundary state have an interpretation in terms of $3d$ polyhedra~\cite{LivineSpeziale2007new-vertex,BianchiDonaSpeziale2011Polyhedra}. The volume-simplicity constraint, however, relates bivectors on different, not necessarily neighbouring $3d$ polyhedra, which have a specific relation to each other depending on the $4d$ geometry. It is therefore arguably much more \emph{dynamical}. This is why the formulation of the constraints should not just be operator equations on the boundary Hilbert space, but should involve the dynamics of the model, i.e.~the amplitude.

Secondly, the so-called `cosine problem' makes the use of the direct amplitude difficult~\cite{BarrettNaish-Guzman2009Ponzano-Regge}. Given of what we just said about the dynamical nature of the volume simplicity constraints, the most straightforward implementation of the Hopf-link volume-simplicity constraint would be to us the EPRL amplitude~\cite{EPRL-FK2008finiteImmirzi}:
\begin{eqnarray}
\mathcal{A}(\Psi)\;:=\;\Psi(\mathds{1},\ldots,\mathds{1}),
\end{eqnarray}

\noindent and define
\begin{eqnarray}\label{Eq:HopfVolumesQuantum}
V_\Psi\;:=\;\mathcal{A}\big(\hat{V}\Psi\big),\qquad V_{\mathpzc{H},\Psi}\;:=\;\mathcal{A}\big(\hat{V}_H\Psi\big),
\end{eqnarray}

\noindent demanding that 
\begin{eqnarray}\label{Eq:HopfLinkVolumeSimplicity}
V_\mathpzc{H} \text{ coincide for different Hopf links }\mathpzc{H}.
\end{eqnarray}

\noindent  However, the EPRL amplitude defines a dynamics in which both space-time orientations are being taken into account. This is easily seen in the semiclassical asymptotics~\cite{Barrett-etal2009asympt4simplex-euclid} of the amplitude $\mathcal{A}$, where not only $\exp(iS_{\text{EC}})$, but $\exp(iS_{\text{EC}})+\exp(-iS_{\text{EC}})=2\cos(S_{\text{EC}})$ appears (cf.~\eqref{eq:partition-2}). Since both of these contributions have $4$-dimensional volumes with opposite signs,  the semiclassical asymptotics of~\eqref{Eq:HopfVolumesQuantum} vanishes for the hypercuboid, as can be easily shown. This makes the comparison of~\eqref{Eq:HopfVolumesQuantum} for different $\mathpzc{H}$, as quantum version of the Hopf-link volume-simplicity constraints questionable.

There are several possible ways out of this: One would be to replace $\mathcal{A}$ with the so-called \emph{proper vertex} $\mathcal{A}^\text{pr}$~\cite{Engle2013proper-vertex}, which aims at defining the model to only include one of the two orientations. While its definition for the $4$-simplex is well-understood, the definition for arbitrary graphs is open, and relies on a definition of $4$-volume in that context. The definition of the constraint could become circular in that case. The semiclassical asymptotics of the proper vertex is, however, conjectured to be understood for all graphs, and in the hypercuboidal case leads to the right answer, as we will see below.

Alternatively, if one were to use the original EPRL amplitude, one could alter the definitions~\eqref{Eq:HopfVolumesQuantum}, by replacing $\hat{V}_\mathpzc{H}$ by $\hat{V}_\mathpzc{H}^2$ or $|\hat{V}_\mathpzc{H}|$. Either of those should alleviate the cosine problem, since both expressions are independent of the sign of space-time orientation.

In the semiclassical asymptotics, all three of these propositions give the same result. Let us consider the quantum cuboid states that have been introduced in~\cite{BahrSteinhaus2016Cuboidal-EPRL}. These states have played a crucial role in the investigation of the critical behaviour of the EPRL model at small scales~\cite{BahrSteinhaus2016phase-transition}. These states depend on six spins $j_i$, $i=1,\ldots, 6$ associated to the six great circles on the boundary, as depicted in Fig.~\ref{Fig:HopfLinkColour}. The intertwiners are all that of Livine-Speziale~\eqref{eq:Livine-Speziale}
\begin{eqnarray}
\iota_{j_{a_1}j_{a_2}j_{a_3}}\;=\;\int_{\mr{SU}(2)}du\;u\triangleright\left[\bigotimes_{i=1}^3|j_{a_i},e_i\rangle\langle j_{a_i},e_i|\right],
\end{eqnarray}

\noindent corresponding to the eight cuboids in Fig.~\ref{Fig:HypercuboidGraph}. This results in the boundary state 
\begin{eqnarray}
\Psi\;=\;\Phi_\ga\circ\left[\bigotimes_{a=1}^8 \iota_a\right],
\end{eqnarray}

\noindent where $\Phi_\gamma$ is the EPRL map~\eqref{eq:EPRL-map}, performing the embedding of the $\mr{SU}(2)$ spin-network into a boundary Hilbert space of the SF amplitude. 

It has been shown that the large-$j$-asymptotic formula for the amplitude $\mathcal{A}(\Psi)$ in the case $\gamma\in(0,1)$ can be written as follows:
\begin{eqnarray}
\mathcal{A}(\Psi)\;\stackrel{j\to\lambda j,\lambda\to\infty}{\longrightarrow}\;\lambda^{-21}\left(\frac{1}{D}+\frac{2}{|D|}+\frac{1}{D^*}\right),
\end{eqnarray}

\noindent where $D$ is a complex polynomial of order $21$ in the $j_i$.

The proper vertex in that limit can be defined the following way: The asymptotics relies on the extended stationary phase approximation of the integral over $g^\pm_v\in (\mr{SU}(2)\times \mr{SU}(2))^{N}$ defining $\mathcal{A}$. For all known examples, there are at most two solutions $S_{1,2}$ for either $g^+_v$ and $g^-_v$, leading to four solutions in total. The proper vertex is defined by choosing $S_1$ for $g^+_v$ and $S_2$ for $g^-_v$ (see~\cite{Barrett-etal2009asympt4simplex-euclid,Engle2013proper-vertex} for details; they do not matter for the discussion that much). The proper vertex for the three different Hopf links $\mathpzc{H}_{\,i}$ in $\Gamma$  lead to
\begin{eqnarray}
\mathcal{A}^\text{pr}\big(\hat{V}_{\mathpzc{H}_{\,1}}\Psi\big)\;\longrightarrow\;\frac{\gamma^2}{6}j_1j_6\,\lambda^{-19}\left(\frac{1}{|D|}\right),
\end{eqnarray}

\noindent with $j_1j_6$ being replaced by $j_2j_5$ and $j_3j_4$ for $\mathpzc{H}_{\,2}$ and $\mathpzc{H}_{\,3}$, respectively. 

Using the square of the Hopf link volume operator and the EPRL amplitude results in 
\begin{eqnarray}
\mathcal{A}\big(\hat{V}_{\mathpzc{H}_{\,1}}^2\Psi\big)\;\longrightarrow\;\frac{\gamma^4}{36}\big(j_1j_6\big)^2\,\lambda^{-17}\left(\frac{1}{D}+\frac{1}{D^*}\right),
\end{eqnarray}

\noindent again with $j_1j_6$ being replaced by $j_2j_5$ and $j_3j_4$ for $\mathpzc{H}_{\,2}$ and $\mathpzc{H}_{\,3}$, respectively. Using $|\hat{V}_\mathpzc{H}|$ instead of $\hat{V}_\mathpzc{H}^2$ leads to a similar result. 

We see that in all of these cases, the Hopf-link volume-simplicity constraints can be formulated by demanding the corresponding quantum expressions to coincide for different Hopf links $\mathpzc{H}$ in the boundary graph. While they all agree in the asymptotic limit, in the regime of small spins they could, very well, differ. The question of which of the possibilities should be the right one, is still open. 

Nevertheless, we can see that in the semi-classical limit, the proper version of quantum condition turns into the classical condition~\eqref{Eq:VolumeSimplicityHypercuboid} on the boundary data $j_i,\iota_a$ of the coherent state $\Psi$, which in turn is precisely the geometricity condition discussed in~\cite{BahrSteinhaus2016Cuboidal-EPRL,Dona-etal2017KKL-asympt,Belov2018Poincare-Plebanski}.

\subsection{Extension beyond the hypercuboidal case?}

First we note that both operators~\eqref{Eq:HopfVolumesQuantum} are invariant under ambient isotopies, i.e.~they do not depend on which way the graph $\Gamma$ is projected to the plane~\cite{Bahr2018polytope-volume}. 

Let us comment about the $4$-simplex case. The boundary graph of the $4$-simplex is the complete graph in five vertices, which does not contain \emph{any} Hopf links (see Fig.~\ref{Fig:FourSimplex}). This is not hard to see, as any non-trivial cycle needs at least three vertices, so a Hopf link needs at least six in total, since the two circles in it are not allowed to intersect. So in this case the condition~\eqref{Eq:HopfLinkVolumeSimplicity} is empty, which is agreeing with the fact that already classically, the diagonal~\ref{diag-simplicity} and cross-simplicity~\ref{cross-simplicity} constraints plus closure~\eqref{eq:closure-3d} imply the volume-simplicity. 

\begin{figure}[!hb]
\begin{center}
\includegraphics[scale=1.0]{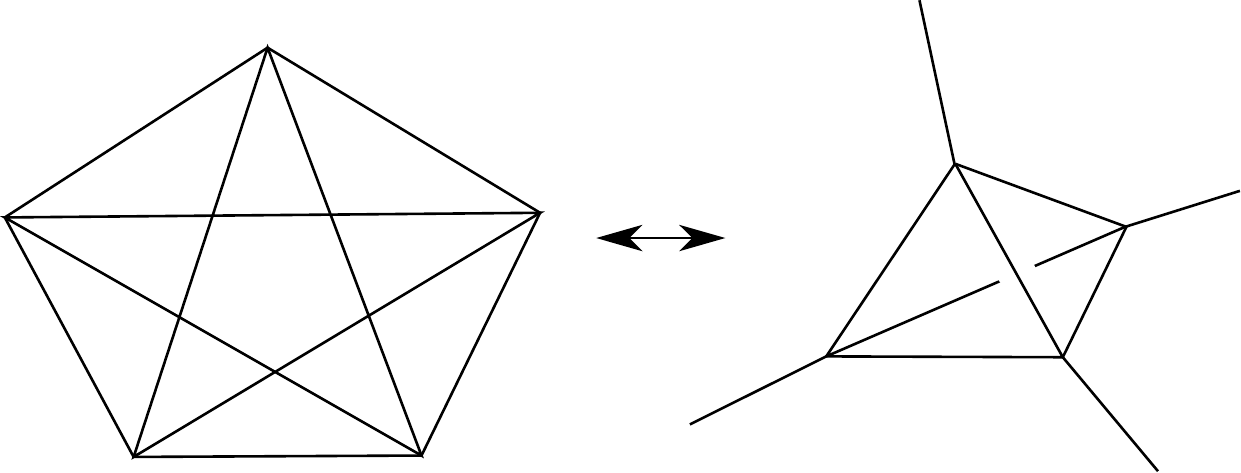}
\end{center}
\caption{A $4$-simplex, as polytope $P$. The boundary graph of $P$ is the complete graph in five vertices, and looks itself like a $4$-simplex. It can be rearranged to have only one crossing (additionally, we have moved one vertex to the infinitely far point, to make the image clearer).}\label{Fig:FourSimplex}
\end{figure}

\pagebreak

On the other hand, the larger the graph $\Gamma$, the more independent conditions are given by~\eqref{Eq:HopfLinkVolumeSimplicity}. This is in line with the fact that the classical condition~\ref{volume-quadr} $\displaystyle{\star B_f\cdot B_{f'}=:\tilde{V}_v(f,f')}$  gets more complicated for polyhedra with more faces.

Still, it is an unsolved question whether, for arbitrary graphs $\Gamma$, conditions~\eqref{Eq:HopfLinkVolumeSimplicity} are sufficient to restrict the bivector geometry enough to capture the right discretised version of metric geometry. This is related to the question of whether closure, linear simplicity and Hopf-link-simplicity are enough to prove a reconstruction theorem analogous to the $4$-simplex case in~\cite{BarrettCrane1998Euclid}. We feel that there is a very good chance for this to hold, since the Hopf-link constraint seems to capture the essential features of the $4$-volume in all cases we have looked at so far. Another hint might be the fact that the Hopf-link-volume $V_\mathpzc{H}$ could serve as a discretised version of the intersection form $Q$ of $4$-dimensional manifolds $M$. In its real formulation, it assigns a number $Q(\Sigma_1,\Sigma_2)$ to $\Sigma_i\in H^2(M)$, which, in the de Rham cohomology, are nothing but closed $2$-forms modulo translation symmetry. The expression in~\eqref{Eq:DefinitionTotalVolume} is then nothing but a discrete version of
\begin{eqnarray}
Q(\Sigma_1,\Sigma_2)=\int_M \Sigma_1\wedge \Sigma_2,
\end{eqnarray}

\noindent lifted to the context of $\Sigma_i\in H^2(M,\bigwedge^2\mathbb{R}^4)$ in the natural way. Due to the way in which the intersection form can also be evaluated by summing over the points of intersection of two  surfaces embedded in $P$, it might be possible to turn this into a sum over crossings of $1$-dimensional submanifolds in the boundary, i.e.~Hopf links in the boundary graph. It might therefore, indeed be enough to consider not the total volume, but only the Hopf link volume, to compute the volume of $P$. This is a point we may come back to in the future. 


\chapter{Proposition \textnumero2: Dual formulation with the co-frames}
\label{ch:prop-2}
 
The present Chapter approaches the problem of volume constraints in Sec.~\ref{sec:problem-linear} from the different perspective. It was already stressed that the normals are considered to be as characteristic for discrete picture as bivectors are, in our view. In what follows, we revisit the field-theoretic content of the \emph{linear formulation} in Sec.~\ref{subsec:linear-formulation}. One is able to achieve the closer contact with Cartan geometric degrees of freedom, by switching from normals directly to tetrad/co-frame, in the classical continuum formulation. Following the approach adopted in Spin Foams, we first provide the characterization in Sec.~\ref{sec:Poincare-BF} for the (unconstrained) Poincar\'{e} BF-type theory on the extended configuration space~\footnote{Apparently, an analogous model was first considered in~\cite{Bi2013PoincareBF}. It is also closely related (equivalent via integration by parts, i.e. up to boundary terms) to the corresponding topological Higher Gauge Theory~\cite{GirelliPfeifferPopescu2008BFCG}. However, no 2-group Categorical Generalization is implied in our case.}. The focus is on gauge symmetries: we first verify the invariances of lagrangian, and then show how they are preserved also on the canonical level in the Hamiltonian analysis of Sec.~\ref{subsec:Poincare-BF-hamiltonian}. The canonical generator of gauge transformations is explicitly constructed. In Sec.~\ref{sec:continuum-linear}, the equivalent set of linear simplicity constraints, reformulated in terms of co-frames, is presented and read geometrically as describing $3$-volume of a hypersurface.

\paragraph{Some preliminary remarks.} The initial incentive was to extend the 3d closure condition~\eqref{eq:closure-3d}, of the discrete Gauss law in BF theory, to the 4d normals of the linear EPRL constraints. It could then have provided the description of polytope geometry, through the generalized Minkowski theorem~\cite{BianchiDonaSpeziale2011Polyhedra}, on the `kinematical' level. Unlike the 3d case, however, the prospective 4d closure~\eqref{eq:closure-4d} does not follow immediately from the (discretized) equations of motion, but rather exploits the simplicity of bivectors themselves.




In this way, consider the hypersurface trivector $[X_1\wedge X_2\wedge X_3]$, spanning the infinitesimal 3-volume $\det[\theta(X_1),\theta(X_2),\theta(X_3)]$. Construct its complementary  dual vector, having the same magnitude and compatible orientation. The corresponding coefficients define the 3-forms:
\begin{align}\label{eq:3-forms}
\vartheta^A \ & = \ \frac{1}{3!}\vartheta^A_{bcd}\, dx^b\wedge dx^c\wedge dx^d \nonumber \\ 
& = \ \frac{1}{3!}\varepsilon^A_{\phantom{A}BCD}\, \theta^B_b \theta^C_c \theta^D_d \, dx^b\wedge dx^c\wedge dx^d \ = \ \frac{1}{3!}\varepsilon^A_{\phantom{A}BCD}\, \theta^B\wedge \theta^C\wedge \theta^D,
\end{align}
applicable to arbitrary vector arguments in $T\mc M$. This clarifies the relation of 4d normals~\eqref{eq:discrete-normals} to (co-)tetrads in the classical continuum theory.

By the Hodge $\star$ duality, we may prefer an alternative parametrization and use directly the tetrad field  $\theta^A_a$, without any loss of generality -- the number of components is the same as $\vartheta^A_{bcd}$. In order to identify the appropriate kinetic term, suppose that a cellular decomposition $\bigcup_v \mc T_v\simeq \mc M$ of the spacetime manifold is given. Express the l.h.s. of~\eqref{eq:closure-4d} through the Stokes' theorem:
\begin{equation}\label{eq:Stokes}
\sum_{e\supset v} \mc V^A_e \ = \ \int_{\pa\mc T_v} \vartheta^A \ = \ \int_{\mc T_v} d\vartheta^A, \qquad \text{where} \quad \pa\mc T_v = \bigcup_{e\supset v}\tau_e
\end{equation}
-- the boundary of the 4-simplex $\mc T_v$ dual to the vertex $v$ (more generally, any 4-polyhedron). One possibility is to introduce it with the Lagrange multipliers (and covariant derivative for non-trivial connection):
\begin{gather}
\int \mathpzc{b}_A^{\phantom{A}} \, D \vartheta^A \ = \ \int \frac12 \mathpzc{b}_A^{\phantom{A}}\varepsilon^A_{\phantom{A}BCD}\, \theta^B\wedge \theta^C\wedge D^\om \theta^D \ \equiv \ -\int  \be_A^{\phantom{A}}\wedge \Theta^A, \nonumber\\
\text{where} \quad \be_A^{\phantom{A}} \ := \ \frac12 \varepsilon_{ABCD}^{\phantom{ABCD}} \, \mathpzc{b}^B \theta^C\wedge \theta^D \ \equiv \ \frac12 \be_{Abc}^{\phantom{A}}\, dx^b\wedge dx^c, \quad \be^A_{cd}\ =\ \varepsilon^{AB}_{\phantom{AB}CD}\, \mathpzc{b}_B^{\phantom{B}}\theta^C_c \theta^D_d \label{eq:T-multiplier}
\end{gather}
-- the 2-form taking values in $\mathbb{R}^{3,1}$, whilst $\mathpzc{b}_A$ are scalar functions (0-forms). 
The first equality utilizes the relation~\eqref{eq:3-forms}, which endows arbitrary 3-forms with geometric meaning of complementary vectors, that is carries the metric information. Given this, the vanishing of torsion $\Theta^A:=D\theta^A$ in the first line is sufficient for $D^\om\vartheta^A=0$ (4d closure), but not the other way around. 


The condition of the vanishing `curl' $D^\om\theta^A=0$ is apparently stronger then $D\vartheta^A=0$, implying that dual vector (for fixed $A$) is divergence free. One can see this in terms of components of the torsion tensor:
\begin{equation}
\de\mathpzc{b}_A^{\phantom{A}} \quad \Rightarrow \quad D \vartheta_A \ = \ \frac14 \Theta^D_{\phantom{D}EF} \, \varepsilon^{\phantom{A}}_{ABCD} \, \theta^B \wedge \theta^C\wedge \theta^E\wedge \theta^F \ = \ \Theta^D_{\phantom{D}AD} \, \det(\theta) \, d^4x \ = \ 0.
\end{equation}
Thus, only the vector irreducible component of torsion vanishes in the first case [geodesics are still riemannian], whereas the entire tensor $\Theta$ is put to zero, if we choose to take generic~$\be_A$ as an independent Lagrange multipliers (conjugate to $\theta^A$)~\footnote{There is still a possibility left for exploration if we had chosen independent $\mathpzc{b}_A$ and weaker equation $D^\om\vartheta^A=0$. In this regard, it would be interesting to compare the two constraints and to identify the sets of compatible geometric configurations in each case (especially in the light of allowed `conformal shape-mismatch' in the recent work~\cite{Dona-etal2017KKL-asympt}).}.

Before studying gravity per se, let us examine first the simplifying theory, which one gets by modifying correspondingly the pure BF action. Thereby, we suggest to consider the following action with Lagrange multipliers imposing the conditions of zero curvature and torsion:
\begin{equation}\label{eq:BF-Poincare}
S_{0}[B,\be,\om,\theta] \ = \ \int B_{AB}^{\phantom{AB}}\wedge \Om^{AB} + \be_A^{\phantom{A}}\wedge \Theta^A.
\end{equation}
The main assumption made about the frame field is that it is non-degenerate, and the matrix $\theta^A_a$ is invertible. Even without knowing the Poincar\'{e} gauge structure, in the form of the action we may recognize the special case of the theory~\cite{MontesinosCuesta2008BF-Cartan-Ndim,MontesinosCuesta2007BF-Cartan}, designed to obey the Cartan's (also Bianchi's) structure equations. One could anticipate and announce~\eqref{eq:BF-Poincare} as the Poincar\'{e} BF theory -- this assertion is made precise, as we develop in this section the structure of the theory in full detail. 

\section{Poincar\'{e} BF theory}
\label{sec:Poincare-BF}
\subsection{Lagrangian and gauge symmetries}
\label{subsec:Poincare-BF-lagr}

The theories in physics are defined by their symmetries, which are usually presupposed. To find them out from the given action principle we follow the route, in a sense, reverse to the renowned Noether's 2nd theorem. Namely, the gauge symmetries (by which we understand the dependence of the dynamics on an arbitrary functions) manifest themselves through the differential identities among equations of motion. Calculate the Euler-Lagrange derivatives (putting the results into convenient language of forms):

\begin{subequations}\label{eq:BF-Poincare-variations}
\begin{align}
&& \frac{\de S_0}{\de B_{AB}}&&& := && \frac{1}{2!}\eps_{abcd} \frac{\de S_0}{\de B_{ABab}} dx^c \wedge dx^d &&= && \Om^{AB} , && \\
&& \frac{\de S_0}{\de \be_A} &&& := && \frac{1}{2!}\eps_{abcd} \frac{\de S_0}{\de \be_{Aab}} dx^c \wedge dx^d &&= && \Theta^A , && \\
&& \frac{\de S_0}{\de \om_{AB}} &&& := && \frac{1}{3!}\eps_{abcd} \frac{\de S_0}{\de \om_{aAB}} dx^b \wedge dx^c \wedge dx^d &&= && D^\om B^{AB} - \theta^{[A}\wedge \be^{B]} , && \label{eq:variation-omega}\\
&& \frac{\de S_0}{\de \theta_A} &&& := && \frac{1}{3!}\eps_{abcd} \frac{\de S_0}{\de \theta_{aA}} dx^b \wedge dx^c \wedge dx^d &&= && D^\om \be^A. &&
\end{align}
\end{subequations}

Setting variations to zero (with some fixed boundary conditions), we obtain the field equations, which every physical motion must satisfy:

\begin{equation}\label{eq:BF-Poincare-e.o.m.}
\frac{\de S_0}{\de B}, \frac{\de S_0}{\de \be}, \frac{\de S_0}{\de \om}, \frac{\de S_0}{\de \theta} \ = \ 0.
\end{equation}
The torsion now manifestly vanishes (as well as the curvature), and one recognizes in the third line~\eqref{eq:variation-omega} the generalized covariant conservation of $B$, i.e. the lifted 3d closure of bivectors in the discrete.


It is now straightforward to derive the corresponding differential relations between functional derivatives~\eqref{eq:BF-Poincare-variations}: 
\begin{subequations}\label{eq:differential-identities}
\begin{align}
&&& \mc I_{AB}^{(1)}  && := &&  D^\om \frac{\de S_0}{\de B^{AB}} \ = \ 0 , &&& \\
&&& \mc I_A^{(2)}  && := &&  D^\om \frac{\de S_0}{\de \be^A} - \theta^B \wedge \frac{\de S_0}{\de B^{AB}} \ = \ 0 , &&& \\
&&& \mc I_{AB}^{(3)}  && :=  && D^\om \frac{\de S_0}{\de \om^{AB}} - 2 B_{[A}^{\phantom{[B}C}\wedge\frac{\de S_0}{\de B^{B]C}}  - \be_{[A}^{\phantom{[}}\wedge \frac{\de S_0}{\de \be^{B]}} - \theta_{[A}^{\phantom{[}}\wedge \frac{\de S_0}{\de \theta^{B]}} \ = \ 0 , &&& \\
&&& \mc I_A^{(4)}  && := &&  D^\om \frac{\de S_0}{\de \theta^A} - \be^B\wedge\frac{\de S_0}{\de B^{AB}} \ = \ 0 , &&&
\end{align}
\end{subequations}
which are vanishing identically (off-shell), without using the equations of motion (only their form). To obtain the above relations we used the 2nd and 1st Bianchi's identities, i.e. the commutation of covariant derivatives:
\begin{align*}
D^\om\Om^{AB} \ & = \ 0, &  D^\om D^\om B^{AB} \ & = \ \Om^A_{\phantom{A}C}\wedge B^{CB} + \Om^B_{\phantom{B}C}\wedge B^{AC}, \\
D^\om D^\om\theta^A  \ & = \ \Om^A_{\phantom{A}B}\wedge \theta^B, & D^\om D^\om \be^A \ & = \ \Om^A_{\phantom{A}B}\wedge \be^B.
\end{align*} 
From the complete set of independent differential identities~\eqref{eq:differential-identities} we now form a generic linear combination, which is identically equal to zero, and integrate by parts:
\begin{gather*}
\int \Xi^{AB}\wedge \mc I_{AB}^{(1)} + \xi^A\wedge \mc I_A^{(2)} + \mathpzc{U}^{AB}\mc I_{AB}^{(3)} + \mathpzc{u}^A\mc I_A^{(4)} \\
= \int \left(D\Xi^{AB} - \xi^A\wedge \theta^B + \mathpzc{U}^A_{\phantom{A}C}B^{CB}+\mathpzc{U}^B_{\phantom{B}C}B^{AC} - \mathpzc{u}^A\be^B\right)\wedge \frac{\de S_0}{\de B^{AB}} \\
+\left(D\xi^A+\mathpzc{U}^A_{\phantom{A}B}\be^B\right)\wedge \frac{\de S_0}{\de \be^A} - D\mathpzc{U}^{AB}\wedge \frac{\de S_0}{\de \om^{AB}} - \left(D\mathpzc{u}^A -\mathpzc{U}^A_{\phantom{A}B}\theta^B\right)\wedge \frac{\de S_0}{\de \theta^A},
\end{gather*}
for some arbitrary coefficient functions $(\mathpzc{U}^{AB},\mathpzc{u}^C)$ and 1-forms $(\Xi^{AB},\xi^C)$. The transformations that leave the action invariant, up to divergence, are readily seen:
\begin{subequations}\label{eq:BF-Poincare-gauge-transform}
\begin{align}
\de \om^{AB} \ & = \ -  d\mathpzc{U}^{AB} + f^{AB}_{CD,EF}\, \om^{CD}\mathpzc{U}^{EF}, \label{eq:Poincare-connection-transform-om} \\ 
\de \theta^A \ & = \ - d\mathpzc{u}^A + f^A_{CD,B}\left(\om^{CD}\mathpzc{u}^B-\mathpzc{U}^{CD}\theta^B\right), \label{eq:Poincare-connection-transform-e}\\
\de B^{AB} \ & = \ -f^{AB}_{CD,EF}\mathpzc{U}^{CD}B^{EF}-\mathpzc{u}^{[A}\be^{B]} +D\Xi^{AB} - \xi^{[A}\wedge \theta^{B]}, \label{eq:null-B-transform} \\ 
\de \be^A \ & = \ -f^A_{CD,B}\mathpzc{U}^{CD}\be^B + D\xi^A, \label{eq:null-T-transform} \\
\text{where} \quad f^{AB}_{CD,EF} \ : & = \ \eta^{\phantom{[}}_{E[C}\de^{[A}_{D]}\de^{B]}_{F\phantom{]}} - \eta^{\phantom{[}}_{F[C}\de^{[A}_{D]}\de^{B]}_{E\phantom{]}}, \quad  f^A_{CD,B} \ := \ \eta_{B[C}^{\phantom{AB}}\de^A_{D]}. \label{eq:Poincare-structure-constants}
\end{align}
\end{subequations}

We may combine the connection $\om\in\mathfrak{so}(3,1)$ and the gauge potential of translations $\theta$ into a single (Cartan) connection~\cite{Wise2010Cartan-geometry} of the Poincar\'{e} gauge group~\footnote{This is quite appealing from another point of view that the usual non-degeneracy condition on $\mr{det}\,\theta^A_a\neq0$, which is unclear how to achieve in practise, is also the requirement for the $\mathfrak{g}$-connection to be a Cartan connection.}:
\begin{equation}\label{eq:Poincare-connection}
\mathfrak{so}(3,1)\ltimes\mathfrak{p}^{3,1}\ \ni\ \varpi \ := \ \om + \theta \ = \ \om^{AB}\mc J_{AB} + e^C\mc P_C,
\end{equation}
whose generators $\mc J_{AB},\mc P_C$ satisfy the algebra (of which~\eqref{eq:Poincare-structure-constants} are structure constants):
\begin{align}
\begin{split}\label{eq:Poincare-algebra}
[\mc J_{AB},\mc J_{CD}] \ & = \ i  \left(\eta_{C[A}\mc J_{B]D}-\eta_{D[A}\mc J_{B]C}\right), \\ 
[\mc J_{AB},\mc P_C] \ & = \ i \, \eta_{C[A}\mc P_{B]}, \\ 
[\mc P_A,\mc P_B] \ & = \ 0.
\end{split}
\end{align}
A local gauge transformation, taking values in $G=SO(3,1)\ltimes\mathbb{P}^{3,1}$, can be split into
\begin{equation}\label{eq:Poincare-split}
g(x) \ = \ u(x)U(x), \qquad u \ = \ e^{-i\mathpzc{u}\mc P}, \qquad U \ = \ e^{-i\mathpzc{U}\mc J},
\end{equation}
s.t. $u(x)$ changes the zero section, i.e. changes the local identification of points of tangency at each spacetime event. The transformation law for the connection is then
\begin{equation}\label{eq:Poincare-gauge-transform}
\varpi \ \rightarrow \ \varpi' \ = \ g^{-1}(\varpi+d)g \ = \ U^{-1}u^{-1}(\om + \theta) u U + U^{-1}u^{-1}(du)U + U^{-1}dU,
\end{equation}
whose infinitesimal form is~\eqref{eq:Poincare-connection-transform-om}, \eqref{eq:Poincare-connection-transform-e}. The combined curvature transforms in the adjoint representation:
\begin{subequations}\label{eq:Poincare-adjoint-transform}
\begin{gather}
\mathpzc{F}[\varpi] \ \rightarrow \ \mathpzc{F}' \ = \ g^{-1}\mathpzc{F} g \ = \ U^{-1}u^{-1}(\Om+\Theta)u U, \\
\intertext{or, infinitesimally:}
\de \Om^{AB} \ = \ f^{AB}_{CD,EF}\Om^{CD}\mathpzc{U}^{EF}, \qquad \de \Theta^A \ = \ f^A_{CD,B} \left(\Om^{CD}\mathpzc{u}^B - \mathpzc{U}^{CD} \Theta^B\right).
\end{gather}
\end{subequations} 

However, the analogous quantity, comprised of the conjugate variables  
\begin{equation}\label{eq:Poincare-B-field}
\mathfrak{so}(3,1)\ltimes\mathfrak{p}^{3,1}\ \ni\ \mathpzc{B} \ := \ B + \be \ = \ B^{AB}\mc J_{AB} + \be^C\mc P_C, 
\end{equation}
demonstrates slightly different behaviour~\eqref{eq:null-B-transform},\eqref{eq:null-T-transform}, compensating for~\eqref{eq:Poincare-adjoint-transform} in order to make the action invariant. The inhomogeneity is transferred from $\be$ to $B$ part. For instance, if $\be$ is of the form~\eqref{eq:T-multiplier}, the corresponding addition would be $-\de_\mathpzc{u}\star B^{AB}=\frac12 (\mathpzc{u}\cdot \mathpzc{b}) \theta^{[A}\wedge \theta^{B]} + \mathpzc{b}^{[A}\theta^{B]}\wedge(\mathpzc{u}\cdot \theta)$. We are tempted to interpret the latter as being responsible for mixing up the simple bivectors -- the symmetry which is usually `broken' in order to obtain GR, in Pleba\'{n}ski constrained approach. In addition to the usual (internal) gauge transformations, the action~\eqref{eq:BF-Poincare} is also invariant w.r.t. the shifts
\begin{equation}\label{eq:null-topological-symmetry}
\de\varpi \ = \ 0, \qquad \de \mathpzc{B} \ = \ \left(D\Xi^{AB}-\theta^A\wedge\xi^B\right)\mc J_{AB} + D \xi^C\mc P_C \ \equiv \ D_\varpi(\Xi+\xi) ,
\end{equation}
which extends the usual `topological' BF symmetry~\footnote{The relation with diffeomorphism Lie derivatives is the same as in the ordinary BF theory, i.e. the above symmetry~\eqref{eq:null-topological-symmetry} is not `independent'.}. In fact, it is always possible to gauge away any local d.o.f.; however, if $\mc M$ is topologically non-trivial, $\varpi$ and $\mathpzc{B}$ can have non-trivial solutions globally (hence the name).

Despite the non-convential form of~\eqref{eq:null-B-transform}, we still have obtained the right connection~\eqref{eq:Poincare-connection} and the algebra~\eqref{eq:Poincare-algebra}. So that allows us to conclude that we have constructed a topological theory of the BF type for the Poincar\'{e} gauge group. We may guess that the departure of $\mathpzc{B}$ transformation properties from the adjoint~\eqref{eq:Poincare-adjoint-transform} is a forced decision, due to the degenerate nature of the Killing form: for the semidirect product algebra with the abelean ideal of translations it reduces to $\Tr \big[\mathrm{ad}_{(u,U)}\circ\mathrm{ad}_{(v,V)}\big]=2\,\Tr \big[\mathrm{ad}_U\circ \mathrm{ad}_V\big]\propto\eta_{A[C}^{\phantom{b}}\eta_{D]B}^{\phantom{b}}=-\frac14 f^{GH\phantom{]}}_{AB,EF}f^{EF\phantom{]}}_{CD,GH}$. Thus, the temptation to read the expression in~\eqref{eq:BF-Poincare} as the Cartan-Killing pairing, similar to the $\int\bra \mathpzc{B}, \mathpzc{F}\ket$ with semisimple group, is valid only with certain caveats. The contraction for $\be - \Theta$ is performed with the metric $\eta$, e.g. obtained from the $\mc J-\mc P$ vector couplings. We find this to be an interesting arena for the imposition of simplicity constraints.

\subsection{Hamiltonian and the gauge generator}
\label{subsec:Poincare-BF-hamiltonian}

Having characterized the system completely at the covariant Lagrangian level, we now pass to studying it using the canonical approach (of symplectic geometry). Ordinarily the Hamiltonian methods imply the manifest breaking of covariance by explicitly separating spatial from temporal field-components, and considering equal-time Poisson brackets. We would like to stress that the chosen preferred status of time coordinate in Hamiltonian analysis is not associated \textit{a priori} with an explicit separation of spacetime \textit{itself} into ``space and time'' (not at this stage at least), as suggested e.g. by the ADM change of coordinates and their geometrical interpretation. We therefore are being cautious with usage of such notions which are usually referred to as 3+1-decomposition, or slicing/splitting/foliation/etc. In particular, all 4d symmetries persist at the canonical level in the form of gauge generators, mapping solutions into solutions, as we show for this particular example.

Following the Dirac's general treatment of singular Lagrangian systems~\cite{Dirac1964lectures}, one starts by defining the conjugate momenta (for \textit{all} configuration variables):
\begin{subequations}
\begin{align}\label{eq:null-conjugate-momenta}
&&\Pi^a_{AB} &&& := & \frac{\de L_0}{\de \dot{\om}^{AB}_a}  &&& \equiv  && \left(\Pi^0_{AB},\Pi^i_{AB}\right) && \approx && \left(0,\frac12 \eps^{ijk}B_{jkAB}^{\phantom{jkA}}\right), && \\
&&\pi^a_A  &&& := & \frac{\de L_0}{\de \dot{\theta}^A_a} &&& \equiv && \left(\pi^0_A,\pi^i_A\right) && \approx && \left(0,\frac12 \eps^{ijk}\be_{jkA}^{\phantom{jkA}}\right), && \\
&&\Phi^{ab}_{AB}  &&& := & \frac{\de L_0}{\de \dot{B}^{AB}_{ab}} &&& \equiv && \left(\Phi^{0i}_{AB},\Phi^{jk}_{AB}\right) && \approx && \quad 0,  && \\
&&\phi^{ab}_A  &&& := & \frac{\de L_0}{\de \dot{\be}^A_{ab}} &&& \equiv && \left(\phi^{0i}_A,\phi^{jk}_A\right) && \approx  && \quad 0, &&
\end{align}
\end{subequations}
where dot denotes the time derivative of the generalized variables $\dot{q} := \pa_0 q$, and $\eps^{ijk} := \varepsilon^{0ijk},\, i,j,k=1,2,3$. We have the totality of primary constraints for the generalized coordinates $q$ and momenta $p$, in the sense that none of the velocities $\dot{q}$ enter the above relations and cannot be inverted -- the system defined by $L_0$ is, thus, maximally singular. The conjugate pairs $(\om,\Pi),(\theta,\pi),(B,\Phi),(\be,\phi)$ satisfy the canonical commutation relations (c.c.r.):
\begin{subequations}\label{eq:null-ccr}
\begin{align}
\big\{\om_a^{AB}(\mathbf{x}),\Pi^b_{CD}(\mathbf{y})\big\} \ & = \ \de^{[A}_{\,C}\de^{B]}_D \de_a^b \de(\mathbf{x},\mathbf{y}), & \big\{\theta_a^A(\mathbf{x}),\pi^b_B(\mathbf{y})\big\} \ & = \ \de^A_B \de_a^b \de(\mathbf{x},\mathbf{y}), \\
\big\{B_{ab}^{AB}(\mathbf{x}),\Phi^{cd}_{CD}(\mathbf{y})\big\} \ & = \ \de^{[A}_{\,C}\de^{B]}_D \de_{[a}^{\,c}\de_{b]}^d \de(\mathbf{x},\mathbf{y}), & \big\{\be_{ab}^A(\mathbf{x}),\phi^{cd}_B(\mathbf{y})\big\} \ & = \ \de^A_B \de_{[a}^{\,c}\de_{b]}^d \de(\mathbf{x},\mathbf{y}).
\end{align}
\end{subequations}

One then constructs the Hamiltonian, schematically $H(q,p)=p\,\dot{q}-L(q,\dot{q},p)=\phi\,\dot{q}+H_{\mathrm{c}}(q,p)$, following Dirac~\cite{Dirac1964lectures}, sometimes also called ``total'' in order to distinguish it from the ``canonical'' part $H_{\mathrm{c}}$, which does not contain primary constraints $\phi$. It is straightforward to verify that $H_{\mathrm{c}}$ is indeed explicitly independent of the velocities $\dot{q}$, which enter as undetermined functions in front of~$\phi$. By this straightforward procedure one gets for the Hamiltonian density:
\begin{align}
\ms H \ = & \ \Pi^a_{AB} \dot{\om}_a^{AB} + \pi^a_A \dot{\theta}_a^A + \Phi^{ab}_{AB}\dot{B}^{AB}_{ab} + \phi^{ab}_A\dot{\be}^A_{ab} - \ms L \nonumber \\
 = & \ \left(\Pi^i_{AB} -\frac12 \eps^{ijk}B_{jkAB}^{\phantom{jkAB}}\right)\dot{\om}_i^{AB} + \left(\pi^i_A -\frac12 \eps^{ijk}\be_{jkA}^{\phantom{jkA}}\right)\dot{\theta}_i^A + \Phi^{ij}_{AB}\dot{B}^{AB}_{ij} + \phi^{ij}_A\dot{\be}^A_{ij}  \nonumber \\
& \ \ + \Pi^0_{AB} \dot{\om}_0^{AB} + \pi^0_A \dot{\theta}_0^A + 2\Phi^{0i}_{AB}\dot{B}^{AB}_{0i} + 2\phi^{0i}_A\dot{\be}^A_{0i} + \ms H_{\mathrm{c}}, \label{eq:null-Hamiltonian}
\end{align}
the canonical part simply consists of spatial components of Lagrangian $\mathscr{H}_c=-\mathscr{L}\big|_{\dot{\varpi}=\dot{\mathpzc{B}}=0}$. It contains no momenta whatsoever at this stage, due to primary constraints, and we retain full covariance w.r.t. Lorentz indices. The form of $\mathscr{H}_c$ will be specified shortly, after reduction is made (compare with the detailed expressions unfold in the Hamiltonian analysis of closely related BFCG theory~\cite{Mikovic-etal2016Hamiltonian-BFCG-Poincare}).

Next we calculate the development of the primary constraints $\dot{\phi}=\{\phi,H\}\equiv\chi$ in order to find out the additional consistency requirements for them to preserve in time -- the secondary constraints $\chi\approx 0$. It may happen that some combinations of constraints form a 2nd class (sub)system, failing to commute. This is precisely our situation, since 
\begin{equation*}
\Big\{\Phi^{ij}_{AB},\Pi^{kCD} -\frac12 \eps^{klm}B_{lm}^{CD}\Big\} = \frac12 \eps^{ijk}\de^{[C}_{\,A}\de^{D]}_B , \qquad \Big\{\phi^{ij}_A,\pi^{kB} -\frac12 \eps^{klm}\be^B_{lm}\Big\} = \frac12 \eps^{ijk}\de^B_A.
\end{equation*}
Note that the second class nature of the initial Lagrangian $L_0$ is not the specialty of our Poincar\'{e} modification but is common to any BF theory in various spacetime dimensions. This apparent fact of the full-fledged Dirac's generalized Hamiltonian analysis is often overlooked in the canonical description of BF and related theories~\cite{CMPR2012BF+Immirzi-Hamiltonian,Buffenoir-etal2004Hamiltonian-Plebanski,MontesinosCuesta2008BF-Cartan-Ndim,Perez2013SF-review} (with rare notable exceptions, e.g.~\cite{Escalante-etal2012Hamiltonian-BF-complete,Mikovic-etal2016Hamiltonian-BFCG-Poincare}).

The presence of second class constraints signals about the degrees of freedom which are physically non-relevant, in our case these are spatial $B$ and $\be$ components. Their velocities are the Lagrange multipliers to be determined by requiring the time preservation of the corresponding 2nd class set -- this allows us to express them in terms of other variables (Lagrangian equations of motion):
\begin{equation}\label{eq:Lagrangian-velocities}
\begin{aligned}
\Big\{H, \Pi^i_{AB} -\frac12 \eps^{ijk}B_{jkAB}\Big\} \ & = \ \frac12 \eps^{ijk} \Big(\dot{B}_{jkAB}^{\phantom{jkAB}}+\om_{0A}^{\phantom{0A}C}B_{jkCB}^{\phantom{jkCB}}+\om_{0B}^{\phantom{jB}C}B_{jkAC}^{\phantom{jkAC}} -\theta_{0[A}\be_{B]jk} \\
& \: \qquad -2\Big(\pa_j^{\phantom{I}} B_{0kAB}^{\phantom{okAB}}+\om_{jA}^{\phantom{jA}C}B_{0kCB}^{\phantom{0kCB}}+\om_{jB}^{\phantom{jA}C}B_{0kAC}^{\phantom{0kAC}}-\theta_{j[A}\be_{B]0k}\Big)\Big)\ = \ 0 , \\
\Big\{H, \pi^i_A -\frac12 \eps^{ijk}\be_{jkA}\Big\} \ & = \ \frac12 \eps^{ijk} \Big(\dot{\be}_{jkA}^{\phantom{jkA}}+\om_{0A}^{\phantom{0A}B}\be_{jkB}^{\phantom{jkB}}-2\Big(\pa_j^{\phantom{I}} \be_{0kA}^{\phantom{0kA}}+\om_{jA}^{\phantom{jA}B}\be_{0kB}^{\phantom{0kB}}\Big) \Big) \ = \ 0 , \\
\Big\{\Phi^{ij}_{AB},H\Big\} \ & = \ \frac12 \eps^{ijk} \Big(\dot{\om}_{kAB}^{\phantom{kAB}} +\om_{0A}^{\phantom{0A}C}\om_{kCB}^{\phantom{kCB}} -\pa_k^{\phantom{I}}\om_{0AB}^{\phantom{0AB}} -\om_{kA}^{\phantom{kA}C}\om_{0CB}^{\phantom{0CB}}\Big) \ = \ 0 , \\
\Big\{\phi^{ij}_A,H\Big\} \ & = \ \frac12 \eps^{ijk} \Big(\dot{\theta}_{kA}^{\phantom{kA}} +\om_{0A}^{\phantom{0A}B}\theta_{kB}^{\phantom{kB}} -\pa_k^{\phantom{I}} \theta_{0A}^{\phantom{0A}} -\om_{kA}^{\phantom{kA}B}\theta_{0B}^{\phantom{0B}}\Big) \ = \ 0 .
\end{aligned}
\end{equation}

One can reduce the system by solving the second class constraints as strong equations. The formal procedure includes passing to the Dirac brackets in order not to sum over variables, which have been thrown away (cf.~\cite{Escalante-etal2012Hamiltonian-BF-complete}). In our case the constraints are of special type, such that we can make a shortcut and simply solve for the spatial $B$ and $\be$ components (together with their identically vanishing momenta), since these just serve the purpose of identifying (the spatial part of) the $\varpi$-connection's conjugate momenta. The rest of the canonical commutation relations are unaltered, as can be easily verified, and one is left in~\eqref{eq:null-Hamiltonian} with the last line $\ms H'$, where prime now signals that $B,\be$ have been solved for $\Pi,\pi$. The rest of the primary constraints are first class and all commute among themselves. They give rise to the secondary
\begin{subequations}\label{eq:secondary-constraints}
\begin{align}
&&&&\Big\{\Pi^0_{AB},H\Big\} &&& \equiv && \chi^0_{AB} && = && \mc D_i^{\phantom{I}}\Pi^i_{AB}- \theta^{\phantom{A}}_{i[A}\pi^i_{B]} && \approx && 0, &&&& \\
&&&&\Big\{\pi^0_A,H\Big\}  &&& \equiv && \chi^0_A && = && \mc D_i^{\phantom{I}}\pi^i_A && \approx && 0, &&&& \\
&&&&\Big\{\Phi^{0i}_{AB},H \Big\} &&& \equiv && \chi^{0i}_{AB} &&  = && \frac12 \eps^{ijk}\left(\pa_j^{\phantom{I}} \om_{kAB}^{\phantom{kAB}} + \om_{jA}^{\phantom{jA}C}\om_{kCB}^{\phantom{kCB}} \right) && \approx && 0, &&&& \\
&&&&\Big\{\phi^{0i}_A,H \Big\} &&& \equiv && \chi^{0i}_A && = && \frac12 \eps^{ijk}\left(\mc D^{\phantom{I}}_j \theta_{kA}^{\phantom{kA}} \right) && \approx && 0. &&&&
\end{align}
\end{subequations}
The derivative $\mc D_i^{\phantom{I}}$ is taken w.r.t. the spatial connection $\om_i^{AB}$~\footnote{The constraint $\chi^0_{AB}$ is again the modified Gauss law, whose appearance was anticipated in~\cite{GielenOriti2010Plebanski-linear}. Here we encounter no need to artificially enlarge the phase space, which follows naturally from the covariant action, together with the nice transformation properties.}. The canonical part of the Hamiltonian takes form 
\begin{equation}\label{eq:null-canonical-Hamiltonian}
-\ms H_{\mathrm{c}}' \ = \ \om_0^{AB}\chi^0_{AB} + \theta_0^A\chi^0_A +2B^{AB}_{0i}\chi^{0i}_{AB} +2\be^A_{0i}\chi^{0i}_A -\pa_i^{\phantom{|}}\! \left(\om_0^{AB}\Pi^i_{AB} +\theta_0^A\pi^i_A\right).
\end{equation}
The bulk contribution to $\ms H'$ vanishes as the sum of (the primary as well as secondary) constraints. We did not specify any form of the boundary conditions and kept the surface term explicit. A good cross check is the consistency between the Hamiltonian $\dot{f} = \{f,H'\}$ and the Lagrangian~\eqref{eq:Lagrangian-velocities} equations of motion, once the solution to second class constraints is taken into account. We warn the reader not to discard the primary constraints from the outset. Although the present case of reduction is very simple, in general, it may affect the symplectic structure of the original action. Moreover, the primary constraints are essential for the equivalence between the Lagrangian and Hamiltonian formulations. By keeping only the (reduced) canonical part $H_{\mathrm{c}}'$ we cannot even address the gauge transformations on the full phase-space -- only the spatial ones. 

The completion of Dirac's procedure consists in proving that $\dot{\chi}\approx 0$ are conserved. This follows from the closure of the algebra:
\begin{equation}
\begin{aligned}\label{eq:null-constraint-algebra}
\left\{\chi^0_{AB},\chi^0_{CD}\right\} \ & = \ \chi^0_{C[A}\eta_{B]D}^{\phantom{0}}-\chi^0_{D[A}\eta_{B]C}^{\phantom{0}}, \\
\left\{\chi^0_{AB},\chi^0_C\right\} \ & = \ \chi^0_{[A}\eta_{B]C}^{\phantom{0}}, \\
\left\{\chi^0_{AB},\chi^{0i}_{CD}\right\} \ & = \ \chi^{0i}_{C[A}\eta_{B]D}^{\phantom{0i}}-\chi^{0i}_{D[A}\eta_{B]C}^{\phantom{0i}}, \\
\left\{\chi^0_{AB},\chi^{0i}_C\right\} \ & = \ \chi^{0i}_{[A}\eta_{B]C}^{\phantom{0i}}, \\
\left\{\chi^0_A,\chi^{0i}_B\right\} \ & = \ -\chi^{0i}_{AB},
\end{aligned}
\end{equation}
the rest of the commutators being trivially zero. 

It is worth at this point to perform the physical degrees of freedom count, in order to make sure that the theory is indeed topological. Starting from the initial phase space of dimensionality (which we denote by putting variables in brackets) $\big[\om^{AB}_a\big]+\big[\Pi^a_{AB}\big]+\big[\theta^A_a\big]+\big[\pi^a_A\big]+\big[B^{AB}_{ab}\big]+\big[\Phi^{ab}_{AB}\big]+\big[\be^A_{ab}\big]+\big[\phi^{ab}_A\big]=200$, we eliminate some variables through the strong second class equalities $\big[B^{AB}_{ij}\big]+\big[\Phi^{ij}_{AB}\big]+\big[\be^A_{ij}\big]+\big[\phi^{ij}_A\big]=60$. Finally, we perform the symplectic reduction as follows: put the system on the surface of first class constraints, then gauge away the redundant modes by factoring out the action of the first class constraints. In effect, we subtract twice the amount of all the first class constraints, taking in account that some of them are reducible (namely, the secondary constraints with the vector index are not independent, but related through the spatial Bianchi's identities): $2\cdot\big(\big[\Pi^0_{AB}\big]+\big[\pi^0_A\big]+\big[\Phi^{0i}_{AB}\big]+\big[\phi^{0i}_A\big]+\big[\chi^0_{AB}\big]+\big[\chi^0_A\big]+\big[\chi^{0i}_{AB}\big]+\big[\chi^{0i}_A\big]-\big[\mc D_i^{\phantom{I}} \chi^{0i}_{AB}\big]-\big[\mc D_i^{\phantom{I}} \chi^{0i}_A \- +\theta^B_i\chi^{0i}_{BA}\big]\big)=140$. We conclude that the theory is devoid of local degrees of freedom, the only relevant ones being that of global nature, those coming from non-trivial topologies. This makes it a potential candidate for spinfoam quantization.

\paragraph{The gauge generator.} If one expects the Hamiltonian picture to represent the original theory, then it has to be shown that it correctly reproduces results of the manifestly covariant approach, in particular, the gauge symmetries, in the form of canonical transformations. The Dirac's old conjecture that all first class constraints do generate such a transformations was formalized later by Castellani~\cite{Castellani1982Hamiltonian-generator} and others into a precise algorithm. This procedure defines the gauge generator (for arbitrary functions of time $\varepsilon(t)$)
\begin{equation}\label{eq:Castellani-generator}
\mc G(t) \ = \ \sum_{n=0}^N \sum_\al \varepsilon^{(n)}_\al \mc G^\al_{(N-n)}, \qquad  \varepsilon^{(n)}_\al =\frac{d^n}{dt^n}\varepsilon_\al^{\phantom{\al}},
\end{equation}
through the chains of first class constraints, unambiguously constructed once the set of primary ones (first class) $\{\al\}$ is given. The multi-index $\al$ is linked to the tensorial nature of transformations, while $(N-n)$ gives the generation number (primary/secondary/tertiary/etc.). As a by-product, knowing the derivative order of gauge transformations, one can predict the overall number $N$ of generations of constraints, and vice versa. 

The chains $\mc G^\al_{(N-n)}$ in~\eqref{eq:Castellani-generator} are constructed iteratively as follows:
\begin{equation}\label{eq:Castellani-chains}
\begin{aligned}
\mc G^\al_{(0)} \ & = \ \text{primary}, \\
\mc G^\al_{(1)} + \big\{\mc G^\al_{(0)},H\big\}\ & = \ \text{primary}, \\
& \ \vdots \\
\mc G^\al_{(N)} + \big\{\mc G^\al_{(N-1)},H\big\}\ & = \ \text{primary}, \\
\big\{\mc G^\al_{(N)},H\big\}\ & = \ \text{primary}.
\end{aligned}
\end{equation}
In the present situation, the primary ones are
\begin{equation}
\mc G^0_{(0)AB} = \Pi^0_{AB} , \qquad \mc G^0_{(0)A} = \pi^0_A , \qquad \mc G^{0i}_{(0)AB} = \Phi^{0i}_{AB}, \qquad \mc G^{0i}_{(0)A} = \phi^{0i}_A,
\end{equation}
and the procedure terminates already at the secondary $n=0,1$:
\begin{equation}
\mc G^\al_{(1)}(\mathbf{x}) \ = \ -\chi^\al(\mathbf{x}) + \int d^3\mathbf{y} \, \mathpzc{A}^\al_{\phantom{\al}\be}(\mathbf{x},\mathbf{y}) \phi^\be(\mathbf{y}),
\end{equation}
where coefficient kernels $\mathpzc{A}^\al_{\phantom{\al}\be}$ are fixed by the last requirement in~\eqref{eq:Castellani-chains} to close onto the primary constraint surface (we have the identical zero due to commutation $\{\phi,\phi\}=\{\phi,\chi\}=0$). Straightforward calculation gives the total (smeared) generator
\begin{equation}\label{eq:null-generator}
\mc G(\mathpzc{U},\mathpzc{u},\Xi,\xi) \ = \ -\ms J(\mathpzc{U}) - \ms P(\mathpzc{u}) + \ms F(\Xi) + \ms T(\xi)
\end{equation}
as a combination of elementary ones:
\begin{equation}\label{eq:null-generators-elementary}
\begin{aligned}
\ms J(\mathpzc{U}) \ & = \ \int \dot{\mathpzc{U}}^{AB}\Pi^0_{AB} - \mathpzc{U}^{AB}\left(2\om^{\phantom{0A}C}_{0A}\Pi^0_{CB}-\theta^{\phantom{0A}}_{0A}\pi^0_B +4B^{\phantom{0iA}C}_{0iA}\Phi^{0i}_{CB}-2\be^{\phantom{0iA}}_{0iA}\phi^{0i}_B + \chi^0_{AB}\right), \\
\ms P(\mathpzc{u}) \ & = \ \int \dot{\mathpzc{u}}^A\pi^0_A - \mathpzc{u}^A\left(\om^{\phantom{0A}B}_{0A}\pi^0_B +2\be^{B}_{0i}\Phi^{0i}_{BA} + \chi^0_A\right), \\
\ms F(\Xi) \ & = \ \int \dot{\Xi}^{AB}_i\Phi^{0i}_{AB} - \Xi^{AB}_i\left(2\om^{\phantom{0A}C}_{0A}\Phi^{0i}_{CB} + \chi^{0i}_{AB}\right), \\
\ms T(\xi) \ & = \ \int \dot{\xi}^A_i\phi^{0i}_A - \xi^A_i\left(\om^{\phantom{0A}B}_{0A}\phi^{0i}_B + \theta^B_0\Phi^{0i}_{BA} + \chi^{0i}_A\right).
\end{aligned}
\end{equation}
This generalizes the result for the canonical gauge generator of $SO(3,1)$ BF theory, reported in~\cite{Escalante-etal2012Hamiltonian-BF-complete}. The construction provides the correct transformation properties via
\begin{equation}
\de f \ = \ \{f, \mc G \},
\end{equation}
mapping solutions into solutions (gauge symmetry). Unlike the secondary constraints~\eqref{eq:secondary-constraints}, it acts on the full phase-space of the theory:
\begin{align*}
\de\om^{AB}_0 \ & = \ - \left(\dot{\mathpzc{U}}^{AB}+\om_{0\phantom{A}C}^{\phantom{0}A}\mathpzc{U}^{CB}+\om_{0\phantom{B}C}^{\phantom{0}B}\mathpzc{U}^{AC}\right), \\ 
\de \om^{AB}_i \ & = \ - \left(\pa_i^{\phantom{I}} \mathpzc{U}^{AB}+\om_{i\phantom{A}C}^{\phantom{i}A}\mathpzc{U}^{CB}+\om_{i\phantom{B}C}^{\phantom{i}B}\mathpzc{U}^{AC}\right), \\ 
\de \theta^A_0 \ & = \ - \left(\dot{\mathpzc{u}}^A+\om_{0\phantom{A}B}^{\phantom{0}A}\mathpzc{u}^B\right) +\mathpzc{U}^A_{\phantom{A}B}\theta^B_0, \\
\de \theta^A_i \ & = \ - \left(\pa_i^{\phantom{I}} \mathpzc{u}^A+\om_{i\phantom{A}B}^{\phantom{i}A}\mathpzc{u}^B\right) +\mathpzc{U}^A_{\phantom{A}B}\theta^B_i, \\
\de \be^A_{0i} \ & = \ \mathpzc{U}^A_{\phantom{A}B}\be^B_{0i} + \left(\dot{\xi}^A+ \om_{0\phantom{A}B}^{\phantom{0}A}\xi^B_i\right), \\
\de B^{AB}_{0i} \ & = \ \mathpzc{U}^A_{\phantom{A}C}B^{CB}_{0i} + \mathpzc{U}^B_{\phantom{B}C}B^{AC}_{0i} + \left(\dot{\Xi}^{AB}_i+\om_{0\phantom{A}C}^{\phantom{0}A}\Xi^{CB}_i+\om_{0\phantom{B}C}^{\phantom{0}B} \Xi^{AC}_i\right)-\theta^{[A}_0\xi^{B]}_i, \\
\de \Pi^0_{AB} \ & = \ \mathpzc{U}^{\phantom{A}C}_A\Pi^0_{CB}+\mathpzc{U}^{\phantom{B}C}_B\Pi^0_{AC}- \mathpzc{u}^{\phantom{A}}_{[A}\pi^i_{B]} - \Xi_{iA}^{\phantom{iA}C}\Phi^{0i}_{CB}-\Xi_{iB}^{\phantom{iB}C}\Phi^{0i}_{AC} +\xi^{\phantom{iA}}_{i[A}\phi^{0i}_{B]}, \\
\de \Pi^i_{AB} \ & = \ \mathpzc{U}^{\phantom{A}C}_A\Pi^i_{CB}+\mathpzc{U}^{\phantom{B}C}_B\Pi^i_{AC}- \mathpzc{u}^{\phantom{A}}_{[A}\pi^i_{B]} +\eps^{ijk}\left(\pa_j^{\phantom{I}} \Xi_{kAB}^{\phantom{kAB}} + \om_{jA}^{\phantom{jA}C}\Xi_{kCB}^{\phantom{kCB}} + \om_{jB}^{\phantom{jB}C}\Xi_{kAC}^{\phantom{kAC}}+\theta^{\phantom{jA}}_{j[A}\xi_{B]k}^{\phantom{Bk}}\right), \\
\de\pi^0_A \ & = \ \mathpzc{U}_A^{\phantom{A}B}\pi^0_B + \xi^B_i\Phi^{0i}_{BA}, \\
\de\pi^i_A \ & = \ \mathpzc{U}_A^{\phantom{A}B}\pi^i_B + \eps^{ijk}\left(\pa_j^{\phantom{I}}\xi_{kA}^{\phantom{kA}} + \om_{jA}^{\phantom{jA}B}\xi_{kB}^{\phantom{kB}}\right), \\
\de \phi^{0i}_A \ & = \ \mathpzc{U}_A^{\phantom{A}B}\phi^{0i}_B + \mathpzc{u}^B\Phi^{0i}_{BA}, \\
\de\Phi^{0i}_{AB} \ & = \ \mathpzc{U}_A^{\phantom{A}C}\Phi^{0i}_{CB}+\mathpzc{U}_B^{\phantom{B}C}\Phi^{0i}_{AC}.
\end{align*}
The correct covariant expressions~\eqref{eq:BF-Poincare-gauge-transform} for all the Lagrangian field components (spatial as well as temporal; using also the 2nd class relations $\de B^{AB}_{ij}=\eps_{ijk}^{\phantom{I}}\de\Pi^{kAB},\de \be^A_{ij}=\eps_{ijk}^{\phantom{I}}\de\pi^{kA}$) are reproduced within the Hamiltonian framework, thus exhibiting the equivalence between the two pictures. 

Using the Jacobi identity, the commutator between the two consecutive transformations is given:
\begin{equation}\label{eq:gauge-algebra}
(\de_1\de_2-\de_2\de_1)f \ = \ \{\{\mc G_1,\mc G_2\}, f\}.
\end{equation}
Its elementary constituents realize the generalized matrix commutators:
\begin{align*}
\{\ms J(\mathpzc{U}_1), \ms J(\mathpzc{U}_2)\} \ & = \ -\ms J([\mathpzc{U}_1,\mathpzc{U}_2]), & [\mathpzc{U}_1,\mathpzc{U}_2]^{AB} \ & = \  \mathpzc{U}^A_{1\; C}\mathpzc{U}_2^{CB} - \mathpzc{U}^B_{1\; C}\mathpzc{U}_2^{CA} ,\\
\{\ms J(\mathpzc{U}), \ms P(\mathpzc{u})\} \ & = \ -\ms P(\mathpzc{U}\rt\mathpzc{u}), & (\mathpzc{U}\rt\mathpzc{u})^{A\phantom{B}} \ & = \ \mathpzc{U}^A_{\phantom{A}B}\mathpzc{u}^B,\\
\{\ms J(\mathpzc{U}), \ms F(\Xi)\} \ & = \ -\ms F([\mathpzc{U},\Xi]), & [\mathpzc{U},\Xi]^{AB}_i \ & = \ \mathpzc{U}^A_{\phantom{A}C}\Xi^{CB}_i-\mathpzc{U}^B_{\phantom{B}C}\Xi^{CA}_i,\\
\{\ms J(\mathpzc{U}), \ms T(\xi)\} \ & = \ -\ms T(\mathpzc{U}\rt\xi), & (\mathpzc{U}\rt\xi)^{A\phantom{B}}_i \ & = \ \mathpzc{U}^A_{\phantom{A}B}\xi^B_i, \\
\{\ms P(\mathpzc{u}), \ms T(\xi)\} \ & = \ +\ms F([\mathpzc{u},\xi]), & [\mathpzc{u},\xi]^{AB}_i \ & = \ \mathpzc{u}_{\phantom{i}}^{[A}\xi^{B]}_i.
\end{align*}

\newpage

We use a chance to comment here on the relation between the dynamics and gauge in reparametrization invariant systems (cf. ``the problem of time'').
The Hamiltonian -- generator of time evolution -- in such model is a combination of first class constraints, which are also known to generate the gauge transformations (i.e. ``unphysical'' changes in the description of the system). Working on the full phase space allows to disentangle these notions: the specific combinations of first class constraints are different for two objects $\ms H$ and $\mc G$; the key role is played by the primary sub-set. 

One usually defines the notion of Dirac observables w.r.t. individual constraints $\{\chi,f\}=0$ (often disregarding the primary $\phi\approx0$, working on the smaller phase space). In the quantum theory, one represents the canonical variables via operators on the appropriate Hilbert space $\mc{H}$ of states of the system, the Poisson (Dirac) brackets being replaced by a commutator $[\;,\,]=i\{\;,\,\}$. For instance, in our example:
\begin{equation}
\ms J(\mathpzc{U}) \ \mapsto \ \mathpzc{U}\cdot\hat{\mc J}, \qquad \ms P(\mathpzc{u}) \ \mapsto \ \mathpzc{u}\cdot\hat{\mc P}
\end{equation}
are the elements of the (local) Poincar\'{e} algebra~\eqref{eq:Poincare-algebra}. The Dirac prescription then consists in imposing on states $\hat{\chi}_\al|\Psi\ket =0$ individually for each $\al$, which are then consistent for the first class system.

With the distinction just pointed out between $\ms H$ and $\mc G$ on the full phase space, the function of canonical variables may satisfy two a priori distinct conditions: $\{\ms H,F\} \ = \ 0$ and/or $\{\mc G,F\} \ = \ 0$. The first can be thought of as characterizing ``evolving constants of motion'', i.e. uniquely associated with the state -- solution of the e.o.m. The state itself, however, is not uniquely defined by the Cauchy data and depends on the arbitrary functions, entering the Hamiltonian (in our reduced case, the velocities of $(\theta,\om,\be,B)$ time components, associated with the 1st class primary constraints on momenta, are not defined by the evolution equations). Thus, the first type functionals may depend on the gauge choice for particular Hamiltonian, whereas the second condition then characterizes ``gauge-invariant'' functionals. 

Requiring the time preservation of vanishing of currents $\mc G^\al$ on every hypersurface, by construction we have then $0=\pa_t\mc G^\al =\{\ms H,\mc G^\al\}$, and the state has to satisfy Hamiltonian e.o.m. So that symmetry generators provide an example of the first type functionals. The commutation relations~\eqref{eq:gauge-algebra} express the fact that the canonical (pre-)symplectic structure is degenerate on the constraint surface. One then passes to the quotient w.r.t. the gauge directions, considering the gauge equivalent classes of solutions as ``physical'' states. The construction of the appropriate observable algebra is of primary importance for the quantization, especially in gravitational theories, so the Dirac's ``rule of thumb'' for \emph{all} 1st class constraints should be applied with certain care.

\pagebreak

\section{(Dual) linear simplicity constraints using co-frames}
\label{sec:continuum-linear}

Now, when we have the (co-)frames $\theta$ at our disposal among the legitimate dynamical variables, it is straightforward to implement the simplicity of the bivectors, in order to reproduce gravity sub-sector. Multiple choices of how to do this are conceivable. First of all, one can simply replace $B \ra\star \theta\wedge \theta$ directly in the action integral:
\begin{equation}\label{eq:EC+0torsion}
S[\theta,\om,\be] \ = \ \int \frac12 \varepsilon_{ABCD}^{\phantom{ABCD}}\,\theta^A\wedge \theta^B\wedge \Om^{CD}+ \be_A^{\phantom{A}}\wedge \Theta^A.
\end{equation}
Secondly, one could try to achieve the same effect via the Lagrange multipliers approach, imposing the simplicity in its most direct sense:
\begin{equation}\label{eq:Lambda-constraints}
S[\mathpzc{B},\varpi,\La] \ = \ S_0[\mathpzc{B},\varpi] + \int \La^{AB}
\wedge \left(B_{AB}^{\phantom{AB}}-\frac12\eps_{ABCD}^{\phantom{ABCD}}\,\theta^C\wedge \theta^D\right),
\end{equation}
with the free independent multiplier 2-forms $\La$.

It turns out that the simplicity constraints can also be put into form, linear in both $B$ and $\theta$, which is more in the vein of current Spin Foam models. In order to stay self-contained and explicit, let us formulate the following

\begin{lemma}[linear simplicity]
Provided that the tetrad field is non-degenerate, and hence the map $\theta$ is invertible, the bivector field $B$ is simple if and only if either of the two equivalent sets of constraints is satisfied:
\begin{equation}\label{eq:linear-simplicity-dual}
  \left\{\begin{aligned}
  \ast B_{AB}^{ab}\, \theta^B_c \ &= \ 0 &&&& \forall \, c\notin\{a,b\}, \\
  \ast B_{AB}^{ab}\, \theta^B_b \ = \ \ast B_{AB}^{ac}\, \theta^B_c \ & = \ \ast B_{AB}^{ad}\, \theta^B_d &&&& \forall \, a\notin\{b,c,d\},
  \end{aligned}\right.
  \quad\Leftrightarrow\quad
  \left\{\begin{aligned}
  B_{ABab}^{\phantom{AB}}\, \theta^B_c \  & = \ 0  &&&& \forall \, c\in\{a,b\}, \\
  B_{ABa(b}^{\phantom{AB}}\, \theta^B_{c)} \  & = \ 0 &&&& \forall \, c\notin\{a,b\}.
  \end{aligned}\right.
\end{equation}
\end{lemma}

\begin{proof}
First, it is straightforward to verify that the two systems imply each other by noting that for some fixed $a\neq b$:
\begin{equation*}
\ast B_{AB}^{ab}\, \theta^B_{b'} \ = \ \sum_{cd} \frac12 \varepsilon^{abcd}B_{ABcd}^{\phantom{A}} \, \theta^B_{b'} \ = \ \varepsilon^{abcd}B_{ABcd}^{\phantom{A}} \, \theta^B_{b'} \quad \text{(no sum over $[cd]\neq[ab]$)}.
\end{equation*}
Assuming the conditions on the one side, the other one follows from here. Notice also that in our notation the dual becomes $\ast B_{AB}^{ab}=(\Si^{-1})_{AB}^{ab}$ on the constrained surface, where $\Si=\theta\wedge \theta$.

To show the necessity of these conditions for simplicity of bivectors, it is enough to cast $B=\star \theta\wedge \theta$ into form~\footnote{Having in mind the possible interpretation in terms of discrete geometry, \eqref{eq:3volume-simplicity} can be suggestively referred to as ``pyramid'', or ``3-volume'' form of simplicity constraints -- namely, by the nature of the object appearing on the right, and since the whole formula can be read as the expression for the volume of a pyramid with the base $\Si$ and height $\theta$.}
\begin{subequations}\label{eq:pyramid-simplicity}
\begin{equation}\label{eq:3volume-simplicity}
B_{ABab}^{\phantom{AB}}\, \theta^B_c \ = \ \varepsilon_{ABCD}^{\phantom{ABCD}}\, \theta^B_a \theta^C_b \theta^D_c \ = \ (\det \theta)\, \theta^d_A \, \varepsilon_{dabc}^{\phantom{a}},
\end{equation}
or, equivalently:
\begin{equation}\label{eq:3volume-simplicity-dual}
\ast B_{AB}^{ab}\, \theta^B_c \ = \ (\det \theta) \, (\theta^a_A \de^b_c - \theta^b_A \de^a_c).
\end{equation}
\end{subequations}
Observing that the l.h.s. is already linear leads us to the (dual) analogue of `cross-simplicity' constraint in the first line of~\eqref{eq:linear-simplicity-dual}, when the r.h.s. in~\eqref{eq:pyramid-simplicity} is zero. In turn, the non trivial expression on the right in~\eqref{eq:pyramid-simplicity} restricts the l.h.s. in~\eqref{eq:3volume-simplicity} to be totally antisymmetric in $[abc]$, whereas it is independent of $b=c\neq a$ in~\eqref{eq:3volume-simplicity-dual}, leading to the second line of~\eqref{eq:linear-simplicity-dual}, respectively (no sum over spacetime indexes).

In order to demonstrate that the conditions~\eqref{eq:linear-simplicity-dual} are also sufficient, one follows the same reasoning as in~\cite{GielenOriti2010Plebanski-linear}. Namely, the generic bivector field can be expanded over the basis, spanned by the skew-symmetric products of $e$:
\begin{equation*}
\ast B_{AB}^{ab} \ = \ G^{ab}_{cd} \, \theta^c_{[A}\theta^d_{B]}, \qquad G^{ab}_{cd} \ = \ G^{[ab]}_{[cd]},
\end{equation*}
which after substitution into the first line of~\eqref{eq:linear-simplicity-dual} leads to
\begin{equation*}
\ast B_{AB}^{ab} \ = \ G_{ab}^{\phantom{a}} \, \theta^a_{[A}\theta^b_{B]} \qquad \text{(no sum over $ab$)}.
\end{equation*}
The individual normalization coefficients have to  satisfy symmetry $G_{ab}^{\phantom{a}} = G_{ba}^{\phantom{a}}$, $G_{aa}^{\phantom{a}} = 0$, but apart from that can be arbitrary. It is only after substitution of this ansatz into the second line of~\eqref{eq:linear-simplicity-dual} that we get the restriction
\begin{equation*}
G_{ab}^{\phantom{a}} \ = \ G_{ac}^{\phantom{a}} \ = \ G_{ad}^{\phantom{a}} \qquad \forall \, a\notin\{b,c,d\},
\end{equation*}
leading to the equality among all $G$'s. Thus, the $B$ is simple up to an overall factor, which can be `eaten' by appropriate normalization.
\end{proof}
  
The continuous formulation that we are advocating for is somewhat different from that of Gielen-Oriti's linear proposal~\cite{GielenOriti2010Plebanski-linear}, which uses 3-forms $\vartheta=\star \theta\wedge \theta\wedge \theta$, and bivectors $\Si = \theta\wedge \theta$ as independent variables, but rather represents its dual version. Before introducing the action principle, and in order to make closer contact between the two formulations, we first recall the corresponding constraint term in~\cite{GielenOriti2010Plebanski-linear} and notice that this can be rewritten as
\begin{equation}\label{eq:lin-simpl-cont}
\int d^4x \ \widetilde{\Xi}^{[ab][cde]}_A\Si^{AB}_{ab}\vartheta_{Bcde}^{\phantom{Bcde}}  \ = \ \int d^4x \ \Xi^{[ab]}_{Ac}\Si^{AB}_{ab}\tilde{\vartheta}^c_B.
\end{equation}
One can choose to work either with 3-forms~\eqref{eq:3-forms} or, equivalently, their dual densitiezed vectors:
\begin{equation}
\tilde{\vartheta}^a_A \ = \ \frac{1}{3!}\varepsilon^{abcd}\vartheta_{Abcd}^{\phantom{A}} \ = \ (\det \theta^B_b)\, \theta^a_A.
\end{equation}
Correspondingly, the Lagrange multipliers $\Xi^{[ab]}_{Ac}$ should be exact tensors, s.t. $\widetilde{\Xi}^{[ab][cde]}_A = \frac{1}{3!}\varepsilon^{cdef\phantom{]}}_{\phantom{Af}}\Xi^{[ab]}_{Af}$ -- tensor densities. The somewhat convoluted index symmetries that $\widetilde{\Xi}$ has to satisfy can be restated as the traceless condition on $\Xi^{[ab]}_{Ab}=0$, which upon variation then leads to the appearance of non-trivial Kronecker deltas on the right:
\begin{equation}\label{eq:Xi-variation}
\de \Xi \quad \Rightarrow \quad \tilde{\vartheta}^c_A\Si^{AB}_{ab} \ = \ \de^c_a v^B_b - \de^c_b v^B_a \qquad \text{for some} \ v^B_b.
\end{equation}
The antisymmetry in $[AB]$ and the tensorial nature of $\Si$ leaves no choice other than $v^A_a\propto \tilde{\vartheta}^A_a=\theta\, \theta^A_a$, and we get the simplicity up to an overall normalization, which is irrelevant. Applying the Hodge dual $\ast$, one restates this in terms of 3-forms, resulting from variation w.r.t. $\widetilde{\Xi}$, correspondingly:
\begin{equation}\label{eq:lin-simpl-dual}
\Si^{AB}_{ab}\vartheta_{Bcde}^{\phantom{A}} \ = \ v^A_a\varepsilon_{bcde}^{\phantom{A}} - v^A_b \varepsilon_{acde}^{\phantom{A}},
\end{equation}
which are essentially the original Gielen-Oriti's constraints.

Comparing~\eqref{eq:Xi-variation} with~\eqref{eq:3volume-simplicity-dual}, and juxtaposing them against the constraint term in \eqref{eq:lin-simpl-cont}, then suggests the respective least action principle in terms of dual variables $B\leftrightarrow\Si$ and $\theta\leftrightarrow\vartheta$, correspondingly:
\begin{equation}\label{eq:Poincare-Plebanski}
S_{\mathrm{PP}}[\mathpzc{B},\varpi,\Xi] \ := \ S_0[\mathpzc{B},\varpi] + \int 
\left(\Xi^A\lrcorner \, \theta^B\right)\wedge B_{AB}^{\phantom{AB}},
\end{equation}
which we coined, referring to its gauge group, the Poincar\'{e}-Pleba\'{n}ski formulation (although such a name might be as well attributed either to the ``$\La$-version'', or essentially to any formulation of this flavour). The $4\times4\times6=96$ Lagrange multipliers constitute the tangent $T\mc M$-valued 2-forms, that is
\begin{equation}\label{eq:PP-Lagrange-multipliers}
\Xi^A \ = \ \frac12 \Xi^{Ac}_{ab}\, dx^a\wedge dx^b \otimes \pa_c^{\phantom{a}}
\end{equation}
are the sections of the fiber bundle $\bigwedge^2 T^\ast\mc M\bigotimes T\mc M$. In the constraint term of the action~\eqref{eq:Poincare-Plebanski} they contract with tetrad 1-forms using the pairing $dx^a_{\phantom{b}}\lrcorner\,\pa_b^{\phantom{a}}=\de^a_b$ in the tangent vector index: $\Xi^A\lrcorner \, \theta^B=\frac12 \Xi^{Ac}_{ab}\theta^B_c dx^a\wedge dx^b$. The $\Theta$'s are restricted to be traceless $\Xi^{Ab}_{[ab]}=0$, that is possess the components of the form $\Xi^{Ac}_{ab}+\frac23\de^{\,c}_{[a}\Xi^{Ad}_{b]d}$; we can formulate this in the coordinate independent way as the full contraction with the canonical tangent-valued form on $\mc M$ being zero:
\begin{subequations}
\begin{equation}\label{eq:traceless-multipliers}
\bar{\theta}_{\mc M} \ := \ dx^a\otimes \pa_a, \qquad \bar{\theta}_{\mc M}\lrcorner\, \Xi^A \ = \ 0.
\end{equation}

Lets count the number of independent $\Xi$ components, in order to verify that we have enough of them to eliminate $36$ $B^{CD}_{cd}$ in favour of $16$ $\theta^A_a$. Apart from the $4\times4$ traceless conditions~\eqref{eq:traceless-multipliers}, from the contraction with $B$ in the action~\eqref{eq:Poincare-Plebanski} follow $10\times6$ antisymmetrization equations
\begin{equation}\label{eq:Lagrange-multipliers-symmetries}
\Xi^{(A}\lrcorner \, \theta^{B)} \ = \ 0,
\end{equation}
\end{subequations}
which $\Xi$ and $\theta$ have to satisfy. Subtracting from this $16$ d.o.f. corresponding to $\theta$'s (they just serve the purpose to isomorphically map indices $\theta(x):T_x\mc M\rightarrow \mathbb{R}^{3,1}$), we are left with $60-16=44$ additional requirements on $\Xi$. Thereby we get the total number of $96-16-44=36$ independent $\Xi$'s -- exactly the right amount to enforce simplicity.

It should be more or less evident after our exposition that the symmetries of the Lagrange multipliers lead to the variation, constrained by the system~\eqref{eq:linear-simplicity-dual}, depending whether we choose to vary w.r.t. $\Xi$ or its dualized version $\widetilde{\Xi}$. The lemma then implies that this is the same as performing variation on the simplicity constraint surface. The manifest presence of the Hodge-star $\star$ in constraint~\eqref{eq:Lambda-constraints} becomes shrouded, instead one has the restriction on the multipliers~$\Xi$. The free variation of $\de\La$ equates the constraint pre-factor to zero exactly, whilst for~$\de\Xi$ obeying additional conditions -- we get the non-vanishing expression on the r.h.s. In an analogous situation within the  standard Pleba\'{n}ski quadratic approach, the corresponding quantity on the right is usually interpreted in geometric terms as a definition of the 4-volume (on the solution of constraints), whilst a non-trivial symmetrization conditions are put on the l.h.s. It is these latter conditions that actually constitute the substance of the respective `volume' part of simplicity constraints. They require that the definition of the 4-volume be consistent, i.e. does not depend on the multiple choices that could be made for its parametrization on the l.h.s.
Note that in~\eqref{eq:pyramid-simplicity} we get the very same picture, now with the quantity on the r.h.s. being precisely the non-trivial 3-volume (cf.~\eqref{eq:volume-unique}). At the same time this last bit now is \emph{``localized''} at the level of each tetrahedron, irregardless of the whole 4-simplex, which was the case for quadratic version~\ref{volume-quadr}. Lastly, the analogue of the `cross-simplicity', when the r.h.s. is zero, now expresses that the corresponding (discrete) $\theta$ is collinear with the face $S_f$, being orthogonal to its dual bivector $B_f=\star\Si_f$.

These 3 \textit{a priori} distinct choices for constraint imposition, tabulated above, all seem to represent the same physical content. In either of the $\La$ or $\Xi$ versions, inserting further the solution for $B$ back into action, one reduces the initial topological theory~\eqref{eq:BF-Poincare} to that of~\eqref{eq:EC+0torsion}, that is the Einstein-Cartan action~\eqref{eq:Einstein-Cartan} supplemented with an extra term for (zero) torsion. Thus, we expect the equivalence should hold with the Einstein-Hilbert variational principle, through the Palatini constrained variation. The relations between different action principles can be schematically depicted in a diagram:
\begin{displaymath}
\begin{tikzcd}[column sep = large, row sep=large]
\text{1st order:} & S_{\mathrm{PP}}[\mathpzc{B},\varpi,\Xi] \arrow{r}{\de\Xi} \arrow{d}[swap]{\de \be}
& S_{\mathrm{EC}}[\theta,\om] +\int \be\wedge \Theta \arrow{d}{\de \be}\\
\text{2nd order:} & S_{\mathrm{GO}}[B,\theta,\om[\theta],\Xi] \arrow{r}{\de\Xi} & S_{\mathrm{EC}}[\theta,\om[\theta]] \equiv S_{\mathrm{EH}}[\theta]
\end{tikzcd}
\end{displaymath}
In the bottom left corner appears a variant of the ``hybrid'' action of the form dual to that of Gielen-Oriti~\cite{GielenOriti2010Plebanski-linear}, but with the unique $\theta$-compatible torsion-free spin connection. This is to be contrasted with their 1st order formulation, where $\om$ is independent and the gauge status of non-dynamical $\vartheta\sim \theta$ is less obvious, which enters a separate sequence:
\begin{displaymath}
\begin{tikzcd}
S_{\mathrm{PP}}[\mathpzc{B},\varpi,\Xi]+\int \la\wedge \be \arrow{r}{\de\la} & S_{\mathrm{GO}}[B,\theta,\om,\Xi] \arrow{r}{\de\Xi} & S_{\mathrm{EC}}[\theta,\om] .
\end{tikzcd}
\end{displaymath}

We stress that the reduction of the Einstein-Cartan theory to that of GR is achieved only \textit{on-shell} in vacuum, by solving the dynamical e.o.m. for $\om$. In contrast, one puts additional restrictions on the allowed variations of the generalized coordinates by the use of (non-dynamical) Lagrange multipliers $\Xi,\be$, which then acquire the physical meaning of ``reaction forces'', corresponding to variations that violate the constraints. The discussions of the relation between two approaches have been recurrent in the literature in the past, in particular, regarding the higher order Lagrangians and matter couplings (e.g., see~\cite{Kichenassamy1986Lagrangian-multipliers-grav} and references therein). 

\newpage

\thispagestyle{empty}
\chapter{Summary and discussion}

In the present thesis we performed thorough investigation of the role and status of the field theoretic degrees of freedom associated with the metric properties of space-time geometry -- both in the classical Einstein-Cartan theory, and for some directly related models of Quantum Gravity.


\paragraph{The context of the study} is provided by the so-called EPRL-FK spinfoam amplitude, which is the particular implementation of the path-integral quantization of GR. This is closely related to the canonical Loop Quantum Gravity approach, where the theory of Einstein is reformulated in terms of Ashtekar-Barbero connection variables. Both theories are non-perturbative and share the common principles of `background independence', as well as certain techniques and practices, adapted from lattice gauge theory. In particular, the discrete structures are commonly invoked as a regularization tool, in order to make the formalism well-defined.

With the discrete graphs and cell structures there is associated a `discrete geometry', so that one encounters a particular combination of gauge theoretic and metric properties. While the quantization of connections is fairly well-developed, the complementary geometric part is quite a non-trivial subject. There exist various versions of what should be understood under `quantum space-time', but the main guidance is usually provided by the classical GR and its discrete Regge version. The major consistency check for the model thus becomes the ability to reproduce the known physics/geometry in some semi-classical regime of large quantum numbers.

\newpage

\paragraph{The results of our analysis} in the second part of the thesis~\footnote{One could roughly divide it into ``classical'' first part, and ``quantum'' second half, respectively, where one focuses on two different aspects of the subject. They are not disconnected, but rather mutually illuminate each other.} concern the interpretation of the asymptotic formula, and implications it may have on the construction of current spinfoam models. 
\begin{itemize}
\item In Ch.~\ref{ch:problem}, the potential problem with the current EPRL-FK-KKL vertex amplitude is revealed. It is demonstrated that the construction, which works fine for triangulations, does not allow an immediate extension to general cell-complexes. In particular, we show that the (sub-)set of certain `volume simplicity constraints' is not implemented properly in the model. We considered both quadratic and linear versions of the `volume simplicity' in Sec.~\ref{sec:problem-quadratic} and~\ref{sec:problem-linear}, respectively.
This does not allow to associate with the variables at the vertex the geometric picture of a flat 4-dimensional polyhedron (usually implied). One thus could raise doubts on the viability of the generalized amplitude and its interpretation. We address some of them in the next two chapters.
\item In Sec.~\ref{sec:Hopf-link}, we proposed a missing set of constraints, related to the knot-invariant of a Hopf-link. Staying within the usual context of Pleba\'{n}ski quadratic constraints, this can be defined for any bivector geometry associated with the graph. If the latter is induced on the boundary of some flat 4d polyhedron, the derived quantity is directly related with the volume of a polytope. It is shown how the `non-geometricity' problem is resolved in the simplistic case of a hypercuboid, so that it can be uniquely reconstructed, if one imposes an additional condition of invariance on the particular Hopf-link chosen. 
\item One further showed in Sec.~\ref{sec:volume-quant} how this condition may be implemented in the quantum amplitude, in the sense of operator expectation values, leading to corrected asymptotics. The resulting formula bears striking resemblance to earlier proposed ad-hoc recipe to include the cosmological constant.
\item In Ch.~\ref{ch:prop-2}, we revisit the classical continuous theory behind the linear version of simplicity constraints. As suggested by analysis in~\ref{sec:problem-linear} and the Hodge duality isomorphism, one switches from normals (which together with bivectors indirectly characterize edge lengths) directly to tetrad/co-frame variables. The configuration space of the usual BF-theory is extended to that of the Poincar\'{e} BF. The symmetry properties of the classical theory are derived in Sec.~\ref{subsec:Poincare-BF-lagr} from the covariant lagrangian perspective.
\item In Sec.~\ref{subsec:Poincare-BF-hamiltonian}, we perform the corresponding Dirac constrained Hamiltonian analysis. In particular, by constructing explicitly the canonical gauge generator, we show that the full set of 4d symmetries is reproduced in the Hamiltonian picture.
\item In the spirit of Pleba\'{n}ski constrained formulation of gravity, the dual set of linear simplicity constraints for bivector is introduced in Sec.~\ref{sec:continuum-linear}. They bear the meaning of prescribing 3d volume for the polyhedral faces in the boundary of a 4d polytope.

\end{itemize}

\paragraph{The problems with geometry, in our viewpoint,} might be ultimately caused by insufficient understanding of the metric degrees of freedom and the related symmetry w.r.t. diffeomorphisms. Even at the classical level, as mentioned, the dualistic character of the metric/connection picture is apparent in the choice of variables $(\theta,\om)$. Given all the successes of the gauge theoretic viewpoint, it thus desirable to have the corresponding description encompassing also the $\theta$-variables and the associated group of translations. This is generally known under the name of `Poincar\'{e} gauge theories', however, there seems to be no general agreement on the status and role of translations as of yet.

In Sec.~\ref{ch:background}, we developed the classical picture of geometry based on the \textbf{gauge theory of Cartan connections}. This is largely a work with sources, accumulating the known results, and can serve as introduction to the subject. Nevertheless, we may have introduced certain novelties (or brought up lesser known aspects), allowing for a better comprehension of Cartan's original insights and the geometric intuition, accompanying his theory. 

\begin{itemize}
\item In Sec.~\ref{sec:Klein-bundles}, the manifold with the $G$-structure is defined that we called a `(locally) Klein bundle', in order to set up the framework for the affine gauge theory. The action of translations is naturally defined, and we extend the notion of gauge symmetries to the full principal group $G$ of the model homogeneous Klein geometry. The Cartan connection $\mf{g}$-valued 1-form, defining the (infinitesimal) parallelization in Sec.~\ref{sec:Cartan-connections}, thus enjoys the wider use of symmetry transformations.
\item The equation corresponding to the Cartan's notion of `osculation' of geometries is provided in Sec.~\ref{subsec:structure} (along with usual equations of structure). This naturally leads to the key concept of `development' in Sec.~\ref{subsec:parallel-transport} (borrowed from~\cite{Sharpe1997Diff-Geometry-Cartan}), which generalizes the usual holonomies/parallel transport of vectors to include also the translation of points. 
\item In Sec.~\ref{subsec:universal-der} the space of generalized tensors is defined for affine group, taking into account the important separation between notions of `free' and `bound/or sliding' vectors (essentially, attached to the point). 
\item Then the universal covariant derivative is defined, following~\cite{Sharpe1997Diff-Geometry-Cartan}. One of the surprising findings is that due to enlarged symmetry this is essentially a Lie derivative, which can `drag' not only vectors but points themselves. We thus obtain the natural gauge-theoretic realization of local translations. This corresponds precisely to the ``rolling'' of affine spaces along the integral curves of vector fields on the manifold (and hence to the action of diffeomorphisms).
\item The notion of development is used to address the geometric discretization of tensor-forms in Sec.~\ref{sec:geom-sum}, in terms of vector summation. This is viewed as a natural `coarse-graining' operation. The series of results is obtained in Sec.~\ref{subsec:non-closure}, relating non-trivial torsion with `defects' of surfaces and their failure to form the closed boundary.
\end{itemize}

\paragraph{An outlook} for the future research could be the possible discretization and quantization of geometry in terms of Cartan connections, which we prefer over the standard Ehresmann's notion. We briefly discussed the issues related to discretization in Sec.~\ref{subsec:curvature}. The formal similarity with discrete elasticity was pointed out, hence one could use this as a guidance. 

As for the prospective quantization, we touched it upon in Sec.~\ref{subsec:Cartan-quntization}. The development functor conceptually and technically (almost) coincides with the generalized connections of LQG, so that familiar techniques are applicable, in principle. The Einstein tensor in the form of Cartan's `moment-of-rotation' has the appearance of a Pauli-Lubanski vector, which led us to speculate about the spectrum of corresponding operator in covariant theory. Obviously, one encounters certain technical issues (s.a. measure, and non-compactness), as well as conceptual regarding the interpretation, in general. The corresponding canonical/phase-space picture is not at all clear, and warrants a future investigation.


\appendix


\chapter{Synopsis of Ehresmann connections and gauge theory}
\label{app-1}

The notion of the general Ehresmann $G$ connection is fundamental to the Yang-Mills gauge theories of particle interactions (strong and electro-weak forces of the Standard Model). They are formulated such as to give the parallel transport on the bundle $E$ of Def.~\ref{def:fiber-bundle}, and the associated principal bundle $P$ of Def.~\ref{def:PFB}, with abstractly given (vector) fibers and the structure group $G$ (typically, unitary $\mr{SU}(N)$). The bundle may not be trivial product $P \approx\mc{M}\times G$, but due to natural right action of $G$, it admits (canonical) parallelization of the vertical fiber directions. The non-trivial connection then specifies horizontal subbundle of $TP$ through the complement construction (non-canonically).

\begin{definition} The \textbf{geometric definition of (Ehresmann, principal) connection on $P$} consists of:\label{def:connection-Ehresmann-1}
\begin{enumerate}[label={\upshape(\roman*)}, align=left, widest=iii]
\item The horizontal distribution as assignment of vector subspaces $H_pP\subset T_pP$, complementary to $V_pP\equiv\{X\in T_pP|\pi_\ast(X)=0\}$, varying smoothly with $p\in P$, such that \label{prop:connection-distr-1}
\begin{equation}
TP \ = \ VP\oplus HP, \qquad VP \ = \ \ker\pi_\ast;
\end{equation} 
\item Right $G$ invariance is required $R_{g\ast}(H_pP)=H_{pg}P$, for each $g\in G$.\label{prop:connection-distr-2}
\end{enumerate}
\end{definition}

On the principal $G$ bundle $P$, an alternative characterization is available.

\begin{definition} Let $\mf{g}$ be the Lie algebra of $G$. The \textbf{principal connection on $P$ in terms of $\mf{g}$-valued 1-form} $\om:TP\ra \mf{g}$ is specified, if the following key properties hold:\label{def:connection-Ehresmann-2}
\begin{enumerate}[label={\upshape(\roman*)}, align=left, widest=iii]
\item Restricted to fibers, $\om|_{VP}^{\ph{i}}$ provides the linear isomorphism between Lie algebra $\mf{g}$ and the vertical sub-bundle $\ker\pi_{\ast p}\equiv V_pP\subset T_pP$, spanned by the \textbf{fundamental vector fields}:\label{prop:connection-form-1}
\begin{equation}\label{eq:fundamental-vector}
X^\mb{A}_p \ = \ \frac{d}{ds}(p\exp s\mb{A})\bigg|_{s=0} , \qquad \om(X^\mb{A}) \ = \ \mb{A}\in \mf{g};
\end{equation}
\item $R_g^\ast\om =\mr{Ad}(g^{-1})\om$, that is $\om_{pg}(R_{g\ast}X)=\mr{Ad}(g^{-1})\om_p(X)$ for all $g\in G,p\in P, X\in T_pP$. \label{prop:connection-form-2}
\end{enumerate}
\end{definition}

\begin{proposition}
There is a natural 1-to-1 correspondence $HP=\ker \om$ between the two definitions of the principal connection above (cf.~\cite[Th.1.2.4]{Bleecker1981gauge-variational-principles},~\cite[Prop.2.1.1]{KobayashiNomizu1963vol-1}).
\end{proposition}

Moreover, the map $\pi_{\ast p}:T_pP\ra T_{\pi(p)}\mc{M}$ restricts to give an isomorphism between horizontal vectors $X^H_p:=X_p-X^{\mb{A}}_p\in\ker\om_p$ in $P$ and tangent vectors $\bar{X}_{\pi(p)}=\pi_{\ast p}(X)$ in $\mc{M}$, respectively. Hence, every path in the base $\bar{\ga}:I\ra\mc{M}$, $\bar{\ga}'(s)\in T_{\bar{\ga}(s)}\mc{M}$, can always be covered by the unique \textbf{horizontal lift} in the bundle $\ga:I\ra P$, $\ga'(s)\in H_{\ga(s)}P$, such that $\ga^\ast\om=0$ and $\bar{\ga}(s)=\pi(\ga(s))$.

The \emph{finite} parallel transport is then given by the \textbf{holonomy} as the path-dependent $G$-valued solution of the ordinary differential equation:
\begin{equation}\label{eq:holonomy}
h^{-1}_{\bar{\ga}}dh_{\bar{\ga}}^{\ph{1}}(\pa_s) \ = \ \om(\ga_\ast(\pa_s)), \qquad h_{\bar{\ga}}(s_0) \ = \ g_0,
\end{equation}
of the same type as in the main text. This establishes the correspondence between frames $p_i\in P_{x_i}$ and $p_j\in P_{x_j}$ at the endpoints of the path $\bar{\ga}:(I,s_i,s_j)\ra (\mc{M},x_i,x_j)$, as follows: $p_j=p_ih_{ij}$, where $h_{ij}\equiv h_{\bar{\ga}}$ is the solution of~\eqref{eq:holonomy}. The latter depends both on path in $\mc{M}$, as well as the frame's initial configuration at the source point (though in a well-behaved covariant manner).~\footnote{In general, the connection $\om$ results in the groupoid morphisms $h:[\bar{\ga}]\ra G$ from the set of paths in $\mc{M}$ into the gauge group of frame transformations. This is often used for discretization over graphs/plaquettes (e.g. in the lattice gauge theory, or LQG with Ashtekar-Barbero connection on the 3d hypersurface $\mc{S}_3\subset\mc{M}$).}

Regarding the correspondence between PFB and associated fiber bundles, the connection on $P$ determines also the \textbf{general Ehresmann connection}~\cite{Marle2014fromCartan-toEhresmann} on $E=P\times_G^{\ph{1}}F$, namely:
\begin{enumerate}[label={\upshape(\roman*)}, align=left, widest=iii]
\item There is a distribution of horizontal subspaces $H_zE\subset T_zE$, complementary to the vertical sub-bundle $VE\equiv\ker\pi_{E\ast}$, so that: $T_zE = V_zE\oplus H_zE$ for every  $z\in E$; \label{prop:Ehresmann-1}
\item Given any smooth path $\bar{\ga}: (I,s_0,s_1) \ra (\mc{M},x_0,x_1)$ in the base, there exists a unique \emph{horizontal lift} $\hat{\ga}:(I,s_0,s_1) \ra (E,z_0,z_1)$ through any $z_0=\hat{\ga}(s_0)\in E_{x_0}$, covering $\bar{\ga}$, such that: \label{prop:Ehresmann-2}
\begin{equation}\label{eq:hor-lift-E}
\pi_E(\hat{\ga}(s)) \ = \ \bar{\ga}(s), \quad \text{and} \quad \hat{\ga}'(s) \in H_{\hat{\ga}(s)}E \quad \text{for all} \quad s\in[s_1,s_2];
\end{equation}
\item The diffeomorphism map $\hat{\phi}_{s_1s_0}:E_{x_0}\ra E_{x_1}$, given by the horizontal lift $z_0\mapsto\hat{\phi}_{s_1s_0}(z_0)=\hat{\ga}(s_1)$, is called \emph{parallel transport} of fibers along $\bar{\ga}$. The connection on $E$ is compatible with that on $P$ if the parallel transport is provided by $\hat{\phi}_{s_1s_0}=\ga(s_1)\circ\ga(s_0)^{-1}$, where $\ga$ is any horizontal lift in $P$ of a given smooth path in $\mc{M}$. Conversely, a ($G$-compatible) connection in $E$ determines the one in $P$, s.t. the horizontal lifts through any $p_0\in P_{x_0}$ are provided by $\ga(s_1)=\hat{\phi}_{s_1s_0}\circ p_0$.
\end{enumerate}

Closely related to the parallel transport is another notion of `development of a path in the fiber'. Let $z:(I,s_0,s_1)\ra (E,z_0,z_1)$ be \emph{some} path in the associated bundle, with the corresponding projected path $\bar{\ga}(s)=\pi_E(z(s))$ in the base. Let $\hat{\phi}_{s_1s_0}$ be the parallel transport along $\bar{\ga}$ obtained via \emph{horizontal} lift in $E$. The smooth path $\hat{z}(s):=\hat{\phi}_{ss_0}^{-1}\circ z(s)$ \emph{in the fiber} $E_{x_0}$ is called the \textbf{development of the path $z(s)$ in the fiber at $s_0$}. Naturally, this path is constant if and only if $z$ is horizontal~\footnote{Unlike the holonomies, which are routinely used as the basis for discretization in non-perturbative approaches to Quantum Gravity, which favour the connective aspects of GR (like LQG and Spin Foams), this notion of development is rarely the focus point of attention. Indeed, it basically retrieves the same information already contained in holonomies -- at least for the Ehresmann's vertical parallelism in the abstractly defined vector fiber bundles, with some generic $G$. On the other hand, there is more meaning associated with this concept in the framework of affine connections, describing the geometry of spacetime itself. The fibers become affine spaces, while the vectorial part is associated with velocities. The notion of horizontality then implies freely falling (inertial) frames, whereas the developed curves represent some actual stretch in $\mc{M}$ -- that is, correspond to the metric degrees of freedom.} 

In the language of equivariant forms: if the two frames in $P$ at the endpoints of the path $\bar{\ga}:(I,s_i,s_j)\ra (\mc{M},x_i,x_j)$ are related by means of the corresponding holonomy $h_{ij}$ as follows $p_j=p_ih_{ij}$, then the value of the field $\psi$ (see~Defs.~\ref{def:equiv-map}-\ref{def:forms} in the main text) for $p_i$ at the source $x_i=\pi(p_i)$ is obtained from that of $p_j$ at the target $x_j=\pi(p_j)$ by means of the pull-back: 
\begin{equation}\label{eq:fiber-development}
(R^\ast_{h_{ij}^{-1}}\psi)(p_i)=h_{ij}\cdot\psi(p_j).
\end{equation}

The notion of covariant (absolute) derivative is normally specified by the horizontal vector fields, determined by connection (and the base direction). 

\begin{definition}\label{def:covariant-differential}
The \textbf{exterior covariant differential} of $\varphi\in\La^k(P,V)$ (w.r.t. connection $\om$) is $D^{\om}\varphi(X_1,...,X_{k+1})\equiv (d\varphi)^H(X_1,...,X_{k+1}):=d\varphi(X_1^H,...,X_{k+1}^H)$. For the horizontal forms $\varphi\in\bar{\La}^k(P,V)$, we have $D^{\om}\varphi=d\varphi+\om\dot{\wedge}\,\varphi\in\bar{\La}^{k+1}(P,V)$. The form is called \textbf{parallel} (or covariantly constant, w.r.t. connection), if $D^\om\varphi=0$. 
\end{definition} 

\begin{definition}[cf.~\cite{Bleecker1981gauge-variational-principles}, \S 3.2]\label{def:gauge-transform}
An \textbf{automorphism} of a PFB $\pi:P\ra\mc{M}$ is a diffeomorphism $f:P\ra P$ such that $f(pg)=f(p)g$ for all $g\in G,p\in P$. It induces the well-defined diffeomorphism of the base $\bar{f}:\mc{M}\ra\mc{M}$ given by $\bar{f}(\pi(p))=\pi(f(p))$, so that the following diagram commutes:
\begin{displaymath}
\begin{tikzcd}[column sep=small]
P \arrow[d,swap, "\pi "] \arrow[rr, "f"] & & P \arrow[d,"\pi"] \\
\mc{M}  \arrow[rr, "\bar{f}"] & & \mc{M}
\end{tikzcd}.
\end{displaymath}
 A \textbf{gauge transformation} (in this narrower sense) is a vertical automorphism over the identity $\bar{f}=\mr{Id}_\mc{M}$ (i.e. $\pi(p)=\pi(f(p))$). The group of all gauge transformations $\mr{GA}(P)$ is isomorphic to the space of $G$-valued functions, transforming in the (anti-)adjoint representation $C(P,G)=\{\tau:P\ra G|\tau(pg)=g^{-1}\tau(p)g\}$, via the correspondence $f(p)=p\tau(p)$. Accordingly, it is generated by the map $\exp:C(P,\mf{g})\ra\mr{GA}(P)$ as $\exp(\zeta)(p)=p\exp(\zeta(p))$, with an obvious multiplication $[\zeta,\zeta'](p)=[\zeta(p),\zeta'(p)]$ making $C(P,\mf{g})$ into a \textbf{gauge algebra}. 
\end{definition}

If $\om$ is a connection form and $f\in\mr{GA}(P)$ is a gauge transformation, then the pull-back $f^\ast\om$ is also a connection (i.e. satisfying~\ref{prop:connection-form-1}-\ref{prop:connection-form-2}, cf.~\cite[Th.3.2.5]{Bleecker1981gauge-variational-principles}), related to the first one as follows: $(f^\ast\om)_p=(\tau^\ast\om_G)_p+\mr{Ad}(\tau(p)^{-1})\om_p$, i.e. $\tau^\ast\om_G=L_{\tau(p)^{-1}\ast}\tau_{\ast p}^{\ph{1}}$. As a consequence of horizontality, the good behaviour of the forms $\varphi\in\bar{\La}^k(P,V)$ w.r.t. gauge transformations follows: $f^\ast\varphi=\tau^{-1}\cdot\varphi$.

\begin{definition}\label{def:curvature-Ehresmann}
The covariant differential of the connection 1-form $\om\in\La^1(P,\mf{g})$ is generally known as the \textbf{curvature}~\footnote{Important: do not confuse the defining formula~\eqref{eq:cuvature-Ehresmann} for $\Om$ with that of the Def.~\ref{def:covariant-differential} $D^\om\varphi=d\varphi+[\om\wedge\varphi]$ for the $\mf{g}$-valued \emph{horizontal} form $\varphi\in\bar{\La}^k(P,\mf{g})$, transforming in the (anti-)adjoint representation. Though, if for a moment we assume the convention to denote by ``dot'' the multiplication of the \emph{matrix group} and its Lie algebra $\mf{g}$, then from $[\varphi_1\wedge\varphi_2]=\varphi_1\dot{\wedge}\,\varphi_2-(-1)^{k_1k_2}\varphi_2\dot{\wedge}\,\varphi_1$ follows $\frac12[\om\wedge\om]=\om\dot{\wedge}\,\om$ for the connection 1-form, leading to the expression in components in the certain representation $\rho_\ast(\om)=\om^\al\rho_\ast(\mb{E}_\al)$ (taken to be defining).}.:
\begin{equation}\label{eq:cuvature-Ehresmann}
\Om \ :=  \ D^\om\om \ \equiv \ (d\om)^H \ = \ d\om+\frac12[\om\wedge\om] \ \in\bar{\La}^2(P,\mf{g}).
\end{equation}
\end{definition} 

It is horizontal by definition, and transforms in the (anti-)adjoint representation under arbitrary gauge transformations: $f^\ast\Om\equiv R^\ast_{\tau}\Om=\mr{Ad}(\tau^{-1})\Om$, for $\tau\in C(P,G)$. Along with the two properties of Def.~\ref{def:connection-Ehresmann-2}, this characterize $\om$ as ``fiber-wise Maurer-Cartan''. The fundamental vectors~\eqref{eq:fundamental-vector} play the role of the left-invariant vector fields on $G$, and satisfy commutativity $[X^\mb{A},X^\mb{B}]=X^{[\mb{A},\mb{B}]}$. Applying Theorem~\ref{th:G-characterization}, the vertical distribution is completely integrable and fibers are isomorphic to $G$ (which is postulated a priori for PFB). The curvature of the Ehresmann's connection $\Om\neq0$ provides the measure of the non-involutivity of horizontal distribution.

The well-known 2nd Bianchi identities are satisfied:
\begin{equation}\label{eq:Bianchi-Ehr}
d\Om \ = \ [\Om\wedge\om] \qquad \Leftrightarrow \qquad D^\om\Om \ = \ 0,
\end{equation}
being inherently linked to the gauge symmetry. Their geometric meaning is revealed in~\eqref{eq:Bianchi-2-meaning} .



\clearpage

\bibliography{LQG-SpinFoams-bibl}

\newpage
\thispagestyle{empty}
\mbox{}

\thispagestyle{empty}
\clearpage


\newpage
\noindent \textbf{\Large This thesis is based on the publications:} 

\vspace{1.0cm}
 
\begin{itemize}
\item
    V.~Belov, {\it {Poincar\'{e}-Pleba\'{n}ski formulation of GR and dual simplicity constraints}}, Class. Quant. Grav. 35 (2018) 235007,  \href{http://arxiv.org/abs/1708.03182}{{\tt arXiv:1708.03182}}.	\\
	
\item
	B.~Bahr, V.~Belov, {\it {Volume simplicity constraint in the Engle-Livine-Pereira-Rovelli spin foam model}},  Phys. Rev. D97 (2018) 086009, \href{http://arxiv.org/abs/1710.06195}{{\tt arXiv:1710.06195}}.
\end{itemize}

\thispagestyle{empty}

\newpage
\thispagestyle{empty}
\chapter*{Declaration\markboth{Declaration}{Declaration}}
\thispagestyle{empty}
\noindent {\Large Eidesstattliche Versicherung / Declaration on oath }\\

\noindent Hiermit versichere ich an Eides statt, die vorliegende Dissertationsschrift selbst verfasst und keine anderen als die angegebenen Hilfsmittel und Quellen benutzt zu haben.\\

\noindent Die eingereichte schriftliche Fassung entspricht der auf dem elektronischen Speichermedium.\\

\noindent Die Dissertation wurde in der vorgelegten oder einer {\"a}hnlichen Form nicht schon einmal in einem fr{\"u}heren Promotionsverfahren angenommen oder als ungen{\"u}gend beurteilt.\\

Hamburg, den M\"{a}rz 14th 2019

\vspace*{1.5cm}

\rule{0.3\textwidth}{.5pt}\\

Vadim Belov

\end{document}